\documentclass[11pt,reqno]{amsart}
\usepackage{fullpage,graphicx,subfigure,mathpazo,color}
\usepackage{amsmath,amscd,tikz,mathrsfs}
\usepackage[normalem]{ulem}
\usepackage{extarrows}
\usepackage{setspace}
\usepackage[colorlinks,urlcolor=blue,citecolor=blue]{hyperref}
\usepackage{cleveref}
\usepackage{mathtools}
\usepackage{enumitem}

\newcommand{\ii}{\mathrm{i}}
\newcommand{\sn}{\mathrm{sn}}
\newcommand{\cn}{\mathrm{cn}}
\newcommand{\dn}{\mathrm{dn}}

\newcommand{\ee}{\mathrm{e}}
\newcommand{\dd}{\mathrm{d}}
\newcommand{\diag}{\mathrm{diag}}
\newcommand{\off}{\mathrm{off}}
\newcommand{\nc}{\mathrm{nc}}
\newcommand{\tn}{\mathrm{tn}}
\newcommand{\dc}{\mathrm{dc}}
\newcommand{\cs}{\mathrm{cs}}
\newcommand{\ns}{\mathrm{ns}}
\newcommand{\ds}{\mathrm{ds}}
\newcommand{\cd}{\mathrm{cd}}
\newcommand{\sd}{\mathrm{sd}}
\newcommand{\nd}{\mathrm{nd}}

\newtheorem{theorem}{Theorem}
\newtheorem{lemma}{Lemma}
\newtheorem{prop}{Proposition}

\newtheorem{define}{Definition}

\newtheorem{remark}{Remark}

\numberwithin{equation}{section}

\usepackage{graphicx}      
\usepackage{titlesec}

\titleformat{\section}{\centering\LARGE\bfseries}{\thesection}{1em}{}
\titleformat{\subsection}{\Large\bfseries}{\thesubsection}{1em}{}
\begin{document}

\title{Focusing mKdV equation: Two-phase solutions and their stability analysis}

\author{Liming Ling}
\address{School of Mathematics, South China University of Technology, Guangzhou, China, 510641}
\email{linglm@scut.edu.cn}

\author{Xuan Sun}
\address{School of Mathematics and Statistics, Donghua University, Shanghai, China, 201620}
\email{sunx@dhu.edu.cn}

\begin{abstract} 
	In this work, we primarily focus on the two-phase solutions and their stability to the focusing mKdV equation.
	By employing the algebro-geometric approach in combination with an effective integration method, we construct explicit two-phase solutions and their corresponding wave-functions expressed in terms of the Riemann theta function. 
	The spectral stability of two-phase solutions is examined via a modified squared-eigenfunction approach, and their stability with respect to subharmonic perturbations is further analyzed under spectrally unstable conditions.
	In addition, the orbital stability of the two-phase solutions is investigated.
	To the best of our knowledge, this study provides the first rigorous stability theory for the two-phase solutions of the focusing mKdV equation.
	
\end{abstract}

\date{\today}

\maketitle

\section{Introduction}

In this work, we primarily investigate genus-two periodic traveling wave solutions (alias for two-phase solutions) and their stability for the focusing modified Korteweg-de Vries (mKdV) equation:
\begin{equation}\label{eq:mKdV}\tag{mKdV}
	u_t+u_{xxx}+6u^2u_{x}=0,\qquad 
	u\equiv u(x,t)\in \mathbb{R}, \quad 
	(x,t)\in \mathbb{R}^2,
\end{equation} 
which arises in various physical contexts, including water waves \cite{AblowitzS-81,Schamel-1973-modified} and plasma physics\cite{AblowitzC-91}.
As a completely integrable model, the \ref{eq:mKdV} equation can be derived from the third positive flow of the Ablowitz-Kaup-Newell-Segur (AKNS) system via two symmetry reductions \cite{AblowitzKNS-1974-inverse}.
It admits the Lax pair \cite{Lax-68-INE}, possesses a bi-Hamiltonian structure \cite{MaF-96-biHamiltonian}, and supports an infinite hierarchy of conserved quantities $\mathcal{H}_i$, $i=0,1,2,\cdots$ \cite{Dickey-2003-soliton,MiuraGK-68-mKdV-conservation}. 
The associated Lax pair is given by:
\begin{equation}\label{eq:Lax-pair} 
	\Phi_x(x,t;\lambda)=\mathbf{U}(\lambda;\mathbf{Q})\Phi(x,t;\lambda), \qquad \Phi_t(x,t;\lambda)=\mathbf{V}(\lambda;\mathbf{Q})\Phi(x,t;\lambda), 
\end{equation}
where the spectral parameter $\lambda\in\mathbb{C}\cup\{\infty\}$,
\begin{equation}\label{eq:Lax-U-V-define}
	\begin{split}
		&\mathbf{U}(\lambda; \mathbf{Q})
		=-\ii \lambda \sigma_3+\mathbf{Q}, \qquad \mathbf{Q}=
		\begin{bmatrix}
			0 &   u \\  -u & 0 
		\end{bmatrix}, \qquad \sigma_3=\begin{bmatrix}
		1 &   0 \\  0 & -1 
		\end{bmatrix},\\
		& \mathbf{V}(\lambda; \mathbf{Q})=4\lambda^2\mathbf{U}(\lambda; \mathbf{Q})
		+2\ii\lambda  \sigma_3(\mathbf{Q}_x-\mathbf{Q}^2)-(\mathbf{Q}_{xx}-2 \mathbf{Q}^3). 
	\end{split}
\end{equation}
It can be readily verified that both matrices $\mathbf{U}(\lambda;\mathbf{Q})$ and $\mathbf{V}(\lambda;\mathbf{Q})$ satisfy the following symmetries:
\begin{equation}\label{eq:symmety-U-V}
	\begin{split}
		\mathbf{U}^{\dagger}(\lambda^*;\mathbf{Q})=-\mathbf{U}(\lambda;\mathbf{Q}),\qquad 
		\mathbf{U}^{\top}(-\lambda;\mathbf{Q})=-\mathbf{U}(\lambda;\mathbf{Q}).
	\end{split}
\end{equation}
The compatibility condition of the linear system \eqref{eq:Lax-pair}, expressed as $\Phi_{tx}(x,t;\lambda)=\Phi_{xt}(x,t;\lambda)$, is equivalent to the zero-curvature equation $\mathbf{U}_t(\lambda;\mathbf{Q})-\mathbf{V}_x(\lambda;\mathbf{Q})+\left[\textbf{U}(\lambda;\mathbf{Q}),\textbf{V}(\lambda;\mathbf{Q})\right]=0$, which leads to the \ref{eq:mKdV} equation.

As a classical integrable equation, the \ref{eq:mKdV} equation admits a rich variety of solutions \cite{AblowitzKNS-1974-inverse,BilmanM-19,GardnerGKM-67,FritzH-2003-Algebro-Geometric,ZakharovS-72}, including solitons, breathers, elliptic function solutions.
Over the years, the study of elementary function solutions, such as solitons and breathers, has become relatively well developed.
More recently, researchers have increasingly focused on periodic solutions, in particular elliptic function solutions, which arise in applications such as photonic crystal fibers and nonlinear metamaterials \cite{LapineSK-2014,Russell-2003}.
The finite-gap solutions for the \ref{eq:mKdV} equation, together with their dynamics and stability properties, also have long attracted significant interest. 
Studying the stability of these solutions provides valuable insight into the underlying structure and dynamics of the equation.

In this work, continuing our previous work on investigated the subharmonic stability of $\cn$- and $\dn$-type solutions \cite{LingS-23-mKdV-stability}, we investigate the subharmonic stability of genus-two periodic traveling wave solutions  of the focusing \ref{eq:mKdV} equation, constructed from three pairs of branch points (also known as two-gap or two-phase solutions).
To the best of our knowledge, this problem has not been previously addressed, due to the non-self-adjoint of Lax operator \cite{Deconinck-10}. 
In the following, we review relevant studies.

\subsection{Review on the stability analysis of finite-gap solutions}

We provide a brief overview of earlier research on finite-gap solutions and their stability. 
These solutions have a long-standing role in integrable systems, providing explicit representations of nonlinear evolution equations. 
Understanding their dynamical properties and long-term behavior relies critically on stability analyses.

\vspace{0.5cm}
\noindent 
\textbf{\Large Algebro-geometric solutions}
\vspace{0.2cm}

The finite-zone theory provides solutions of nonlinear integrable equations, initially applied to the nonlinear Schr\"odinger (NLS) and sine-Gordon (SG) equations via Abelian varieties \cite{DubrovinMN-1976}.
This approach was subsequently extended to the treatment of periodic  problems for nonlinear systems \cite{Dubrovin-1977}, and algebro-geometric Poisson brackets were formulated for real finite-zone solutions of the Korteweg–de Vries (KdV) equation \cite{VeselovN-1982}.
Real finite-zone solutions were then uniformly represented in
terms of Riemann theta functions \cite{DubrovinN-1982},
and the spectral properties of matrix finite-zone operators were linked to algebraic curves through theta functions \cite{Dubrovin-1983}.
Hyperelliptic quasi-periodic (g-gap) solutions of the NLS equation were also constructed using Riemann theta functions \cite{PreviatoE-1985-NLS}. 
In addition, algebro-topological methods were introduced to effectively classify real finite-zone solutions of the SG equation \cite{Novikov-1985}.
More advanced techniques were later developed to systematically analyze finite-gap solutions of integrable equations, including the KdV equation \cite{BelokolosBME-1986} and the Schr\"{o}dinger operator \cite{BelokolosE-1994}.

In recent decades, the construction of finite-gap solutions for integrable equations experienced continuous development, leading to increasingly diverse analytical representations.
Among the various approaches, algebro-geometric methods played a central role in generating finite-gap solutions of integrable nonlinear equations.
Through the development of this methodology, finite-gap solutions were constructed for a wide range of equations, including the SG equation \cite{GrinevichN-2003}, the Camassa-Holm hierarchy \cite{Qiao-03}, the vortex filament equation \cite{CaliniI-05-VFE}, the coupled mKdV hierarchy \cite{GengZD-2014}, the three-wave interaction system \cite{HeGW-2014}, the NLS equation \cite{Wright-16,Wright-19,Wright-2017}, all of which could be expressed in terms of Riemann theta functions.
Reduction theory of theta functions further formalized the construction of algebro-geometric solutions for nonlinear integrable systems \cite{BelokolosE-2001,BelokolosE-2002}.
Finite-gap solutions of the NLS equation were also analyzed via the Riemann–Hilbert method \cite{BertolaT-2017}.
In addition to theta-functional solutions, explicit one-gap and two-gap solutions, as well as localized waves on periodic traveling wave backgrounds, were derived for both the focusing \cite{ChenPW-19,Chen-19-mKdV} and defocusing mKdV equations \cite{ArrudaP-2025} in terms of Jacobi elliptic functions. 
Based on Jacobi theta functions, the elliptic-localized wave solutions of the NLS equation \cite{FengLT-20} and the SG equation \cite{LingS-22-SG} were obtained, and higher-order rational elliptic rogue wave solutions of the integrable nonlinear soliton equations \cite{LingS-24-RW} were constructed.
Moreover, by means of the Miura transformation, finite-gap solutions of the KdV equation and the mKdV equation \cite{SmirnovA-24} were also generated.

Building on the various forms of solutions discussed above, significant progress has also been made in the study of their spectral and orbital stability.  
In what follows, we review relevant studies, particularly those focusing on the spectral and orbital stability periodic traveling wave solutions  to integrable equations.

\vspace{0.5cm}
\noindent 
\textbf{\Large The stability analysis}
\vspace{0.2cm}

As early as the 20th century, researchers such as Benjamin \cite{Benjamin-72},
Bona \cite{Bona-75,BonaSS-87}, Grillakis, Shatah, Strauss \cite{GrillakisSS-87,GrillakisSS-90}, and Weinstein \cite{Weinstein-85,Weinstein-86}, made significant contributions to the analysis of solitary wave stability. 
%
Kapitula, Kevrekidis, and Sandstede \cite{KapitulaKS-04} proposed a method for studying the spectral stability of nonlinear waves using the Krein signature.
Additionally, H\v{a}r\v{a}gus and Kapitula \cite{HaragusK-08}, employing the Floquet–Bloch decomposition, 
to establish relationships among the operators $\mathcal{L}$, $\mathcal{JL}$, and the eigenvalue $\Omega$.
Alejo and Mu\~{n}oz \cite{Alejo-2013-nonlinear} analyzed the nonlinear stability of breather solutions by introducing a new Lyapunov functional and utilizing the higher-order conservation laws to characterize the dynamics of small perturbations.
Building on this work, Semenov \cite{Semenov-2022-orbital} investigated the orbital stability of multi-soliton/breather solutions of the mKdV equation by modifying the Lyapunov functional.
Recently, the spectral stability of non-degenerate vector soliton solutions
and the nonlinear stability of breather solutions for the coupled nonlinear Schr\"{o}dinger equation have been studied based on the integrability structure and the Lyapunov method \cite{LingPS-2024}.

Building on these foundational methods, various approaches to the subharmonic stability analysis of finite-gap solutions have been developed in recent years.
Pava \cite{Pava-07} established the orbital stability of $\dn$-type solutions for both the mKdV equation and the NLS equation. 
Additionally, Gallay and H\u{a}r\u{a}gu\c{s} \cite{KapitulaH-07-small,KapitulaH-07-per} studied the spectral stability of periodic solutions for the NLS equation under co-periodic perturbations. Deconinck and Kapitula \cite{DeconinckK-10} studied the orbital stability of cnoidal waves of the KdV equation under the subharmonic perturbations.
Bottman, Deconinck, and Nivala \cite{BottmanDN-11} examined both the spectral and orbital stability of elliptic function solutions for the defocusing NLS equation.
Deconinck and Segal \cite{DeconinckS-17} demonstrated that $\dn$-type solutions of the focusing NLS equation are spectrally stable with respect to co-periodic perturbations, while $\cn$-type solutions exhibit spectral stability under subharmonic perturbations, employing special functions such as the Weierstrass $\wp$ function, the $\zeta$ function.  
Continuing this work, Deconinck and Upsal \cite{DeconinckyU-20} further explored the orbital stability of elliptic function solutions, including $\cn$-type, $\dn$-type, and elliptic function solutions with nontrivial phase, by constructing a novel Lyapunov function based on higher-order conserved quantities.

Deconinck and Nivala \cite{Deconinck-10} showed that the periodic traveling wave solutions for the defocusing mKdV equation are spectrally stable and pointed out that in the focusing case, the spectral parameter is no longer confined to the real axis, which poses additional challenges for stability analysis.
To overcome this difficulty, authors \cite{LingS-23-mKdV-stability} used the theta function theory to solve the subharmonic stability analysis partially, i.e., the spectral and orbital stability of $\cn$-type and $\dn$-type solutions under the subharmonic perturbations. Building on this work, we will analyze the stability of genus-two periodic traveling wave solutions  for the focusing mKdV equation. 
Compared with the genus-one case, the new difficulties are from deriving explicit and simplified forms of genus-two periodic traveling wave solutions  and the corresponding wave functions of the Lax pair, which are crucial for facilitating the subsequent stability analysis. 
The genus-two algebro-geometric solutions are associated with the two-dimensional Riemann theta function. For the general periodic waves of mKdV equation, these solutions can be represented by the elliptic functions. Notably, the two-dimensional Riemann theta functions are not directly related to elliptic functions; however, under a specific symmetric condition, they can be reduced to a product of two Jacobi theta functions. The expressions of Riemann theta functions also involve certain hyperelliptic integrals, which, under the symmetry of mKdV equation, can be transformed into elliptic integrals through a variable transformation. These elliptic integrals arise from the Abel maps, Abelian integrals and other integral constants. We will develop a systematic way to tackle these integrals in a uniform framework. By systematically combining these two steps, the Riemann theta function solutions can be reduced to the elliptic functions. Furthermore, the wave functions of the corresponding Lax pair will be derived explicitly.
Building on these explicit representations, we aim to investigate the subharmonic spectral stability of the solutions.
Furthermore, we explore the existence of an appropriate functional framework in which the genus-two periodic traveling wave solutions  of the focusing mKdV equation exhibit orbital stability.

\subsection{Main results}

As is well known, the classical AKNS system \cite{AblowitzKNS-1974-inverse} generates an infinite hierarchy of integrable nonlinear soliton equations.
Within this hierarchy, the NLS equation and the mKdV equation correspond to positive power flows, while the SG equation is derived from the negative power flow of this system. 
The mKdV equation possesses an infinite number of conserved quantities
	\begin{equation}\label{eq:H0}
		\!\!\!
		\mathcal{H}_1=\frac{1}{2}\int_{-PT}^{PT} u^2 \dd  x  , \,\,\,
		\mathcal{H}_3=\frac{1}{2}\int_{-PT}^{PT}\left(u_x^2-u^4\right) \dd  x ,  \,\,\,
		\mathcal{H}_5= \frac{1}{2}\int_{-PT}^{PT} \left( u_{xx}^2-10u^2u_{x}^2+2u^6 \right)\dd x,  \,\,\, \cdots
		\!\!\!
	\end{equation} 
where the period of the function $u$ is $2PT$. 
The Hamiltonian flows in the mKdV hierarchy are given by $u_{t_n}=\partial_x\mathcal{H}'_{2n+1}(u)$, $i=0,1,\cdots$, where the prime denotes the variational derivative of the Hamiltonian $\mathcal{H}_n$ with respect to $u$. 
In particular, $n=0$ corresponds to the equation $u_{t_0}=u_x$;  $n=1$ yields the mKdV equation; and $n=2$ produces the fifth-order mKdV equation.
Introducing the time variables $\eta_n$, the equation can be expressed in a moving coordinate form $(\xi,\eta_n)$ as
\begin{equation}\label{eq:H-mKdV}
	u_{\eta_n}=\mathcal{J}\hat{\mathcal{H}}'_n(u),\qquad \mathcal{J}=\partial_{\xi}, \qquad
	\hat{\mathcal{H}}_n:=\mathcal{H}_{2n+1}+\sum_{i=0}^{n-1}c_{n,i}\mathcal{H}_{2i+1}, \qquad \hat{\mathcal{H}}_0:=\mathcal{H}_0,
\end{equation}
where $c_{n,i}\in \mathbb{R},i=0,1,...,n-1$. 
The stationary solution of the $n$-th mKdV equation satisfies the ordinary differential equation $\mathcal{J}\hat{\mathcal{H}}'_n(u)=0$. 
Furthermore, by introducing the recursion operator $\mathcal{F}$ defined as 
$\mathcal{F}:=-(\partial_x^2+4u^2 - 4u \partial_x^{-1} u_x)$, the Hamiltonians satisfy $\mathcal{H}^{\prime}_{2n+1}=\mathcal{FH}_{2n-1}^{\prime}$, $n=1,2,\cdots$.
For ease of expression, we introduce a vector $\mathbf{t}=\left(\cdots,t_{-2}, t_{-1},t_0, t_1, t_2 ,\cdots\right)\in \mathbb{R}^{\infty}$, where $t_i$, $i=1,2,\cdots$, correspond to positive power flows and $t_{-i}$, $i=1,2,\cdots$, correspond to negative power flows. 
Considering the positive power flow, we set the wave function $\Phi(x,\mathbf{t};\lambda)$ as 
\begin{equation}\label{eq:Phi-m}
	\Phi(x,\mathbf{t};\lambda)=m(x,\mathbf{t};\lambda)\exp\left(-\ii\lambda \sigma_3 \left(x+\sum_{n=1}^{\infty}\lambda^n t_n\right) \right),
\end{equation}
where the $2\times 2$ matrix function $\Phi(x,\mathbf{t};\lambda)$ is called the wave function and $m(x,\mathbf{t};\lambda)$ is a meromorphic matrix function in $\mathbb{C}\setminus \Gamma$ smoothly depending on $x$ and $\mathbf{t}$ \cite{BealsC-84,Terng-97}. 
By considering the third positive flow and setting $t=4t_3$, we obtain the Lax pair \eqref{eq:Lax-pair} and the mKdV equation. 
Defining $\Psi(x,\mathbf{t};\lambda):=m(x,\mathbf{t};\lambda)\sigma_3 m^{-1}(x,\mathbf{t};\lambda)$,
it can be readily verified that $\Psi(x,\mathbf{t};\lambda)$ satisfies the zero-curvature equations:
\begin{equation}\label{eq:zero-curve-equation}
	\Psi_x(x,\mathbf{t};\lambda)=\left[\mathbf{U}(\lambda;\mathbf{Q}), \Psi(x,\mathbf{t};\lambda)\right],\qquad
	\Psi_{t}(x,\mathbf{t};\lambda)=\left[\mathbf{V}(\lambda;\mathbf{Q}), \Psi(x,\mathbf{t};\lambda)\right],
\end{equation}
where the matrices $\mathbf{U}(\lambda;\mathbf{Q})$ and $\mathbf{V}(\lambda;\mathbf{Q})$ are defined in \eqref{eq:Lax-U-V-define}.
Detailed derivations and definitions are presented in \Cref{sec:Preliminaries}.

\begin{prop}\label{prop:L-matrix}
	Define the matrix function $\mathbf{L}(\lambda)=\mathbf{L}(x,\mathbf{t};\lambda)$ as
	\begin{equation}\label{eq:L-matrix-genus-two}
		\mathbf{L}(\lambda)=-\ii \sum_{i=0}^{g+1}\left(\alpha_i\lambda^i\Psi(x,\mathbf{t};\lambda)\right)_{+}=-\ii \sum_{i=0}^{g+1}\sum_{j=i}^{g+1}\alpha_j \Psi_{j-i}\lambda^{i},
	\end{equation}
	with the matrix functions $\Psi_i(x,t)$ defined in equation \eqref{eq:Theta-expand-lambda-infty}.
	If the matrix function $\mathbf{L}(\lambda)$ satisfies the stationary zero-curvature equations \eqref{eq:zero-curve-equation},
	then it imposes an additional constraint ordinary differential equation
	\begin{equation}\label{eq:ode}
		\sum_{i=1}^{g+1}\alpha_i\Psi_{i+1}^{\rm off}+\ii\alpha_0\mathbf{Q}=0.
	\end{equation}
	Here, $\Psi_i^{\rm off}$ denotes the off-diagonal part of $\Psi_i$.
\end{prop}

The proof of \Cref{prop:L-matrix} is presented in \Cref{sec:Preliminaries}. 
For the two-phase solution of the mKdV equation (i.e., setting $g=2$ in equation \eqref{eq:L-matrix-genus-two}), we impose the constraints $\alpha_2=\alpha_0=0$. 
Under these conditions, the constraint ordinary differential equation \eqref{eq:ode} reduce to $\alpha_3(u_{xxx}+6u^2u_x)-4\alpha_1 u_x=0$, which further implies $\alpha_3 u_t+4\alpha_1 u_x=0$. 
In other words, $u(x,t)$ must be a traveling solution.
Indeed, \Cref{prop:L-matrix} corresponds precisely to equation \eqref{eq:H-mKdV}.

Without loss of generality, setting $\alpha_3=1$, the expression of the matrix function \eqref{eq:L-matrix-genus-two} can be rewritten as $\mathbf{L}(\lambda):=\mathbf{V}(\lambda;\mathbf{Q})/4+\alpha_1 \mathbf{U}(\lambda;\mathbf{Q})$, where the $(i,j)$-elements of the matrix function $\mathbf{L}(\lambda)$ are given by
\begin{equation}\label{eq:L-elements}
	\begin{split}
		\mathbf{L}_{11}(\lambda)=&\  -\ii \lambda^3 -\ii\lambda (\alpha_1-u^2/2), \qquad \qquad \mathbf{L}_{22}(\lambda)=-\mathbf{L}_{11}(\lambda), \\
		\mathbf{L}_{12}(\lambda)=&\ u\left(\lambda^2+\ii u_x \lambda /(2u) -u_{xx}/(4u)+\alpha_1-u^2/2\right)= u(\lambda-\mu_1)(\lambda-\mu_2),\\
		\mathbf{L}_{21}(\lambda)=&\ -u\left(\lambda^2-\ii u_x \lambda /(2u) -u_{xx}/(4u)+\alpha_1-u^2/2\right)= -u(\lambda-\mu_1^*)(\lambda-\mu_2^*),
	\end{split}
\end{equation}
 and functions $\mu_1$, $\mu_2$, $\mu_1^*$, $\mu_2^*$ are expressed by $u$, $u_x$, and $u_{xx}$.
Based on the symmetric properties of matrices $\mathbf{U}(\lambda;\mathbf{Q})$ and $\mathbf{V}(\lambda;\mathbf{Q})$ as given in equation \eqref{eq:symmety-U-V}, one can readily verify that if matrix $\mathbf{L}(\lambda)$ is the solution of the stationary zero-curvature equation \eqref{eq:zero-curve-equation}, then transformed matrices $\mathbf{L}^{\dagger}(\lambda^*)$ and $\mathbf{L}^{\top}(-\lambda)$ also satisfy this equation.
If the spectral parameter $\lambda=\lambda_i$ is a root of $\det(\mathbf{L}(\lambda))=0$, then $\lambda=\lambda_i^*$ and $\lambda=-\lambda_i$ must also be roots.
Consequently, the determinant of the matrix function $\mathbf{L}(\lambda)$ can be expressed as
\begin{equation}\label{eq:det-L-lambda}
	\det\left(\mathbf{L}(\lambda)\right)
	=\prod_{i=1}^{3}(\lambda-\lambda_i)(\lambda-\lambda_i^*), \qquad i=1,2,3,
\end{equation}
where $\lambda_i,\lambda_i^* \in \mathbb{C}\backslash \mathbb{R}$, $i=1,2,3$, are the six roots of the equation $\det\left(\mathbf{L}(\lambda)\right)=0$.
These roots can be categorized into the following two cases:
\begin{itemize}
	\item[\textbf{Case 1:}] All of the above roots are purely imaginary numbers:
	\begin{equation}\label{case1}\tag{Case 1}
	\lambda_i=-\lambda_i^*\in \ii \mathbb{R},\,\,\,\, i=1,2,3.
	\end{equation}
	\item[\textbf{Case 2:}] Two roots are purely imaginary, while the remaining four are complex numbers:
	\begin{equation}\label{case2}\tag{Case 2} 
	\lambda_2=-\lambda_2^*\in \ii \mathbb{R}\quad\text{and}\quad \lambda_i,\lambda_i^*\in \mathbb{C}\backslash(\mathbb{R}\cup \ii \mathbb{R}), \,\,\,\, i=1,3 \quad\text{with}\quad \lambda_1=-\lambda_3^*.
	\end{equation}
\end{itemize} 
Without loss of generality, let $\pm \ii y$ be two eigenvalues of the matrix function $\mathbf{L}(\lambda)$ defined in equation \eqref{eq:L-matrix-genus-two}, which implies $\det\left(\pm \ii y-\mathbf{L}(\lambda)\right)=0$. 
We then define the associate algebraic curve as:
\begin{equation}\label{eq:define-curve-algebro}
	y^2=\prod_{i=1}^{3}\left(\lambda-\lambda_i\right)\left(\lambda-\lambda_i^*\right)\xlongequal{\text{\ref{case1} or \ref{case2}}}\prod_{i=1}^{3}\left(\lambda^2-\lambda_i^2\right).
\end{equation} 
The compact Riemann surface $\mathcal{R}_2$ of genus-two can be described by $\mathcal{R}_2:=\{ (\lambda,y) | y^2=\prod_{i=1}^{3}(\lambda^2-\lambda_i^2)\}$,
with the standard projection $\boldsymbol{\pi}$: $\mathcal{R}_2\rightarrow \mathbb{CP}^1$ defined by
\begin{equation}\label{eq:define-p-y}
	\boldsymbol{\pi}(P)=\lambda,\qquad P=(\lambda,y),
\end{equation}
so that $\mathcal{R}_2$ forms a two-sheeted covering of $\mathbb{CP}^1$.
There are exactly two points $\infty^{\pm}\in \mathcal{R}_2$ such that $\boldsymbol{\pi} (\infty^{\pm})=\infty \in \mathbb{CP}^1$, with the local behavior
\begin{equation}\label{eq:define-standard-projection}
	P\rightarrow \infty ^{\pm} \quad \Leftrightarrow \quad  \lambda\rightarrow \infty, \quad y \rightarrow \pm 
	\lambda^{3}.
\end{equation}
Based on this construction, we proceed to develop the algebro-geometric approach to obtain the explicit solutions of the \ref{eq:mKdV} equation in terms of the Riemann theta function.

Define the divisor $\mathcal{D}$ on $\mathcal{R}_2$ as a map $\mathcal{D}: \mathcal{R}_2 \rightarrow \mathbb{Z}$, where $\mathcal{D}(P)\neq 0$ for only finitely many points $P\in \mathcal{R}_2$. 
The periodic lattice $L_2(\mathcal{R}_2)\subset \mathbb{C}^2$ is defined by $L_2(\mathcal{R}_2)=\{\mathbf{z}\in \mathbb{C}^2 \left| \,\, \mathbf{z}=2\pi \ii (\mathbf{n}+\mathbf{Bm}),\right.$ $\mathbf{n}, \mathbf{m}\in \mathbb{Z}^2\}$, where $\mathbf{B}$ is a periodic matrix of $\mathcal{R}_2$.
The Jacobi variety $J(\mathcal{R}_2)$ of $\mathcal{R}_2$ is then given by $J(\mathcal{R}_2)=\mathbb{C}^2/L_2(\mathcal{R}_2)$.
 The Abel maps are defined by 
\begin{equation}\label{eq:abel_map}
	\mathcal{A}_{P_0}:\mathcal{R}_2 \rightarrow J(\mathcal{R}_2), \quad
	P \mapsto \mathcal{A}_{P_0}(P)=\left(\mathcal{A}_{P_0,1}(P), \mathcal{A}_{P_0,2}(P) \right)=\left(\int_{P_0}^{P} w_1\, \dd \lambda, \int_{P_0}^{P} w_2\, \dd \lambda\right),  	
\end{equation}
and $\alpha_{P_0}:\mathrm{Div}(\mathcal{R}_2)\rightarrow J(\mathcal{R}_2)$,
$\mathcal{D} \mapsto \alpha_{P_0}(\mathcal{D})=\sum_{P\in \mathcal{R}_2} \mathcal{D}(P)\mathcal{A}_{P_0}(P)$,
where $P_0\in \mathcal{R}_2$ is a fixed base point.
 For convenience, the same path is chosen from $P_0$ to $P$ for $j=1,2$.
Considering the algebraic curve defined in equation \eqref{eq:define-curve-algebro},
the differential $w_i \ \dd \lambda$ form a basis in the space of holomorphic $1$-form defined on $\mathcal{R}_2$ satisfying 
\begin{equation}\label{eq:define-basis-B}
	\oint_{a_j}w_i \ \dd \lambda=2\pi\ii \delta_{ij}, \quad 
	\oint_{b_j}w_i \ \dd \lambda=\mathbf{B}_{ij}, \quad w_i=\sum_{k=0}^{1}d_{ik}\lambda^{k}y^{-1}, \quad  d_{ik}\in \mathbb{C}, \quad 
	\delta_{ij}=
	\left\{ \begin{aligned}
		1, \quad i=j, \\
		0, \quad i\neq j,
	\end{aligned} \right.
\end{equation}
where $a_1$, $a_2$, $b_1$, $b_2$, are the homology basis for $\mathcal{R}_2$, 
such that $a_i \circ a_j=0$, $b_i \circ b_j=0$, and $a_i\circ b_j=\delta_{ij}=-b_j \circ a_i$.
By the above two cases with respect to branch points (\ref{case1} and \ref{case2}), 
we would like to define the basis of the above algebraic curves in two cases (shown in \Cref{fig:genus-two-figure}).

\begin{figure}[h]
	\centering
	\subfigure[\ref{case1}]{\includegraphics[width=0.38\textwidth]{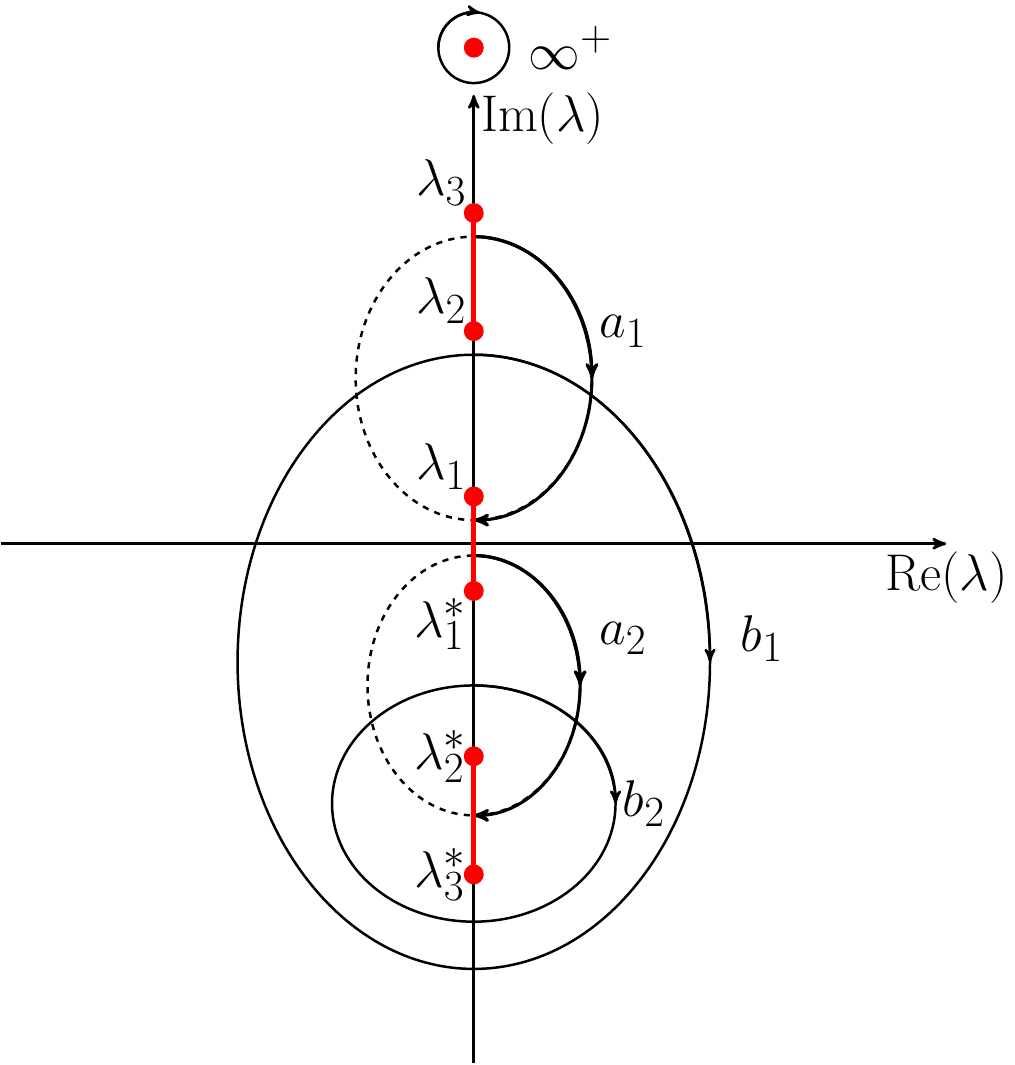}\label{fig:genus-two-figure-p}}
	\subfigure[\ref{case2}]{\includegraphics[width=0.38\textwidth]{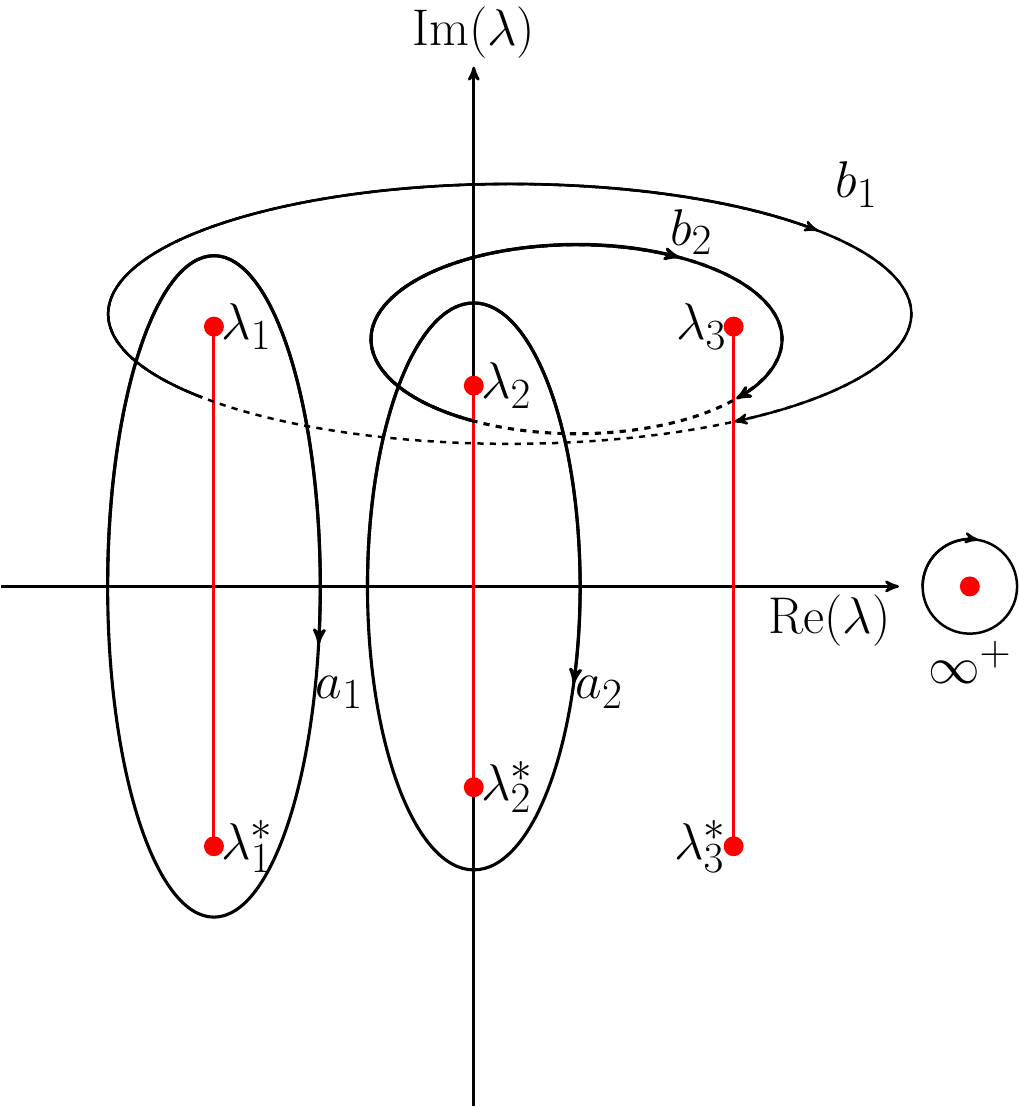}\label{fig:genus-two-figure-c}}
	\caption{The homology basis for the curve $y^2=\prod_{i=1}^{3}\left(\lambda^2-\lambda_i^2\right)$ defined in equation \eqref{eq:define-curve-algebro}. 
	The \Cref{fig:genus-two-figure-p} and \Cref{fig:genus-two-figure-c} corresponds to the \ref{case1} and \ref{case2}, respectively.}
	\label{fig:genus-two-figure}
\end{figure}

\begin{define}[Riemann theta function {\cite[p.35]{BelokolosBME-1986}}]\label{define:Riemann-Theta-function}
	The Riemann theta function (dimension-g) is defined as 
	\begin{equation}\label{eq:define-Riemann-Theta}
		\Theta(\mathbf{z})=\sum_{\mathbf{n}\in \mathbb{Z}^g}\exp\left\{ \left(\langle \mathbf{n}/2, \mathbf{B}\mathbf{n}\rangle+\langle \mathbf{n},\mathbf{z}\rangle\right)\right\}, 
	\end{equation}
	where $\mathbf{B}\in \mathbb{C}^{g\times g}$ is a $g\times g$ matrix; parameters $\mathbf{n}=[n_1, n_2, \cdots, n_g]^{\top}\in \mathbb{Z}^g$ and $\mathbf{z}=[z_1, z_2, \cdots, z_g]^{\top}\in \mathbb{C}^g$ are both $g$-dimensional vectors and $\langle 
	\mathbf{u}, \mathbf{v} \rangle=\mathbf{u}^{\top}\mathbf{v}$ with $\mathbf{u}\in \mathbb{Z}^g$ and $\mathbf{v}\in \mathbb{C}^g$. 
	Naturally, the Riemann theta function also satisfies equations
	\begin{equation}\label{eq:Riemann-Theta-prop}
		\Theta(\mathbf{z}+2\pi \ii\mathbf{e}_k)=\Theta(\mathbf{z}), \qquad 
		\Theta(\mathbf{z}+ \mathbf{B}\mathbf{e}_k)=\Theta(\mathbf{z})\ee^{-\mathbf{z}_k-\mathbf{B}_{kk}/2}, \qquad 
		\Theta(-\mathbf{z})=\Theta(\mathbf{z}),
	\end{equation}
	where the vector $\mathbf{e}_k\in \mathbb{Z}^{g}$ is defined as the $k$-th column of identity matrix $\mathbb{I}_g$. 
\end{define}

For the above definitions and combining with the algebro-geometric approach \cite{BelokolosBEI-94,FritzH-2003-Algebro-Geometric}, we can obtain two-phase solutions of the mKdV equation expressed by the Riemann theta function as follows:
\begin{equation}\label{eq:q1-q0}
	u(x,t)= C_{u_0}\frac{\Theta(D+\Delta+\ii Ux+\ii V t)}{\Theta(D+\ii Ux+\ii V t)}\ee^{2\ii \omega_2 x+2\ii \omega_3  t},\qquad C_{u_0}=\frac{\Theta(D)}{\Phi_{11,0}^{+}\Theta(D+\Delta)},\\
\end{equation}
where $\Phi_{11,0}^{+}$ (defined in equation \eqref{eq:Phi-lim-infty-initial}) is a constant depending on the initial points we setting. 
The solutions $\Phi_{11}=\Phi_{11}(x,t;\lambda)$ and $\Phi_{21}=\Phi_{21}(x,t;\lambda)$ of the related Lax pair \eqref{eq:Lax-pair} can be expressed as 
\begin{equation}\label{eq:Phi-expression}
	\begin{split}
		\!\!\!
		\Phi_{11}
		&\!
		\xlongequal[\eqref{eq:define-D},\eqref{eq:define-D-1}]{\eqref{eq:define-Phi},\eqref{eq:D-mu-define}}\!\! \frac{\Theta(D)\Theta(D+\mathcal{A}_{\infty^{-}}(P)+\ii Ux+\ii Vt)}{\Theta(D+\mathcal{A}_{\infty^{-}}(P))\Theta(D+\ii Ux+\ii Vt)}
		\ee^{\ii \left(\Omega_2(P) + \omega_2 \right)x+\ii\left(\Omega_3(P) + \omega_3 \right)t}
		,\\
		\!\!\!
		\Phi_{21}
		&\!
		\xlongequal[ \eqref{eq:define-D},\eqref{eq:define-D-1},\eqref{eq:C_u0-define}]{\eqref{eq:r-0},\eqref{eq:D-mu-define},\eqref{eq:relation-mu}}  \!\!
		\frac{2\ii\Theta(D)\Theta(D+\mathcal{A}_{\infty^{-}}(P)+\ii Ux+\ii Vt-\Delta)}
		{C_{u_0}\omega_1 \Theta(D+\mathcal{A}_{\infty^{-}}(P))\Theta(D+\ii Ux+\ii Vt)}
		\ee^{\ii \left(\Omega_2(P)-\omega_2 \right)x+\ii\left(\Omega_3(P) - \omega_3 \right)t+\Omega_1(P)}.\!\!\!\!
	\end{split}
\end{equation}
Additional details are given in \Cref{sec:algebro-geometric-approach}.
The Abelian integrals $\Omega_{1,2,3}(P)$ are defined as follows.
\begin{define}
	The Abelian integrals $\Omega_{1,2,3}(P)$, which have no singularities at points different from $\infty^{\pm}$, are defined as follows:
	\begin{itemize}
		\item The function $\Omega_1(P)$ is defined as
		\begin{equation}\label{eq:define-int-w1}
			\begin{split}
				&\Omega_1(P):=\int_{P_0}^{P} \dd\Omega_1, 
				\quad
				\oint_{a_i} \dd \Omega_1=0,
				\quad
				\text{and} \quad \Omega_1(P)=\pm(\ln(\lambda) 
				+ \ln(\omega_1)+o(1)), \quad  P\rightarrow \infty^{\pm}; 
			\end{split}
		\end{equation}
		\item The function $\Omega_2(P)$ is defined as 
		\begin{equation}\label{eq:define-int-w2}
			\Omega_2(P):=\int_{P_0}^{P}\dd\Omega_2, 
			\quad 	
			\oint_{a_i}\dd\Omega_2=0,  \quad
			\text{and} \quad
			\Omega_2(P)=\pm(\lambda
			+ \omega_2+o(1)), \quad  P\rightarrow \infty^{\pm}; 
		\end{equation}
		\item The function $\Omega_3(P)$ is defined as 
		\begin{equation}\label{eq:define-int-w3}
			\Omega_3(P):=\int_{P_0}^{P} \dd\Omega_3, \quad 
			\oint_{a_i}\dd \Omega_3=0, \quad 
			\text{and} \quad	
			\Omega_3(P)=\pm (4\lambda^3
			+ \omega_3+o(1)), 
			\quad 
			P\rightarrow \infty^{\pm};		
		\end{equation}
		where $P_0=(\lambda_3,0)$ and the integral path $a_i$, $i=1,2$ are shown in \Cref{fig:genus-two-figure}.	
	\end{itemize}
	Parameters $\omega_{1,2,3}$ are determined by the branch points $\lambda_{1,2,3}$ and independent of the spectral parameter $\lambda$.
\end{define}

Solutions \eqref{eq:q1-q0}, expressed in terms of hyperelliptic integrals, are not convenient for analyzing their stability. To address this, we introduce appropriate transformations that convert the hyperelliptic integrals into three standard forms of elliptic integrals, providing essential theoretical support for deriving explicit solutions. The hyperelliptic integrals required to determine the parameters of the corresponding solutions primarily include the following two types:
\begin{subequations} \label{eq:hyper}
	\begin{align}
		&\int \frac{\lambda^{2n}\dd \lambda }{((\lambda^2-\lambda_1^2)(\lambda^2-\lambda_2^2)(\lambda^2-\lambda_3^2))^{1/2}}\xlongequal{\lambda^2=\Lambda}\pm\int \frac{\Lambda^n\dd \Lambda}{2(\Lambda (\Lambda-\lambda_1^2)(\Lambda-\lambda_2^2)(\Lambda-\lambda_3^2))^{1/2}}, \label{eq:hyper-1}\\
		& \int \frac{\lambda^{2n+1}\dd \lambda }{((\lambda^2-\lambda_1^2)(\lambda^2-\lambda_2^2)(\lambda^2-\lambda_3^2))^{1/2}} \xlongequal{\lambda^2=\Lambda}\int \frac{\Lambda^n\dd \Lambda }{2((\Lambda-\lambda_1^2)(\Lambda-\lambda_2^2)(\Lambda-\lambda_3^2))^{1/2}}, \label{eq:hyper-2}
	\end{align}
\end{subequations}
by setting $\lambda^2=\Lambda$.
The sign $``\pm"$ is deduced by the correspondence between $\lambda$ and $\Lambda^{1/2}$, i.e., $\lambda=\Lambda^{1/2}$ or $\lambda=-\Lambda^{1/2}$.
Furthermore, we aim to demonstrate that the above hyperelliptic integrals can be expressed in terms of the three canonical forms of elliptic integrals defined in \Cref{define:elliptic-function} as well as in terms of elliptic functions, through the following steps.
\begin{enumerate}
	\item[\textbf{Step 1}] \textbf{Introduce a suitable conformal map with the aim of converting the integrals into the standard form.} 
	Introduce a conformal map between parameters $\Lambda$ and $z$, to transform elliptic integrals $\int [(\Lambda-\lambda_1^2)(\Lambda-\lambda_2^2)(\Lambda-\lambda_3^2)]^{-1/2} \dd \Lambda$ and $\int [\Lambda(\Lambda-\lambda_1^2)(\Lambda-\lambda_2^2)(\Lambda-\lambda_3^2)]^{-1/2}$ $\dd \Lambda$ into Legrangde standard elliptic integrals, which are listed in 
	\Cref{prop:elliptic-int-1} and \Cref{prop:elliptic-int-2}.
	The different forms of elliptic integrals $\int [(1-z^2)(1-k^2z^2)]^{-1/2}\dd z$ and $\int [(1-z^2)((k^{\prime })^2+k^2z^2)]^{-1/2}\dd z$ we choosing is dependent on the different cases of the branch points $\lambda_i$, $i=1,2,3$, i.e.  \ref{case1},  and \ref{case2} respectively.
	The detailed processes are provided in \Cref{appendix:map}.
	\item[\textbf{Step 2}] \textbf{Convert the general hyperelliptic integrals into general elliptic integrals form.}
	Providing the suitable transformations between the spectral parameter $\lambda$ and the new parameter $z$ expressed by the rational forms of elliptic functions, we deduce the hyperelliptic integrals provided in equation \eqref{eq:hyper} into the integral expressed by Jacobi elliptic functions in terms of linear fractional transformations.
	The transformation in this paper are listed in equations \eqref{eq:fs-2}, \eqref{eq:fs-1}, \eqref{eq:fs-4} and \eqref{eq:fs-3}.  
	\item[\textbf{Step 3}] \textbf{Determine the integration path after applying the above transformations.}
	Building on the previous two steps, the most crucial part is to identify the appropriate integration path, as hyperelliptic integrals are inherently path-dependent. Furthermore, it is also necessary to account for the specific sheet of $y$ selected during the evaluation of the integrals in equation \eqref{eq:hyper}.
	\item[\textbf{Step 4}] \textbf{Derive the recursive formula to get the exact expressions of the general hyperelliptic integrals in terms of three kinds of normal elliptic integrals.} 
	In order to obtain the recursive formulas associated with elliptic integrals, the hyperelliptic integrals given in equation \eqref{eq:hyper} are transformed into a combination of the three canonical forms of elliptic integrals. 
\end{enumerate}
Based on the above steps, we obtain the \Cref{prop:case-1-P}-\ref{prop:case-1-C},
and get the explicit two-phase solutions of the mKdV equation and related Lax pair through the algebro-geometric method. 
Further details are provided in \Cref{section:exact-solution} and the Appendix \ref{appendix:map}.
Notably, the transformations we utilizing in this work are not unique.
In summary, we obtain the explicit expressions of the solution $u(x,t)$ of the \ref{eq:mKdV} equation as follows.

\begin{theorem}\label{theorem:solution-u}
	The two-phase solutions of the \ref{eq:mKdV} equation can be expressed as
	\begin{equation}\label{eq:solutions-dn}
		u(x,t)=C_{u_0}^{(i)}\frac{\Theta(\ii U^{(i)}(x+vt)+\Delta^{(i)}+D^{(i)})}{\Theta(\ii U^{(i)}(x+vt)+D^{(i)})}\ee^{2\ii \omega_{2}^{(i)}(x+vt)}, 
	\end{equation}
	where $v=2(\lambda_1^2+\lambda_2^2+\lambda_3^2)$.
	The superscript ``$(i)$" denotes Case $i$, $i = 1, 2$. 
	The rest parameters are given in two distinct cases respectively.
	\begin{itemize}
		\item For the \ref{case1}, without loss of generality, we set $0<\Im(\lambda_1)<\Im(\lambda_2)<\Im(\lambda_3)$, $\lambda_{1,2,3}\in \ii \mathbb{R}$. 
		The related parameters of the solution $u(x,t)$ are $\omega_{2}^{(1)}=0$, $	U^{(1)}=\kappa^{(1)}\mathbf{1}$, $\Delta^{(1)}=(\ii \pi+ \nu^{(1)}\sigma_3) \mathbf{1}$,
		\begin{subequations}\label{eq:u-parameters-dn}
			\begin{align}
			&
				\mathbf{B}^{(1)}=\ii \pi \begin{bmatrix}
					  \tau^{(1)}_2 + \tau^{(1)}_1 &   \tau^{(1)}_2 - \tau^{(1)}_1 \\[3pt]   \tau^{(1)}_2 - \tau^{(1)}_1 &   \tau^{(1)}_2 + \tau^{(1)}_1
				\end{bmatrix},
				\qquad
				C_{u_0}^{(1)}=\frac{2\ii \vartheta_1(\nu^{(1)},  \tau^{(1)}_1)\vartheta_2(\nu^{(1)},  \tau^{(1)}_1)\lambda_3}{\vartheta_2(0,  \tau^{(1)}_1)\vartheta_1(2 \nu^{(1)},  \tau^{(1)}_1)}, \label{eq:u-parameters-dn-B}\\
				\kappa^{(1)}&=\frac{\pi(\lambda_1^2-\lambda_3^2)^{1/2}}{  K^{(1)}_2}, \quad
				  k^{(1)}_2=\frac{(\lambda_1^2-\lambda_2^2)^{1/2}}{(\lambda_1^2-\lambda_3^2)^{1/2}}, \quad
				  k^{(1)}_1=   k^{(1)}_2\frac{\lambda_3}{\lambda_2}, \quad \nu^{(1)}=\frac{\ii \pi  }{  K^{(1)}_1}F\left(\frac{\lambda_2}{\lambda_3},  k^{(1)}_1\right),  \label{eq:u-parameters-dn-k}
			\end{align}
		\end{subequations} 
		with $K^{(1)}_i=K(k^{(1)}_i)$, $K^{{(1)}\prime}_i=K(k^{{(1)}\prime}_i)$, $\tau^{(1)}_i=\ii K^{{(1)}\prime}_i/K^{(1)}_i$ and $\mathbf{1}=[1 \,\,\, 1]^{\top}$.
		The parameters choosing $D^{(1)}=\mathbf{0}$ and $D^{(1)}=\ii   \tau^{(1)}_2\pi \mathbf{1}$ will correspond to two distinct solutions of the mKdV equation.
		\item For the \ref{case2}, without loss of generality, we set $\Re(\lambda_1)<0<\Re(\lambda_3)$, $\lambda_{1,3}\in \mathbb{C}\backslash(\ii \mathbb{R}\cup \mathbb{R})$, and $\lambda_2\in \ii \mathbb{R}$. 
		The related parameters of the solution $u(x,t)$ are $\omega_{2}^{(2)}=\kappa^{(2)}/2$, $U^{(2)}=-\kappa^{(2)}\mathbf{2}$,
		\begin{subequations}\label{eq:u-parameters-cn}
			\begin{align}
				&
				\Delta^{(2)}=\ii \pi \begin{bmatrix}
					1-\tau^{(2)}_2 \\
					\frac{\tau^{(2)}_2+\tau^{(2)}_1}{-2}+\frac{\ii\nu^{(2)} }{2 \pi}
				\end{bmatrix},\,\,\, \nu^{(2)}=\frac{\ii \pi}{ K^{(2)}_1}F\left(\frac{2\ii(AB)^{1/2}}{A-B}, k^{(2)}_1\right)\!, \,\,\, \label{eq:u-parameters-cn-U}\\
					\mathbf{B}^{(2)}&=\ii \pi \begin{bmatrix}
					2 \tau^{(2)}_2  & \tau^{(2)}_2-1 \\ 
					\tau^{(2)}_2-1 & (\tau^{(2)}_1+\tau^{(2)}_2)/2
				\end{bmatrix},  \qquad 	C_{u_0}^{(2)}=\frac{(A-B)\vartheta_4(0,\tau^{(2)}_1)\ee^{\frac{\ii\tau^{(2)}_2\pi}{4}}}{\lambda_2\vartheta_1(\ii\tau^{(2)}_1\pi+\nu ^{(2)},\tau^{(2)}_1)}, \label{eq:u-parameters-cn-B}\\			 
				 \kappa^{(2)}&=\frac{A^{1/2} \pi}{K^{(2)}_2}
, \,\,\,\,  k^{(2)}_1=\frac{(\lambda_2^4-(A-B)^2)^{1/2}}{2(AB)^{1/2}}, \,\,\,
				 k^{(2)}_2=\frac{(2(A+\lambda_2^2)-\lambda_3^2-\lambda_1^2)^{1/2}}{2(A)^{1/2}},\label{eq:u-parameters-cn-k}
			\end{align}
		\end{subequations}
		with $A=|\lambda_1^2-\lambda_2^2|$, $B=|\lambda_1^2|$, $D^{(2)}=\ii \pi \mathbf{1}$, $K^{(2)}_i=K(k^{(2)}_i)$, $K^{{(2)}\prime}_i=K(k^{{(2)}\prime}_i)$, $\mathbf{2}=[2,1]^{\top}$, and $\tau^{(2)}_i=\ii K^{{(2)}\prime}_i/K^{(2)}_i$.
	\end{itemize}
	These two-phase solutions of the mKdV equation are the traveling waves with velocity $-v$. 
\end{theorem}

	Under the different cases of branch points $\lambda_i$ and $\lambda_i^*$, $i=1,2,3$ satisfying the \ref{case1} and \ref{case2}, the related solutions exhibiting different cases are divided into the following cases:
	\begin{itemize}
		\item If all of branch points are nonzero, the related solutions can be constructed by genus-two algebraic curves and expressed as the Riemann theta function form shown in \Cref{theorem:solution-u}.
		\item If a pair of branch points on the imaginary axis vanish, i.e., $\lambda_1=\lambda_1^*=0$ in \ref{case1} and $\lambda_2=\lambda_2^*=0$ in \ref{case2}, the corresponding solutions can degenerate into $\cn$-type and $\dn$-type elliptic solutions.
	\end{itemize}
The detailed process can be found in Section \ref{section:exact-solution}. Furthermore, by examining the relationships among Riemann theta functions and Jacobi theta functions, we establish the equivalence between solutions expressed in terms of Riemann theta functions and those expressed in terms of elliptic functions.

\begin{prop}\label{prop:solutions-equivalent}
	All periodic traveling wave solutions of the mKdV equation can be expressed by Riemann theta functions, as shown in equation \eqref{eq:solutions-dn}. In other words, the genus-two periodic traveling wave solutions represented by the Riemann theta function \eqref{eq:solutions-dn} are equivalent to those expressed in terms of elliptic functions \eqref{eq:u1-elliptic},\eqref{eq:u2-elliptic}.
\end{prop}

By employing the algebro-geometric method, the fundamental solutions of the associated Lax pair can be expressed as follows:
\begin{theorem}\label{theorem:solution-Phi}
	The vector solutions of the Lax pair with the above two-phase solutions of the mKdV equation can be expressed as follows:
	\begin{equation}\label{eq:Phi-solution}
		\Phi(x,t;P)=
		\begin{bmatrix}
			\frac{\Theta(\ii U^{(i)}(x+v t)+\mathcal{A}^{(i)}_{\infty^{-}}(P)+D^{(i)})}{\Theta(\ii U^{(i)}(x+vt)+D^{(i)})}\ee^{\ii \omega_2^{(i)}(x+vt)}\\
			\frac{\Theta(\ii U^{(i)}(x+vt)+\mathcal{A}^{(i)}_{\infty^{-}}(P)+D^{(i)}-\Delta^{(i)})}{\Theta(\ii U^{(i)}(x+vt)+D^{(i)})}\ee^{\Omega_1^{(i)}-\ii \omega_2^{(i)} (x+vt)}
		\end{bmatrix}\ee^{\ii (\Omega_2^{(i)}(P) x+ \Omega_3^{(i)}(P) t)}, \quad i=1,2,
	\end{equation}
	where parameters $U^{(1,2)}$, $D^{(1,2)}$ are provided in \Cref{theorem:solution-u};
	and functions $\Omega_{1,2,3}^{(1)}$, $\Omega_{1,2,3}^{(2)}$, $\mathcal{A}^{(1)}_{ \infty^{-}}(P)$, and $\mathcal{A}^{(2)}_{ \infty^{-}}(P)$ are defined in equations \eqref{eq:define-Omega-123-P}, \eqref{eq:define-Omega-123-C}, \eqref{eq:define-AP-P}, and  \eqref{eq:define-AP-C}, respectively.
\end{theorem}

For the periodic traveling wave solutions  of the mKdV equation, we would consider the spectral stability under the following transformations:
\begin{equation}\label{eq:transformation-xi-eta}
	(x,t) \xlongequal[\eta=t]{\xi=x+vt} (\xi,\eta),
\end{equation}
where the parameter $-v$ is called the velocity of the solution defined in \Cref{theorem:solution-u}.
Under this transformation, the \ref{eq:mKdV} equation would be converted into
\begin{equation}\label{eq:mKdV1}
	u_{\eta}+vu_{\xi}+u_{\xi\xi\xi}+6u^2u_{\xi}=0.
\end{equation}
To study the spectral stability of genus-two traveling wave solutions , we introduce perturbations of the stationary solution $\hat{u}(\xi,\eta)=u(\xi)+\epsilon w(\xi,\eta)+\mathcal{O}(\epsilon^2)$,
where $\epsilon$ is a small parameter and $w(\xi,\eta)$ is a real-valued function of $(\xi,\eta)\in \mathbb{R}^2$. 
Plugging $\hat{u}(\xi,\eta)$ into \eqref{eq:mKdV1} and considering the first-order term of $\epsilon$, we obtain the linearized equation
\begin{equation}\label{eq:linearized-mKdV}
	\partial_t w +\partial_{\xi}^3w +v\partial_{\xi}w +6u^2\partial_{\xi}w +12w u\partial_{\xi}u=0, \qquad u\equiv u(\xi), \quad w\equiv w(\xi,\eta).
\end{equation}
Since equation \eqref{eq:linearized-mKdV} is autonomous in time, we can decompose $w(\xi,\eta)$ into the following form
\begin{equation}\label{eq:w}
	w(\xi,\eta)=W(\xi;\Omega)\exp(\Omega \eta)+W^*(\xi;\Omega)\exp(\Omega^* \eta),
\end{equation}
by separating variables.
Then, we obtain the linearized spectral problem of equation \eqref{eq:linearized-mKdV}:
\begin{equation}\label{eq:spectral}
	\partial_\xi  (-\partial_\xi^2-v-6u^2)W=\mathcal{JL}W=\Omega W, \qquad W(\xi;\Omega) \in C_b^0(\mathbb{R}),
\end{equation} 
where $\mathcal{J}=\partial_\xi,\mathcal{L}=-\partial_\xi^2-v-6u^2$, $\Omega\in \mathbb{C}$, and $C_b^0(\mathbb{R})$ denotes the space of bounded continuous functions on the real line. 
The spectrum is defined as 
\begin{equation}\label{eq:spectrum}
	\sigma(\mathcal{JL}):=\{\Omega\in \mathbb{C}| W(\xi)\in C^0_b(\mathbb{R})\}. 
\end{equation}
Due to the Hamiltonian structure of the spectrum \cite{HaragusK-08}, the genus-two solution $u$ is spectrally stable with respect to perturbations in $C_b^0(\mathbb{R})$ if $\sigma(\mathcal{JL})\subset \ii \mathbb{R}$. 
Then, the definition of spectral stability is given as follows: 
\begin{define}\label{define:spect-stable}
	The genus-two periodic traveling solution $u(\xi)$ is spectrally stable to perturbations $w(\xi,\eta)$ in $C_b^0(\mathbb{R})$, where $w(\xi,\eta)$ is defined in equation \eqref{eq:w}, if all $\Omega\in \ii \mathbb{R}$. In brief, the stability spectrum is defined as $\sigma(\mathcal{JL})\subset \ii \mathbb{R}$, where $\sigma(\mathcal{JL})$ is defined in equation \eqref{eq:spectrum}.
\end{define}

As Deconinck and Kapitula pointed out in \cite{DeconinckK-10}, the Lax spectrum of the focusing mKdV equation is no longer confined to the real axis, which makes the detailed analysis of the bounded eigenfunctions more difficult. 
To overcome this difficulty, we use the Riemann theta function to express the squared eigenfunction $W(\xi;\Omega)$, which converts the problem of analyzing bounded functions into studying the algebraic problems on Zeta function and the radical fraction with respect to the spectral parameter $\lambda$. 
According to the Floquet theorem (Theorem in \cite{DeconinckK-06,Floquet-83}), we know that the solution $W(\xi;\Omega)$ in the linear homogeneous differential equation \eqref{eq:linearized-mKdV} is of the form
$W(\xi;\Omega)=\ee^{\ii \hat{\eta} \xi} \hat{W}(\xi;\Omega), \hat{W}(\xi+2T;\Omega)=\hat{W}(\xi;\Omega),\hat{\eta}\in \mathbb{C}$, where $2T$ is the period of the function $\hat{W}(\xi ;\Omega )$. 
Every bounded solutions of spectral problem \eqref{eq:spectral} is of the form
\begin{equation}\label{eq:W-W-eta}
	W(\xi ;\Omega )=\ee^{\ii \hat{\eta} \xi  } \hat{W}(\xi ;\Omega ),\qquad \hat{W}(\xi  +2T;\Omega)=\hat{W}(\xi ;\Omega ),  \qquad \hat{\eta} \in  \left[-\frac{\pi}{2T},\frac{\pi }{2T}\right). 
\end{equation}
Utilizing the squared-eigenfunction method \cite{BottmanD-2009}, we get the squared-eigenfunction $W(\xi;\Omega)$, which can be used to gain all solutions of the equation \eqref{eq:spectral}. Additional details are given in \Cref{section:spectral-stability}.
By the explicit expression of the function $W(\xi;\Omega)$ shown in equation \eqref{eq:define-W}, we get
\begin{equation}\nonumber
	\exp{(2\ii \hat{\eta} T)}=\frac{W(\xi+2T;\Omega)}{W(\xi;\Omega)}
	=\exp{\left(4\ii \left(\Omega_2(P) + \omega_2 \right)T\right)}.
\end{equation}
For the stability analysis, we just consider the bounded function $W(\xi;\Omega)$, which implies that the real part of the exponent of the function $W(\xi;\Omega)$ is zero, i.e.,  the parameter $\lambda$ must locate in the set $Q$ defined as
\begin{equation}\label{eq:set-Q}
	Q:=\left\{  \lambda\in \mathbb{C}\cup\{\infty\}\left| {\Im}\left(\mathcal{I} (\lambda)\right)=0  \right. \right\},
\end{equation}
where the function $\mathcal{I}(\lambda)$ is defined as
\begin{equation}\label{eq:M}
	M(\lambda):=2\hat{\eta} T=4 \left(\Omega_2(P) + \omega_2 \right)T, \qquad 
	\mathcal{I}(\lambda):=M(\lambda)/(2T)=2\Omega_2(P) + 2\omega_2,
\end{equation}
where the parameter $\hat{\eta}$ is defined in equation \eqref{eq:W-W-eta}.
We also divide the analysis of the spectral stability into the following two cases -- \ref{case1} and \ref{case2}.

When the \ref{case1} holds, we rewrite the set $Q$ as $Q^{(1)}$ and define three sets as: 
\begin{equation}\label{eq:define-hat-Q}
	\begin{split}
		& Q^{(1)} _R:=\left\{\lambda \in  \mathbb{R}\right\}, \qquad
		 Q^{(1)} _{P_1}:=
		\left\{\lambda \in \ii \mathbb{R}\left| \,
		|\Im(\lambda)|\le\Im(\lambda_1) \right.\right\}, \\
		&  Q^{(1)} _{P_2}:=
		\left\{\lambda \in \ii \mathbb{R} \left| \,
		\Im(\lambda_2) \le |\Im(\lambda)|\le\Im(\lambda_3) \right.\right\}.
	\end{split}
\end{equation}
We get that the set $Q^{(1)}$ can be expressed as the union of the above three sets, i.e. $Q^{(1)}= Q^{(1)} _R\cup  Q^{(1)} _{P_1} \cup  Q^{(1)} _{P_2}$, which is proved in \Cref{lemma:dn-bounded}.
Moreover, for any $\lambda\in Q^{(1)}$, the corresponding eigenvalues $\Omega(\lambda)$ are pure imaginary.
Then, we get the consequence for the spectral stability.
\begin{theorem}\label{theorem:spectral-image}
	Under the \ref{case1}, the two-phase solutions of the mKdV equation are spectrally stable.
\end{theorem}

When the \ref{case2} holds, we denote the set $Q$ as $Q^{(2)}$ and define three sets as: 
\begin{equation}\label{eq:define-Q-C-IRp}
	 Q^{(2)}_{R}:=\left\{\lambda\in \mathbb{R}\right\}, \qquad 
	 Q^{(2)}_{I}:=\left\{\lambda\in \ii \mathbb{R}\left| \ |\Im(\lambda)|\le \Im(\lambda_2)\right.\right\}, \qquad 
	 Q^{(2)}_{P}:=\left\{\lambda_1,\lambda_1^*,\lambda_3,\lambda_3^* \right\}.
\end{equation}
In this case, we obtain that the set satisfies $Q^{(2)}\neq Q^{(2)}_{R} \cup Q^{(2)}_{I} \cup Q^{(2)}_{P}$ and $(Q^{(2)}_{R} \cup Q^{(2)}_{I} \cup  Q^{(2)}_{P})\subset Q^{(2)}$. 
In \Cref{prop:Q-M-cn}, we prove that for any $\lambda \in (Q^{(2)}_{R} \cup Q^{(2)}_{I} \cup  Q^{(2)}_{P})$, the eigenvalue $\Omega(\lambda)$ satisfies $\Omega(\lambda)\in \ii \mathbb{R}$.
Furthermore, we can prove that $\Omega(\lambda)\notin \ii \mathbb{R}$ for any $\lambda \in Q^{(2)}\backslash(Q^{(2)}_{R} \cup Q^{(2)}_{I} \cup  Q^{(2)}_{P})$.
In conclusion, we obtain the following Theorem:
\begin{theorem}\label{theorem:spectral-complex}
	The two-phase solutions satisfying the \ref{case2} of the mKdV equation are spectrally unstable.
\end{theorem}

Under this case, we would like to study the subharmonic perturbation and to explore the subharmonic perturbation stability.

\begin{define}\label{defin:P-sub} 
	For the two-phase solutions $u(\xi)$ with period $2T$, 
	if the perturbation of this solution is $2PT$ periodic function $P\in \mathbb{N}^{+}$, 
	it is called a P-subharmonic perturbation of solution $u(\xi)$. 
	If the period of perturbations is the same as the solution $u(\xi)$, we call it co-periodic perturbation.
\end{define}

Combining \Cref{define:spect-stable} with \Cref{defin:P-sub}, we obtain the definition of subharmonic perturbations.

\begin{define}\label{define:spect-P}
	If perturbations $W(\xi;\Omega)$ are $2PT$ periodic functions and $\Omega\in \ii \mathbb{R}$, 
	i.e., the spectrum $\sigma_P(\mathcal{JL})$ satisfies 
	\begin{equation}\nonumber
		\sigma_P(\mathcal{JL}):=\{\Omega\in \mathbb{C}| W(\xi;\Omega)\in C^0_b(\mathbb{R})\cap L^2_{per}([-PT,PT]) \} \subset \ii \mathbb{R},
	\end{equation} 
	then the solution $u(\xi)$ is P-subharmonic perturbation spectrally stable.
\end{define}
 
In the following cases, we are going to study the period of the function $W(\xi;\Omega)$, i.e., consider the parameter $\hat{\eta}$.
Based on \Cref{defin:P-sub}, for the $P$-subharmonic perturbation problems,
$\hat{\eta}$ can be defined in any interval of length $2\pi/(2T)$, i.e.,
\begin{equation}\label{eq:eta}
		\hat{\eta}=\frac{2m\pi}{2PT}+\frac{(2n+1)\pi}{2T} ,\qquad m=-P,-P+1,\cdots,-1, \quad \text{and} \quad n\in \mathbb{Z}.
\end{equation}
Together with equations \eqref{eq:M} and \eqref{eq:eta}, the $P$-subharmonic perturbation problems must satisfy $M(\lambda)=2 n\pi /P$, $n\in \mathbb{Z}$. 
The spectral stability with respect to the subharmonic perturbations of period $2PT$ is that all eigenvalues $\Omega$ of $2PT$ periodic function $W(\xi;\Omega)$ satisfying \eqref{eq:spectral} are imaginary, i.e., $\Omega(\lambda)\in \ii \mathbb{R}$. 
We set 
\begin{equation}\label{eq:Q_P}
		Q_{sub}:=\left\{ z\in Q|M(\lambda)=\frac{2\pi}{P}m+(2n+1)\pi, \quad m=-P,-P+1,\cdots,-1, \quad n \in \mathbb{Z}\right\}.
\end{equation}
When for any $\lambda\in Q_{sub}$, the value $\Omega(\lambda)\in \ii \mathbb{R}$ and then the corresponding solutions are spectrally stable with respect to perturbations of period $2PT$. 

\begin{theorem}\label{theorem:spectral-complex-P}
	Under the \ref{case2}, the two-phase solutions of the mKdV equation are $P$-subharmonic spectrally stable. 
	The parameter $P$ is dependent on the modulus $k_2^{(2)}$ as follows:
	\begin{itemize}
		\item When $1<2 E^{(2)}_2/  K^{(2)}_2$, solutions are $P$-subharmonic spectrally stable with $P\le 4\pi/(\pi+M(\lambda_0))$ with $M(\lambda_0)$ provided in equation \eqref{eq:M-lambda_0};
		\item When $1=2 E^{(2)}_2/  K^{(2)}_2$, solutions are $2$-subharmonic spectrally stable;
		\item When $1>2 E^{(2)}_2/  K^{(2)}_2$, solutions are co-subharmonic spectrally stable.
	\end{itemize}
	The function $M(\lambda)$ and parameter $\lambda_0$ are defined in equations \eqref{eq:M} and \eqref{eq:define-lambda-0}, respectively.
\end{theorem}

Therefore, we obtain that the condition of the maximize parameter $P$ is dependent on the modulus $k_2^{(2)}$. 
In \Cref{fig:P}, we exhibit the correspondence between the maximize parameter $P$ and the modulus $k_2^{(2)}$.
When solutions are subharmonic spectrally stable, we further study the orbital stability of the above two-phase solutions in a suitable function space.
For $2PT$-periodic functions $u(\xi):[-PT,PT]\mapsto \mathbb{C}$, we define the space $H^{k}_{per}([-PT,PT])$, as
\begin{equation}\nonumber
	H^{k}_{per}([-PT,PT]):=\left\{u\left|\left(\sum_{j=0}^{k}\int_{-PT}^{PT}\left|\partial_{\xi}^ju(\xi)\right|^2\dd \xi\right)^{1/2}< \infty\right.\right\}.
\end{equation}

\begin{define}\label{define:orbital stable}
	The genus-two solution $u(\xi)$ of the mKdV equation is orbitally stable with respect to perturbations in a Hilbert space $X$ if for any solution $v(\xi,\eta)$ of the mKdV equation and any given $\epsilon>0$, there exists $\delta>0$ satisfying
	\begin{equation} \nonumber
		\| v(\xi,0)- \mathcal{T}(\gamma_0)u(\xi) \|_{X}\le \delta,
	\end{equation}
	such that for any $ \eta \in \mathbb{R} $,
	\begin{equation} \nonumber
		\inf_{\gamma\in \mathbb{R}} \| v(\xi,\eta)-\mathcal{T}(\gamma)u(\xi,\eta ) \|_{X} \le \epsilon, 
	\end{equation} 
	where $\|\cdot \|$ denotes the norm obtained through $\left\langle\cdot, \cdot\right\rangle $ in the space $X$ and the operator $\mathcal{T}(\gamma)$ is defined here as
	\begin{equation}\label{eq:defin-operator}
		\mathcal{T}(\gamma)u(\xi)\equiv u(\xi+\gamma).
	\end{equation}
\end{define}

As is well known, the mKdV equation possesses an infinite number of conserved quantities \eqref{eq:H0}, where the period of the function $u$ is $2PT$.
Define the $n$-th mKdV equation with time variables $\eta_n$ under the moving coordinate form $(\xi,\eta_n)$ in equation \eqref{eq:H-mKdV}.

For any conserved quantities $\mathcal{H}_i$, $i=1,3,5,\cdots$ in the mKdV hierarchy (equation \eqref{eq:H0}), the corresponding operator $\mathcal{L}_i$ and Krein signature $\mathcal{K}_i(\lambda)$ are defined in \Cref{defin:Krein}. 
Based on the stationary solution $u$, we linearize equations $u_{\eta_i}=\mathcal{J}\hat{\mathcal{H}}'_i(u),i=1,2,\cdots,n$ about $u$ with 
\begin{equation}\nonumber
	v(\xi,\boldsymbol{\eta})=u(\xi,\boldsymbol{\eta})+\epsilon w(\xi,\boldsymbol{\eta})+\mathcal{O}(\epsilon^2), \qquad \boldsymbol{\eta}=\left(\eta_1,\eta_2,\cdots,\eta_n\right),
\end{equation}
and result in the linear system: $w_{\eta_i}=\mathcal{JL}_{i}w, i=1,2,\cdots,n$, where $\mathcal{L}_i$ is the variational derivative $\hat{\mathcal{H}}''_{i}$, $i=1,2,\cdots$, evaluated at the stationary solution. Then, we obtain
\begin{equation} \label{eq:spectral-Omega}
	\Omega_nW=\mathcal{JL}_nW,\qquad \Omega_n^*W^*=\mathcal{JL}_nW^*, 
\end{equation} 
where $W=W(\xi;\Omega_n)$.
\begin{define}\label{defin:Krein}
	Krein signature is the sign of 
	\begin{equation}\label{eq:define-krein}
		\mathcal{K}_n(\lambda):=\left\langle W_n,\mathcal{L}_n W_n \right\rangle_{L^2}, \qquad \left\langle W_n,\mathcal{L}_n W_n \right\rangle_{L^2}=\int_{-PT}^{PT}W_n^*\mathcal{L}_n W_n \dd \xi,
	\end{equation}
	where $W_n=W(\xi;\Omega_n)$ is an eigenfunction of the $n$-th mKdV equation \eqref{eq:spectral-Omega}. 
	The inner product is defined in the $L^2([-PT,PT])$ inner product space.
\end{define}
Under the two different cases, the Krein signatures $\mathcal{K}_{1,2}(\lambda)$ satisfy the following Proposition.
\begin{prop}\label{prop:kerin}
	When $2 E^{(2)}_2/  K^{(2)}_2\le 1$, for any $\lambda \in  Q^{(2)} _R\cup  Q^{(2)} _I$, $\mathcal{K}_1(\lambda)\ge 0$; when $2 E^{(2)}_2/  K^{(2)}_2>1$, not all $\lambda \in  Q^{(2)} _R\cup  Q^{(2)} _I$ such that $\mathcal{K}_1(\lambda)\ge 0$.
	For any $k\in (0,1)$, $\lambda\in Q^{(2)}_I\cup Q^{(2)}_R$, $\mathcal{K}_2(\lambda)\ge 0$. If and only if $\lambda=0,\pm \lambda_0$, the equality holds.
\end{prop}
With the aid of methods in \cite{GrillakisSS-87,GrillakisSS-90,KapitulaP-13}, we provide an orbital stability analysis and come to the following theorems. All proofs are provided in \Cref{sec:orbital-stability}.

\begin{theorem}\label{theorem:orbital-cn}
	If the two-phase solutions of the mKdV equation constructed by branch points satisfying \ref{case2} are spectrally stable with respect to perturbations of period $2PT,P\in \mathbb{Z}_{+}$ and $P< \frac{4\pi}{\pi+ M(\lambda_0)}$, then they are orbitally stable in the space $H^2_{per}([-PT,PT])$.
\end{theorem}

\begin{theorem}\label{theorem:orbital-dn}
	 Under the \ref{case1}, the two-phase solutions $u(\xi)$ are orbitally stable in the space $H^2_{per}([-PT,PT])$, $P\in \mathbb{Z}_{+}$.
\end{theorem}

The main contributions of this work are summarized as follows:
\begin{itemize}
	\item We present a method to establish the relationship between genus-two hyperelliptic Riemann theta function solutions and Jacobi elliptic function solutions of the mKdV equation. This approach can be extended to the other AKNS system in the genus-two case, for instance, the two-phase solution of  sine-Gordon equation \cite{DubrovinN-1982} and double-periodic solutions of NLS equation \cite{ChenPW-19}. 
	\item Building on the general properties of hyperelliptic integrals and their recursive relations, we express the required hyperelliptic integrals in terms of the three canonical elliptic integrals. Subsequently, all hyperelliptic integrals in the Riemann theta function solution can be evaluated using Jacobi elliptic integrals.
	These formulations provide the essential foundation for the explicit evaluation of the Riemann theta function solutions of the mKdV equation and their associated Lax pairs. 
	They further lay the groundwork for analyzing the stability of genus-two periodic traveling wave solutions under two distinct scenarios.	
	\item We investigate the subharmonic spectral stability of genus-two traveling wave solutions and derive the necessary and sufficient conditions for their spectral stability under subharmonic perturbations with the aid of squared eigenfunction method. 
	Furthermore, we analyze the orbital stability of these genus-two periodic traveling wave solutions. To the best of our knowledge, this constitutes the first rigorous proof to the nonlinear stability for genus-two traveling solutions of mKdV equation.
\end{itemize}

\subsection{Outline for this work}

The organization of this work is as follows. 
In \Cref{sec:algebro-geo-solution}, by applying the algebro-geometric method, we construct the two-phase solutions of the mKdV equation together with the corresponding Lax pair solution expressed in terms of Riemann theta functions. 
Using the relations among Jacobi theta functions, Jacobi elliptic functions, and Riemann theta functions, we demonstrate that the resulting two-phase solutions are equivalent to their elliptic-function representations.
In \Cref{section:spectral-stability}, we investigate the linearized spectral problem of the focusing mKdV equation for these two-phase solutions via the squared eigenfunction method and analyze the spectral stability of periodic waves with respect to subharmonic perturbations.
In \Cref{sec:orbital-stability}, based on the conditions for subharmonic spectral stability, we further establish the orbital stability of periodic traveling waves in an appropriate Hilbert space.

\section{Two-phase solutions  of the mKdV equation}\label{sec:algebro-geo-solution}
In this section, we mainly deduce the two-phase solutions  of the focusing \ref{eq:mKdV} equation and its associated Lax pair, through the algebro-geometric approach. 
These solutions are expressed by Riemann theta functions with related branch points $\lambda_{1,2,3}$, which greatly facilitate the analysis of the spectral and orbital stability of genus two periodic traveling wave solutions  in the following sections.
Firstly, we introduce the integrable hierarchy, which is given in \Cref{sec:Preliminaries}.
Secondly, we provide the two-phase solutions in terms of Riemann theta functions in \Cref{sec:algebro-geometric-approach}.
Based on the general properties of hyperelliptic integrals and their recursive relationships, all hyperelliptic integrals listed in \Cref{sec:algebro-geometric-approach} are expressed by the canonical forms of the elliptic integrals with related branch points in \Cref{section:exact-solution}. 
Finally, we provide the explicit expressions and their Lax pair solutions of the genus-two periodic traveling wave solutions .

\subsection{Integrable hierarchy}\label{sec:Preliminaries}
We start from matrix function $\Phi \equiv \Phi(x,\mathbf{t};\lambda)$ defined in equation \eqref{eq:Phi-m} and taking the derivative of variables $x$ and $t_n$, it follows that
\begin{equation}\label{eq:Lax-ti-x-part}
	\begin{split}
		\Phi_{x}\Phi^{-1}
		\xlongequal{\eqref{eq:Phi-m}} &\ m_{x}m^{-1}-\ii \lambda  m\sigma_3m^{-1}
		= \mathbf{U}(\lambda;\mathbf{Q})+\mathcal{O}\big(\lambda^{-1}\big),\\
		\Phi_{t_n}\Phi^{-1}
		\xlongequal{\eqref{eq:Phi-m}} &\ m_{t_n}m^{-1}-\ii \lambda^{n} m\sigma_3m^{-1}
		= \mathbf{V}_{n}(\lambda;\mathbf{Q})+\mathcal{O}\big(\lambda^{-1}\big).
	\end{split}
\end{equation}
 Thus, through the Liouville theorem, matrices $\mathbf{U}(\lambda;\mathbf{Q})$ and $\mathbf{V}_n(\lambda;\mathbf{Q})$ can be determined by
\begin{equation}\label{eq:U-V-n}
	\begin{split}
		& \mathbf{U}(\lambda;\mathbf{Q})
		:=-\ii( \lambda  m\sigma_3m^{-1})_{+},  \qquad
		\mathbf{V}_n(\lambda;\mathbf{Q})
		:=-\ii( \lambda^{n} m\sigma_3m^{-1})_{+},
	\end{split}
\end{equation}
where $(\cdot)_{+}$ defines the regular part of the function with respect to the spectral parameter $\lambda$. 
As $\lambda\rightarrow \infty$, the matrix function $m$ is expressed as $m=\mathbb{I}_2+m_1(x,\mathbf{t})\lambda^{-1}+m_2(x,\mathbf{t})\lambda^{-2}+\mathcal{O}(\lambda^{-3})$. Combined with equation \eqref{eq:U-V-n}, the $x$-part of the Lax pair can be written as
\begin{equation}\label{eq:Lax-x-part}
	\Phi_x(x,\mathbf{t};\lambda)=\mathbf{U}(\lambda;\mathbf{Q}) \Phi(x,\mathbf{t};\lambda), \qquad \mathbf{U}(\lambda;\mathbf{Q})=-\ii \lambda \sigma_3 +\mathbf{Q}, \qquad 
	\mathbf{Q}
	=\begin{bmatrix}
		0 &  u(x,\mathbf{t}) \\  -u(x,\mathbf{t}) & 0
	\end{bmatrix},
\end{equation}
where $\left[\mathbf{A},\mathbf{B}\right]=\mathbf{AB}-\mathbf{BA}$ denotes the commutator and $\mathbf{Q}$ is called the potential function. Furthermore, we will show that the matrix $\mathbf{V}_n(\lambda;\mathbf{Q})$ in terms of $\mathbf{Q}$.
The matrix function $\Psi(x,\mathbf{t};\lambda):=m\sigma_3 m^{-1}$ could be represented as a summation form and the matrix $\mathbf{V}_n(\lambda;\mathbf{Q})$ could be rewritten by matrix functions $\Psi_i=\Psi_i(x,\mathbf{t})$ as follows:
\begin{equation}\label{eq:Theta-expand-lambda-infty}
	\Psi(x,\mathbf{t};\lambda)=\sum_{i=0}^{\infty}\Psi_i\lambda^{-i},\qquad \Psi^2(x,\mathbf{t};\lambda)=\mathbb{I}_2, \qquad 	\mathbf{V}_n(\lambda;\mathbf{Q})=-\ii \lambda^{n} \sum_{i=0}^{n}\Psi_i\lambda^{-i}.
\end{equation} 
Since the matrix function $\Psi(x,\mathbf{t};\lambda)$ satisfies the stationary zero-curvature equation defined in equation \eqref{eq:zero-curve-equation}, then
matrices $\Psi_i$ can be determined recursively as follows: $\Psi_0=\sigma_3$, $\Psi_{1}=\ii \mathbf{Q}$, and
\begin{subequations}\label{eq:Theta-i-expression}
\begin{align}
	\Psi_{i+1}^{\off}
	&=\frac{\ii \sigma_3}{2}\left(\Psi_{i,x}^{\off} -\left[\mathbf{Q},
	\Psi_{i}^{\diag}\right]\right), \label{eq:Theta-i-expression-off}\\
	\Psi_{i+1}^{\diag}
	&=-\frac{\sigma_3}{2}\sum_{j=1}^{i}\left(\Psi_{j}^{\diag}\Psi_{i+1-j}^{\diag}+\Psi_{j}^{\off}\Psi_{i+1-j}^{\off}\right), \label{eq:Theta-i-expression-diag}
\end{align}
\end{subequations}
$i=1,2,\cdots$, where $\Psi^{\diag}$ denotes the diagonal part of the matrix $\Psi$ and $\Psi^{\off}=\Psi-\Psi^{\diag}$ denotes its off-diagonal part.
The terms $\Psi_i$, $i=2,3,4,5$ are given by
\begin{equation}\label{eq:Psi-i}
\begin{split}
\Psi_2= &\  \frac{\sigma_3}{2}(\mathbf{Q}^2-\mathbf{Q}_x), \\
\Psi_3= &\ \frac{\ii}{4}(2\mathbf{Q}^3-\mathbf{Q}_{xx}+\mathbf{Q}_x\mathbf{Q}-\mathbf{Q}\mathbf{Q}_x), \\
\Psi_4=&\  \frac{\sigma_3}{8}(\mathbf{Q}_{xxx}-6\mathbf{Q}_x\mathbf{Q}^2)+\frac{\sigma_3}{8}(3\mathbf{Q}^4+\mathbf{Q}_x^2-\mathbf{Q}\mathbf{Q}_{xx}-\mathbf{Q}_{xx}\mathbf{Q}), \\
\Psi_5=&\ \frac{\ii}{16}(\mathbf{Q}_{xxxx}-8\mathbf{Q}_{xx}\mathbf{Q}^2-4\mathbf{Q}_{x}^2\mathbf{Q}-6\mathbf{Q}_{x}\mathbf{Q}\mathbf{Q}_{x}-2\mathbf{Q}\mathbf{Q}_{xx}\mathbf{Q}+6\mathbf{Q}^5) \\
&\ -\frac{\ii}{16}(\mathbf{Q}_{xxx}\mathbf{Q}-\mathbf{Q}\mathbf{Q}_{xxx}+6\mathbf{Q}^3\mathbf{Q}_x-6\mathbf{Q}_x\mathbf{Q}^3+\mathbf{Q}_x\mathbf{Q}_{xx}-\mathbf{Q}_{xx}\mathbf{Q}_x).
\end{split}
\end{equation}
Plugging matrices $\Psi_i$ into equation \eqref{eq:Theta-expand-lambda-infty}, we get expressions for matrices $\mathbf{V}_n(\lambda;\mathbf{Q})$. 
In such a case, the compatibility conditions of ordinary differential equations
\begin{equation}\nonumber
	\Phi_{x}(x,\mathbf{t};\lambda)=\mathbf{U}(\lambda;\mathbf{Q})\Phi(x,\mathbf{t};\lambda), \qquad \Phi_{t_n}(x,\mathbf{t};\lambda)=\mathbf{V}_n(\lambda;\mathbf{Q})\Phi(x,\mathbf{t};\lambda),
\end{equation} 
deduce the related integrable nonlinear soliton equation \cite{Wright-19} under different symmetries. 
Moreover, by taking a linear combination of the aforementioned ordinary differential equations with $\hat{t}_{n}=\sum_{i=1}^{n}a_{i}t_i$, we can derive various integrable nonlinear soliton equations. 

In the following work, we consider the third flow of the AKNS system (the \ref{eq:mKdV} equation).
Choosing $t=4t_3$, we obtain the Lax pair \eqref{eq:Lax-pair} and the mKdV equation. 

\newenvironment{proof-prop-matrix-L}{\emph{\textbf{Proof of the \Cref{prop:L-matrix}.}}}{\hfill$\Box$\medskip}
	\begin{proof-prop-matrix-L}
	Since the matrix function $\Psi(x,t;\lambda)$ satisfies the stationary zero-curvature equation \eqref{eq:zero-curve-equation}, the matrices $\Psi_i\equiv\Psi_i(x,t)$, $i=0,1,\cdots$, must satisfy
	\begin{equation}\label{eq:define-Psi-x}
		\begin{split}
			 \ii \Psi_{i,x}= \left[\Psi_{0}, \Psi_{i+1}\right]+ \left[\Psi_{1},\Psi_{i}\right], \quad
			 \ii \Psi_{i,t}=4 \left(\left[\Psi_{0}, \Psi_{i+3}\right]+ \left[\Psi_{1},\Psi_{i+2}\right] + \left[ \Psi_2, \Psi_{i+1}\right] + \left[\Psi_3,\Psi_{i}\right]\right).
		\end{split}
	\end{equation}
	Comparing the definition of the matrix functions $\Psi(x,t;\lambda)$ and $\mathbf{L}(x,t;\lambda)$, we deduce that the $x$-part of stationary zero-curvature equation of $\mathbf{L}(x,t;\lambda)$ satisfies the additional constraint ordinary differential equation:
	\begin{equation}\label{eq:define-x}
		\begin{split}
			\mathbf{L}_x(x,t;\lambda)-[\mathbf{U}(\lambda;\mathbf{Q}),\mathbf{L}(x,t;\lambda)]=-\sum_{j=0}^{g+1}\alpha_j\left[\Psi_0,\Psi_{j+1}\right]=0, \quad \text{and} \quad \sum_{j=0}^{g+1} \alpha_j \Psi_{j+1}^{\off}=0,
		\end{split}
	\end{equation}
	which can deduce the equation \eqref{eq:ode}. 
		
	Then, we would like to prove that the $t$-part of the stationary zero-curvature equation holds automatically (i.e. $\mathbf{L}_t(x,t;\lambda)=[\mathbf{V}(\lambda;\mathbf{Q}),\mathbf{L}(x,t;\lambda)]$).
	Similarly, we consider the coefficients of the spectral parameter $\lambda$.
	For the coefficients of $\lambda^2$ for the quantity $\mathbf{L}_t(x,t;\lambda)-[\mathbf{V}(\lambda;\mathbf{Q}),\mathbf{L}(x,t;\lambda)]=0$, we obtain
	\begin{equation}\nonumber
		\begin{split}
			\mathcal{O}(\lambda^2): \quad 
			& \ -\ii \sum_{j=2}^{g+1}\alpha_j \Psi_{j-2,t}+4\left[ \Psi_1,\sum_{j=0}^{g+1}\alpha_j\Psi_{j}\right]+4\left[ \Psi_2,\sum_{j=1}^{g+1}\alpha_j\Psi_{j-1}\right]+4\left[ \Psi_3,\sum_{j=2}^{g+1}\alpha_j\Psi_{j-2}\right] \xlongequal[\eqref{eq:define-x}]{\eqref{eq:define-Psi-x}}0.
	\end{split}
	\end{equation}
Differentiating both sides of the equation \eqref{eq:define-x} with respect to $x$, we obtain 
\begin{equation}\label{eq:sum-1}
	\begin{split}
		&0\xlongequal{\eqref{eq:define-x}}\ii\sum_{j=0}^{g+1}\alpha_j\Psi_{j+1,x}^{\off}\xlongequal{\eqref{eq:define-Psi-x}} \sum_{j=0}^{g+1}\alpha_j\left([\Psi_0,\Psi_{j+2}]+[\Psi_1,\Psi_{j+1}]\right)^{\off}.
	\end{split}
\end{equation}
Utilizing the above equations, we compute the coefficient of $\lambda$ and get
\begin{equation}\nonumber
	\begin{split}
		\mathcal{O}(\lambda): \qquad 
		& \ -\ii \sum_{j=1}^{g+1}\alpha_j \Psi_{j-1,t}+4\left[ \Psi_2,\sum_{j=0}^{g+1}\alpha_j\Psi_{j}\right]+4\left[ \Psi_3,\sum_{j=1}^{g+1}\alpha_j\Psi_{j-1}\right] 
		\xlongequal[\eqref{eq:define-x}]{\eqref{eq:define-Psi-x},\eqref{eq:sum-1}} 0.
	\end{split}
\end{equation}
Due to the Jacobi-identity $[[\mathbf{A},\mathbf{B}],\mathbf{C}]+[[\mathbf{B},\mathbf{C}],\mathbf{A}]+[[\mathbf{C},\mathbf{A}],\mathbf{B}]=0$, we get
\begin{equation}\nonumber
		\begin{split}
			0=&\ -\sum_{j=0}^{g+1}\alpha_j\Psi_{j+1,xx}^{\off}
			\xlongequal[\eqref{eq:define-x}]{\eqref{eq:define-Psi-x}, \eqref{eq:sum-1}}
			  \sum_{j=0}^{g+1}\alpha_j \left([\Psi_0,[\Psi_0,\Psi_{j+3}]+[\Psi_1,\Psi_{j+2}]+[\Psi_2,\Psi_{j+1}]]\right)^{\off}.
		\end{split}
\end{equation}
From equation \eqref{eq:define-x}, it is easy to get
\begin{equation}\nonumber
		\begin{split}
			0\xlongequal{\eqref{eq:define-x}}&
		\sum_{j=0}^{g+1}\alpha_j	\left(\left[\sigma_3\Psi_{j+1,x},\Psi_{1}\right]+\left[\Psi_{1,x},\sigma_3\Psi_{j+1}\right]\right)^{\diag}
			\xlongequal{\eqref{eq:define-Psi-x}}
			 \sum_{j=0}^{g+1}\alpha_j(2[ \Psi_{j+2},\Psi_{1}]+2[\Psi_{j+1},\Psi_2])^{\diag}.
		\end{split}
	\end{equation}
	So, we get
	\begin{equation}\nonumber
		\begin{split}
			\mathcal{O}(1): \quad 
			& \ -\ii \sum_{j=0}^{g+1}\alpha_j \Psi_{j,t}+4\left[ \Psi_3,\sum_{j=0}^{g+1}\alpha_j\Psi_{j}\right] 
			\xlongequal{\eqref{eq:define-Psi-x}} 4\sum_{j=0}^{g+1}\alpha_j  \left(\left[\Psi_{j+3}, \Psi_{0}\right]+\left[\Psi_{j+2}, \Psi_{1}\right]+\left[\Psi_{j+1}, \Psi_{2}\right]\right) \\
			\xlongequal{\eqref{eq:sum-1}}&	\  4\sum_{j=0}^{g+1}\alpha_j  \left(\left[\Psi_{j+3}, \Psi_{0}\right]+\left[\Psi_{j+2}, \Psi_{1}\right]+\left[\Psi_{j+1}, \Psi_{2}\right]\right)^{\off}+ \left(\left[\Psi_{j+2}, \Psi_{1}\right]+\left[\Psi_{j+1}, \Psi_{2}\right]\right)^{\diag}=0.
		\end{split}
	\end{equation}
	Thus the \Cref{prop:L-matrix} holds.
\end{proof-prop-matrix-L}

In another viewpoint, the $x$-part of the stationary zero-curvature equation defined in equation \eqref{eq:define-x} could be re-expressed as
\begin{equation}\nonumber
	\begin{split}
		\mathbf{L}_x-[\mathbf{U}(\lambda;\mathbf{Q}),\mathbf{L}]\xlongequal[\eqref{eq:Theta-expand-lambda-infty}]{\eqref{eq:L-matrix-genus-two}} &\
		\sum_{i=0}^{g+1}\alpha_i(\mathbf{V}_{i,x}(\lambda;\mathbf{Q})-[\mathbf{U}(\lambda;\mathbf{Q}),\mathbf{V}(\lambda;\mathbf{Q})])
		= \sum_{i=0}^{g+1}\alpha_i \mathbf{Q}_{t_i}=0.
	\end{split}
\end{equation}
When $g+1=n$, letting $\alpha_i=c_{2,i}$, $i = 0, 1,\cdots, g$, $\alpha_{g+1} = 1$, the ordinary differential equation \eqref{eq:ode} is equivalent to the right-hand side of equation \eqref{eq:H-mKdV}.

\subsection{The algebro-geometric approach}\label{sec:algebro-geometric-approach}
The algebro-geometric approach has been developed over several decades \cite{BelokolosBEI-94,FritzH-2003-Algebro-Geometric}. 
In this subsection, we introduce this method and apply it to construct two-phase solutions for the \ref{eq:mKdV} equation.

Suppose that $\pm \ii y$ are two eigenvalues of the matrix function $\mathbf{L}$ defined in equation \eqref{eq:L-matrix-genus-two}, which implies that the equation $\det\left(\pm \ii y -\mathbf{L}\right)=0$ holds. 
Eigenvectors of corresponding eigenvalues $\pm \ii y$ are $[1, r_1]^{\top}$ and $[1, r_2]^{\top}$ respectively, where the fundamental meromorphic functions $r_{1,2}$ on $\mathcal{R}_2$ are defined as 
\begin{equation}\label{eq:define-r-12}
	\begin{split}
		&\ r_1(x,t;P)=\frac{\mathbf{L}_{21}(x,t;\lambda)}{\ii y+\mathbf{L}_{11}(x,t;\lambda)}=\frac{\ii y-\mathbf{L}_{11}(x,t;\lambda)}{\mathbf{L}_{12}(x,t;\lambda)}, \\
		&\ r_2(x,t;P)=\frac{\mathbf{L}_{21}(x,t;\lambda)}{-\ii y+\mathbf{L}_{11}(x,t;\lambda)}=\frac{-\ii y-\mathbf{L}_{11}(x,t;\lambda)}{\mathbf{L}_{12}(x,t;\lambda)}, 
	\end{split}
\end{equation}
where $P=(\lambda,y)\in \mathcal{R}_2$ is defined in equation \eqref{eq:define-p-y}.
Furthermore, the eigenvectors corresponding to the eigenvalues  $\pm \ii y$ can also be expressed in terms of the solution of the Lax pair \eqref{eq:Lax-pair} as $\Phi(x,t;\lambda)[1, r_1(0,0;P)]^{\top}$ and $\Phi(x,t;\lambda)[1, r_2(0,0;P)]^{\top}$, which can refer to \cite{LingS-23-mKdV-stability}.
Comparing these different representations for the kernels of matrices $\mathbf{L}\pm \ii y $, we deduce that functions $\Phi_{ij}(x,t;P)$, $i,j=1,2$, satisfy equations:
\begin{equation}\label{eq:relation-Phi-r}
	\Phi_{2i}(x,t;P)=r_i(x,t;P)\Phi_{1i}(x,t;P), \qquad i=1,2.
\end{equation}
Their details are provided in \cite{FengLT-20,LingS-23-mKdV-stability}.
Based on the above relations, we would like to utilize the algebro-geometric method to construct two-phase solutions of the mKdV equation and their wave functions explicitly.

Building on above facts, we aim to find the relationship among functions $r_{1,2}(x,t;\lambda)$, $u(x,t)$ and $\Phi_{ij}(x,t;\lambda)$ to provide the explicit expressions of the solution $u(x,t)$ of the \ref{eq:mKdV} equation and the fundamental solution $\Phi(x,t;\lambda)$ of the related Lax pair \eqref{eq:Lax-pair}.
Here, we just study the function $r_1(x,t;P)$, since $r_2(x,t;P)$ can be obtained by a shift of sheets on the Riemann surface.
As $\lambda\rightarrow \infty$, based on equation \eqref{eq:define-curve-algebro}, 
the expansion of the parameter $y$ could be expressed as 
\begin{equation}\label{eq:y-limitation}
	y=\pm \left(\lambda^3-v\lambda/4\right)+\mathcal{O}(\lambda^{-1}), \quad \text{as} \qquad P \rightarrow \infty^{\pm},
\end{equation}
where $v$ is defined in equation \eqref{eq:solutions-dn} and $P$ is defined in equation \eqref{eq:define-standard-projection}.
Combining the parameter $y$ in equation \eqref{eq:y-limitation} with the function $r_{1}(x,t;\lambda)$ in equation \eqref{eq:define-r-12} and $\mathbf{L}_{ij}$, $i,j=1,2$, in equation \eqref{eq:L-elements}, we obtain that as $\lambda\rightarrow \infty$, the function $r_{1}\equiv  r_{1}(x,t;\lambda)$ has 
\begin{equation}\label{eq:asy-r-y-infty}
		r_{1}(x,t;P)=\left\{
		\begin{aligned}
			\frac{ \ii y-\mathbf{L}_{11}}{\mathbf{L}_{12}}=&\ \frac{2\ii }{u}\lambda+\frac{ u_x}{u^2}+\ii\frac{(u^4+uu_{xx}-u_x^2)}{2u^3\lambda}-u_x\frac{4\alpha_1-3u^2}{4u^2\lambda^2}\\
			&-u_x\frac{u_x^2-2uu_{xx}}{4u^4\lambda^2}+\mathcal{O}(\lambda^{-3}),\qquad \text{as}\quad 
			P \rightarrow \infty^{+},\!\!\! \\
			\frac{\mathbf{L}_{21}}{ \ii y+\mathbf{L}_{11}}=& -\frac{\ii u}{2\lambda}-\frac{ u_x}{4\lambda^2}+\mathcal{O}(\lambda^{-3}), \qquad \text{as}\quad 
			P \rightarrow \infty^{-}.
		\end{aligned}\right.	
\end{equation}
Together with equations \eqref{eq:relation-Phi-r} and \eqref{eq:asy-r-y-infty}, it is easy to obtain that functions $\Phi_{i}\equiv \Phi_{i1}(x,t;\lambda)$, $i=1,2$, have the following relations:
\begin{equation}\label{eq:Phi-limitation}
	\begin{split}
		&
		\left\{ \begin{aligned}
			&\frac{\Phi_{1,x}}{\Phi_{1}} = - \ii \lambda+ \mathcal{O}(\lambda^{-1}),\\
			&\frac{\Phi_{1,t}}{\Phi_{1}} = - 4\ii\lambda^3+ \mathcal{O}(\lambda^{-1}),
		\end{aligned}\right.  \quad \qquad \text{as} \quad  P \rightarrow \infty^{-};\\
		&
		\left\{ \begin{aligned}
			&\frac{\Phi_{1,x}}{\Phi_{1}} = \ii \lambda+\frac{u_x}{u}+ \mathcal{O}(\lambda^{-1}), \\
			&\frac{\Phi_{1,t}}{\Phi_{1}} = 4\ii\lambda^3+\frac{u_t}{u}+ \mathcal{O}(\lambda^{-1}),
		\end{aligned} \right. \qquad \text{as}  \quad P \rightarrow \infty^{+};
	\end{split}
\end{equation} 
since $\mathbf{Q}_t+4\alpha_1\mathbf{Q}_x=0$ as deduced from  equation \eqref{eq:ode}.
When we consider the function $r_2(x,t;\lambda)$, we could obtain the solution $\Phi_{i2}(x,t;\lambda)$.
Combining equations \eqref{eq:relation-Phi-r}-\eqref{eq:Phi-limitation} with the related Lax pair \eqref{eq:Lax-pair}, it follows that

\begin{subequations}\label{eq:limitation-Phi}
    \begin{align}
	\text{as}\quad P\rightarrow \infty^{-},\qquad 
	&\begin{aligned}
	\begin{bmatrix}
	\Phi_{1}  \\
	\Phi_{2}
	\end{bmatrix}
	=
	\Phi_{1,0}^{-}\left(\begin{bmatrix}
			1\\
			0
		\end{bmatrix}+\mathcal{O}(\lambda^{-1})\right)\ee^{-\ii  \lambda x-4\ii \lambda^3 t},
		\end{aligned} \label{eq:limitation-Phi-n}\\
	\text{as}\quad P\rightarrow \infty^{+},\qquad 
		&\begin{aligned}	
			\begin{bmatrix}
				\Phi_{1}\\
				\Phi_{2}
			\end{bmatrix}=
			&\ \Phi_{1,0}^{+}
			\left(\begin{bmatrix}
				u\\[3pt]
				r_1u
			\end{bmatrix}+\mathcal{O}(\lambda^{-1})\right)\ee^{\ii  \lambda x+4\ii \lambda^3 t}\\
			\xlongequal{\eqref{eq:asy-r-y-infty}}&\ 2\ii \lambda\Phi_{1,0}^{+}\left(\begin{bmatrix}
				0   \\ 	1	\end{bmatrix}+\mathcal{O}(\lambda^{-1})\right)\ee^{\ii  \lambda x+4\ii \lambda^3 t}, 
		\end{aligned}\label{eq:limitation-Phi-p}
	\end{align}
\end{subequations}
where 
\begin{equation}\label{eq:Phi-lim-infty-initial}
	\Phi_{1,0}^{\pm}=\lim_{P\rightarrow \infty^{\pm}}\Phi_{1}(x_0,t_0;P)\ee^{\pm\ii\lambda(x_0+4\lambda^2 t_0)}, \qquad 
	u_0=u(x_0,t_0), 
\end{equation}
and $(x_0,t_0)\in \mathbb{R}^2$ is called the initial point. 

Then we would like to introduce the Abel maps and elliptic integrals to obtain two-phase solution.
The divisor of the function $r_{1}(x,t;P)$ in \eqref{eq:define-r-12} is 
\begin{equation}\label{eq:div-r}
	\mathcal{D}iv\left(r_{1}(x,t;P)\right)=\mathcal{D}_{P_{\infty^{-}},\mu^*}(\lambda)-\mathcal{D}_{P_{\infty^{+}},\mu}(\lambda), 
\end{equation}
	where the abbreviations are given by $\mu=\left\{\hat{\mu}_1,\hat{\mu}_2\right\}\in \mathrm{Sym}^2(\mathcal{R}_2)$, $\mu^*=\left\{\hat{\mu}_1^*,\hat{\mu}_2^*\right\} \in \mathrm{Sym}^2(\mathcal{R}_2)$,
with $\hat{\mu}_i=(\mu_i,-\ii \mathbf{L}_{11}(x,t;\mu_i))\in \mathcal{R}_2$, 
$\hat{\mu}_i^*=(\mu_i^*,\ii \mathbf{L}_{11}(x,t;\mu_i^*))\in \mathcal{R}_2$, $i=1,2$ and $P_{\infty^{\pm}}$ denote the points at infinity.
The equation \eqref{eq:div-r} reveals that $P_{\infty^{-}}, \hat{\mu}_1^*,\hat{\mu}_2^*$ are three zeros and $P_{\infty^{+}}, \hat{\mu}_1,\hat{\mu}_2$ are three poles of the function $r_1(x,t;P)$.
Thus, by the Riemann-Roch theorem \cite{FritzH-2003-Algebro-Geometric}, the function $r_1(x,t;P)$ can be expressed as:
\begin{equation}\label{eq:define-r1}
	r_1(x,t;P)= \hat{r}_1(x,t)\frac{\Theta(\mathcal{C}+\mathcal{A}_{P_0}(P)-\alpha_{P_0}(\mathcal{D}_{\mu^*}))}{\Theta(\mathcal{C}+\mathcal{A}_{P_0}(P)-\alpha_{P_0}(\mathcal{D}_{\mu}))}\ee^{\Omega_1(P) -\ln(\omega_1 )}, \quad i=1,2,
\end{equation}
in which $\Theta(\cdot)$ is the Riemann theta function defined in \Cref{define:Riemann-Theta-function}; the Abel map $\mathcal{A}_{P_0}(P)$ is defined in equation \eqref{eq:abel_map};  zeros and poles divisors of the meromorphic function $r_1(x,t;P)$ are shown in equation \eqref{eq:div-r};
the function $\hat{r}_1(x,t)$ is independent of the spectral parameter $\lambda$; and the parameter $\mathcal{C}$ is called the Riemann constant \cite[p.41]{BelokolosBEI-94}.
The definition of the Jacobi theta function is defined in the \Cref{define:theta}. 
Furthermore, combining equations \eqref{eq:define-int-w1} and \eqref{eq:define-r1}, the function $r_1(x,t;P)$ would deduce into 
\begin{equation}\label{eq:r12-P-pm-infty}
	\begin{split}
		&r_1(x,t;P)
		\xlongequal{\eqref{eq:define-r1}}\left\{
	\begin{aligned}
		&\hat{r}_1(x,t)\frac{\Theta(\mathcal{C}+\mathcal{A}_{P_0}(\infty^{+})-\alpha_{P_0}(\mathcal{D}_{\mu^*}))}{\Theta(\mathcal{C}+\mathcal{A}_{P_0}(\infty^{+})-\alpha_{P_0}(\mathcal{D}_{\mu}))}\lambda+\mathcal{O}(1),
		\qquad \quad 	P\rightarrow \infty^{+},\\
		&\frac{\hat{r}_1(x,t)\Theta(\mathcal{C}+\mathcal{A}_{P_0}(\infty^{-})-\alpha_{P_0}(\mathcal{D}_{\mu^*}))}{\omega_1^2\Theta(\mathcal{C}+\mathcal{A}_{P_0}(\infty^{-})-\alpha_{P_0}(\mathcal{D}_{\mu}))}\frac{1}{\lambda }+\mathcal{O}(\lambda^{-2}),\qquad  \,
	 	P\rightarrow \infty^{-}.
	\end{aligned}\right.
	\end{split}	
\end{equation}

Without loss of generality, we set the initial point $(x_0,t_0)=(0,0)$, $\Phi_{1,0}^{-}=1$ from equation \eqref{eq:Phi-lim-infty-initial}.
By equations \eqref{eq:Lax-pair} and \eqref{eq:relation-Phi-r}, functions $\Phi_{i}(x,t;P)$, $i=1,2,$ are meromorphic functions on $\mathcal{R}_2\backslash\left\{P_{\infty^{+}},P_{\infty^{-}}\right\}$, except at the poles $\mu$ of the function $r_1(x,t;P)$.
Furthermore, by equations \eqref{eq:define-int-w2}, \eqref{eq:define-int-w3}, \eqref{eq:limitation-Phi-n} and \eqref{eq:define-r1}, functions $\Phi_{i}$, $i=1,2$, can be expressed as
\begin{equation}\label{eq:define-Phi}
	\begin{split}
		\Phi_{1}
		= & \ \frac{\Theta(\mathcal{C}+\mathcal{A}_{P_0}(P)-\alpha_{P_0}(\mathcal{D}_{\mu}))\Theta(\mathcal{C}+\mathcal{A}_{P_0}(\infty^{-})-\alpha_{P_0}(\mathcal{D}_{\mu_0}))}{\Theta(\mathcal{C}+\mathcal{A}_{P_0}(P)-\alpha_{P_0}(\mathcal{D}_{\mu_0}))\Theta(\mathcal{C}+\mathcal{A}_{P_0}(\infty^{-})-\alpha_{P_0}(\mathcal{D}_{\mu}))}
		\ee^{\ii \left(\omega_2+ \Omega_2(P) \right)x+\ii\left(\omega_3+ \Omega_3(P)\right)t},\\
		\Phi_{2}
		= & \  \frac{\hat{r}_1(x,t)\Theta(\mathcal{C}+\mathcal{A}_{P_0}(P)-\alpha_{P_0}(\mathcal{D}_{\mu^*}))\Theta(\mathcal{C}+\mathcal{A}_{P_0}(\infty^{-})-\alpha_{P_0}(\mathcal{D}_{\mu_0}))}{\omega_1\Theta(\mathcal{C}+\mathcal{A}_{P_0}(P)-\alpha_{P_0}(\mathcal{D}_{\mu_0}))\Theta(\mathcal{C}+\mathcal{A}_{P_0}(\infty^{-})-\alpha_{P_0}(\mathcal{D}_{\mu}))}\\
		& \ \ 
		\ee^{\ii \left(\omega_2+\Omega_2(P) \right)x+\ii\left(\omega_3+\Omega_3(P) \right)t+\Omega_1(P)},
	\end{split}
\end{equation}
where $\mu_0=\mu(x_0,t_0)$, functions $\Omega_{i}(P)$ and parameters $\omega_i$, $i=1,2,3$, are defined in equations \eqref{eq:define-int-w1}-\eqref{eq:define-int-w3}.
Considering $P\rightarrow \infty^{+}$, we get 
\begin{equation}\nonumber
	\begin{split}
		& \Phi_{1}(x,t;P) \\
		\xlongequal{\eqref{eq:define-Phi}} &  \  \frac{\Theta(\mathcal{C}+\mathcal{A}_{P_0}(\infty^{+})-\alpha_{P_0}(\mathcal{D}_{\mu}))\Theta(\mathcal{C}+\mathcal{A}_{P_0}(\infty^{-})-\alpha_{P_0}(\mathcal{D}_{\mu_0}))}{\Theta(\mathcal{C}+\mathcal{A}_{P_0}(\infty^{+})-\alpha_{P_0}(\mathcal{D}_{\mu_0}))\Theta(\mathcal{C}+\mathcal{A}_{P_0}(\infty^{-})-\alpha_{P_0}(\mathcal{D}_{\mu}))}\ee^{\ii (\lambda+2\omega_2 )x+\ii (4\lambda^3+ 2 \omega_3 )t}
		+\mathcal{O}(\lambda^{-1}).
	\end{split}
\end{equation}
Combining the above equation with the equation \eqref{eq:limitation-Phi-p}, we obtain the explicit expression of the solution 
\begin{equation}\label{eq:q-1-q0}
	u(x,t)= \frac{\Theta(\mathcal{C}+\mathcal{A}_{P_0}(\infty^{+})-\alpha_{P_0}(\mathcal{D}_{\mu}))\Theta(\mathcal{C}+\mathcal{A}_{P_0}(\infty^{-})-\alpha_{P_0}(\mathcal{D}_{\mu_0}))}{\Phi_{1,0}^{+}\Theta(\mathcal{C}+\mathcal{A}_{P_0}(\infty^{+})-\alpha_{P_0}(\mathcal{D}_{\mu_0}))\Theta(\mathcal{C}+\mathcal{A}_{P_0}(\infty^{-})-\alpha_{P_0}(\mathcal{D}_{\mu}))}\ee^{2\ii \omega_2 x+2\ii \omega_3 t},
\end{equation}
where $\Phi_{1,0}^{+}$ is defined in equation \eqref{eq:Phi-lim-infty-initial}.
Letting $P \rightarrow \infty^{+}$ and together with equations \eqref{eq:asy-r-y-infty}, \eqref{eq:r12-P-pm-infty} and \eqref{eq:q-1-q0}, we obtain
\begin{equation}\label{eq:r-0}
	\begin{split}
		\hat{r}_1(x,t)=
		2\ii\Phi_{1,0}^{+}\frac{ \Theta(\mathcal{C}+\mathcal{A}_{P_0}(\infty^{+})-\alpha_{P_0}(\mathcal{D}_{\mu_0}))\Theta(\mathcal{C}+\mathcal{A}_{P_0}(\infty^{-})-\alpha_{P_0}(\mathcal{D}_{\mu}))}
		{\Theta(\mathcal{C}+\mathcal{A}_{P_0}(\infty^{+})-\alpha_{P_0}(\mathcal{D}_{\mu^*}))\Theta(\mathcal{C}+\mathcal{A}_{P_0}(\infty^{-})-\alpha_{P_0}(\mathcal{D}_{\mu_0}))}\ee^{-2\ii \omega_2 x-2\ii \omega_3 t}.
	\end{split}
\end{equation}
Utilizing equations \eqref{eq:asy-r-y-infty}, \eqref{eq:r12-P-pm-infty} and \eqref{eq:r-0}, we obtain 
\begin{equation}\label{eq:u-2}
	\!\! u(x,t)=\frac{(2\ii)^2\Phi_{1,0}^{+} \Theta(\mathcal{C}+\mathcal{A}_{P_0}(\infty^{+})-\alpha_{P_0}(\mathcal{D}_{\mu_0}))\Theta(\mathcal{C}+\mathcal{A}_{P_0}(\infty^{-})-\alpha_{P_0}(\mathcal{D}_{\mu^*}))}{\omega_1^2\Theta(\mathcal{C}+\mathcal{A}_{P_0}(\infty^{+})-\alpha_{P_0}(\mathcal{D}_{\mu^*}))\Theta(\mathcal{C}+\mathcal{A}_{P_0}(\infty^{-})-\alpha_{P_0}(\mathcal{D}_{\mu_0}))}\ee^{-2\ii \omega_2 x-2\ii \omega_3 t}.
\end{equation}
If the solutions $u(x,t)$ given in equation \eqref{eq:q-1-q0} and \eqref{eq:u-2} are equal to each other and are real ones, then the functions $u(x,t)$ are the solutions of the mKdV equation. 
Then, we are going to provide relations between these two functions, which is the crucial step to get the two-phase solutions of mKdV equation.
Based on the definition of the Abel map, the following conditions hold:
\begin{itemize}
	\item The Abel map linearizes the auxiliary divisors \cite{FritzH-2003-Algebro-Geometric}:
	\begin{equation}\label{eq:D-mu-define}
		\alpha_{P_0}(\mathcal{D}_{\mu})=\alpha_{P_0}(\mathcal{D}_{\mu_0})-\ii U(x-x_0)-\ii V(t-t_0),
	\end{equation} 
	since $\partial_x\left(\alpha_{P_0}(\mathcal{D}_{\mu})\right)=-\ii U$ and $\partial_t\left(\alpha_{P_0}(\mathcal{D}_{\mu})\right)=-\ii V$.
	\item By equation \eqref{eq:div-r}, $\mathcal{D}_{P_{\infty^-},\mu^*}(\lambda)$ and $\mathcal{D}_{P_{\infty^+},\mu}(\lambda)$ are linearly equivalent \cite{CaliniI-05-VFE,FritzH-2003-Algebro-Geometric}, which implies
	\begin{equation}\label{eq:relation-mu}
		\alpha_{P_0}(\mathcal{D}_{\mu^*}) = \alpha_{P_0}(\mathcal{D}_{\mu})+\Delta, \qquad \Delta=\mathcal{A}_{\infty^-}(\infty^{+}).
	\end{equation}
	\item The matrices $U$ and $V$ can be expressed as
	\begin{equation}\label{eq:U-V-define}
		U=\left(U_1,U_2\right),\quad U_i=\oint_{b_i}\dd \Omega_2, \quad 
		V=\left(V_1,V_2\right),\quad V_i=\oint_{b_i}\dd \Omega_3, \quad 
		i=1,2,
	\end{equation}
	as proved in \cite{BelokolosBEI-94,CaliniI-05-VFE,FritzH-2003-Algebro-Geometric}. Functions $\dd \Omega_{2,3}$ are defined in equations \eqref{eq:define-int-w2} and \eqref{eq:define-int-w3}.
	\item Parameters $\omega_2$ and $\omega_3$ \cite{CaliniI-05-VFE} could be obtained by 
	\begin{equation}\label{eq:omega-2-3-define}
		\omega_2 =-\frac{\ii}{4\pi} \sum_{i=1}^{2}U_i\oint_{a_i}\frac{\lambda^2}{y}\dd \lambda, \qquad
		\omega_3 =-\frac{\ii}{4\pi} \sum_{i=1}^{2}V_i\oint_{a_i}\frac{\lambda^2}{y}\dd \lambda.
	\end{equation}
	\item  The parameter $\omega_1$  \cite{CaliniI-05-VFE} defined in equation \eqref{eq:define-int-w1} satisfies the equation
	\begin{equation}\label{eq:omega-1-define}
		\ln(\omega_1 )=-\ln(\lambda_3)+\lim_{ P\rightarrow \infty^{+}}\int_{\Gamma_p}\dd\Omega_1-\frac{\dd \lambda}{\lambda}+\ii \pi,
	\end{equation}
	where $\Gamma_p$ is the integration path from $P_0$ to $P$.
\end{itemize}
The detailed proofs of above conclusions were provided in \cite{BelokolosBEI-94,CaliniI-05-VFE,FritzH-2003-Algebro-Geometric}.
If functions $\hat{\Omega}$ and $\tilde{\Omega}$ are holomorphic on  $\mathcal{R}_2$, then by applying Stokes' Theorem we obtain
	\begin{equation}\nonumber
		0=\int_{\mathcal{R}_2}\dd (\hat{\Omega}\dd \tilde{\Omega})=\oint_{\partial \mathcal{R}_2}  \hat{\Omega}\dd \tilde{\Omega}=\sum_{i=1}^{2} \oint_{a_i} \dd \hat{\Omega} \oint_{b_i} \dd \tilde{\Omega}-\oint_{b_i} \dd \hat{\Omega} \oint_{a_i} \dd \tilde{\Omega}.
	\end{equation}
	If $\hat{\Omega}\dd \tilde{\Omega}$ has singularities at $\pm\infty$, the above equation should be modified as
	\begin{equation}\label{eq:Stokes-int}
		0=\sum_{i=1}^{2} \oint_{a_i} \dd \hat{\Omega} \oint_{b_i} \dd \tilde{\Omega}-\oint_{b_i} \dd \hat{\Omega} \oint_{a_i} \dd \tilde{\Omega}+\oint_{l+\boldsymbol{\pi}(l)}\hat{\Omega}\dd \tilde{\Omega},
	\end{equation}
	where the curve $l+\boldsymbol{\pi}(l)$ corresponds to clockwise cycles around $\infty^{\pm}$.
	As $\lambda \rightarrow \pm \infty$, expanding the function $1/y$ yields
	\begin{equation}\nonumber
		\begin{split}
			\frac{1}{y}=\pm\left(\frac{1}{\lambda^3}+\frac{v}{4\lambda^5}+\mathcal{O}\left(\frac{1}{\lambda^7}\right)\right), \qquad \lambda \rightarrow \pm \infty.
		\end{split}
	\end{equation}
	Setting $\hat{\Omega}=\Omega_2$ and $\dd \tilde{\Omega}=\lambda^2/y \dd \lambda$ and combining with equations \eqref{eq:define-int-w2} and \eqref{eq:Stokes-int}, we obtain 
	\begin{equation}\nonumber
		0\xlongequal[\eqref{eq:U-V-define},\eqref{eq:define-int-w2}]{\eqref{eq:Stokes-int}} 4\pi \ii \omega_2-\sum_{i=1}^2U_i \oint_{a_i}\frac{\lambda^2}{y}\dd \lambda,
	\end{equation}
	which gives the first equation in \eqref{eq:omega-2-3-define}. 
	Similarly, we obtain the second equation in \eqref{eq:omega-2-3-define}.
	If we instead set $\dd \tilde{\Omega}=w_{1,2}\dd \lambda$ in \eqref{eq:define-basis-B}, it is straightforward to verify that
	\begin{equation}\label{eq:U-v}
		0\xlongequal[\eqref{eq:define-basis-B},\eqref{eq:U-V-define}]{\eqref{eq:Stokes-int}}4\pi \ii d_{j1}-\sum_{i=1}^2U_i \oint_{a_i}w_j\dd \lambda\xlongequal{\eqref{eq:define-basis-B}}4\pi \ii d_{j1}-2\pi \ii \sum_{i=1}^2U_i \delta_{ij}.
	\end{equation}
	Thus, we obtain $U_i=2d_{i1}$ for $i=1,2$.
	Similarly, taking $\hat{\Omega}=\Omega_3$ yields
	\begin{equation}\label{eq:V-v}
		0=4\pi \ii v d_{j1}-\sum_{i=1}^2V_i \oint_{a_i}w_j\dd \lambda = 4\pi \ii v d_{j1}-2\pi \ii \sum_{i=1}^2V_i \delta_{ij}.
	\end{equation}
	Therefore, in combination with equations \eqref{eq:U-v} and \eqref{eq:V-v}, we obtain 
	\begin{equation}\label{eq:U-V-v}
		V_i\xlongequal[\eqref{eq:V-v}]{\eqref{eq:U-v}}vU_i=2vd_{i1}.
\end{equation}

Without loss of generality, setting the initial point $(x_0,t_0)=(0,0)$ and
\begin{equation}\label{eq:define-D}
	\mathcal{C}+\mathcal{A}_{P_0}(\infty^{-})-\alpha_{P_0}(\mathcal{D}_{\mu_0})=D,
\end{equation}
it is easy to obtain
\begin{equation}\label{eq:define-D-1}
	\mathcal{C}+\mathcal{A}_{P_0}(\infty^{+})-\alpha_{P_0}(\mathcal{D}_{\mu_0})=\mathcal{C}+\mathcal{A}_{P_0}(\infty^{-})+\Delta-\alpha_{P_0}(\mathcal{D}_{\mu_0})=D+\Delta,
\end{equation} 
where $\mathcal{A}_{P_0}(\infty^{\pm})$ is called the Abel map defined in equation \ref{eq:abel_map} and $\Delta$ is defined in equation \eqref{eq:relation-mu}.
By combining with equations \eqref{eq:q-1-q0},\eqref{eq:D-mu-define},\eqref{eq:define-D} and \eqref{eq:define-D-1},
we obtain that the solution $u(x,t)$ in equations \eqref{eq:q-1-q0} and \eqref{eq:u-2} can be rewritten as \eqref{eq:q1-q0} and 
\begin{equation}\label{eq:C_u0-define}
	 u(x,t)\xlongequal[\eqref{eq:define-D},\eqref{eq:define-D-1}]{\eqref{eq:u-2},\eqref{eq:D-mu-define}} \frac{ -4\Theta(D-\Delta+\ii Ux+\ii V t)}{\omega_1^2C_{u_0} \Theta(D+\ii Ux+\ii V t)}\ee^{-2\ii \omega_2 x-2\ii \omega_3 t},
\end{equation}
with $C_{u_0}$ defined in equation \eqref{eq:q1-q0}, respectively.
Together with equations  \eqref{eq:define-Phi},\eqref{eq:r-0}, \eqref{eq:D-mu-define},\eqref{eq:relation-mu}, \eqref{eq:define-D}, \eqref{eq:define-D-1}, and \eqref{eq:C_u0-define}, the simple form of functions $\Phi_{i}=\Phi_{i}(x,t;P)$, $i=1,2$, defined in equation \eqref{eq:define-Phi} are expressed in equation  \eqref{eq:Phi-expression}.
Then, we calculate the explicit expressions of the above equations by providing suitable formulas and utilizing the properties of hyperelliptic integrals.

\subsection{The two-phase solutions of the mKdV equation}\label{section:exact-solution}

In this subsection, we aim to represent the two-phase solutions of the \ref{eq:mKdV} equation exactly for two different cases. 
To analyze the stability of the two-phase solutions of mKdV equation, we should determine the parameters in the Riemann theta functions exactly, which are related to the hyperelliptic integrals.
To address this issue, appropriate transformations will be introduced to convert the associated hyperelliptic integrals into the three canonical forms of elliptic integrals provided in \Cref{define:elliptic-function}. 
According to the two possible configurations of the branch points described in equation \eqref{eq:det-L-lambda}, the parameters must satisfy \ref{case1} or \ref{case2}.
The corresponding homology bases for these two cases are depicted in \Cref{fig:genus-two-figure}.

In the following, we investigate the hyperelliptic integrals involved in the formulas of two-phase solutions and their corresponding vector solutions of the Lax pair.
For clarity and convenience, we introduce the following notations:
\begin{equation}\label{eq:Y-define}
	\mathcal{Y}_{a_i}^{[j]}=\oint_{a_i}\frac{\lambda^{j}}{y}\dd \lambda, \qquad \mathcal{Y}_{b_i}^{[j]}=\oint_{b_i}\frac{\lambda^{j}}{y}\dd \lambda, \qquad i=1,2, \quad j=0,1,2,\cdots,5. 
\end{equation}
The hyperelliptic integrals described above can be categorized into the two cases presented in equations \eqref{eq:hyper-1} and \eqref{eq:hyper-2}. 
Furthermore, we demonstrate that all of these integrals can be expressed in terms of the three canonical forms of elliptic integrals defined in \Cref{define:elliptic-function}. 
A key step is to apply suitable transformations reducing general hyperelliptic integrals to standard forms of elliptic integrals. 
Specifically, by introducing suitable transformations between the parameter $\Lambda$ and elliptic functions 
$\sn(\nu,k)$, $\cn(\nu,k)$, and $\dn(\nu,k)$, the integrals in equations \eqref{eq:hyper-1} and \eqref{eq:hyper-2} can be rewritten as rational functions of these elliptic functions.
Further, using the three canonical forms of elliptic integrals, we are able to express the hyperelliptic integral given in equation \eqref{eq:Y-define} in terms of standard elliptic integrals.

In the subsequent analysis, we focus on constructing the appropriate transformations, which are provided in equations \eqref{eq:fs-2}, \eqref{eq:fs-1}, \eqref{eq:fs-4} and \eqref{eq:fs-3}, that enable us to express the hyperelliptic integral \eqref{eq:Y-define} in terms of the three standard forms of elliptic integrals, for two distinct configurations of branch points.

\vspace{0.1cm}
\noindent $\bullet$\quad
\textbf{\large \ref{case1}}

Without loss of generality, we assume that $\Im(\lambda_3)>\Im(\lambda_2)>\Im(\lambda_1)$.
\Cref{appendix:map} provides a detailed introduction to the method for calculating elliptic integrals. 
\Cref{prop:elliptic-int-1} presents the transformation that converts the hyperelliptic integrals into the canonical form. 
From the \Cref{prop:case-1-P} and \Cref{prop:case-2-P}, we derive explicit expressions for hyperelliptic integrals along the $a_{1,2}$- and $b_{1,2}$-circles defined in equation \eqref{eq:hyper}.
Accordingly, the explicit values of these hyperelliptic integrals in \eqref{eq:hyper} are obtained as follows.

\begin{lemma}\label{lemma:dn-1-int}
	When $j=0,2$, the hyperelliptic integrals defined in equation \eqref{eq:Y-define}, along the $a_i$-circle and $b_i$-circle, $i=1,2$ are
	\begin{equation}\label{eq:Y-dn-1}\nonumber
		\begin{split}
			&\mathcal{Y}_{a_1}^{[0]}=-\mathcal{Y}_{a_2}^{[0]}=\frac{2\ii   K^{(1)}_1}{\lambda_2\sqrt{\lambda_1^2-\lambda_3^2}}, \qquad
			\mathcal{Y}_{a_1}^{[2]}=-\mathcal{Y}_{a_2}^{[2]}=\frac{2\ii\lambda_1^2 \Pi((\lambda_2^2-\lambda_1^2)/\lambda_2^2,  k^{(1)}_1)}{\lambda_2\sqrt{\lambda_1^2-\lambda_3^2}},\\
			&\mathcal{Y}_{b_2}^{[0]}=-\mathcal{Y}_{b_1}^{[0]}=\frac{2   K^{(1)\prime}_1}{\lambda_2\sqrt{\lambda_1^2-\lambda_3^2}},
			\qquad
			\mathcal{Y}_{b_2}^{[2]} =\frac{2  \lambda_3^2\Pi((\lambda_2^2-\lambda_3^2)/\lambda_2^2,  k^{(1)\prime}_1)}{\lambda_2\sqrt{\lambda_1^2-\lambda_3^2}}, \\
			&
			\mathcal{Y}_{b_1}^{[2]} =\mathcal{Y}_{b_2}^{[2]}+2\lambda_3^2\mathcal{Y}_{b_1}^{[0]}+\frac{4 \lambda_3^2\Pi(\lambda_1^2/(\lambda_1^2-\lambda_3^2),  k^{(1)\prime}_1)}{\lambda_2\sqrt{\lambda_1^2-\lambda_3^2}},
		\end{split}
	\end{equation}
	where $K_1^{(1)}=K(k_1^{(1)})$ and $\Pi(\lambda_1^2/(\lambda_1^2-\lambda_3^2),  k^{(1)\prime}_1)$ the elliptic integrals of the first and third kinds, defined in equations \eqref{eq:define-first-integral} and \eqref{eq:define-third-integral} respectively.
\end{lemma}

\begin{proof}
	Based on the \Cref{prop:case-1-P}, combined the equation \eqref{eq:integral-dn-1} with the integration formulas \eqref{eq:define-third-integral}, it is easy to obtain 
	$\mathcal{Y}_{a_1}^{[2n]}$, $n=0,1$:
	\begin{equation}\nonumber
		\begin{split}
			&\ \mathcal{Y}_{a_1}^{[0]}\xlongequal{\eqref{eq:integral-dn-1}}\frac{2\ii   K^{(1)}_1}{\lambda_2\sqrt{\lambda_1^2-\lambda_3^2}}, \\ 
			&\ 
			\mathcal{Y}_{a_1}^{[2]}\xlongequal{\eqref{eq:integral-dn-1}}\frac{2\ii \lambda_1^2}{\lambda_2\sqrt{\lambda_1^2-\lambda_3^2}}\int_{0}^{   K^{(1)}_1}\frac{\dd \nu}{1-\frac{\lambda_2^2-\lambda_1^2}{\lambda_2^2}\sn^2(\nu,  k^{(1)}_1)} \xlongequal{\eqref{eq:define-third-integral}}\frac{2\ii\lambda_1^2 \Pi((\lambda_2^2-\lambda_1^2)/\lambda_2^2,  k^{(1)}_1)}{\lambda_2\sqrt{\lambda_1^2-\lambda_3^2}}.
		\end{split}
	\end{equation}
	 Similarly, by utilizing steps listed in \Cref{prop:case-1-P}, we obtain $\mathcal{Y}_{a_2}^{[2n]}$ and $\mathcal{Y}_{b_i}^{[2n]}$, $i=1,2$.
\end{proof}

\begin{lemma}\label{lemma:dn-2-int}
	Along the $a_i$-circle and $b_i$-circle, $i=1,2$, the hyperelliptic integrals defined in equation \eqref{eq:Y-define} with $j=1,3,5$, are
	\begin{equation}\nonumber
		\begin{split}
			&\mathcal{Y}_{a_1}^{[1]}=\mathcal{Y}_{a_2}^{[1]}=\frac{2\ii K^{(1)}_2}{\sqrt{\lambda_1^2-\lambda_3^2}}, \qquad \mathcal{Y}_{a_1}^{[3]}=\mathcal{Y}_{a_2}^{[3]} =\frac{2\ii  \left(\lambda_3^2  K^{(1)}_2+(\lambda_1^2-\lambda_3^2) E^{(1)}_2\right)}{\sqrt{\lambda_1^2-\lambda_3^2}},\\
			&\mathcal{Y}_{b_1}^{[1]}=\mathcal{Y}_{b_2}^{[1]}=\frac{-2  K^{(1)\prime}_2}{\sqrt{\lambda_1^2-\lambda_3^2}},\qquad
			\mathcal{Y}_{b_1}^{[3]}=\mathcal{Y}_{b_2}^{[3]}=\frac{-2(\lambda_1^2  K^{(1)\prime}_2+(\lambda_3^2-\lambda_1^2) E^{(1)\prime}_2)}{\sqrt{\lambda_1^2-\lambda_3^2}},\\
			& \mathcal{Y}_{a_1}^{[5]}=\mathcal{Y}_{a_2}^{[5]} = \frac{v\mathcal{Y}_{a_1}^{[3]}-v_1\mathcal{Y}_{a_1}^{[1]}}{3}, \qquad \mathcal{Y}_{b_1}^{[5]}=\mathcal{Y}_{b_2}^{[5]} =\frac{v\mathcal{Y}_{b_1}^{[3]}-v_1\mathcal{Y}_{b_1}^{[1]}}{3},
		\end{split}
	\end{equation}
	where $v$ is defined in equation \eqref{eq:solutions-dn} and $v_1$ is defined as $v_1=\lambda_1^2\lambda_2^2+\lambda_1^2\lambda_3^2+\lambda_2^2\lambda_3^2$.
\end{lemma}
\begin{proof}
Based on the \Cref{prop:case-2-P}, combined with the equation \eqref{eq:integral-dn-2} and the integration formulas \eqref{eq:define-first-integral} and \eqref{eq:define-second-integral}, it is easy to obtain $\mathcal{Y}_{a_1}^{[2n+1]}$, $n=0,1,2$:
\begin{equation}\nonumber
	\begin{split}
	    \mathcal{Y}_{a_1}^{[1]}= &\ \int_{\lambda_2}^{\lambda_1}\frac{2\lambda}{y}\dd \lambda	\xlongequal{\eqref{eq:integral-dn-2}} \frac{2\ii K^{(1)}_2}{\sqrt{\lambda_1^2-\lambda_3^2}},\\
		\mathcal{Y}_{a_1}^{[3]}	\xlongequal{\eqref{eq:integral-dn-2}} &\ \frac{2\ii }{\sqrt{\lambda_1^2-\lambda_3^2}}\int_{0}^{   K^{(1)}_2}(\lambda_3^2+(\lambda_1^2-\lambda_3^2)\dn^2(\nu,  k^{(1)}_2)) \dd \nu\xlongequal[\eqref{eq:define-second-integral}]{\eqref{eq:define-first-integral}}\frac{2\ii  \left(\lambda_3^2  K^{(1)}_2+(\lambda_1^2-\lambda_3^2) E^{(1)}_2\right)}{\sqrt{\lambda_1^2-\lambda_3^2}}, \\
		\mathcal{Y}_{a_1}^{[5]}	\xlongequal{\eqref{eq:integral-dn-2}} &\ \frac{2\ii}{\sqrt{\lambda_1^2-\lambda_3^2}}\int_{0}^{   K^{(1)}_2}\left(\lambda_3^2+(\lambda_1^2-\lambda_3^2)\dn^2(\nu,  k^{(1)}_2)\right)^2 \dd \nu  \xlongequal[\eqref{eq:define-second-integral}]{\eqref{eq:recursive-formula-dn},\eqref{eq:define-first-integral}}  \frac{v\mathcal{Y}_{a_1}^{[3]}-v_1\mathcal{Y}_{a_1}^{[1]}}{3}.\\
	\end{split}
\end{equation}
Similarly, we obtain $\mathcal{Y}_{a_2}^{[2n+1]}$ and  $\mathcal{Y}_{b_i}^{[2n+1]}$, $i=1,2$.
\end{proof} 

\begin{lemma}\label{lemma:U-V-R-dn}
	When branch points satisfy \ref{case1}, parameters $U^{(1)}$, $\Delta^{(1)}$, $\mathbf{B}^{(1)}$, $\omega_2^{(1)}$, and $\omega_3^{(1)}$ are given in equation \eqref{eq:u-parameters-dn}, and $V^{(1)}=vU^{(1)}$, $\omega_3^{(1)}=v\omega_2^{(1)}$,
	\begin{equation}\nonumber
		\omega_1^{(1)}=
		\frac{-\vartheta_2(0,  \tau^{(1)}_1)\vartheta_1(2 \nu^{(1)} ,  \tau^{(1)}_1)}{\lambda_3\vartheta_1( \nu^{(1)} ,  \tau^{(1)}_1)\vartheta_2( \nu^{(1)} ,  \tau^{(1)}_1)}.
	\end{equation}
\end{lemma}

\begin{proof}
	By equation \eqref{eq:C_u0-define}, to obtain the explicit solution of the \ref{eq:mKdV} equation, we need to provide the values of parameters $U^{(1)}$, $V^{(1)}$, $\Delta^{(1)}$, $\omega_2^{(1)}$, and $\omega_3^{(1)}$.
	From the definition of $w_{1,2}$ in equation \eqref{eq:define-basis-B}, parameters $d_{ij}$ are given by $d_{11}=d_{21}= \pi\sqrt{\lambda_1^2-\lambda_3^2}/(2   K^{(1)}_2)$ and $d_{10}=-d_{20}=\pi \lambda_2\sqrt{\lambda_1^2-\lambda_3^2}/(2  K^{(1)}_1)$.
	Thus, we obtain $ \mathbf{B}^{(1)}$.
	Utilizing the equation \eqref{eq:U-V-v}, we get
		\begin{equation}\nonumber
			V_i^{(1)}=vU_i^{(1)}=2vd_{i1}=v \kappa^{(1)}, \qquad i=1,2,
		\end{equation}
		where $\kappa^{(1)}$ is defined in equation \eqref{eq:u-parameters-dn-B}.
	From the above results, parameters $\omega_2^{(1)}$ and $\omega_3^{(1)}$, defined in equation \eqref{eq:omega-2-3-define} as $\omega_2^{(1)}=\omega_3^{(1)}=0$.
	According to the definition of $\Omega_{1}^{(1)}(P)$, we set 
		\begin{equation}\label{eq:define-Omega-1-c}
			\dd\Omega_1^{(1)}=\sum_{i=0}^2 c_{1i}\frac{\lambda^i}{y}\dd \lambda.
		\end{equation}
		Considering $\lambda \rightarrow \infty$ and combining with equation \eqref{eq:define-int-w1}, we can obtain parameters $c_{12}=1$.
		Considering the case $\oint_{a_i}\dd \Omega_1^{(1)}=0$, $i=1,2$, we obtain $c_{11}=0$, and $c_{10}=-\lambda_1^2\Pi((\lambda_2^2-\lambda_1^2)/\lambda_2^2,  k^{(1)}_1)/  K^{(1)}_1$.
		Then, we get
		\begin{equation}\nonumber
			\begin{split}
				\Delta_1^{(1)}=&\oint_{b_1}\dd\Omega_1^{(1)}=\mathcal{Y}_{b_1}^{[2]}+c_{10}\mathcal{Y}_{b_1}^{[0]}
				=\ii \pi +2\nu^{(1)}, \quad 
				\Delta_2^{(1)}=\oint_{b_2}\dd\Omega_1^{(1)}
				=\mathcal{Y}_{b_2}^{[2]}+c_{10}\mathcal{Y}_{b_2}^{[0]}
				=\ii \pi -2\nu^{(1)} ,
			\end{split}
		\end{equation}
		through utilizing the formulas shown in \cite{ByrdF-54} and the definition of the parameter $\nu$ defined in equation \eqref{eq:u-parameters-dn}. 	
	In summary, we obtain parameters $U^{(1)}$, $V^{(1)}$, $\Delta^{(1)}$, $ \mathbf{B}^{(1)} $, $\omega_2^{(1)}$, and $\omega_3^{(1)}$.
	
	Considering the transformation (defined  in equation \eqref{eq:fs-2}) and the correspondence between $\lambda$ and $\nu$ shown in \Cref{prop:appendenx-f-1-2}, 
	we can obtain that when $\lambda=+\infty$, the corresponding parameter $\nu$ is $\nu= \nu^{(1)} _{\infty}$.
	By combining this with equation \eqref{eq:integral-dn-1}, it is easy to obtain 
	\begin{equation}\label{eq:define-nu-hat-infinity}
		\int_{\lambda_3}^{+\infty}\frac{c_{10}}{y}\dd \lambda=\frac{\ii c_{10}\left( \nu^{(1)}_{\infty}-  K^{(1)}_1\right)}{\lambda_2\sqrt{\lambda_1^2-\lambda_3^2}}
		, \qquad  
		\nu^{(1)} _{\infty}=F\left(\sqrt{\frac{\lambda_1^2-\lambda_3^2}{-\lambda_3^2}},  k^{(1)}_1\right).
	\end{equation}
	By combining equation \eqref{eq:integral-dn-2} with the properties of Jacobi elliptic functions, we obtain
	\begin{equation}\label{eq:int-infini-1}
		\begin{split}
			&\ \int_{\lambda_3}^{+\infty}\frac{\lambda^2}{y}-\frac{1}{\lambda}\dd \lambda
			\xlongequal[\eqref{eq:add-app}]{\eqref{eq:integral-dn-2}} \frac{\ii \lambda_1^2( \nu^{(1)} _{\infty}-  K^{(1)}_1) }{\lambda_2\sqrt{\lambda_1^2-\lambda_3^2}}
			-\int_{  K^{(1)}_1}^{ \nu^{(1)} _{\infty}}\frac{\sn( \nu^{(1)} _3+\nu,  k^{(1)}_1)}{\sn(\nu,  k^{(1)}_1)\sn(\nu^{(1)} _3,  k^{(1)}_1)}\dd \nu\\
			\!\!	\xlongequal[\eqref{eq:Jacobi-shift}]{\eqref{eq:add-app}}&\ \frac{\ii \lambda_1^2( \nu^{(1)} _{\infty}-  K^{(1)}_1) }{\lambda_2\sqrt{\lambda_1^2-\lambda_3^2}}
			-\int_{  K^{(1)}_1}^{ \nu^{(1)} _{\infty}}Z(\nu+\ii  K^{(1)\prime}_1,  k^{(1)}_1)+Z( \nu^{(1)} _3+\ii  K^{(1)\prime}_1,  k^{(1)}_1)-Z( \nu^{(1)} _3+\nu+2\ii  K^{(1)\prime}_1,  k^{(1)}_1)\dd \nu\\
			\xlongequal{\eqref{eq:Zeta-define}}&\ \frac{\ii \lambda_1^2( \nu^{(1)} _{\infty}-  K^{(1)}_1) }{\lambda_2\sqrt{\lambda_1^2-\lambda_3^2}}
			-\ln\left(\frac{\vartheta_4(\frac{\ii( \nu^{(1)} _{\infty}+\ii  K^{(1)\prime}_1)\pi}{  K^{(1)}_1},  \tau^{(1)}_1)\vartheta_4(\frac{\ii ( \nu^{(1)} _3+  K^{(1)}_1+2\ii  K^{(1)\prime}_1)\pi}{  K^{(1)}_1},  \tau^{(1)}_1)}{\vartheta_4(\frac{\ii(  K^{(1)}_1+\ii  K^{(1)\prime}_1)\pi}{  K^{(1)}_1},  \tau^{(1)}_1)\vartheta_4(\frac{\ii ( \nu^{(1)} _3+ \nu^{(1)} _{\infty}+2\ii  K^{(1)\prime}_1)\pi }{  K^{(1)}_1},  \tau^{(1)}_1)}\right)\\
			&\
			-Z( \nu^{(1)} _3+\ii  K^{(1)\prime}_1,  k^{(1)}_1)( \nu^{(1)} _{\infty}-  K^{(1)}_1), \qquad 
			\nu^{(1)} _3=F\left(\sqrt{\frac{-\lambda_2^2}{\lambda_1^2-\lambda_2^2}},  k^{(1)}_1\right).
		\end{split}	
	\end{equation}
	Since $  K^{(1)}_1- \nu^{(1)} _{\infty}= \nu^{(1)}   K^{(1)}_1/(\ii \pi)$ and $ \nu^{(1)} _3=  K^{(1)}_1- \nu^{(1)}   K^{(1)}_1/(\ii \pi)-\ii   K^{(1)\prime}_1$
	by the transformation of parameters in \cite{ByrdF-54} and equation \eqref{eq:Jacobi-shift}, we can simplify the above integration as 
	\begin{equation}\nonumber
		\begin{split}
			\int_{\lambda_3}^{+\infty}\dd \Omega_1^{(1)}-\frac{\dd \lambda}{\lambda}\xlongequal[\eqref{eq:integral-dn-2}]{\eqref{eq:int-infini-1}} &\ 
			\frac{\ii(c_{10}+ \lambda_1^2)(  K^{(1)} _1- \nu^{(1)} _{\infty}) }{-\lambda_2\sqrt{\lambda_1^2-\lambda_3^2}}
			\! -\! \int_{  K^{(1)}_1}^{ \nu^{(1)} _{\infty}}\frac{\sn( \nu^{(1)} _3+\nu,  k^{(1)}_1)}{\sn(\nu,  k^{(1)}_1)\sn( \nu^{(1)} _3,  k^{(1)}_1)}\dd \nu \\
			\xlongequal[\eqref{eq:Jacobi-Theta-K-iK}]{\eqref{eq:formula-trans-Pi-Zeta-theta}} &\ \! -\! \ln\left(\! \frac{\vartheta_1( \nu^{(1)} ,  \tau^{(1)} _1)\vartheta_2( \nu^{(1)} ,  \tau^{(1)} _1)}{\vartheta_2(0,  \tau^{(1)} _1)\vartheta_1(2 \nu^{(1)} ,  \tau^{(1)} _1)}\! \right)\!.
		\end{split}
	\end{equation}
	By equation \eqref{eq:omega-1-define}, we obtain
	\begin{equation}\nonumber
		\omega_1^{(1)}=\exp\left(-\ln(\lambda_3)+\lim_{ P\rightarrow \infty^{+}}\int_{\Gamma_p}\dd\Omega_1^{(1)}-\frac{\dd \lambda}{\lambda}+\ii \pi \right)=\frac{\vartheta_2(0,  \tau^{(1)} _1)\vartheta_1(2 \nu^{(1)} ,  \tau^{(1)} _1)}{\lambda_3\vartheta_1( \nu^{(1)} ,  \tau^{(1)} _1)\vartheta_2( \nu^{(1)} ,  \tau^{(1)} _1)}.
	\end{equation}
	Thus, we express the explicit expression of the $\omega_1^{(1)}$ in terms of Jacobi theta function.
\end{proof}

In summary, we obtain parameters we needed with three pairs of imaginary branch points. 
By the same way, we also obtain parameters with two pairs of complex branch points and a pair of imaginary branch point with the different conformal map we introduced.

\vspace{0.1cm}
\noindent $\bullet$\quad\textbf{\large \ref{case2}}

Without loss of generality, we set $\Re(\lambda_3)>\Re(\lambda_2)=0>\Re(\lambda_1)$.
Since $\lambda_2^*=-\lambda_2$, $\lambda_1^*=-\lambda_3$, $\lambda_3^*=-\lambda_1$, the definition of the parameter $y$ (in equation \eqref{eq:define-curve-algebro}) could be rewritten as $y^2=(\lambda^2-\lambda_1^2)(\lambda^2-\lambda_2^2)(\lambda^2-\lambda_3^2)$ with homology basis shown in \Cref{fig:genus-two-figure}. 
Then, we calculate the related hyperelliptic integrals. 
Similarly, we would introduce two different functions to obtain all explicit expressions of hyperelliptic integrals.
From the Propositions \ref{prop:elliptic-int-2}-\ref{prop:case-2-C}, 
the hyperelliptic integrals \eqref{eq:hyper} along the $a_{1,2}$-circle and $b_{1,2}$-circle could be expressed by elliptic integrals and branch points $\lambda_{1,2,3}$ as follows: 
\begin{lemma}\label{lemma:cn-1-int}
 	The hyperelliptic integrals defined in equation \eqref{eq:Y-define} with $j=0,2$, along the $a_{1,2}$-circle and $b_{1,2}$-circle are
	\begin{equation}\nonumber
		\begin{split}
			&\mathcal{Y}_{a_1}^{[0]}=\frac{2\ii   K^{(2)}_1}{\sqrt{AB}}, 
			\qquad
			\mathcal{Y}_{a_2}^{[0]}=-2\mathcal{Y}_{a_1}^{[0]}, 
			\qquad
			\mathcal{Y}_{a_1}^{[2]}=\frac{\ii\lambda_2^2 }{\sqrt{AB}}\left(  K^{(2)}_1 +\frac{B+A}{B-A}\Pi\left(\frac{(B-A)^2-\lambda_2^4}{(B-A)^2},  k^{(2)}_1\right)\right), \\
			& \mathcal{Y}_{b_1}^{[0]}=\mathcal{Y}_{a_1}^{[0]},\qquad \mathcal{Y}_{b_1}^{[2]}=\mathcal{Y}_{a_1}^{[2]}, \qquad
			\mathcal{Y}_{a_2}^{[2]}=\frac{2\ii \lambda_2^2}{\sqrt{AB}(A-B)}\left(2B  K^{(2)}_1-(A+B)\Pi\left(\frac{(A-B)^2}{-4AB},  k^{(2)}_1\right)\right),  \\
			& 
			\mathcal{Y}_{b_2}^{[0]}=\frac{K^{(2)\prime}_1-\ii  K^{(2)}_1}{\sqrt{AB}}, \qquad
			\mathcal{Y}_{b_2}^{[2]}=\frac{\mathcal{Y}_{b_1}^{[2]}+\ii \pi}{2}+\lambda_2^2\frac{2B  K^{(2)\prime}_1+(A-B)\Pi\left(\frac{(A+B)^2}{4AB},  k^{(2)\prime}_1\right)}{2(A+B)\sqrt{AB}},
		\end{split}
	\end{equation}
	where notations $\mathcal{Y}_{a_i}^{[j]}$, $\mathcal{Y}_{b_i}^{[j]}$ are defined in equation \eqref{eq:Y-define}, and parameters $k^{(2)}_{1,2}$, $A$ and $B$ are defined in equation \eqref{eq:u-parameters-cn}.
\end{lemma}

\begin{proof}
	By  the equation \eqref{eq:Y-a-2n+1-cn-calculate-1}, it is easy to obtain 
	\begin{equation}\nonumber
		\begin{split}
			\mathcal{Y}_{a_1}^{[0]}=&\ 2\int_{\lambda_1^*}^{\lambda_1}\frac{\dd  \lambda}{y}
			\xlongequal{\eqref{eq:Y-a-2n+1-cn-calculate-1}} \int_{-3  K^{(2)}_1+\ii   K^{(2)\prime}_1}^{-  K^{(2)}_1+\ii   K^{(2)\prime}_1} \frac{\ii }{\sqrt{AB}} \dd \nu=\frac{2\ii K^{(2)}_1}{\sqrt{AB}}, \\
			\mathcal{Y}_{a_1}^{[2]}
			\xlongequal{\eqref{eq:Y-a-2n+1-cn-calculate-1}}
			&\  \int_{-3  K^{(2)}_1+\ii   K^{(2)\prime}_1}^{-  K^{(2)}_1+\ii   K^{(2)\prime}_1} \frac{\ii \lambda_2^2B}{\sqrt{AB}(B-A)} \left(1-\frac{2A(B+A)}{4AB+(B-A)^2\sn^2(\nu,  k^{(2)}_1)}\right) \dd \nu\\
			=&\  \frac{\ii\lambda_2^2 }{\sqrt{AB}}\left(  K^{(2)}_1 +\frac{B+A}{B-A}\Pi\left(\frac{(B-A)^2-\lambda_2^4}{(B-A)^2},  k^{(2)}_1\right)\right)
		\end{split}
	\end{equation}
	since $(\cn(\nu+2  K^{(2)}_1-\ii   K^{(2)\prime}_1,  k^{(2)}_1))/(a+b\sn^2(\nu+2  K^{(2)}_1-\ii   K^{(2)\prime}_1,  k^{(2)}_1))$ is odd.
	Combined with equation \eqref{eq:Pi-K-alpha}, it is easy to obtain $\mathcal{Y}_{a_1}^{[0]}$ and $\mathcal{Y}_{a_1}^{[2]}$.
	Together with the recursive formulas provided in \Cref{prop:case-1-C}, we obtain the $\mathcal{Y}_{a_1}^{[2n]}$, $n=0,1$.
	Similarly, we obtain the integration $\mathcal{Y}_{a_2}^{[j]}$ and $\mathcal{Y}_{b_i}^{[j]}$, $i=1,2$, $j=0,2$.
\end{proof}

\begin{lemma}\label{lemma:cn-2-int}
	The hyperelliptic integrals defined in equation \eqref{eq:Y-define} along the $a_i$-circle and $b_i$-circle are $	\mathcal{Y}_{a_2}^{[1]}=\mathcal{Y}_{a_2}^{[3]}=\mathcal{Y}_{a_2}^{[5]}=0$,
	\begin{equation}\nonumber
		\begin{split}
			&\mathcal{Y}_{a_1}^{[1]} =-\frac{2\ii   K^{(2)}_2}{\sqrt{A}}, \quad  
			\mathcal{Y}_{a_1}^{[3]} =-\frac{2\ii \left((\lambda_2^2-A)  K^{(2)}_2+2 A  E^{(2)}_2\right)}{\sqrt{A}},\quad
			\mathcal{Y}_{a_1}^{[5]} = \frac{v}{3}\mathcal{Y}_{a_1}^{[3]}- \frac{v_1}{3}\mathcal{Y}_{a_1}^{[1]},\\
			&\mathcal{Y}_{b_1}^{[1]}=\frac{2  K^{(2)\prime}_2}{\sqrt{A}}, \qquad	
			\mathcal{Y}_{b_1}^{[3]}=\frac{2(A+\lambda_2^2)  K^{(2)\prime}_2-4A E^{(2)\prime}_2}{\sqrt{A}}, \qquad 
			\mathcal{Y}_{b_1}^{[5]}=\frac{v}{3}\mathcal{Y}_{b_1}^{[3]}-\frac{v_1}{3}\mathcal{Y}_{b_1}^{[1]}, \\
			&\mathcal{Y}_{b_2}^{[1]}=\frac{\ii   K^{(2)}_2+  K^{(2)\prime}_2}{\sqrt{A}}, \qquad
			\mathcal{Y}_{b_2}^{[3]}=\ii\frac{2A  E^{(2)}_2-(A-\lambda_2^2)  K^{(2)}_2}{\sqrt{A}}+\frac{\mathcal{Y}_{b_1}^{[3]}}{2},\qquad \mathcal{Y}_{b_2}^{[5]}=\frac{v}{3}\mathcal{Y}_{b_2}^{[3]}-\frac{v_1}{3}\mathcal{Y}_{b_2}^{[1]},
		\end{split}
	\end{equation}
	where notations $\mathcal{Y}_{a_i}^{[j]}$ and $\mathcal{Y}_{b_i}^{[j]}$ are defined in equation \eqref{eq:Y-define}.
	The parameter $E_2^{(2)}$ represents $E(k_2^{(2)})$ with $E(\cdot)$ defined in \Cref{define:elliptic-function}.
\end{lemma}

\begin{proof}
	Utilizing the equation \eqref{eq:Y-a-2n+1-cn-calculate-1}, we obtain
	\begin{equation}
		\begin{split}
			\mathcal{Y}_{a_1}^{[2n+1]}=&\ \int_{\lambda_1^*}^{\lambda_1}\frac{2\lambda^{2n+1}}{y}\dd \lambda
			\xlongequal[\eqref{eq:fs-3-y}]{\eqref{eq:fs-3}} \int_{-  K^{(2)}_2-\ii   K^{(2)\prime}_2}^{  K^{(2)}_2-\ii   K^{(2)\prime}_2}-\frac{\ii}{\sqrt{A}} \left(\lambda_2^2-\frac{A+A\cn(\nu,  k^{(2)}_2)}{1-\cn(\nu, k^{(2)}_2)}\right)^n \dd \nu\\
			\xlongequal{\eqref{eq:Jacobi-shift}}&\ -\frac{\ii}{\sqrt{A}}\int_{-  K^{(2)}_2}^{  K^{(2)}_2}\left(\lambda_2^2 -A + 2A\dn^2(\nu,  k^{(2)}_2)\right)^n \dd \nu,
		\end{split}
	\end{equation}
	since $\sn(\nu,  k^{(2)}_2)\dn(\nu,  k^{(2)}_2)$ is odd.
	Applying the recursive formulas provided in \Cref{prop:case-2-C}, we obtain $\mathcal{Y}_{a_1}^{[2n+1]}$, $n=0,1,2$.
	Similarly, we obtain the integrals $\mathcal{Y}_{a_i}^{[j]}$ and $\mathcal{Y}_{b_i}^{[j]}$, $i=1,2$, $j=1,3,5$.
\end{proof}
	
\begin{lemma}\label{lemma:U-V-R-cn}
	Parameters $U^{(2)}$, $\Delta^{(2)}$, $\mathbf{B}^{(2)}$, $\omega_2^{(2)}$, could be expressed in equation \eqref{eq:u-parameters-cn}, and $V^{(2)}=vU^{(2)}$, $\omega_3^{(2)}=v\omega_2^{(2)}$,
	\begin{equation}\nonumber
		\omega_1^{(2)} = \frac{2\lambda_2\vartheta_1(\ii \tau^{(2)}_1+ \nu^{(2)},\tau^{(2)}_1)}{\ii(A-B)\vartheta_4(0,\tau^{(2)}_1)\ee^{\ii \tau^{(2)}_2\pi/4}}.
	\end{equation}
\end{lemma}

\begin{proof}
	We would calculate the values of parameters $U^{(2)}$, $V^{(2)}$, $\Delta^{(2)}$, $\mathbf{B}^{(2)}$, $\omega_2^{(2)}$, $\omega_3^{(2)}$,  and $\omega_1^{(2)}$, one by one.
	By the definition of $w_{1,2}$ shown in equation \eqref{eq:define-basis-B}, parameters $d_{ij}$ are $d_{11}=-\kappa^{(2)}$, $d_{10}=0$, $d_{21}=-\kappa^{(2)}/2$ and $d_{20}=\pi \sqrt{AB}/(2  K^{(2)}_1)$.
	Then, we get $
	\mathbf{B}^{(2)}_{11}=2\pi \ii \tau^{(2)}_2$, 
	$\mathbf{B}^{(2)}_{12}=\pi \ii (\tau^{(2)}_2-1)$, 
	$\mathbf{B}^{(2)}_{21}=\pi \ii (\tau^{(2)}_2-1)$, and
	$\mathbf{B}^{(2)}_{22}=\pi \ii (\tau^{(2)}_1+\tau^{(2)}_2)/2$.
	Together with the equation \eqref{eq:U-V-v}, it is easy to obtain 
	\begin{equation}\label{eq:define-U-1-2-C}
		V_1^{(2)}=vU_1^{(2)}=2vd_{11}
		=-2v\kappa^{(2)}, \qquad  
		V_2^{(2)}=vU_2^{(2)}=2vd_{21}
		=-v\kappa^{(2)}, 
	\end{equation}
	where $\kappa^{(2)}$ is defined in \Cref{theorem:solution-u}.
	Furthermore, we obtain parameters $\omega_2^{(2)} $ and $\omega_3^{(2)} $ as 
	\begin{equation}\label{eq:define-omega-2-3-C}
		\omega_2 ^{(2)}
		=-\frac{\ii}{4\pi} \sum_{i=1}^{2}U_i^{(2)} \oint_{a_i}\frac{\lambda^2}{y}\dd \lambda=\frac{\kappa^{(2)}_2}{2}, \qquad 
		\omega_3^{(2)} 
		=-\frac{\ii}{4\pi} \sum_{i=1}^{2}V_i^{(2)} \oint_{a_i}\frac{\lambda^2}{y}\dd \lambda=\frac{v\kappa^{(2)}_2}{2},
	\end{equation}
	by utilizing the equation \eqref{eq:omega-2-3-define}.
	Combining the definition of $\dd \Omega_{1}^{(2)}$ in equation \eqref{eq:define-Omega-1-c} with the conditions that all integrals on the $a_i$-circle are zero, 
	we get parameters $c_{ij}$.
	Considering the case $\oint_{a_i}\dd \Omega_1^{(2)}=0$, similarly, we can obtain $c_{12}=1$, $c_{11}=-\kappa^{(2)}/2$, and $c_{10}=\lambda_2^2(2B  K^{(2)}_1-(A+B)\Pi(-(A-B)^2/(4AB),  k^{(2)}_1))/(2(A-B)  K^{(2)}_1)$, which implies 
	\begin{equation}\nonumber
	\begin{split}
	\Delta_1^{(2)}&=\oint_{b_1}\dd\Omega_1^{(2)}=\sum_{i=0}^{2}c_{1i}\mathcal{Y}_{b_1}^{[i]}
	=\ii \pi -\ii \pi \tau^{(2)}_2 \! , \\
	\Delta_2^{(2)}&=\oint_{b_2}\dd\Omega_1^{(2)}
	=\sum_{i=0}^{2}c_{1i}\mathcal{Y}_{b_2}^{[i]}
	=\frac{ \ii \pi(\tau^{(2)}_1+\tau^{(2)}_2)+\nu^{(2)}}{2}, 
	\end{split}
	\end{equation}
	through utilizing the formulas shown in \cite{ByrdF-54}.
	In summary, we obtain parameters $U^{(2)}$, $V^{(2)}$, $D^{(2)}$, $\Delta^{(2)}$, $\omega_2^{(2)} $, and $\omega_3^{(2)}$.	
	
	Consider the limit, as $P\rightarrow \infty^{+}$. 
	The method we used is similar to the one we provided in \Cref{lemma:U-V-R-dn}. 
	Thus, we only list the differences between them as follows:
	\begin{equation}\nonumber
		\begin{split}
			\lim_{\lambda\rightarrow \infty}\int_{\lambda_3}^{\lambda}\frac{c_{11}\lambda}{y} \dd \lambda
			=& \int_{  K^{(2)}_2+\ii   K^{(2)\prime}_2}^{0}\frac{\ii \kappa^{(2)}}{4\sqrt{A}}\dd \nu
			=-\frac{\ii \pi+\ii \pi \tau^{(2)}_2}{4},\\
			\lim_{\lambda\rightarrow \infty} 
			\int_{\lambda_3}^{\lambda} \frac{\lambda^2}{y}\dd \lambda
			\xlongequal[\eqref{eq:define-third-integral}]{\eqref{eq:elliptic-funda-1-cn}} 
			& \ \frac{\ii \lambda_2^2\sqrt{B} \nu}{2\sqrt{A}(B-A)}+\frac{1}{4}\ln\left(\frac{2\sqrt{AB}\dn(\nu,  k^{(2)}_1)+\ii \lambda_2^2\sn(\nu,  k^{(2)}_1)}{2\sqrt{AB}\dn(\nu,  k^{(2)}_1)-\ii \lambda_2^2\sn(\nu,  k^{(2)}_1)}\right)\\
			&  -\left.\frac{\ii \lambda_2^2(A+B)}{4\sqrt{AB}(A-B)}\Pi\left(\nu,\frac{(B-A)^2}{-4AB},  k^{(2)}_1\right)\right|_{\nu=  K^{(2)}_1+\ii   K^{(2)\prime}_1}^{\nu^{(2)}_{1}K_1^{(1)}/(\ii \pi)},
		\end{split}
	\end{equation}
	and
	\begin{equation}\nonumber
		\begin{split}
			&\lim_{\lambda\rightarrow \infty}\int_{\lambda_3}^{\lambda} \dd \Omega_1-\frac{1}{\lambda}\dd \lambda 
			\xlongequal{\eqref{eq:formula-trans-Pi-Zeta-theta}} \ln\left(\frac{\vartheta_1(\ii( \nu^{(2)} _{\infty}/  K^{(2)}_1+ \tau^{(2)} _1)\pi,\tau_1^{(2)})}{\vartheta_4(0)}\right)-\frac{\ii \tau^{(2)}_2 \pi}{4} -\ln\left(\frac{(A-B)}{2\lambda_2\lambda_3}\right)-\ii\pi,\\
		\end{split}
	\end{equation}
	where $ \nu^{(2)} _{\infty}= \nu^{(2)}   K^{(2)}_1/(\ii \pi )$, $ \nu^{(2)} $ defined in equation \eqref{eq:u-parameters-cn} and $\nu^{(2)}_1$ is defined in equation \eqref{eq:define-Omega-1-c}.
	Similarly as the proof of the \Cref{lemma:U-V-R-dn}, we obtain 
	\begin{equation}\nonumber
		\omega_1=\exp\left(-\ln(\lambda_3)+\lim_{ P\rightarrow \infty^{+}}\int_{\Gamma_p}\dd\Omega_1-\frac{\dd \lambda}{\lambda}+\ii \pi \right)
		=\frac{2\lambda_2\vartheta_1(\ii \tau^{(2)} _1\pi+  \nu^{(2)} , \tau^{(2)} _1)}{\ii (A-B)\vartheta_4(0, \tau^{(2)} _1)}\ee^{-\ii  \tau^{(2)} _2 \pi/4}. 
	\end{equation}
\end{proof}

Using the results from Lemmas \ref{lemma:dn-1-int}-\ref{lemma:U-V-R-cn}, we substitute them into equation \eqref{eq:q1-q0} to obtain the two-phase solutions as given in equation \eqref{eq:solutions-dn}.

\newenvironment{proof-solution-dn}{\emph{\textbf{Proof of \Cref{theorem:solution-u}.}}}{\hfill$\Box$\medskip}
\begin{proof-solution-dn}
	Based on the explicit expression given in equation \eqref{eq:C_u0-define}, as well as \Cref{lemma:U-V-R-dn} and \Cref{lemma:U-V-R-cn}, the determination of the explicit form of the function $u(x,t)$ requires a detailed examination of the parameters $D$ and $C_{u_0}$. 
	When parameters $U$, $V$ and $\omega_{2,3}$ are all real, the real-valued solution $u(x,t)$ provided in equation \eqref{eq:q1-q0} can be rewritten accordingly as:
	\begin{equation}\label{eq:u-complex}
		u(x,t)\xlongequal[\eqref{eq:Riemann-Theta-prop}]{\eqref{eq:q1-q0}}C^*_{u_0}\frac{\Theta(-D^*-\Delta^*+\ii Ux+\ii V t)}{\Theta(-D^*+\ii Ux+\ii V t)}\ee^{-2\ii \omega_2 x-2\ii \omega_3  t}.
	\end{equation}
	A comparison between the above expression of solution $u(x,t)$ and equation \eqref{eq:C_u0-define}
	yields the following result:
	\begin{equation}\label{eq:u=u}
		C^*_{u_0}\frac{\Theta\left(-D^*-\Delta^*+\ii Ux+\ii V t\right)}{\Theta\left(-D^*+\ii Ux+\ii V t\right)}
		=\frac{-4\, \Theta\left(D-\Delta+\ii Ux+\ii Vt\right)}{C_{u_0}\omega_1^2\Theta\left(D+\ii Ux+\ii Vt\right)}.
	\end{equation}
	From equation \eqref{eq:Riemann-Theta-prop}, since the poles of the function must coincide, we obtain $-D^*+\ii Ux+\ii Vt+2\pi \mathbf{n}+\mathbf{Bm}=D+\ii Ux+\ii Vt$, which implies
	\begin{equation}\label{eq:define-D-value}
		D^*+D=2\pi \ii\mathbf{n}+\mathbf{B}\mathbf{m}, \qquad \mathbf{n}, \mathbf{m}\in \mathbb{Z}^2.
	\end{equation}
	Referring to the definition of the Riemann theta function in \Cref{define:Riemann-Theta-function} and  equation \eqref{eq:u=u},
	we deduce 
	\begin{equation}\label{eq:C_u_0-r-c}
		|C_{u_0}|^2=\frac{-4\, \Theta(D-\Delta+\ii U x+\ii V t)\exp\langle \mathbf{m}, \Delta^*\rangle}{\omega_1^2 \, \Theta(D-\Delta^*+\ii U x+\ii V t)},
	\end{equation}	
	where the parameter $\mathbf{m}$ is provided in equation \eqref{eq:define-D-value}.
	Furthermore, by equations \eqref{eq:q1-q0} and \eqref{eq:u-complex}, we obtain
	\begin{equation}\label{eq:C_u_0-r-c-1}
		C^*_{u_0}=C_{u_0}\frac{\Theta(D+\Delta+\ii Ux+\ii V t)\ee^{4\ii \omega_2 x+4\ii \omega_3  t}\exp\langle \mathbf{m}, \Delta^*\rangle}{\Theta(D-\Delta^*+\ii Ux+\ii V t)}.
	\end{equation}
	Since the solution $u(x,t)$ is a regular function, the parameter matrix $D$ must also satisfy the condition $\Theta(D+\ii Ux+\ii Vt)\neq0$ for all $(x,t)\in \mathbb{R}^2$. 
	The chosen parameters $C_{u_0}$ and $D$ must satisfy equations \eqref{eq:C_u_0-r-c} and \eqref{eq:C_u_0-r-c-1}.
	Guided by this constraint, we proceed to determine all the parameters $D$ and $C_{u_0}$ under different configurations of branch points.
	
	In the \ref{case1}, by \Cref{lemma:U-V-R-dn}, we have $U^{(1)},V^{(1)}\in \mathbb{R}^2$, $\omega_{2,3}^{(1)}=0$, and $\Delta^{(1)}\in \ii \mathbb{R}^2$, which further implies $(\Delta^{(1)})^*=-\Delta^{(1)}$.
	Since the right-hand side of equation \eqref{eq:C_u_0-r-c} must be independent of the variables $x$ and $t$, functions $\Theta(D^{(1)}-\Delta^{(1)}+\ii U^{(1)} x+\ii V^{(1)}t)$ and $\Theta(D^{(1)}+\Delta^{(1)}+\ii U^{(1)} x+\ii V^{(1)}t)$ must have the same poles and zeros.
	Combining the formula \eqref{eq:formula-Rieman-shift-1} with the definition of functions $\Delta^{(1)}$,
	we obtain 
	\begin{equation}\nonumber
		\begin{split}
			&\ \frac{\Theta(D^{(1)}-\Delta^{(1)}+\ii U^{(1)} x+\ii V^{(1)}t)}{\Theta(D^{(1)}+\Delta^{(1)}+\ii U^{(1)} x+\ii V^{(1)}t)}\\
			\xlongequal{\eqref{eq:formula-Rieman-shift-1}} &\ \frac{\vartheta_3(D^{(1)}_1-D^{(1)}_2-2\nu^{(2)},2\tau_1^{(1)})\frac{\vartheta_3(D^{(1)}_1+D^{(1)}_2+2\ii U^{(1)}_1x+2\ii V^{(1)}_1t ,2\tau_2^{(1)})}{\vartheta_2(D^{(1)}_1+D^{(1)}_2+2\ii U^{(1)}_1x+2\ii V^{(1)}_1t ,2\tau_2^{(1)})}-\vartheta_2(D^{(1)}_1-D^{(1)}_2-2\nu^{(2)},2\tau_1^{(1)})}{\vartheta_3(D^{(1)}_1-D^{(1)}_2+2\nu^{(2)},2\tau_1^{(1)})\frac{\vartheta_3(D^{(1)}_1+D^{(1)}_2+2\ii U^{(1)}_1x+2\ii V^{(1)}_1t ,2\tau_2^{(1)})}{\vartheta_2(D^{(1)}_1+D^{(1)}_2+2\ii U^{(1)}_1x+2\ii V^{(1)}_1t ,2\tau_2^{(1)})}-\vartheta_2(D^{(1)}_1-D^{(1)}_2+2\nu^{(2)},2\tau_1^{(1)})},
		\end{split}
	\end{equation}
	which implies that when $D_1^{(1)}=D_2^{(1)}$, functions $\Theta(D^{(1)}-\Delta^{(1)}+\ii U^{(1)} x+\ii V^{(1)}t)$ and $\Theta(D^{(1)}+\Delta^{(1)}+\ii U^{(1)} x+\ii V^{(1)}t)$ have the same poles and zeros. 
	By the equation \eqref{eq:define-D-value},
	we set $D^{(1)}= \ii n \pi \mathbf{1}+m/2 \mathbf{B}^{(1)} \mathbf{1}$,  with $m\in \mathbb{Z}$, $n\in \mathbb{R}$. Substituting into equations \eqref{eq:C_u_0-r-c} and \eqref{eq:C_u_0-r-c-1} yields $C_{u_0}^{(1)}=(C_{u_0}^{(1)})^*=2\ii /\omega_1^{(1)} \in \mathbb{R}$.
	The parameter $n$ in $D^{(1)}$ may take any real integers, which can be eliminated via a translation in either the $x$- or $t$-direction.
	Therefore, without loss of generality,	we set $n = 0$.
	From equation \eqref{eq:Riemann-Theta-prop}, when $m$ is even, the solution coincides with the case $m=0$, (i.e., $D^{(1)}=\mathbf{0}$).
	When $m$ is odd, the solution is equal to the case with $m=1$, i.e., $D^{(1)}= \mathbf{B}^{(1)} \mathbf{1}/2=\ii \pi   \tau^{(1)}_2 \mathbf{1}$.
	The parity of the parameter $m$ (even or odd) determines different pole configurations, indicating that the two cases yield distinct solutions.
	Thus, we conclude that the explicit expression for the solution $u(x,t)$ can be written in only two forms, corresponding to $D^{(1)}=\mathbf{0}$ and $D^{(1)}=\ii \pi   \tau^{(1)}_2 \mathbf{1}$.
		
	For the \ref{case2}, by \Cref{lemma:U-V-R-cn}, we have $U^{(2)},V^{(2)}\in \mathbb{R}^2$, $\omega_{2,3}^{(2)}\in \mathbb{R}$, and $(\Delta^{(2)})^*=\Delta^{(2)} \, \mod \, 2\pi \ii$.
	By equations \eqref{eq:Riemann-Theta-prop} and \eqref{eq:C_u_0-r-c-1}, we obtain 
	\begin{equation}\nonumber
		\begin{split}
			\frac{(C_{u_0}^{(2)})^*}{C_{u_0}^{(2)}}\xlongequal[\eqref{eq:u-parameters-cn}]{\eqref{eq:Riemann-Theta-prop}}
%
			&\ \frac{\Theta(D^{(2)}+\hat{\Delta}+\ii U^{(2)} x+\ii  V^{(2)} t)\ee^{-2D_2^{(2)}+\langle\mathbf{m},(\Delta^{(2)})^*\rangle}}{\Theta(D^{(2)}-\hat{\Delta}+\ii U^{(2)} x+\ii  V^{(2)} t)},
		\end{split}
	\end{equation}
	where $\hat{\Delta}=[0 \,\,\, \nu^{(2)}/2 ]^{\top}$.
	The right-hand side of the above equation must be independent of the variables $x$ and $t$, which implies that $\Theta(D^{(2)}+\hat{\Delta}+\ii U^{(2)} x+\ii  V^{(2)} t)/\Theta(D^{(2)}-\hat{\Delta}+\ii U^{(2)} x+\ii  V^{(2)} t)=\mathrm{const}$ and $-2D_2^{(2)}+\langle \mathbf{m}, (\Delta^{(2)})^* \rangle\in \ii \mathbb{R}$. Similarly, by equation \eqref{eq:formula-Rieman-shift-1} we obtain $2(D_1^{(2)}+2D_2^{(2)})=2\ii \pi (2n+1)$, $\mathbf{m}=\mathbf{0}$ and $D_2^{(2)}\in \ii \mathbb{R}$, which implies $(C_{u_0}^{(2)})^*=\ee^{-2D_2^{(2)}}C_{u_0}^{(2)}$.
	Plugging $\mathbf{m}=\mathbf{0}$ into equation \eqref{eq:C_u_0-r-c}, 
	we get $|C_{u_0}^{(2)}|^2=-4/\omega_1^2$.
	Thus, we get $D_2^{(2)}=\ii \pi n$, $n\in \mathbb{Z}$ and $C_{u_0}^{(2)}=2\ii /\omega_1^{(2)}$, which implies that $D_1^{(2)}=\ii \pi$.
	Using parameters $U^{(2)},V^{(2)}$, $\omega_{1,2,3}^{(2)}$, we consider the following two cases for the parameter $n$ selected in $D_2^{(2)}$: when $n$ is even (with $n=0$ taken as example), and when $n$ is odd--- the latter corresponds to the solution with an additional sign $``-"$.
	Therefore, we obtain $C_{u_0}^{(2)}=2\ii /\omega_1^{(2)}$ and $D^{(2)}=\ii \pi \mathbf{1}$. 
\end{proof-solution-dn}

\subsection{The elliptic form solution and fundamental solutions of Lax pair}\label{sec:effective-integration-method}
Another method to obtain the lower-genus solution is to use the effective integration method.
When we consider the explicit expressions of the matrix function $\mathbf{L}(x,t;\lambda)$, we can also deduce solutions expressed in terms of Jacobi elliptic functions.
The determinant of $\mathbf{L}(x,t;\lambda)$ given in equation \eqref{eq:L-elements} can also be rewritten as $\det\left(\mathbf{L}(x,t;\lambda)\right)
=\lambda^6+2\alpha_1\lambda^4
+s_2\lambda^2+s_0$, where
\begin{equation}\label{eq:det-L-s}
	s_2 = (u_x^2-2uu_{xx}-3u^4)/4+\alpha_1u^2+\alpha_1^2, \qquad 
	s_0 = (u_{xx}-4\alpha_1u+2u^3)^2/16.
\end{equation}
By the equation \eqref{eq:det-L-s}, we get $u_{xx}=4\sqrt{s_0}+4\alpha_1 u -2u^3$ and
\begin{equation}\label{eq:define-R-u}
	u_x^2=-R(u), \qquad R(u):=u^4-4\alpha_1 u^2+8u\sqrt{s_0}-4s_2+4\alpha_1^2.
\end{equation}
If there exists a root $u_1$ such that $R(u_1)=0$, the polynomials $\det(\mathbf{L}(x,t;\lambda))$ (defined in equation \eqref{eq:det-L-s}) can be decomposed into 
\begin{equation}\label{eq:det-L-poly}
	\begin{split}
		\det(\mathbf{L}(x,t;\lambda))=
		& \left(\lambda^3+\ii u_1\lambda^2+(\alpha_1-u_1^2/2-\ii 	u_1(\mu_1+\mu_2))\lambda+\ii u_1 \mu_1\mu_2\right)\\
		& \left(\lambda^3-\ii u_1\lambda^2+(\alpha_1-u_1^2/2+\ii 	u_1(\mu_1+\mu_2))\lambda-\ii u_1 \mu_1\mu_2\right).
	\end{split}
\end{equation}
And the solutions of the \ref{eq:mKdV} equation must satisfy
\begin{equation}\label{eq:solution-u-ux}
	u_x^2=-R(u)=-(u-u_1)(u-u_2)(u-u_3)(u-u_4), 
\end{equation} 
with $u_1+u_2+u_3+u_4=0$.
Together with Eq. \eqref{eq:L-elements}, 
we obtain the relations of functions $\mu_{1,2}$ and $\mu_{1,2}^*$ satisfy equations $u^2(\mu_1\mu_1^*\mu_2+\mu_1\mu_1^*\mu_2^*+\mu_1\mu_2\mu_2^*+\mu_1^*\mu_2\mu_2^*)=0$, $(u^2/2-\alpha_1)^2+u^2(\mu_1\mu_1^*+\mu_1\mu_2+\mu_1\mu_2^*+\mu_2\mu_1^*+\mu_2\mu_2^*+\mu_1^*\mu_2^*)=s_2$,  
$u^2(\mu_1+\mu_2+\mu_1^*+\mu_2^*)=0$, and $u^2\mu_1\mu_2\mu_1^*\mu_2^*=s_0$.
From the above equations, the following equations $\mu_1+\mu_1^*=-(\mu_2+\mu_2^*)$ and $(\mu_1+\mu_1^*)\mu_2\mu_2^*=-(\mu_2+\mu_2^*)\mu_1\mu_1^*$ hold,
which in turn implies $\mu_1\mu_1^*=\mu_2\mu_2^*$. 
Similarly, we obtain $\mu_1\mu_2=\mu_1^*\mu_2^*$ and $\mu_1\mu_2^*=\mu_1^*\mu_2$. 
Combining these results, we derive
\begin{equation}\label{eq:define-R-u-mu}
	\mu_1=\mu_2=-\mu_1^*=-\mu_2^*=\frac{\sqrt{R(u)}}{4u}. 
\end{equation}
Depending on different values of parameters $s_0$ and $s_2$, the roots of $R(u)=0$ could be classified into the following cases---each corresponding to a distinct form of the solution $u(x,t)$:
\begin{itemize}
	\item[i)] When $s_0=0$, parameters satisfy $u_4=-u_1\in \mathbb{R}$ and $u_3=-u_2 \in \mathbb{R}$. 
	Under these conditions, the solution $u(x,t)$ could be expressed as one-gap solutions, and these can be further simplified to $\cn$-type and $\dn$-type solutions. 
	\item[ii)] When $s_0 \neq 0$ and
	four roots of the equation $R(u)=0$ are all real, this case leads to two forms of solution $u(x,t)$ expressed by the Riemann theta functions under the \ref{case1}.
	\item[iii)] When $s_0 \neq 0$ and the equation $R(u)=0$ has real numbers with the remaining two being complex conjugates, this case yields only one solution $u(x,t)$ by the Riemann theta functions corresponding to the \ref{case2}.
\end{itemize}

Via the relation between roots and coefficients in equations \eqref{eq:det-L-lambda} and \eqref{eq:det-L-poly}, we obtain
\begin{equation}\label{eq:define-u-lambda}
	\begin{split}
		u_1&=-\ii(\lambda_1+\lambda_2+\lambda_3), \qquad
		u_2=\ii(\lambda_1+\lambda_2-\lambda_3), \\
		u_3&=\ii(\lambda_1-\lambda_2+\lambda_3),	\qquad
		u_4=\ii(-\lambda_1+\lambda_2+\lambda_3).
	\end{split}
\end{equation}
When $\lambda_{1,2,3}$ satisfy the \ref{case1},  we set $0<\Im(\lambda_1)<\Im(\lambda_2)<\Im(\lambda_3)$, which implies $u_i\in \mathbb{R}$, $i=1,2,3,4$, with $u_4<u_3<u_2<u_1$, $u_4<u_3<0$, and $0<u_1$, corresponding to the case ii).
When $\lambda_{1,2,3}$ satisfy the \ref{case2}, we set $-\Re(\lambda_1)=\Re(\lambda_3)>0$, $\Im(\lambda_1)=\Im(\lambda_3)>0$ and $\Im(\lambda_2)>0$, which deduce $u_{1,3}\in \mathbb{R}$ with $u_1>u_3$ and $u_{2,4}\in \mathbb{C}\backslash(\ii \mathbb{R}\cup\mathbb{R})$ with $\Im(u_2)<0<\Im(u_4)$, corresponding to the case iii).

Considering solutions satisfying hyperelliptic integrals provided by the equation \eqref{eq:solution-u-ux}, we express the above solutions 
into the following forms.
When $u_{1,2,3,4}$ satisfy the case ii), by the elliptic integrals, we get
\begin{subequations}\label{eq:u1-elliptic}
\begin{align}
	& u(x,t)  =  u_2  +  \frac{(u_3-u_2)(u_4-u_2)}{(u_4-u_2) +(u_3-u_4)\sn^2(\alpha (x+vt),k)}\in [u_4,u_3], \label{eq:u1-elliptic-1}\\
	& u(x,t)  =  u_4 +  \frac{(u_1-u_4)(u_2-u_4)}{(u_2-u_4)+(u_1-u_2)\sn^2(\alpha (x+vt),k)}\in [u_2,u_1], \label{eq:u1-elliptic-2}
\end{align}
\end{subequations}
with $k^2=(u_1-u_2)(u_3-u_4)/((u_1-u_3)(u_2-u_4))$ and $\alpha^2=(u_1-u_3)(u_2-u_4)/4$.
The above solutions correspond to \ref{case1}.

When $u_{1,2,3,4}$ satisfy the case iii) with $u_1>u_3\in \mathbb{R}$ and $u_4=u_2^*\in \mathbb{C}\backslash(\ii \mathbb{R}\cup \mathbb{R})$, 
the solution $u(x,t)$ could be expressed as 
\begin{equation}\label{eq:u2-elliptic}
	u(x,t)=u_1+\frac{(u_3-u_1)(1-\cn(\alpha(x+vt),k))}{1+\delta+(\delta-1)\cn(\alpha(x+vt),k)}\in[u_3,u_1],\qquad \delta=\left|\frac{u_3-u_2}{u_1-u_2}\right|,
\end{equation} 	
$\alpha^2=\left|(u_3-u_2)(u_1-u_2)\right|$ and $k^2=(\alpha^2-(u_1-\Re(u_2))(u_3-\Re(u_2))-\Im^2(u_2))/(2\alpha^2)$.
This solution corresponds to the \ref{case2}.
These two group of solutions \eqref{eq:u1-elliptic} and \eqref{eq:u2-elliptic} are presented in the previous literature by the nonlinearization method \cite{Chen-19-mKdV}.

When the function $y^2(\lambda)$ with respect to $\lambda$ has three pairs of complex conjugate roots $\lambda_{1,2,3},\lambda_{1,2,3}^*$ and all of them are not zeros, which corresponds to the cases $s_0\neq 0$.  Combining with the equation \eqref{eq:define-u-lambda}, we get $2\ii \lambda_1=u_1+u_2$, $2\ii \lambda_2=u_1+u_4$, and $2\ii \lambda_3=u_1+u_3$.	
When $u_{1,2,3,4}\in \mathbb{R}$, the spectral parameters $\lambda_{1,2,3}$ must satisfy $\lambda_{1,2,3}\in \ii \mathbb{R}$. 
When $u_{1,4}\in \mathbb{R}$ and $u_{2,3}\in \mathbb{C}\backslash(\ii \mathbb{R}\cup \mathbb{R})$, only the parameter $\lambda_2\in \ii \mathbb{R}$ and the rest parameters $\lambda_{1,3}\in \mathbb{C}\backslash(\ii \mathbb{R}\cup\mathbb{R})$.
Then, the corresponding solution $u(x,t)$ could be expressed by the Riemann theta functions.
	
When one pair of branch points is zero, i.e., $\lambda_2=\lambda_2^*=0$, the determinant of the matrix function $\mathbf{L}$ expressed in equation \eqref{eq:det-L-lambda} could be rewritten as $\lambda^2\left(\lambda^4+2\alpha_1 \lambda^2+s_2\right)$, which corresponds to the cases $s_0=0$ and $s_2\neq 0$.
Excepting zero, the function $y^2(\lambda)$ with respect to $\lambda$ have only two pairs of complex conjugate roots $\lambda_{1,3},\lambda_{1,3}^*$ satisfying $\lambda_1=-\lambda_3^*\in \mathbb{C}$. 
Then, the corresponding solution $u(x,t)$ could be expressed by the dimension-$1$ Riemann theta functions via the algebro-geometric approach, which corresponds to the $\cn$-type solutions studied in our previous work \cite{LingS-23-mKdV-stability}.  When $\lambda_1=\lambda_1^*=0$ and $\lambda_{2,3}=-\lambda_{2,3}^*\in \ii \mathbb{R}$,  the corresponding solution $u(x,t)$ could be expressed by the Riemann theta functions related to the genus-one curves via the algebro-geometric approach, which corresponds to the $\dn$-type solution \cite{LingS-23-mKdV-stability}.

Based on the above relationships, we aim to demonstrate that solutions expressed in terms of Riemann theta functions in equation \eqref{eq:solutions-dn} are equivalent to those expressed in terms of elliptic functions in equations \eqref{eq:u1-elliptic} and \eqref{eq:u2-elliptic}. 
To this end, we seek to derive explicit expressions for certain parameters, which will facilitate a direct comparison between two different formulations.

\begin{lemma}\label{lemma:parameter-hat}
	Utilizing equations \eqref{eq:define-u-lambda}, \eqref{eq:tilde-v-1-2}, \eqref{eq:formula-trans-theta-elliptic}, \eqref{eq:Jacobi-double} and \eqref{eq:formula-theta-2tau-1},
	we obtain the following equations:
	\begin{subequations}\label{eq:equ-1}
		\begin{align}
			\frac{\vartheta_2(2\tilde{\nu}_2,2  \tau^{(1)}_2)}{\vartheta_3(2\tilde{\nu}_2,2  \tau^{(1)}_2)} & \  =
			\frac{\vartheta_3(2 \nu^{(1)} ,2  \tau^{(1)}_1)}{\vartheta_2(2 \nu^{(1)} ,2  \tau^{(1)}_1)},  \label{eq:equ-1-a}\\
			 \frac{\vartheta_2(2\tilde{\nu}_1,2  \tau^{(1)}_2)}{\vartheta_3(2\tilde{\nu}_1,2  \tau^{(1)}_2)}& \  =  -\frac{\vartheta_3(0,2  \tau^{(1)}_1)}{\vartheta_2(0,2  \tau^{(1)}_1)}, \label{eq:equ-1-b}\\
			C_{u_0}^{(1)}\frac{\vartheta_2(2 \nu^{(1)} ,2  \tau^{(1)}_2)}{\vartheta_2(0,2  \tau^{(1)}_2)}& \  =-\frac{u_3(u_2-u_4)  k^{(1)\prime}_2+u_4(u_3-u_2)}{(u_2-u_4)  k^{(1)\prime}_2+(u_3-u_2)}, \label{eq:equ-1-c}
		\end{align}
	\end{subequations}
	where the parameter $\nu^{(1)}$ is defined in equation \eqref{eq:u-parameters-dn}, parameters $\tilde{\nu}_{1,2}$ are defined as
	\begin{equation}\label{eq:tilde-v-1-2}
		\tilde{\nu}_1=\frac{\ii \pi}{  K^{(1)}_2} F\left(\frac{(u_2-u_4)^{1/2}}{(u_3-u_4)^{1/2}},  k^{(1)}_2\right),\qquad
		\tilde{\nu}_2=\frac{\ii \pi}{  K^{(1)}_2} F\left(\frac{(u_3(u_2-u_4))^{1/2}}{(u_2(u_3-u_4))^{1/2}},  k^{(1)}_2\right),
	\end{equation}
	and modulus $  k^{(1)}_{1,2}$ and the parameter $ C_{u_0}^{(1)}$ are defined in equation \eqref{eq:u-parameters-dn}.
\end{lemma}

\begin{proof}
		Consider the left-hand side of \eqref{eq:equ-1-a}.
		Firstly, we apply the formula \eqref{eq:formula-theta-2tau-1} to transform the parameter $\vartheta_i(\cdot,2  \tau^{(1)}_2)$ into the combination of Jacobi theta functions $\vartheta_i( \cdot,\tau^{(1)}_2)$. 
		Secondly, using the relationship between Jacobi theta functions and Jacobi elliptic functions provided in equation \eqref{eq:formula-trans-theta-elliptic}, we express this ratio in terms of elliptic functions.
		Then, based on the definition of the elliptic integral $F(u,k)$ ($\sn(F(u,k),k)=u$ in \Cref{define:elliptic-function}), we derive the explicit expressions in terms of parameters $u_{1,2,3,4}$.
		Finally, substituting the equation \eqref{eq:define-u-lambda} into it and simplifying, we obtain that the desired expressions in terms of $\lambda_{1,2,3}$ and $  k^{(1)}_{1,2}$.
		The detailed calculation is given as follows:
		\begin{equation}\nonumber
			\begin{split}
				\frac{\vartheta_2(2\tilde{\nu}_2,2  \tau^{(1)}_2)}{\vartheta_3(2\tilde{\nu}_2,2  \tau^{(1)}_2)}
				\xlongequal{\eqref{eq:formula-theta-2tau-1}}& \ 
				\frac{\vartheta_2(2\tilde{\nu}_2,  \tau^{(1)}_2)\vartheta_2(0,  \tau^{(1)}_2)-\vartheta_1(2\tilde{\nu}_2,  \tau^{(1)}_2)\vartheta_1(0,  \tau^{(1)}_2)}{\vartheta_3(2\tilde{\nu}_2,  \tau^{(1)}_2)\vartheta_3(0,  \tau^{(1)}_2)+\vartheta_4(2\tilde{\nu}_2,  \tau^{(1)}_2)\vartheta_4(0,  \tau^{(1)}_2)}\\
				\xlongequal{\eqref{eq:formula-trans-theta-elliptic}}&\
				\frac{\cn(2\tilde{\nu}_2  K^{(1)}_2/(\ii\pi),  k^{(1)}_2)  k^{(1)}_2}{\dn(2\tilde{\nu}_2  K^{(1)}_2/(\ii\pi),  k^{(1)}_2)+  k^{(1)\prime}_2}\\
				\xlongequal[\eqref{eq:define-u-lambda}]{\eqref{eq:Jacobi-double},\eqref{eq:tilde-v-1-2}} & \ 
				\frac{(\lambda_{1}^{2}-\lambda_{2}^{2}-\lambda_{3}^{2}) k^{(1)\prime}_1-(\lambda_{1}^{2}-\lambda_{2}^{2}+\lambda_{3}^{2})}{  k^{(1)}_1 \left(\lambda_{1}^{2}+\lambda_{2}^{2}-\lambda_{3}^{2}\right)}.
			\end{split}
		\end{equation}
		By using equations \eqref{eq:define-u-lambda}, \eqref{eq:tilde-v-1-2}, \eqref{eq:formula-trans-theta-elliptic}, \eqref{eq:Jacobi-double} and \eqref{eq:formula-theta-2tau-1}, we obtain that the right-hand side of equation \eqref{eq:equ-1-a} also satisfies the above result.
		Therefore, the equation \eqref{eq:equ-1-a} holds.

		Similarly, we obtain the equation \eqref{eq:equ-1-b}.
			For equation \eqref{eq:equ-1-c}, we obtain
			\begin{equation}\nonumber
				\begin{split}
					C_{u_0}^{(1)}\frac{\vartheta_2(2 \nu^{(1)} ,2  \tau^{(1)}_2)}{\vartheta_2(0,2  \tau^{(1)}_2)}\xlongequal[\eqref{eq:formula-theta-2tau-1},\eqref{eq:formula-theta-2-1}]{\eqref{eq:u-parameters-dn-B}}&\
					\frac{\ii \lambda_3 \vartheta_3(0 ,  \tau^{(1)}_1)\vartheta_4( 0,  \tau^{(1)}_1)(\vartheta_3^2( \nu^{(1)} ,  \tau^{(1)}_2)-\vartheta_4^2( \nu^{(1)} ,  \tau^{(1)}_2))}{\vartheta_3( \nu^{(1)} ,  \tau^{(1)}_1)\vartheta_4( \nu^{(1)} ,  \tau^{(1)}_1)(\vartheta_3^2(0,  \tau^{(1)}_2)-\vartheta_4^2(0,  \tau^{(1)}_2))}\\
					\xlongequal[\eqref{eq:define-second-integral}]{\eqref{eq:formula-trans-theta-elliptic}}&\ \frac{  k^{(1)\prime}_2(\lambda_1^2-\lambda_3^2)+\lambda_1\lambda_2}{\ii\lambda_3}.
				\end{split}	
			\end{equation}
			 Combining the above equations with the relationship between parameters $\lambda_{1,2,3}$ and $u_{1,2,3,4}$, we obtain the equation \eqref{eq:equ-1-c}. 
\end{proof}

\begin{lemma}\label{lemma:parameter-check}
		Through utilizing equations \eqref{eq:u-parameters-cn}, \eqref{eq:define-u-lambda},\eqref{eq:u2-elliptic}, \eqref{eq:tilde-v-1-2}, \eqref{eq:formula-trans-theta-elliptic}, \eqref{eq:Jacobi-double} and  \eqref{eq:formula-theta-2tau-1},
		we obtain the following equations:
	\begin{subequations}\label{eq:equ-2}
		\begin{align}
			&\sqrt{\frac{  k^{(2)\prime}_2  k^{(2)\prime}_1}{k^{(2)}_2  k^{(2)}_1}}=
			\frac{1+\delta}{1-\delta}, \label{eq:equ-2-a}
			\\
			&\frac{u_1\delta+u_3}{u_1\delta-u_3}\xlongequal[\eqref{eq:define-u-lambda}]{\eqref{eq:u2-elliptic}}\frac{(\lambda_2^2-A-B)(A-B)}{\lambda_2(\lambda_1+\lambda_3)(A+B)},  \label{eq:equ-2-b}\\
			&\frac{C_{u_0}^{(2)}\vartheta_2( \nu^{(2)} , \tau^{(2)} _1)}{\vartheta_2(0, \tau^{(2)} _1)}\ee^{\frac{-\ii ( \tau^{(2)} _1+ \tau^{(2)} _2)\pi}{4}- \frac{\nu^{(2)}}{2} }
			=\frac{u_1\delta -u_3}{\delta-1}, \label{eq:equ-2-c}
		\end{align}
	\end{subequations}
	where $C_{u_0}^{(2)}$ is defined in equation \eqref{eq:u-parameters-cn}.
\end{lemma}
\begin{proof}
	Based on the definition of the parameter $\delta$ defined in equation \eqref{eq:u2-elliptic}, we obtain 
	\begin{equation}\label{eq:1-1}
		\begin{split}
			& \frac{1+\delta}{1-\delta}
			\xlongequal[\eqref{eq:u-parameters-cn}]{\eqref{eq:u2-elliptic},\eqref{eq:define-u-lambda}}
			\frac{\lambda_2^2-A-B}{\lambda_2(\lambda_1+\lambda_3)},\\
		\end{split}
	\end{equation}
	since $u_4=u_2^*$. 
    By the equation \eqref{eq:u-parameters-cn}, it is easy to obtain that the left-hand side of equation \eqref{eq:equ-2-a} also satisfy the above result. Thus, we obtain the equation \eqref{eq:equ-2-a}.
	Similarly, we obtain the equation \eqref{eq:equ-2-b}.
	By utilizing the method provided in \Cref{lemma:parameter-hat} for calculating the equation \eqref{eq:equ-1-a}, we obtain the equation \eqref{eq:equ-2-c}, which we would not repeat anymore.
	For equation \eqref{eq:equ-2-c},
	we deduce 
	\begin{equation}
		\begin{split}
			\frac{C_{u_0}^{(2)}\vartheta_2( \nu^{(2)} , \tau^{(2)} _1)}{\vartheta_2(0, \tau^{(2)} _1)}\ee^{\frac{-\ii ( \tau^{(2)} _1+ \tau^{(2)} _2)\pi}{4}-\frac{ \nu^{(2)}}{2} }
			\xlongequal[\eqref{eq:Jacobi-Theta-K-iK}]{\eqref{eq:u-parameters-cn-B}}&\ \frac{(A-B)\vartheta_4(0, \tau^{(2)} _1)\vartheta_2( \nu^{(2)} , \tau^{(2)} _1)}{\ii\lambda_2\vartheta_4(\nu^{(2)},\tau^{(2)} _1)\vartheta_2(0, \tau^{(2)} _1)}
			\xlongequal[\eqref{eq:define-first-integral}]{\eqref{eq:formula-trans-theta-elliptic}}  
			\frac{A+B}{\ii \lambda_2}.
		\end{split}	
	\end{equation}
\end{proof}

\newenvironment{proof-prop-equivalent}{\emph{\textbf{Proof of \Cref{prop:solutions-equivalent}.}}}{\hfill$\Box$\medskip}
\begin{proof-prop-equivalent}
	We focus on demonstrating the equivalence between solutions represented by Riemann theta functions and those expressed via elliptic functions.
	By invoking Liouville’s Theorem and verifying that both forms share the same poles, zeros, and initial values, we establish that the two types of solutions are indeed identical.
	
	Consider the solution $u(x,t)$ with parameters given in \Cref{theorem:solution-u} and $D^{(1)}=\mathbf{0}$, as follows
	\begin{equation}\label{eq:u-convert-p}
		\begin{split}
				u(x,t)
			\xlongequal[\eqref{eq:formula-Rieman-shift-1}, \eqref{eq:Jacobi-Theta-K-iK}]{\eqref{eq:solutions-dn}, \eqref{eq:u-parameters-dn}} 
			&\ C_{u_0}^{(1)}\frac{\vartheta_2(2 \nu^{(1)} ,2  \tau^{(1)}_1)}{\vartheta_2(0 ,2  \tau^{(1)}_1)} \cdot \frac{\frac{\vartheta_3(2 \nu^{(1)} ,2  \tau^{(1)}_1)}{\vartheta_2(2 \nu^{(1)} ,2  \tau^{(1)}_1)}-\frac{\vartheta_2(2\ii  \kappa^{(1)}  (x+vt),2  \tau^{(1)}_2)}{\vartheta_3(2\ii  \kappa^{(1)}  (x+vt),2  \tau^{(1)}_2)}}{\frac{\vartheta_3(0,2  \tau^{(1)}_1)}{\vartheta_2(0,2  \tau^{(1)}_1)}+\frac{\vartheta_2(2\ii  \kappa^{(1)}  (x+vt),2  \tau^{(1)}_2)}{\vartheta_3(2\ii  \kappa^{(1)}  (x+vt),2  \tau^{(1)}_2)}},\\
		\end{split}
	\end{equation}
	by utilizing the formulas provided in \Cref{prop:Riemann-2-1}.
	According to equations \eqref{eq:equ-1-a} and \eqref{eq:equ-1-b}
	and considering $2\ii  \kappa^{(1)} (x+vt)$ as a whole, we deduce that the poles and zeros of the solution are $2\ii  \kappa^{(1)} (x+vt)=2\tilde{\nu}_1+4\ii n\pi +4\ii m  \tau^{(1)}_2\pi $, $n,m\in \mathbb{Z}$ and $2\ii  \kappa^{(1)} (x+vt)=2\tilde{\nu}_2+4\ii n\pi +4\ii m  \tau^{(1)}_2\pi $, $n,m\in \mathbb{Z}$, respectively.
	Substituting equation \eqref{eq:define-u-lambda} into the definitions of $\alpha,k$ in equation \eqref{eq:u1-elliptic}, we find $k=  k^{(1)}_2$ and $\alpha= \kappa^{(1)}   K^{(1)}_2/\pi$. 
	Then, we obtain 
	\begin{equation}\label{eq:u-convert-p-2}
		\begin{split}
			u(x,t)\xlongequal{\eqref{eq:u1-elliptic}} 
			&\ u_2\frac{\sn^2(\tilde{\nu}_2 K_2^{(1)}/(\ii \pi) ,k^{(1)}_2)-\sn^2(\alpha (x+vt),k^{(1)}_2)}{\sn^2(\tilde{\nu}_1 K_2^{(1)}/(\ii \pi) ,k^{(1)}_2)-\sn^2(\alpha (x+vt),k^{(1)}_2)},
		\end{split}
	\end{equation}
	where $\tilde{\nu}_{1,2}$ are defined in equation \eqref{eq:tilde-v-1-2}.
	By the properties of Jacobi elliptic functions, it is easy to get that the zeros and poles of this solution $u(x,t)$ are
	$\alpha (x+vt)= \kappa^{(1)} K^{(1)}_2 (x+vt)/\pi=\tilde{\nu}_2 K_2^{(1)}/(\ii \pi)+2nK^{(1)}_2+2\ii m K^{(1)\prime}_2$ and $\alpha (x+vt)= \kappa^{(1)} K^{(1)}_2 (x+vt)/\pi=\tilde{\nu}_1 K_2^{(1)}/(\ii \pi)+2nK^{(1)}_2+2\ii m K^{(1)\prime}_2$, $n,m\in \mathbb{Z}$, respectively.	
	It follows that the solution $u(x,t)$ in equation \eqref{eq:u-convert-p-2} possesses the same poles and zeros as in equation \eqref{eq:u-convert-p}. 
	By substituting $x+vt=\pi/(2\kappa^{(1)})$ into both representations of $u(x,t)$ and applying \Cref{lemma:parameter-hat} again, we conclude that the two forms of the solution are indeed identical, through utilizing the Liouville Theorem and the equation \eqref{eq:equ-1}.
	
	Similar results can be obtained for other values of the parameter $D^{(1)}=\ii   \tau^{(1)}_2 \mathbf{1}$, and the details are omitted here for brevity.
	Therefore, the solutions expressed in terms of Riemann theta functions with two different values of $D^{(1)}$, as given in equations \eqref{eq:solutions-dn} and \eqref{eq:u-parameters-dn}, are shown to be equivalent to the solution expressed in terms of elliptic functions in equation \eqref{eq:u1-elliptic}.
	
	Next, we consider the solution $u(x,t)$ under the \ref{case2} as follows:
	\begin{equation}\nonumber
		\begin{split}
			 &\ u(x,t)\\
			\xlongequal[\eqref{eq:formula-Rieman-shift-1}, \eqref{eq:Jacobi-Theta-K-iK}]{\eqref{eq:solutions-dn}, \eqref{eq:u-parameters-cn}}&\ \frac{C_{u_0}\left(\vartheta_4( \nu^{(2)} , \tau^{(2)} _1)\vartheta_4(2\ii  \kappa^{(2)} (x+vt), \tau^{(2)} _2)-\vartheta_2( \nu^{(2)} , \tau^{(2)} _1)\vartheta_2 (2\ii  \kappa^{(2)} (x+vt), \tau^{(2)} _2)\right)}{(\vartheta_4(0, \tau^{(2)} _1)\vartheta_4(2\ii  \kappa^{(2)} (x+vt), \tau^{(2)} _2)-\vartheta_2(0, \tau^{(2)} _1)\vartheta_2(2\ii  \kappa^{(2)} (x+vt), \tau^{(2)} _2))\ee^{\frac{\ii ( \tau^{(2)} _1+ \tau^{(2)} _2)\pi+2\nu^{(2)}}{4}}}\\
			\xlongequal{\eqref{eq:formula-trans-theta-elliptic}}&\ \frac{C_{u_0}\vartheta_3( \nu^{(2)} , \tau^{(2)} _1)}{\vartheta_2(0, \tau^{(2)} _1)}\frac{\cn( \alpha^{(2)} (x+vt),  k^{(1)}_2)-(B-A)(  k^{(1)\prime}_1  k^{(1)\prime}_2/(  k^{(1)}_1  k^{(1)}_2))^{1/2}/(A+B)}{\left(\cn(\alpha^{(2)} (x+vt),  k^{(1)}_2)-(  k^{(1)\prime}_1  k^{(1)\prime}_2/(  k^{(1)}_1  k^{(1)}_2))^{1/2}\right)\ee^{\frac{\ii ( \tau^{(2)} _1+ \tau^{(2)} _2)\pi+2\nu^{(2)}}{4}}},
		\end{split}
	\end{equation}
	where $ \alpha^{(2)} = 2\kappa^{(2)}   K^{(2)}_2/\pi$ and $\vartheta_4(0, \tau^{(2)} _1)/\vartheta_2(0, \tau^{(2)} _1)=(  k^{(2)\prime}_1/  k^{(2)}_1)^{1/2}$.
	By \Cref{lemma:parameter-check}, we find that the above solution and the one given in equation \eqref{eq:u2-elliptic} share the same poles and zeros. Setting $(x+vt)=\pi/(2 \kappa^{(2)} )$ and applying equation \eqref{eq:equ-2}, we further conclude that two forms of the solution are identical. Therefore the function $u(x,t)$ expressed in these two different forms is indeed consistent.
\end{proof-prop-equivalent}

Next, we are going to calculate the explicit expressions of the solutions for the Lax pair with the above mentioned two-phase solutions $u(x,t)$ in equation \eqref{eq:solutions-dn}. 
Before providing the explicit expressions of solutions, we would provide the explicit expressions of the Abel integrals.

\begin{lemma}\label{lemma:Abel-integral}
	Under the different cases of the branch points, the Abel map $\mathcal{A}^{(i)}_{\infty^{-}}(P)$ and functions $\Omega_{1,2,3}^{(i)}(P)$, $i=1,2$ could be expressed by elliptic functions as follows.
\end{lemma}

\begin{proof}
	For the \ref{case1}, we would like to study the vector $\mathcal{A}_{\infty^{-}}^{(1)}(P)$ and the functions $\Omega_{1,2,3}^{(1)}(P)$, where $P$ is a point on the hyperelliptic curve lying over $\lambda$.
	For the definition of the function $\mathcal{A}_{\infty^{-}}(P)$ in \eqref{eq:abel_map} and the quantities $w_{1,2}$ listed in equation \eqref{eq:define-basis-B}, we obtain the first component of the Abel map:
	\begin{equation}\label{eq:define-AP-P}
		\begin{split}
			\left(\mathcal{A}^{(1)}_{\infty^{-}}(P)\right)_{1,2}=&\ \int_{\infty^{-}}^P w_{1,2} \dd \lambda 
			\xlongequal[\eqref{eq:fs-1}]{\eqref{eq:fs-2}}
			\frac{ -\nu^{(1)} _2\pm(\nu^{(1)} _1-\ii\pi+ \nu^{(1)}) }{2},
		\end{split}
	\end{equation}
	with $K^{(1)}_1-\nu^{(1)}_{\infty}= \nu^{(1)} K^{(1)}_1/(\ii \pi)$,
	where $ \nu^{(1)} _{\infty}$ is defined in equation \eqref{eq:define-nu-hat-infinity} and $ \nu^{(1)} _{1,2}$ are defined as 
	\begin{equation}\label{eq:define-nu-1-12}\nonumber
		\nu^{(1)} _1=\frac{\ii\pi }{  K^{(1)}_1} F\left(\frac{\lambda(\lambda_1^2-\lambda_3^2)^{1/2}}{\lambda_3(\lambda_1^2-\lambda^2)^{1/2}},  k^{(1)}_1\right), \quad 
		\nu^{(1)} _2=\frac{\ii \pi }{  K^{(1)}_2} F\left(\frac{(\lambda_1^2-\lambda_3^2)^{1/2}}{(\lambda_1^2-\lambda^2)^{1/2}},  k^{(1)}_2\right).
	\end{equation}
	Considering the definition of $\Omega_{i}^{(1)}(P)$, $i=1,2,3$, we set 
	\begin{equation}\label{eq:define-Omega-23-c}
		\dd\Omega_2^{(1)}=\sum_{i=0}^3c_{2i}\frac{\lambda^i}{y}\dd \lambda,
		\qquad 
		\dd\Omega_3^{(1)}=\sum_{i=0}^5 c_{3i}\frac{\lambda^i}{y}\dd \lambda.
	\end{equation}
	Considering $\lambda \rightarrow \infty$ together with equations \eqref{eq:define-int-w2}-\eqref{eq:define-int-w3}, we can obtain parameters $c_{23}=1$, $c_{22}=c_{34}=c_{32}=0$, $c_{35}=12$, $c_{33}=-3v$, where $v$ is defined in equation \eqref{eq:solutions-dn}.
	Combining with the conditions that all integrals over the $a_i$-circle vanish, 
	we determine all the values of coefficients $c_{ij}$.
	From the relations $\oint_{a_i}\dd \Omega_2^{(1)}=0$ and $\oint_{a_i}\dd \Omega_3^{(1)}=0$, for $i=1,2$, in equations \eqref{eq:define-int-w2} and \eqref{eq:define-int-w3}, we obtain $c_{20}=c_{30}=0$, $c_{21}=(\lambda_3^2-\lambda_1^2) E^{(1)}_2/  K^{(1)}_2-\lambda_3^2$, and $c_{31}=4v_1-\lambda_3^2v-v(\lambda_1^2-\lambda_3^2) E^{(1)}_2/  K^{(1)}_2$,  where $v_1$ is defined in \Cref{lemma:dn-2-int}.
	Using the definition of $\Omega_{1,2,3}^{(1)}(P)$ in equations \eqref{eq:define-Omega-1-c} and \eqref{eq:define-Omega-23-c}, together with the elliptic integrals, we obtain
	\begin{equation}\label{eq:define-Omega-123-P}
		\begin{split}
			\Omega_2^{(1)}(P)
			\xlongequal[\eqref{eq:elliptic-funda-2}]{\eqref{eq:define-Omega-23-c}}&\ 
			-\ii \sqrt{\lambda_1^2-\lambda_3^2} \int_{   K^{(1)}_2}^{\frac{ \nu^{(1)} _2  K^{(1)}_2}{\ii \pi}}  \frac{k^2\sn^4(\nu,  k^{(1)}_2)-1}{\sn^2(\nu,  k^{(1)}_2)}+\dn^2(\nu,  k^{(1)}_2)-\frac{ E^{(1)}_2}{  K^{(1)}_2}\dd \nu\\
			\xlongequal[\eqref{eq:define-first-integral}]{\eqref{eq:Zeta-define}}&\ 
			 -\ii \sqrt{\lambda_1^2-\lambda_3^2}\, Z\left(\frac{ \nu^{(1)} _2  K^{(1)}_2}{\ii\pi},  k^{(1)}_2\right)-\ii \sqrt{\frac{(\lambda_3^2-\lambda^2)(\lambda_2^2-\lambda^2)}{(\lambda_1^2-\lambda^2)}},\\
			\Omega_3^{(1)}(P)\xlongequal{\eqref{eq:define-Omega-23-c}}&\ 
			\int_{P_0}^{P}\frac{12\lambda^5-4v\lambda^3+4v_1\lambda}{y}\dd \lambda+v\Omega_2^{(1)}(P)=
			4y+v\Omega_2^{(1)}(P),\\
			\Omega_1^{(1)}(P)
			\xlongequal[\eqref{eq:fs-2}]{\eqref{eq:fs-1}}&\, \frac{-\ii }{\lambda_2\sqrt{\lambda_1^2-\lambda_3^2}}	\int_{  K^{(1)}_1}^{\frac{ \nu^{(1)} _1  K^{(1)}_1}{\ii \pi}}\frac{\lambda_1^2\lambda_3^2\sn^2(\nu,  k^{(1)}_1)}{\lambda_1^2-\lambda_3^2\cn^2(\nu,  k^{(1)}_1)}-\frac{\lambda_1^2}{  K^{(1)}_1}\Pi\left( \frac{\lambda_2^2-\lambda_1^2}{\lambda_2^2},  k^{(1)}_1 \right) \dd \nu\\
			\xlongequal[\eqref{eq:add-app}]{\eqref{eq:formula-trans-Pi-Zeta-theta}}&\, \frac{1}{2}\ln\left(\frac{\vartheta_2( \nu^{(1)}+ \nu^{(1)} _1 ,  \tau^{(1)}_1)}{\vartheta_2( \nu^{(1)} - \nu^{(1)}_1 ,  \tau^{(1)}_1)}\right)+\frac{\ii \pi}{2}.	
		\end{split}
	\end{equation}
	
	For the \ref{case2}, from the definition of the Abel map $\mathcal{A}_{\infty^-}^{(2)}(P)$,
	we get the first component:
	\begin{equation}\label{eq:define-AP-C}
			\!\!\!\! (\mathcal{A}^{(2)}_{\infty^-}(P))_{1}
			=\int_{\infty^{-}}^{P}w_{1} \dd \lambda
			= \Delta_1^{(2)}+ \frac{ \nu^{(2)} _2}{2}, \,
			(\mathcal{A}^{(2)}_{\infty^-}(P))_{2}
			=\int_{\infty^{-}}^{P}w_{2} \dd \lambda
			= \Delta_2^{(2)}+ \frac{ \nu^{(2)}_1- \nu^{(2)}- \nu^{(2)}_2}{4}, 
	\end{equation}
	where $\nu^{(2)}_1 =\ii \pi F\left(\frac{2(AB)^{1/2}(\lambda_2^2-\lambda^2)^{1/2}\lambda}{\lambda^2(A-B)+\lambda_2^2B},  k^{(2)}_1\right)/K^{(2)}_1$,
	 $ \nu^{(2)}_2=\ii \pi F\left(\frac{2A^{1/2}(\lambda_2^2-\lambda^2)^{1/2}}{\lambda_2^2+A-\lambda^2},  k^{(2)}_2\right)/ K^{(2)}_2.$
	From the definition of integrals $\Omega_{1,2,3}^{(2)}(P)$ defined in equations \eqref{eq:define-Omega-1-c} and \eqref{eq:define-Omega-23-c}, together with the formulas of elliptic integrals, we get 
	\begin{equation}\label{eq:define-Omega-123-C}
		\begin{split}
			\Omega_2^{(2)}(P) \xlongequal[\eqref{eq:define-U-1-2-C}]{\eqref{eq:fs-3}}& \int_{  K^{(2)}_2+\ii   K^{(2)\prime}_2}^{\frac{  K^{(2)}_2 \nu^{(2)} _2}{\ii \pi}} \ii \sqrt{A} \left(\frac{\cn(\nu,  k^{(2)}_2)}{1-\cn(\nu,  k^{(2)}_2)} +\frac{ E^{(2)}_2}{  K^{(2)}_2}\right)\dd \nu
			 \\=& \! \left. -\ii \sqrt{A} \left(\frac{\sn(\nu,  k^{(2)}_2)\dn(\nu,  k^{(2)}_2)}{1-\cn(\nu,  k^{(2)}_2)} +Z(\nu,  k^{(2)}_2)\right)\right|_{  K^{(2)}_2+\ii   K^{(2)\prime}_2}^{\frac{  K^{(2)}_2 \nu^{(2)} _2}{\ii \pi}}\\
			=& - \ii \sqrt{A}\, Z\left(\frac{  K^{(2)}_2 \nu^{(2)} _2}{\ii \pi },  k^{(2)}_2\right)- \frac{y}{\left( \lambda^2_2+A-\lambda^2\right)} +\frac{ \sqrt{A} \pi }{2  K^{(2)}_2},\\
			\Omega_1^{(2)}(P)\xlongequal[\eqref{eq:define-U-1-2-C},\eqref{eq:formula-trans-Pi-Zeta-theta}]{\eqref{eq:fs-4},\eqref{eq:fs-3}}&\ \frac{1}{4}\ln\left(\frac{\vartheta_1( \nu^{(2)} _1+ \nu^{(2)} ,\tau_1^{(2)})}{\vartheta_1( \nu^{(2)} _1- \nu^{(2)} ,\tau_1^{(2)})}\right)
			+\frac{\ii (2- \tau^{(2)} _1- \tau^{(2)} _2)\pi- \nu^{(2)} - \nu^{(2)} _1+ \nu^{(2)} _2}{4}\\&-\frac{1}{4}\ln\left(\frac{\sqrt{(\lambda_1^2-\lambda^2)(\lambda_3^2-\lambda^2)}+\ii \sqrt{(\lambda_2^2-\lambda^2)\lambda^2}}{\sqrt{(\lambda_1^2-\lambda^2)(\lambda_3^2-\lambda^2)}-\ii \sqrt{(\lambda_2^2-\lambda^2)\lambda^2}}\right) \! ,
			\\
		\end{split}
	\end{equation}
	and $\Omega_3^{(2)}(P) = 4y + v \Omega_2^{(2)}(P)$.
\end{proof}

\newenvironment{proof-solution-dn-Phi}{\emph{\textbf{Proof of \Cref{theorem:solution-Phi}.}}}{\hfill$\Box$\medskip}
\begin{proof-solution-dn-Phi}
	For the \ref{case1} and \ref{case2}, we consider $\Phi_{1}(x,t;P)$ and $\Phi_{2}(x,t;P)$ in equation \eqref{eq:Phi-expression}. By virtue of Lemmas \ref{lemma:U-V-R-dn} and \ref{lemma:Abel-integral}, it follows readily that the vector solutions to the Lax pair of the mKdV equation are given by \eqref{eq:Phi-solution}
	after eliminating the constant $\Theta(D^{(i)}+\mathcal{A}^{(i)}_{\infty^-}(P))/\Theta(D^{(i)})$ and using the relation $C_{u_0}^{(i)}=2\ii/\omega_1^{(i)}$.
\end{proof-solution-dn-Phi}

\section{The spectral stability of two-phase solutions }\label{section:spectral-stability}

In this section, we investigate the spectral stability of the two-phase solutions for the mKdV equation.  
To carry out the stability analysis, we consider the bounded function $W(\xi;\Omega)$, whose exponential factor is required to have zero real part. 
According to equations \eqref{eq:define-W} and \eqref{eq:M}, the corresponding parameter $\lambda$ must lie in the set $Q$ defined in equation \eqref{eq:set-Q}.
The analysis is further divided into two cases, depending on the configuration of the branch points.

By considering the stationary periodic traveling wave solutions  under the transformation \eqref{eq:transformation-xi-eta},  and utilizing the modified squared eigenfunction method together with  \Cref{theorem:solution-Phi}, we obtain the eigenvalue and eigenfunctions of the linearized spectral problem of the \ref{eq:mKdV} equation as
\begin{equation}\label{eq:define-W}
	\begin{split}
	\!\!	&\ W(\xi;\Omega)\\
	=&\ \left(\Phi_{1}^2-\Phi_{2}^2\right)\ee^{-2\ii (\Omega_3^{(i)}(P)-v\Omega_2^{(i)}(P)) \eta}\\
	\xlongequal{\eqref{eq:Phi-solution}}& \left(\frac{\Theta^2(D^{(i)}+\mathcal{A}_{\infty^{-}}^{(i)}(P)+\ii U^{(i)}\xi)}{\Theta^2(D^{(i)}+\ii U^{(i)}\xi)}-\frac{\Theta^2(D^{(i)}+\mathcal{A}_{\infty^{-}}^{(i)}(P)+\ii U^{(i)}\xi -\Delta^{(i)})}
		{\Theta^2(D^{(i)}+\ii U^{(i)}\xi)\ee^{-4\ii \omega_{2}^{(i)}\xi-2\Omega_1^{(i)}(P)}}\right) \!
		\ee^{2\ii \left(\Omega_2^{(i)}(P) + \omega_2^{(i)} \right)\xi},\!\!
	\end{split}
\end{equation}
with the eigenvalue defined as 
\begin{equation}\label{eq:define-Omega}
		\Omega(\lambda)=2\ii (\Omega_3^{(i)}(P)-v\Omega_2^{(i)}(P))=8\ii y.
\end{equation}
The method for constructing these solutions are provided in previous work \cite{LingS-23-mKdV-stability}
and will not be repeated here for brevity.

For the Floquet theorem (see \cite{DeconinckK-06,Floquet-83}), the solution $W(\xi;\Omega)$ of the linear homogeneous differential equation \eqref{eq:linearized-mKdV} is of the form
$W(\xi;\Omega)=\ee^{\ii \hat{\eta} \xi} \hat{W}(\xi;\Omega), \hat{W}(\xi+2T;\Omega)=\hat{W}(\xi;\Omega),\hat{\eta}\in \mathbb{C}$,
where $2T$ is the period of the function $\hat{W}(\xi ;\Omega )$ and $\hat{\eta}$ is defined in equation \eqref{eq:M}.
According to \Cref{define:spect-stable}, studying the spectral stability of the genus-two periodic traveling solution is equivalent to
examining all values of $\Omega(\lambda)$ for any spectral parameter $\lambda$ satisfying $\Im(\mathcal{I}(\lambda))=0$ defined in equation \eqref{eq:M} and determining whether these eigenvalues $\Omega(\lambda)$ are purely imaginary.
Therefore, it is crucial to analyze the curve $\Im(\mathcal{I}(\lambda))=0$, i.e., the set $Q$ defined in equation \eqref{eq:set-Q}.

If we get a point $\lambda_0$ satisfying $\Im(\mathcal{I}(\lambda_0))=0$, we obtain the curve $\Im(\mathcal{I}(\lambda))=0$ which goes through the point $\lambda_0$ along the tangent vector.
Differentiating with respect to ${\lambda}_{R},{\lambda}_{I}$ on the curve ${\Im}\left(\mathcal{I}(\lambda) \right)=0$, we get the tangent vector 
\begin{equation}\label{eq:tangent-vector}
	\left(-\frac{\dd {\Im} (\mathcal{I})}{\dd \lambda_{I}}, \frac{\dd {\Im} (\mathcal{I})}{\dd \lambda_{R}} \right)= \left(-{\Re}(\mathcal{I}'(\lambda)), {\Im}(\mathcal{I}'(\lambda)) \right), \qquad \mathcal{I}'(\lambda):=\frac{\dd \mathcal{I}(\lambda)}{\dd \lambda},
\end{equation}   
where $\lambda_{R},\lambda_{I}$ denote the real and imaginary part of $\lambda$ respectively. 
Then, we are going to study the spectral stability of the genus-two solution under two different cases of branch points $\lambda_i$, $i=1,2,3$.

\subsection{The spectral stability analysis for \ref{case1}}\label{subsec:spectral-dn}
On the basis of the above analysis, to study the spectral stability of aforementioned periodic solutions, we merely need to consider the set $Q$ defined in equation \eqref{eq:set-Q}. 
By the equation \eqref{eq:M} and the \Cref{theorem:solution-Phi}, we get 
\begin{equation}\label{eq:I-define-P-dn}
	\begin{split}
		\mathcal{I}^{(1)}(\lambda)\xlongequal[\eqref{eq:define-Omega-123-P}]{\eqref{eq:M}}&-2\ii \sqrt{\lambda_1^2-\lambda_3^2}\, Z\left(\frac{ \nu^{(1)} _2  K^{(1)}_2}{\ii \pi},  k^{(1)}_2\right)-2\ii \sqrt{\frac{(\lambda_3^2-\lambda^2)(\lambda_2^2-\lambda^2)}{(\lambda_1^2-\lambda^2)}}.
	\end{split}
\end{equation}

\begin{lemma}\label{lemma:dn-bounded}
	The set $Q$ is equivalent to the set $ Q^{(1)} := Q^{(1)} _R\cup  Q^{(1)} _{P_1}\cup  Q^{(1)} _{P_2}$, defined in equation \eqref{eq:define-hat-Q}.
	
\end{lemma}
\begin{proof}
	The proof of the present Lemma will be divided into the following two steps. 
	We want to prove $ Q^{(1)} \subseteq Q$.
	For any $\lambda\in 
	 Q^{(1)} _{P_1} \cup  Q^{(1)} _{R}$, since $((\lambda_1^2-\lambda_3^2)/(\lambda_1^2-\lambda^2))^{1/2}\in \ii \mathbb{R}$, the function $ \nu^{(1)} _2$ is real, which implies $ \nu^{(1)} _2  K^{(1)}_2/(\ii \pi)\in \ii  \mathbb{R}$.
	Combined with the definition of the Zeta function, it is easy to obtain that $Z( \nu^{(1)} _2  K^{(1)}_2/(\ii \pi),  k^{(1)}_2)\in \ii \mathbb{R}$.
	Thus, we obtain that for any $\lambda\in 
	 Q^{(1)} _{P_1} \cup  Q^{(1)} _{R}$, $\mathcal{I}^{(1)}(\lambda)\in \mathbb{R}$, the equation $\Im(\mathcal{I}^{(1)}(\lambda))=0$ holds.
	Then, we consider the case $\lambda \in  Q^{(1)} _{P_2}$.
	It is well known that the second term on the right-hand side of the equation \eqref{eq:I-define-P-dn} is real.
	Utilizing the addition formula \eqref{eq:add-first}, we get
	\begin{equation}\label{eq:F-trans-case1-2}
		\frac{\nu^{(1)} _2  K^{(1)}_2}{\ii \pi}\xlongequal{\eqref{eq:add-first}}K^{(1)}_2-F\left(\sqrt{\frac{\lambda^2-\lambda_3^2}{\lambda^2-\lambda_2^2}},  k^{(1)}_2\right), \quad \text{with} \quad F\left(\sqrt{\frac{\lambda^2-\lambda_3^2}{\lambda^2-\lambda_2^2}},  k^{(1)}_2\right)\in \ii \mathbb{R}, \quad \lambda\in  Q^{(1)} _{P_2}.
	\end{equation}
	Thus, for any $\lambda\in  Q^{(1)} _{P_2}$, we obtain
	\begin{equation}\nonumber
		\begin{split}
			Z\left(\frac{ \nu^{(1)} _2  K^{(1)}_2}{\ii \pi},  k^{(1)}_2\right)\xlongequal{\eqref{eq:F-trans-case1-2}} &\ Z \left(K^{(1)}_2-F\left(\sqrt{\frac{\lambda^2-\lambda_3^2}{\lambda^2-\lambda_2^2}},  k^{(1)}_2\right),  k^{(1)}_2\right) \\
			\xlongequal{\eqref{eq:add-app}}&\ Z\left( F\left(\sqrt{\frac{\lambda^2-\lambda_3^2}{\lambda^2-\lambda_2^2}},  k^{(1)}_2\right),k^{(1)}_2\right)\\
			&\ +  (k^{(1)}_2)^2\sn\left(\frac{ \nu^{(1)} _2  K^{(1)}_2}{\ii \pi},  k^{(1)}_2\right)\sn\left( F\left(\sqrt{\frac{\lambda^2-\lambda_3^2}{\lambda^2-\lambda_2^2}},  k^{(1)}_2\right),  k^{(1)}_2\right)\in \ii \mathbb{R}.
		\end{split}
	\end{equation}
	Therefore, we obtain $\Im(\mathcal{I}^{(1)}(\lambda))=0$. Furthermore, we get $ Q^{(1)} \subseteq Q$.

	Conversely, we want to prove $Q\subseteq Q^{(1)}$.
	For the definition of the Zeta function (\Cref{define:zeta}) and the elliptic function (\Cref{define:elliptic-function}), the derivative of the function $\mathcal{I}^{(1)}(\lambda)$ with respect to the spectral parameter $\lambda$ is
	\begin{equation}\nonumber
		\begin{split}
			(\mathcal{I}^{(1)}(\lambda))^{\prime}
			 \xlongequal[\eqref{eq:diff-elliptic},\eqref{eq:define-first-integral}]{\eqref{eq:I-define-P-dn}, \eqref{eq:Zeta-define}}&\ \frac{2\left(\dn^2(\frac{ \nu^{(1)} _2  K^{(1)}_2}{\ii \pi},  k^{(1)}_2)-\frac{ E^{(1)}_2}{  K^{(1)}_2}\right)\lambda(\lambda_1^2-\lambda_3^2)}{\sqrt{(\lambda^2-\lambda_1^2)(\lambda^2-\lambda_2^2)(\lambda^2-\lambda_3^2)}}\\
			 &\ +\frac{2\lambda\left((\lambda_3^2-\lambda_1^2)(\lambda^2-\lambda_2^2)+(\lambda^2-\lambda_1^2)(\lambda^2-\lambda_3^2)\right)}{(\lambda^2-\lambda_1^2)\sqrt{(\lambda^2-\lambda_1^2)(\lambda^2-\lambda_2^2)(\lambda^2-\lambda_3^2)}}\\
			\xlongequal{\eqref{eq:define-first-integral}}&\ \frac{2 \lambda \left((\lambda^2-\lambda_3^2)- E^{(1)}_2(\lambda_1^2-\lambda_3^2)/  K^{(1)}_2
				\right) }{\sqrt{(\lambda^2-\lambda_1^2)(\lambda^2-\lambda_2^2)(\lambda^2-\lambda_3^2)}}.
		\end{split}
	\end{equation}
	The zeros of $(\mathcal{I}^{(1)}(\lambda))^{\prime}$ are $\lambda=0\in  Q^{(1)} $ and $\lambda=\pm ( E^{(1)}_2\lambda_1^2/  K^{(1)}_2+\lambda_3^2(1- E^{(1)}_2/  K^{(1)}_2))^{1/2}\in \ii \mathbb{R}\notin  Q^{(1)} $ since $\lambda^2-\lambda_2^2>0$ and $\lambda^2-\lambda_1^2<0$.
	The poles of the function $(\mathcal{I}^{(1)}(\lambda))^{\prime}$ are $\pm \lambda_{1,2,3}$ with the order $1/2$. 
	Then, we want to prove that excepting the curve listed in set $ Q^{(1)} $ without any other curves such that $\Im(\mathcal{I}^{(1)}(\lambda))=0$. 
	Assuming $\lambda_0 \in Q$ but $\lambda_0 \notin  Q^{(1)} $, we get a curve $l_1$ which goes through $\lambda_0$ and satisfies ${\Im}(\mathcal{I}^{(1)}(\lambda))=0$ by the tangent vector. 
	This curve would end at the poles or zeros of the function $(\mathcal{I}^{(1)}(\lambda))^{\prime}$ or $\lambda \rightarrow \infty$. If not, it must be a closed region.
	Since $\pm \lambda_{1,2,3}$  are the half-order poles of the function $(\mathcal{I}^{(1)}(\lambda))^{\prime}$, which implies that only one ray satisfying ${\Im}(\mathcal{I}^{(1)}(\lambda))=0$ goes through this poles. 
	Since within the set $ Q^{(1)} $, there exists a curve ending at these poles, the curve $l_1$ does not end at these branch points.
	Considering the zero points, we know $\lambda=0\in  Q^{(1)} $ and there exist two curves across this point.
	However in set $ Q^{(1)} $, there are two curves across the zero point, which deduce that the curve $l_1$ will not end at the point $\lambda=0$.
	By the equation \eqref{eq:I-define-P-dn}, as $\lambda \not \in \mathbb{R}\rightarrow \infty$, we get the function $\mathcal{I}^{(1)}(\lambda) \not \in \mathbb{R}$, which implies that the curve $l_1$ could not contain the point $\lambda =\infty$.
	Thus, this curve $l_1$ is a closed one. In the interior of a closed curve, by the maximum principle of the harmonic function, we know that all points $\lambda$ satisfy ${\Im}(\mathcal{I}^{(1)}(\lambda))=0$, so $(\mathcal{I}^{(1)}(\lambda))^{\prime}=0$ in this closed region. However, there are only two points such that $(\mathcal{I}^{(1)}(\lambda))^{\prime}=0, \lambda\in Q$, so we get the contradiction. Therefore, $Q \subseteq  Q^{(1)} $.
\end{proof}

Utilizing the above Lemma, we obtain the spectral stability results of the genus-two periodic traveling wave solutions  for the \ref{case1}.

\newenvironment{proof-spec-dn}{\emph{Proof of \Cref{theorem:spectral-image}.}}{\hfill$\Box$\medskip}
\begin{proof-spec-dn}
	For any $\lambda \in Q= Q^{(1)} $, we obtain $ (\lambda^2-\lambda_1^2)(\lambda^2-\lambda_2^2)(\lambda^2-\lambda_3^2)\ge 0$.
	Combining with the definition of $\Omega$ provided in equation \eqref{eq:define-Omega}, we obtain $\Omega\in \ii \mathbb{R}$. 
	The \Cref{lemma:dn-bounded} claims that the set corresponding to all bounded spectral functions of the mKdV equation with the genus-two periodic traveling wave solutions  for the \ref{case1} is $ Q^{(1)} $, and all elements of $ Q^{(1)} $ satisfy $\Omega(\lambda) \in \ii \mathbb{R}$. 
	By \Cref{define:spect-stable}, the two-phase solutions of the \ref{eq:mKdV} equation are spectrally stable. 
\end{proof-spec-dn}

We present a special example to illustrate the set $Q^{(1)}$ and the corresponding eigenvalue $\Omega$.
By selecting the parameters  $\lambda_1=2\ii/5$, $\lambda_2=4\ii /5$, $\lambda_3=7\ii/5$, we display the set  $Q^{(1)}$ and eigenvalues $\Omega(\lambda),\lambda\in Q^{(1)}$ in \Cref{fig:imag-spectral}.
The left panel of \Cref{fig:imag-spectral} shows the values of  $\lambda$ satisfying $\Im(\mathcal{I}^{(1)})=0$, which includes the entire real axis and three bands on the imaginary axis.
The right panel of  \Cref{fig:imag-spectral} depicts the spectrum $\Omega(\lambda)$ corresponding to those values of  $\lambda$. Notably,  the spectrum $\Omega(\lambda)$ always lies around the imaginary axis whether  $\lambda\in \mathbb{R}$ or it lies within three bands on the imaginary axis.

\begin{figure}[h]
	\centering
	{\includegraphics[width=0.8\textwidth]{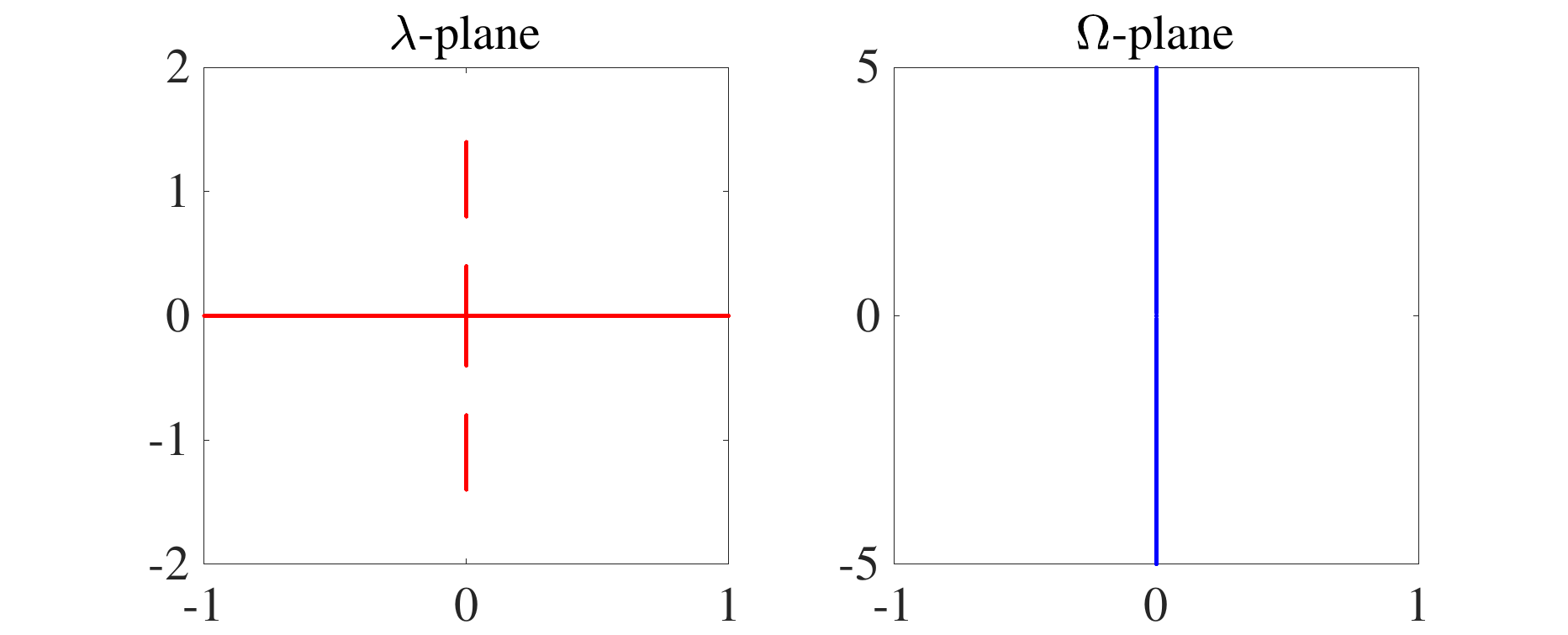}}
	\caption{The set $Q^{(1)}$ and the eigenvalue $\Omega$
		with branch points $\lambda_1=2\ii/5$, $\lambda_2=4\ii /5$, $\lambda_3=7\ii/5$.}
	\label{fig:imag-spectral}
\end{figure}

\subsection{The spectral stability analysis for \ref{case2}} \label{subsec:cn}

In this subsection, we would like to study the spectral stability of the two-phase solutions for the \ref{case2}.
Based on the fundamental solutions of the related Lax pair \eqref{eq:Lax-pair} and combined with equations \eqref{eq:M} and \eqref{eq:define-Omega-123-C}
we obtain  
\begin{equation}\label{eq:I-def-C}
	\mathcal{I}^{(2)}(\lambda)\xlongequal[\eqref{eq:define-Omega-123-C}]{\eqref{eq:M}}-2\ii \sqrt{A}\left(Z\left(\frac{ \nu^{(2)}_2 K^{(2)}_2}{\ii \pi},  k^{(2)}_2\right)-\frac{\sqrt{\left(\lambda_1^2-\lambda^2\right)\left(\lambda_2^2-\lambda^2\right)\left(\lambda_3^2-\lambda^2\right)}}{\sqrt{A}\left(\lambda_2^2-\lambda^2+A\right)}\right)+ 2\kappa^{(2)}.
\end{equation}
And, the derivative of function $\mathcal{I}^{(2)}(\lambda)$ with respect to the spectral parameter $\lambda$ is
\begin{equation}\label{eq:I-def-C-deriv}
	\begin{split}
		(\mathcal{I}^{(2)}(\lambda))^{\prime}
		\xlongequal[\eqref{eq:Zeta-define}]{\eqref{eq:I-def-C}}&-2 \frac{K^{(2)}_2 \sqrt{A}}{\pi}\left( \dn^2\left(\frac{ \nu^{(2)}_2 K^{(2)}_2}{\ii \pi},  k^{(2)}_2\right)-\frac{ E^{(2)}_2}{  K^{(2)}_2} \right)\frac{\dd  \nu^{(2)} _2}{\dd\lambda}\\
		-&2\ii \lambda\frac{(\lambda_2^2-\lambda^2)(\lambda^4-2\lambda_2^2\lambda^2+\lambda_3^2(\lambda_2^2-\lambda_1^2)+\lambda_1^2\lambda_2^2)+(v_1-v\lambda^2+3\lambda^4)A}{\sqrt{(\lambda_1^2-\lambda^2)(\lambda_2^2-\lambda^2)(\lambda_3^2-\lambda^2)}(\lambda_2^2-\lambda^2+A)^2}\\
		\xlongequal[\eqref{eq:define-first-integral}]{\eqref{eq:diff-elliptic}}&-2\frac{\ii \lambda \left(2A E^{(2)}_2/  K^{(2)}_2+\left(\lambda_2^2-A-\lambda^2\right)\right)}{\sqrt{(\lambda_1^2-\lambda^2)(\lambda_2^2-\lambda^2)(\lambda_3^2-\lambda^2)}}.
	\end{split}
\end{equation}

Define the set $ Q^{(2)} _{I}$, $ Q^{(2)} _{R}$, and $ Q^{(2)} _{P}$ in equation \eqref{eq:define-Q-C-IRp}.

\begin{prop}\label{prop:Q-M-cn}
	The branch points $\lambda_{i}$ and $\lambda_{i}^*$, $i=1,2,3$, satisfy the following conditions:
	\begin{itemize}
		\item[(a)] $\lambda_i,\lambda_i^*\in Q^{(2)}$, i.e., $\mathcal{I}^{(2)}(\lambda_i)\in \mathbb{R}$, $i=1,2,3$; for any $\lambda\in  Q^{(2)}_{I}\cup  Q^{(2)} _{R}\cup  Q^{(2)} _{P}$, $\mathcal{I}^{(2)}(\lambda)$ satisfies $\mathcal{I}^{(2)}(\lambda)\in \mathbb{R}$, i.e., $\Im(\mathcal{I}^{(2)}(\lambda))=0$;
		\item[(b)] For any $\lambda\in  Q^{(2)}_{I}\cup  Q^{(2)} _{R}\cup  Q^{(2)} _{P}$, 
		 the corresponding function $\Omega(\lambda)$ satisfies $\Omega(\lambda)\in \ii \mathbb{R}$;
		\item[(c)] $M(\lambda_i)=\pi \mod 2\pi $, $i=1,3$ and $M(\lambda_2)=0 \mod 2\pi $.
	\end{itemize}
\end{prop}
\begin{proof}
	(a): Firstly, we consider three pairs of the spectral numbers $\lambda_i$, $i=1,2,3$. It is easy to obtain $\mathcal{I}^{(2)}(\lambda_2)=2\kappa^{(2)}\in \mathbb{R}$.
	Combining the equation \eqref{eq:fs-3} with the definition of elliptic functions $F(\nu,k)$ in \eqref{eq:define-first-integral}, we get that when $\lambda=\lambda_1$, the equation $\nu_{2}^{(2)}K_{2}^{(2)}/(\ii \pi)=K_{2}^{(2)}-\ii K_{2}^{(2)\prime}$ holds; when $\lambda=\lambda_3$, the equation $\nu_{2}^{(2)}K_{2}^{(2)}/(\ii \pi)=K_{2}^{(2)}+\ii K_{2}^{(2)\prime}$ holds.
	Thus, we get
	\begin{equation}\nonumber
		\mathcal{I}^{(2)}(\lambda_1)\xlongequal{\eqref{eq:I-def-C}}-2\ii \sqrt{A}\left(Z\left(K^{(2)}_2-\ii K^{(2)\prime}_2,  k^{(2)}_2\right)\right)+2\kappa^{(2)}
		\xlongequal[\eqref{eq:u-parameters-cn-k}]{\eqref{eq:add-app}}3 \kappa^{(2)}\in \mathbb{R}.
	\end{equation}
	 Similarly, we obtain
	 \begin{equation}\nonumber
	 	\mathcal{I}^{(2)}(\lambda_3)\xlongequal{\eqref{eq:I-def-C}}-2\ii \sqrt{A}\left(Z\left(K^{(2)}_2+\ii K^{(2)\prime}_2,  k^{(2)}_2\right)\right)+2\kappa^{(2)}
	 	\xlongequal[\eqref{eq:u-parameters-cn-k}]{\eqref{eq:add-app}} \kappa^{(2)}\in  \mathbb{R}.
	 \end{equation}
	Considering $\lambda\in Q_{R}^{(2)}$, it is easy to obtain $\nu_2^{(2)}\in \mathbb{R}$, which deduce that $Z(\nu_2^{(2)}K_2^{(2)}/(\ii \pi ),k_2^{(2)})\in \ii \mathbb{R}$. Since $\left(\lambda_1^2-\lambda^2\right)\left(\lambda_2^2-\lambda^2\right)\left(\lambda_3^2-\lambda^2\right)<0$, we know $\sqrt{\left(\lambda_1^2-\lambda^2\right)\left(\lambda_2^2-\lambda^2\right)\left(\lambda_3^2-\lambda^2\right)}\in \ii \mathbb{R}$. Thus, we get $\mathcal{I}^{(2)}(\lambda)\in \mathbb{R}$, for $\lambda\in Q_{R}^{(2)}$.
	When $\lambda\in Q_{I}^{(2)}$, for the definition of $\nu_2^{(2)}$, we get $\nu_2^{(2)}\in \mathbb{R}$, i.e., $Z(\nu_2^{(2)}K_2^{(2)}/(\ii \pi ),k_2^{(2)})\in \ii \mathbb{R}$. Moreover, we deduce $\left(\lambda_1^2-\lambda^2\right)\left(\lambda_2^2-\lambda^2\right)\left(\lambda_3^2-\lambda^2\right)<0$. Thus, we obtain $\mathcal{I}^{(2)}(\lambda)\in \mathbb{R}$, for $\lambda\in Q_{I}^{(2)}$.
		
	(b): Based on the definition of the parameter $\Omega(\lambda)$ defined in equation \eqref{eq:define-Omega}, it is easy to obtain $\Omega(\lambda_i)=\Omega(\lambda_i^*)=0$.
	Moreover, we also obtain
	\begin{equation}\nonumber
		\begin{split}
			\Omega^2
			=&-64 \left(\lambda^2-\lambda_1^2\right)\left(\lambda^2-\lambda_2^2\right)\left(\lambda^2-\lambda_3^2\right)\\
			=&-64\left(\lambda^2-\lambda_2^2\right) \left(\lambda^4-2\lambda^2\left(\Re(\lambda_1^2)\right)^2+\left(\Re(\lambda_1^2)\right)^2+\left(\Im(\lambda_1^2)\right)^2\right)\\
			=&-64\left(\lambda^2-\lambda_2^2\right) \left(\left(\lambda^2-\left(\Re(\lambda_1^2)\right)^2\right)^2+\left(\Im(\lambda_1^2)\right)^2\right).
		\end{split}
	\end{equation}
	It is easy to obtain $\Omega^2\in \mathbb{R}$ and $\Omega^2<0$, when $\lambda \in \mathbb{R}$.
	Thus, for any $\lambda \in  Q^{(2)} _{I}$, we also could obtain $\left(\lambda^2-\lambda_2^2\right)>0$, so we also could get $\Omega^2<0$, which implies $\Omega\in \ii \mathbb{R}$.
	
	(c): The period of the solution \eqref{eq:solutions-dn} is $2T=2\pi\ii/(2\ii \kappa^{(2)})=\pi/\kappa^{(2)}$, based on the 
	definition of Riemann theta function in \Cref{define:Riemann-Theta-function} and the transformation formula \eqref{eq:formula-Rieman-shift-1}.
	 By the proof of (a), we obtain $M(\lambda_i)=2T\mathcal{I}^{(2)}(\lambda_i)=\pi\mathcal{I}^{(2)}(\lambda_i)/\kappa^{(2)}=\pi \mod 2\pi $, $i=1,3$ and $M(\lambda_2)=0 \mod 2\pi $.
\end{proof}

\begin{remark}\label{remark:Omega-sym}
	The curve ${\Re}(\Omega(\lambda))=0$ is also symmetric about the origin point, lines ${\Im}(\lambda)=0$ and ${\Re}(\lambda)=0$, since $\Omega(-\lambda)=\Omega(\lambda)$ and $\Omega(\lambda^*)=\Omega^*(\lambda).$
	Thus, if $\lambda$ satisfies ${\Re}(\Omega(\lambda))=0$, points $\pm\lambda^*,-\lambda$ also satisfy this equation. 
	For the function $\mathcal{I}^{(2)}(\lambda)$ defined in equation \eqref{eq:I-def-C}, we also get that 
	if $\Im(\mathcal{I}^{(2)}(\lambda))=0$, equations $\Im(\mathcal{I}^{(2)}(\lambda^{*}))=0$ and $\Im(\mathcal{I}^{(2)}(-\lambda))=0$ also holds.
	Therefore,
	the curve $\Im(\mathcal{I}^{(2)}(\lambda))=0$ is also symmetric about the origin point, lines ${\Im}(\lambda)=0$ and ${\Re}(\lambda)=0$.
\end{remark}

\begin{prop}\label{prop: zr0 zi0} 
	By equations \eqref{eq:I-def-C} and \eqref{eq:I-def-C-deriv}, the following properties hold (see the Fig. \ref{fig3}):
	\begin{itemize}
		\item[(i)]
		If $A(2E^{(2)}_2/K^{(2)}_2-1)>-\lambda_2^2$, the set $Q$ not only includes sets $ Q^{(2)} _{R}$ and $ Q^{(2)} _{I}$, but also contains two curves starting at points $\lambda_1,\lambda_3$ intersecting with the real axis at point $\pm\lambda_0 \in  Q^{(2)} _{R}$ and ending at points $\lambda_1^*,\lambda_3^*$;
		\item[(ii)]
		If $A(2E^{(2)}_2/K^{(2)}_2-1)=-\lambda_2^2$, the set $Q$ not only includes sets $ Q^{(2)} _{R}$ and $ Q^{(2)} _{I}$, but also contains two curves starting at points $\lambda_1,\lambda_3$ intersecting with the real axis at point $\lambda_0=0$ and ending at points $\lambda_1^*,\lambda_3^*$;
		\item[(iii)]
		If $0 < A(2E^{(2)}_2/K^{(2)}_2-1) <-\lambda_2^2$, the set $Q$ not only includes sets $ Q^{(2)} _{R}$ and $ Q^{(2)} _{I}$, but also contains two curves starting at points $\lambda_1,\lambda_1^*$ intersecting with the imaginary axis at point $\pm\lambda_0\in  Q^{(2)} _{I}\backslash\{\lambda_2,\lambda_2^*\}$ and ending at points $\lambda_3,\lambda_3^*$;
		\item[(iv)] 
		If $0=A(2E^{(2)}_2/K^{(2)}_2-1)$, the set $Q$ not only includes sets $ Q^{(2)} _{R}$ and $ Q^{(2)} _{I}$, but also contains two curves starting at points $\lambda_1,\lambda_1^*$ intersecting with the imaginary axis at point $\pm\lambda_2$ and ending at points $\lambda_3,\lambda_3^*$;
		\item[(v)] 
		 If $A(2E^{(2)}_2/K^{(2)}_2-1)<0$, the set $Q$ not only includes sets $ Q^{(2)} _{R}$ and $ Q^{(2)} _{I}$, but also contains two curves starting at points $\lambda_1,\lambda_1^*$ intersecting with the imaginary axis at point $\pm\lambda$ satisfying $\Im(\lambda)>\Im(\lambda_2)$ and ending at points $\lambda_3,\lambda_3^*$.
	\end{itemize} 
\end{prop}
\begin{proof}	
	By equation \eqref{eq:I-def-C-deriv}, it is easy to obtain that the zeros of $\mathcal{I}^{\prime}(\lambda)$ are $0$ and $\pm \lambda_0$, where
	\begin{equation}\label{eq:define-lambda-0}
		\lambda_0=(2E^{(2)}_2A/K^{(2)}_2-A+\lambda_2^2)^{1/2}.
	\end{equation}
	In virtue of above results, the roots of $\mathcal{I}^{(2)\prime}(\lambda)=0$ are divided into the following five cases.
	\begin{itemize}
		\item[(i)] When $2E^{(2)}_2A/K^{(2)}_2-A>-\lambda_2^2$, excepting zero point, the roots of $\mathcal{I}^{(2)\prime}(\lambda)=0$ are real numbers, i.e., $\mathcal{I}^{(2)\prime}(\pm\lambda_0)=0$, with $\lambda_0\in \mathbb{R}\neq 0$;
		\item[(ii)] When $2E^{(2)}_2A/K^{(2)}_2-A=-\lambda_2^2$, $\lambda_0=0$ is the third order zero point $\mathcal{I}^{(2)\prime}(0)=0$;
		\item[(iii)] When $0 < 2E^{(2)}_2A/K^{(2)}_2-A<-\lambda_2^2$, excepting zero point, the roots of $\mathcal{I}^{(2)\prime}(\lambda)=0$ are imaginary, satisfying $\mathcal{I}^{(2)\prime}(\pm\lambda_0)=0$, $\lambda_0 \in \ii\mathbb{R}$ and $0<|\Im(\lambda_0)|< \Im(\lambda_2)$;
		\item[(iv)] When $0=2E^{(2)}_2A/K^{(2)}_2-A$, excepting zero point, the roots of $\mathcal{I}^{(2)\prime}(\lambda)=0$ are imaginary, satisfying $\lambda_0=\lambda_2\in \ii \mathbb{R}$ and $\mathcal{I}^{(2)\prime}(\lambda_2)=0$;
		\item[(v)] When $2E^{(2)}_2A/K^{(2)}_2-A<0$, excepting zero point, the roots of $\mathcal{I}^{(2)\prime}(\lambda)=0$ are imaginary, satisfying $\mathcal{I}^{(2)\prime}(\pm\lambda_0)=0$, $\lambda_0 \in \ii\mathbb{R}$ and $\Im(\lambda_2)<|\Im(\lambda_0)|$.
	\end{itemize}
	
On the basis of the above cases, we proceed to examine all possibilities for the components of the set $Q^{(2)}$. 
	The curve $l_1\in Q$ ends at the point $\lambda$ satisfying $\mathcal{I}^{(2)\prime}(\lambda)=\infty$ or crosses to another component at the point $\lambda$ with $\mathcal{I}^{(2)\prime}(\lambda)=0$. 
	If the spectrum contains a closed curve, the cross point satisfies ${\Im}(\mathcal{I}^{(2)}(\lambda))=0$. 
	In the interior of a closed curve, by the maximum value principle of the harmonic function, we have ${\Im}(\mathcal{I}^{(2)}(\lambda))=0$. 
	Then $\mathcal{I}^{(2)}(\lambda)$ is a constant in this closed region.
	However, this is impossible, since there are only three points satisfying $\mathcal{I}^{(2)\prime}(\lambda)=0$ (allowing for the repeated roots).
	Thus there is no closed curve with ${\Im}(\mathcal{I}^{(2)}(\lambda))=0$. Furthermore, by \Cref{prop:Q-M-cn}, we know ${\Im}(\mathcal{I}^{(2)}(\lambda_i))={\Im}(\mathcal{I}^{(2)}(\lambda_i^*))=0, i=1,2,3$. 
	By the implicit function theorem, we know that there exist six curves with ${\Im}(\mathcal{I}^{(2)}(\lambda))=0$ to the harmonic function ${\Im}(\mathcal{I}^{(2)}(\lambda))$ ending at the points $\lambda_i,\lambda_i^*, i=1,2,3$, since $(\mathcal{I}^{(2)\prime}(\lambda_i))=\infty$ and $\mathcal{I}^{(2)\prime}(\lambda_i^*)=\infty$.
	
	(i): When $2E^{(2)}_2A/K^{(2)}_2-A>-\lambda_2^2$, we know $\lambda_0\in \mathbb{R}$. From \Cref{prop:Q-M-cn}, we get ${\Im}(\mathcal{I}^{(2)}(\lambda))=0$, for $\lambda \in Q_R^{(2)} \cup Q_I^{(2)} \cup Q_P^{(2)}$.
	Furthermore, by $(\mathcal{I}^{(2)}(\lambda))^{\prime}|_{\lambda=\pm \lambda_0}=0$ and $(\mathcal{I}^{(2)}(\lambda))^{\prime \prime}|_{\lambda=\pm \lambda_0}\neq 0$, then in the neighborhood of $\lambda=\pm \lambda_0$, we have Taylor expansions $\mathcal{I}^{(2)}(\lambda)=\mathcal{I}^{(2)}(\pm \lambda_0)+\mathcal{I}^{(2)\prime\prime}(\pm \lambda_0)(\lambda\pm \lambda_0)^2+\mathcal{O}((\lambda-\lambda_0)^3)$. 
	By the localized analysis and implicit function theorem, we find two curves ${\Im}(\mathcal{I}^{(2)}(\lambda))=0$ departing from the point $\lambda=\pm \lambda_0$. 
	And we know that the real axis and $l_1$ goes through them. 
	Furthermore, the curve $l_1$ going through $-\lambda_0$ does not across the imaginary axis.
	If not, by the symmetry of the curve $\Im(\mathcal{I}^{(2)}(\lambda))=0$, there must exist a point $\lambda_c$ such that $\Im(\mathcal{I}^{(2)}(\lambda_c))=0$ and two symmetric curves $l_1$ and $l_2$ going through point $\lambda_c$, which implies $\mathcal{I}^{(2)\prime}(\lambda_c)=0$. 
	This is a contradiction.
	Therefore, we conclude that the curve departing from the point $\lambda=\lambda_1$ goes across $\lambda=-\lambda_0$ and ends with $\lambda_1^*$, and another curve departing from the point $\lambda=\lambda_3$ goes across $\lambda=\lambda_0$ and ends with $\lambda_3^*$.
	
	(ii): When $2E^{(2)}_2A/K^{(2)}_2-A=-\lambda_2^2$, we know $\lambda_0=0$. 
	By equation \eqref{eq:I-def-C-deriv}, we get
	$\mathcal{I}^{(2)\prime}(0)=\mathcal{I}^{(2)\prime\prime}(0)=\mathcal{I}^{(2)\prime\prime\prime}(0)=0$ and $\mathcal{I}^{(2)\prime\prime\prime\prime}(0)\neq 0$, then in the neighborhood of $\lambda=0$, we have Taylor expansions $\mathcal{I}(\lambda)=\mathcal{I}(0)+\mathcal{I}^{(2)\prime\prime\prime\prime}(0)\lambda^4+\mathcal{O}(\lambda^5)$. 
	Thus, we can conclude that there are four curves passing through the zero point and terminating at the points $\lambda_{1,2,3}$ and $\lambda_{1,2,3}^*$. Specifically: One is the real axis; one is the imaginary axis with $|\Im(\lambda)|\le \Im(\lambda_2)$; one is the curve connecting branch points $\lambda_1,\lambda_3^*$ and go through the zero point; the last one is the curve connecting branch points $\lambda_3,\lambda_1^*$ and go through the zero point.
		
	(iii): When $0<2E^{(2)}_2A/K^{(2)}_2-A<-\lambda_2^2$, we know $\lambda_0\in \ii \mathbb{R}$ and $0<\Im(\lambda_0)<\Im(\lambda_2)$.
	The proof of this case is similar to the above case (i).
	Thus we will not repeat the details here.
	
	(iv): When $0=2E^{(2)}_2A/K^{(2)}_2-A$, we know $\lambda_0=\lambda_2$. 
	By the exact derivative formulas of $\mathcal{I}^{\prime}(\lambda)$ in equation \eqref{eq:I-def-C-deriv}, we have the expansion $\mathcal{I}(\lambda)=\mathcal{O}((\lambda-\lambda_2)^{3/2})$, in the neighborhood of $\lambda=\lambda_2$, which implies that there are only three radial departing from the point $\lambda_2$ ending at points $0,\lambda_1,\lambda_3$.
		
	(v): When $2E^{(2)}_2A/K^{(2)}_2-A<0$, we know $\Im(\lambda_0)>\Im(\lambda_2)$ but $\Im(\mathcal{I}(\lambda_0))\neq 0$.  
	Similar to the above analysis, we conclude that there are two curves emitting from $\lambda=\lambda_{1}, \lambda_1^*$ that go across the imaginary axis (not the point $\pm\lambda_0$) and end with $\lambda=\lambda_3, \lambda_3^*$, respectively, since $\Im(\mathcal{I}(\pm\lambda_0))\neq 0$.
\end{proof}

The derivative of the function $\Omega(\lambda)$ is
\begin{equation}\label{eq:diff-Omega}
	\Omega^{\prime}(\lambda)=8\ii \lambda\left((\lambda^2-\lambda_1^2)(\lambda^2-\lambda_2^2)+(\lambda^2-\lambda_1^2)(\lambda^2-\lambda_3^2)+(\lambda^2-\lambda_2^2)(\lambda^2-\lambda_3^2)\right)/y.
\end{equation}
Therefore, we get 
\begin{equation}\label{eq:derivative-1}
	\begin{split}
		\frac{\mathcal{I}^{(2)\prime}(\lambda_1)}{\Omega^{\prime}(\lambda_1)}
		\xlongequal[\eqref{eq:diff-Omega}]{\eqref{eq:I-def-C-deriv}}&\frac{A(K^{(2)}_2-2E^{(2)}_2)-\left(\lambda_2^2-\lambda_1^2\right)K^{(2)}_2}{4(\lambda_1^2-\lambda_2^2)(\lambda_1^2-\lambda_3^2)K^{(2)}_2}
		=\frac{(K^{(2)}_2-2E^{(2)}_2)(\lambda_3^2-\lambda_2^2)+AK^{(2)}_2}{4A(\lambda_1^2-\lambda_3^2)K^{(2)}_2},\\
		\frac{\mathcal{I}^{(2)\prime}(\lambda_3)}{\Omega^{\prime}(\lambda_3)}
		\xlongequal[\eqref{eq:diff-Omega}]{\eqref{eq:I-def-C-deriv}}&\frac{(K^{(2)}_2-2E^{(2)}_2)(\lambda_1^2-\lambda_2^2)+AK^{(2)}_2}{4A(\lambda_3^2-\lambda_1^2)K^{(2)}_2},
	\end{split}
\end{equation}
and $\mathcal{I}^{(2)\prime}(\lambda_i)/\Omega^{\prime}(\lambda_i)$, $i=1,3$.

Due to the symmetry of functions $\Im(\mathcal{I}^{(2)}(\lambda))$ and $\Re(\Omega(\lambda))$ in \Cref{remark:Omega-sym}, we just study the case that the parameter $\lambda$ lies in the first quadrant, i.e., $\Re(\lambda)> 0$ and $\Im(\lambda)> 0$. 
For ease of analyzing and studying, we introduce the function $\Upsilon=\Upsilon(\lambda)$ and define it as
\begin{equation}\label{eq:Lambda}
	\Upsilon(\lambda)
	=\frac{\lambda^2-\lambda_2^2}{A}, 
	\qquad  \text{or} \qquad 
	\lambda^2=A\Upsilon +\lambda_2^2.
\end{equation}
We will consider the parameter $\Upsilon$ on upper-half complex plane for any $\lambda$ on the first quadrant.
Then, functions $\mathcal{I}^{(2)}(\lambda)$, $\Omega(\lambda)$ by the parameter $\Upsilon$  and the derivative of $\mathcal{I}^{(2)}(\lambda)$  with respect to $\Upsilon$ could be rewritten as follows:
\begin{equation}\label{eq:I-Lambda-C}
	\begin{split}
		\mathcal{I}^{(2)}(\lambda)\equiv &\  \hat{\mathcal{I}}
		 (\Upsilon)\\ 
		 \xlongequal{\eqref{eq:I-def-C}} &\  -2\ii \sqrt{A}\left(Z\left(F\left(\frac{\sqrt{-4\Upsilon }}{1-\Upsilon },  k^{(2)}_2\right),  k^{(2)}_2\right)-\frac{\sqrt{-\left(\Upsilon_1-\Upsilon \right)\Upsilon \left(\Upsilon_3-\Upsilon \right)}}{\left(1-\Upsilon\right)}\right)+ 2\kappa^{(2)},\\
		\frac{\dd  \hat{\mathcal{I}} (\Upsilon)}{\dd \Upsilon}\xlongequal{\eqref{eq:I-def-C-deriv}} &\  -\frac{\ii\sqrt{A} }{\sqrt{-(\Upsilon_1-\Upsilon)\Upsilon(\Upsilon_3-\Upsilon)}}\left(\frac{2 E^{(2)}_2}{  K^{(2)}_2}-1-\Upsilon\right),\\
		\Omega(\lambda)  \equiv  &\  \hat{\Omega}(\Upsilon) \xlongequal{\eqref{eq:define-Omega}} 8 \ii \sqrt{A^3(\Upsilon-\Upsilon_1)\Upsilon(\Upsilon-\Upsilon_3)}
	\end{split}
\end{equation}
where 
\begin{equation}\nonumber
	\Upsilon_i:=\Upsilon_{iR}+\ii \Upsilon_{iI}=\frac{\lambda_i^2-\lambda_2^2}{A}, \qquad 
	\Upsilon_{iR},\,\,\, \Upsilon_{iI}\in \mathbb{R}, \qquad 
	\Upsilon_3=\Upsilon_1^*,\qquad
	i=1,3.
\end{equation}
Furthermore, we get 
\begin{equation}\label{eq:Lambda-i-1}
	|\Upsilon_i|^2=\Upsilon_{iR}^2+\Upsilon_{iI}^2=\frac{1}{A^2}\left(\lambda_i^2-\lambda_2^2\right)\left(\lambda_i^{*2}-\lambda_2^2\right)=1, \qquad i=1,3.
\end{equation}

\newenvironment{proof-spec-cn}{\emph{Proof of \Cref{theorem:spectral-complex}.}}{\hfill$\Box$\medskip}
\begin{proof-spec-cn}
	We aim to prove that for any $\lambda \in Q \backslash ( Q^{(2)} _R \cup  Q^{(2)} _I \cup  Q^{(2)} _P)$, it holds that $\Omega(\lambda) \notin \ii  \mathbb{R}$, which indicates that the two-phase solutions are spectrally unstable.
	Without loss of generality, we consider the spectral parameter $\lambda$ located in the first quadrant of the $\lambda$-plane. 
	Owing to the symmetry of the curve ${\Re}(\Omega(\lambda))=0$ and the set $Q$, as stated in Remark \ref{remark:Omega-sym} and \Cref{prop:Q-M-cn}, respectively, the computation for $\lambda$ in the second, third and fourth quadrants is identical to that in the first quadrant.
	 We divide the proof into two cases:
    (1): $2E^{(2)}_2A/K^{(2)}_2-A \ge 0$, corresponding to cases (i)-(iv) listed in \Cref{prop: zr0 zi0};
	(2): $2E^{(2)}_2A/K^{(2)}_2-A<0$, corresponding to case (v) listed in \Cref{prop: zr0 zi0}.

	Under the transformation \eqref{eq:Lambda}, the function $\hat{\Omega}(\Upsilon)$ in equation \eqref{eq:I-Lambda-C} could be expressed as
	\begin{equation}\nonumber
		\hat{\Omega}^2(\Upsilon)=-64A^3 \left(\Upsilon -\Upsilon_1\right)\Upsilon\left(\Upsilon -\Upsilon_3\right).
	\end{equation}
	Combining this with equation \eqref{eq:Lambda-i-1}, the real and imaginary parts of the function $\hat{\Omega}^2$ are given by
	\begin{equation}\label{eq:Omega-Im_Re}
		\left\{\begin{aligned}
			{\Im}(\hat{\Omega}^2)&=-64A^3\Upsilon_I\left(3\Upsilon_R^2-4\Upsilon_{3R}\Upsilon_R-\Upsilon_I^2+1\right),\\
			{\Re}(\hat{\Omega}^2)&=-64A^3\left(\Upsilon_R(\Upsilon_R^2-2\Upsilon_{3R}\Upsilon_R-3\Upsilon_I^2+1)+2\Upsilon_{3R}\Upsilon_I^2\right),
		\end{aligned} \right.
	\end{equation}
	where $\Upsilon=\Upsilon_{R}+\ii \Upsilon_I$, $\Upsilon_{3R}$ and $\Upsilon_{3I}$ are defined in \eqref{eq:Lambda-i-1}.
	The necessary and sufficient conditions for ${\Re}(\hat{\Omega})=0$ are ${\Im}(\hat{\Omega}^2)=0$ and ${\Re}(\hat{\Omega}^2)\le 0$. 
	From equation \eqref{eq:Omega-Im_Re}, the curve ${\Re}(\hat{\Omega})=0$ satisfying $\Upsilon\notin \mathbb{R}$ is equivalent to
	\begin{equation}\label{eq:Omega-condition}
		\Upsilon_I^2=3\Upsilon_R^2-4\Upsilon_{3R}\Upsilon_R+1, \qquad \Upsilon_R\le \Upsilon_{3R}, \qquad \Upsilon_I>0.
	\end{equation}
	In the following, 
	we prove that curves ${\Im}(\hat{\mathcal{I}}(\Upsilon))=0$ and ${\Re}(\hat{\Omega}(\Upsilon))=0$ in the $\Upsilon$-plane do not intersect, excepting on the point $\Upsilon_3$.
	Firstly, we introduce several formulas useful for the subsequent analysis. Secondly, we examine the variation of the curve ${\Im}( \hat{\mathcal{I}} (\Upsilon))=0$. Finally, we compare the curve $\Re(\hat{\Omega})=0$ and $\Im(\hat{\mathcal{I}})=0$ to demonstrate that they do not intersect except on the point $\Upsilon_3$.

	Consider the case (1): $2E^{(2)}_2A/K^{(2)}_2-A \ge 0$.
	Along the curve ${\Im}(\hat{\mathcal{I}} (\Upsilon))=0$, the tangent vector in equation \eqref{eq:tangent-vector} can be rewritten as
	\begin{equation}\label{eq:tan-vec-Lambda-I}
		\left(-\frac{\dd {\Im}( \hat{\mathcal{I}} )}{\dd \Upsilon_I},\frac{\dd {\Im}(\hat{\mathcal{I}} )}{\dd \Upsilon_R} \right)=\left(-{\Re}\left(\hat{\mathcal{I}}^{\prime}(\Upsilon) \right),{\Im}\left(\hat{\mathcal{I}}^{\prime}(\Upsilon)\right)\right).
	\end{equation}
Since
	\begin{equation}\label{eq:define-a-b}
		\frac{\Re(\sqrt{\Upsilon_{3R}-\ii \Upsilon_I})}{\Im(\sqrt{\Upsilon_{3R}-\ii \Upsilon_I})}
		=\frac{-\sqrt{\Upsilon_{3R}^2+\Upsilon_I^2}-\Upsilon_{3R}}{\Upsilon_I},  
	\end{equation}
	when $0<\Upsilon_I<\Upsilon_{3I}$, i.e., $\hat{\Upsilon}\in \ii \mathbb{R}$, $\hat{\Upsilon}:=((\Upsilon_I^2-\Upsilon_{3I}^2)(\Upsilon_{3R}^2+ \Upsilon_I^2))^{1/2}$, we obtain
	\begin{equation}\nonumber
		\begin{split}
			\left.\Re\left(\hat{\mathcal{I}}^{\prime}(\Upsilon)\right)\right|_{\Upsilon=\Upsilon_{3R}+\ii \Upsilon_I} \!\!
			\xlongequal[\eqref{eq:define-a-b}]{\eqref{eq:I-Lambda-C}}&\,
			\frac{\sqrt{A}\Im(\sqrt{\Upsilon_{3R}-\ii \Upsilon_I}) }{\ii \hat{\Upsilon}}\cdot\frac{H(\Upsilon_I)}{\Upsilon_I},\\
			\left.\Im\left(\hat{\mathcal{I}}^{\prime}(\Upsilon)\right)\right|_{\Upsilon=\Upsilon_{3R}+\ii \Upsilon_I} \!\!
			\xlongequal[\eqref{eq:define-a-b}]{\eqref{eq:I-Lambda-C}} &\, 	\frac{ \sqrt{A}\Im(\sqrt{\Upsilon_{3R}-\ii \Upsilon_I})}{\ii  \hat{\Upsilon}}\left(\frac{2 E^{(2)}_2}{  K^{(2)}_2}-1 
			+\sqrt{\Upsilon_{3R}^2+\Upsilon_I^2}\right),
		\end{split} 
	\end{equation} 
	where $H(\Upsilon_I):= \Upsilon_I^2+\Upsilon_{3R}^2-(2 E^{(2)}_2/K^{(2)}_2-1-\Upsilon_{3R})\sqrt{\Upsilon_{3R}^2+\Upsilon_I^2}-(2 E^{(2)}_2/K^{(2)}_2-1)\Upsilon_{3R}$.
	When $\Upsilon_{I}>\Upsilon_{3I}$, $\hat{\Upsilon}\in \mathbb{R}$, we have
	\begin{equation}\label{eq:I-U}
		\begin{split}
			\left.\Re\left(\hat{\mathcal{I}}^{\prime}(\Upsilon)\right)\right|_{\Upsilon=\Upsilon_{3R}+\ii \Upsilon_I} \!\!\!
			=\left.\Im\left(\hat{\mathcal{I}}^{\prime}(\Upsilon)\right)\right|_{\Upsilon=\Upsilon_{3R}+\ii \Upsilon_I},\quad 
			\left.\Im\left(\hat{\mathcal{I}}^{\prime}(\Upsilon)\right)\right|_{\Upsilon=\Upsilon_{3R}+\ii \Upsilon_I} \!\!\!
			=	\left.\Re\left(\hat{\mathcal{I}}^{\prime}(\Upsilon)\right)\right|_{\Upsilon=\Upsilon_{3R}+\ii \Upsilon_I}.
		\end{split}
	\end{equation}	
	This implies that 
	\begin{equation}\nonumber
		\lim\limits_{\Upsilon_I\rightarrow \Upsilon_{3I}^{-}}\left.\frac{-\Re\left( \hat{\mathcal{I}}^{\prime} (\Upsilon)\right)}{\Im\left( \hat{\mathcal{I}}^{\prime} (\Upsilon)\right)}\right|_{\Upsilon=\Upsilon_{3R}+\ii \Upsilon_I}
		\xlongequal{\eqref{eq:Lambda-i-1}}\frac{(K^{(2)}_2-E^{(2)}_2)(1+\Upsilon_{3R})}{-\Upsilon_{3I} E^{(2)}_2}<0,
	\end{equation}
	by the second formula in equation \eqref{eq:inequality}.
	
	Then, we aim to show that the curve $\Im(\hat{\mathcal{I}})=0$ lies entirely on the right-hand side of the line $\Upsilon_R=\Upsilon_{3R}$.
	Since $2E^{(2)}_2/K^{(2)}_2-1-\Upsilon_{3R}=2E^{(2)}_2/K^{(2)}_2-1-(1-2(k^{(2)}_2)^2)>0$ by the first formula in \eqref{eq:inequality}, the point $(2E^{(2)}_2/K^{(2)}_2-1,0)$ is located to the right of the line $\Upsilon_R=\Upsilon_{3R}$. Hence, it suffices to show that the curve $\Im(\hat{\mathcal{I}})=0$ does not cross this line.
	Along the line $\Upsilon_R=\Upsilon_{3R}$,
	the derivative of $\Im(\hat{\mathcal{I}})$ with respect to $\Upsilon_I$ is 
	$\dd  {\Im}( \hat{\mathcal{I}})/\dd \Upsilon_I|_{\Upsilon=\Upsilon_{3R}+\ii \Upsilon_I}
	\!
	={\Re}\left( \hat{\mathcal{I}}^{\prime}\right)|_{\Upsilon=\Upsilon_{3R}+\ii \Upsilon_I}$.
	When $\Upsilon_I>\Upsilon_{3I}$, since $2E^{(2)}_2/K^{(2)}_2-1>0$, we have $\Re(\dd \hat{\mathcal{I}}/\dd \Upsilon) \neq 0$, which implies that $\Im(\hat{\mathcal{I}})$ is monotonic. 
	Because $\Im(\hat{\mathcal{I}}(\Upsilon_{1}))=0$, it follows that $\Im(\hat{\mathcal{I}})\neq 0$ for any $\Upsilon_R=\Upsilon_{3R}$ and $\Upsilon_I>\Upsilon_{3I}$.
	When $0<\Upsilon_I<\Upsilon_{3I}$, the zeros of $\Re( \hat{\mathcal{I}}^{\prime})$ correspond to the zeros of $H(\Upsilon_I)$.
	It is straightforward to compute that 
	$H(\Upsilon_{3I})=2(1+\Upsilon_{3R})(1-E^{(2)}_2/K^{(2)}_2)>0$. 
	If $\Upsilon_{3R}>0$, then $H(0)=-2(2 E^{(2)}_2/K^{(2)}_2-1-\Upsilon_{3R})\Upsilon_{3R}<0$; while if $\Upsilon_{3R}\le 0$, we have $H(0)=0$. 
	As $\Upsilon_I$ varies from $0$ to $\Upsilon_{3I}$, the function $(\Upsilon_{3R}^2+\Upsilon_I^2)^{1/2}$ is monotonically increasing.
	Treating $H(\Upsilon_I) $ as a quadratic function in $(\Upsilon_{3R}^2+\Upsilon_I^2 )^{1/2}$, it follows that within the interval $\Upsilon_I\in (0,\Upsilon_{3I})$, there can be at most one points where $\dd H/\dd (\Upsilon_{3R}^2+\Upsilon_I^2)^{1/2}=0$, that is,  at most one point where $H^{\prime}(\Upsilon)=0$.
	If exist a point $ \Upsilon_{I_1}\in (0,\Upsilon_{3I} )$ such that $H^{\prime}(\Upsilon_{I_1})=0$, then combining $H(0)\le 0$ and $H(\Upsilon_{3I})>0$, we obtain that there exist a unique point $\Upsilon_{I_0}$ such that $H(\Upsilon_{I_0})=0$.
	Otherwise, if $H(\Upsilon)$ is monotonic, the same conclusion follows: there exists at most one $\Upsilon_{I_0}$ such that $H(\Upsilon_{I_0})=0$.
	We now show that the curve $\Im(\hat{\mathcal{I}})=0$ dose not intersect the line $\Upsilon_R=\Upsilon_{3R}$, $0<\Upsilon_I<\Upsilon_{3I}$. 
	If it did, there would exist a point on this line satisfying $\Im(\hat{\mathcal{I}})=0$.
	Thus, along $\Upsilon_R=\Upsilon_{3R}$ with $0\le \Upsilon_I \le \Upsilon_{3I}$, there would be three points where $\Im(\hat{\mathcal{I}})=0$, since $\Im(\hat{\mathcal{I}}(0))=0$ and $\Im(\hat{\mathcal{I}}(\Upsilon_3))=0$.
	By Rolle's theorem, this would imply the existence of two points where $H^{\prime}(\Upsilon)=0$, which contradicts with the fact that there can be at most one $\Upsilon_{I_0}$ such that $H(\Upsilon_{I_0})=0$.
	From \Cref{prop: zr0 zi0}, we know that the curve $\Im(\hat{\mathcal{I}})=0$ passes through the point $ (2E^{(2)}_2/K^{(2)}_2-1,0)$ and	its other endpoint is $\Upsilon_3$. 
	Therefore, no other point on the line $\Upsilon_R=\Upsilon_{3R}$ satisfies $\Im(\hat{\mathcal{I}})=0$. 
	Consequently, the curve ${\Im}(\hat{\mathcal{I}})=0$ on the $\Upsilon$-plane satisfies the condition $\Upsilon_R\ge \Upsilon_{3R}$. 
	
	The curve  ${\Re}(\hat{\Omega})=0$ satisfies $\Upsilon_R\le \Upsilon_{3R}$, and on the line $\Upsilon_R=\Upsilon_{3R}$ there exists only one point $(\Upsilon_{3R},\Upsilon_{3I})$ such that ${\Re}(\hat{\Omega})=0$. 
	Similarly, the curve $\Im(\hat{\mathcal{I}})=0$ satisfies $\Upsilon_R\ge \Upsilon_{3R}$, and on the same line $\Upsilon_R=\Upsilon_{3R}$ there exists only one point $(\Upsilon_{3R},\Upsilon_{3I})$ such that ${\Im}(\hat{\mathcal{I}})=0$. 
	Therefore, in the $\Upsilon$-plane, except for the point $\Upsilon_3$, there are no other intersection points satisfying ${\Re}(\hat{\Omega})=0$ and ${\Im}( \hat{\mathcal{I}})=0$ simultaneously. 
	Consequently, in the $\lambda$-plane, except for $\lambda_3$, there are no other intersection points satisfying ${\Re}(\Omega(\lambda))=0$ and ${\Im}( \mathcal{I}^{(2)} (\lambda))=0$. 
	Similar conclusions hold in the second, third and fourth quadrants. 
	Thus, when $2E^{(2)}_2A/K^{(2)}_2-A>0$, for any $\lambda \in Q \backslash (Q^{(2)}_{R}\cup Q^{(2)}_{I}\cup Q^{(2)}_P)$, we have $\Omega(\lambda) \notin \ii  \mathbb{R}$.

	Consider case (2): $2E^{(2)}_2A/K^{(2)}_2-A<0$.
	We employ the same method as in the case (1). 
	Consider the line $\Upsilon_{I}=\Upsilon_{3I}$ with $\Upsilon_R<\Upsilon_{3R}$.
	Let $F(\Upsilon_R)=\Upsilon_R(\Upsilon_R-\Upsilon_{3R})-2\Upsilon_{3I}^2+\ii\Upsilon_{3I}(\Upsilon_{3R}-3\Upsilon_R)$ with $\Im(F(\Upsilon_R))>0$. Then, we get
	\begin{equation}\nonumber
		\frac{\Re(\sqrt{F(\Upsilon_R)})}{\Im(\sqrt{F(\Upsilon_R)})}=\frac{|F(\Upsilon_R)|+\Re(F(\Upsilon_R))}{\Im(F(\Upsilon_R))}>0.
	\end{equation}
	From equation \eqref{eq:I-Lambda-C}, we obtain
	\begin{equation}\nonumber
		\begin{split}
			\Im\left(\hat{\mathcal{I}}^{\prime}(\Upsilon) \right)|_{\Upsilon=\Upsilon_R+\ii \Upsilon_{3I}} \!
			=\frac{-\sqrt{A}\Im(\sqrt{F(\Upsilon_R)})}{(\Upsilon_{3R}-\Upsilon_R)^{1/2}|F(\Upsilon_R)|}\left(\left(\frac{2E_2^{(2)}}{K_2^{(2)}}-1-\Upsilon_{R}\right)\frac{|F(\Upsilon_R)|+\Re(F(\Upsilon_R))}{\Im(F(\Upsilon_R))}+\Upsilon_{3I}\right).
		\end{split}
	\end{equation}
	Since $2E_2^{(2)}/K_2^{(2)}-1-\Upsilon_{R}\ge 2E_2^{(2)}/K_2^{(2)}-1-\Upsilon_{3R}>0$ and $\Upsilon_{3I}>0$, we have $\Im(\dd \hat{\mathcal{I}}/ \dd \Upsilon)|_{\Upsilon=\Upsilon_R+\ii \Upsilon_{3I}}\neq 0$, which implies that along the line $\Upsilon_R<\Upsilon_{3R}$, $\Upsilon_{3I}>0$, the function $\Im(\hat{\mathcal{I}})$ is monotonic.
	Since $\Im(\hat{\mathcal{I}}(\Upsilon_{3R}+\ii \Upsilon_{3I}))=0$, it follows that $\Im(\hat{\mathcal{I}}(\Upsilon_{R}+\ii \Upsilon_{3I}))\neq 0$ for $\Upsilon_R<\Upsilon_{3R}$. 
	When $\Upsilon_{I}>\Upsilon_{3I}$ and $\Upsilon_R=\Upsilon_{3R}$, the condition $(\Upsilon_{3R}^2+\Upsilon_I^2)^{1/2}-1>(\Upsilon_{3R}^2+\Upsilon_{3I}^2)^{1/2}-1=0$ ensures $\Re(\dd \hat{\mathcal{I}}/ \dd \Upsilon)\neq 0$ from equation \eqref{eq:I-U}, implying that $\Im(\hat{\mathcal{I}})$ is also monotonic.
	Hence $\Im(\hat{\mathcal{I}})\neq 0$.
	Since the curve $\Im(\hat{\mathcal{I}})=0$ starts at the point $(\Upsilon_{3R},\Upsilon_{3I})$ and terminates at a point on the real line, we conclude that this curve does not cross the lines $\Upsilon_R=\Upsilon_{3R}$, $\Upsilon_{I}>\Upsilon_{3I}$ or  $\Upsilon_R<\Upsilon_{3R}$, $\Upsilon_{I}=\Upsilon_{3I}$.
	Moreover, the curve $\Re(\hat{\Omega})=0$ coincides with that defined by equation \eqref{eq:Omega-condition}. For any $\Upsilon_R<\Upsilon_{3R}<0$, equation \eqref{eq:Omega-condition} gives $\Upsilon_I>\Upsilon_{3I}$. 
	Hence, the curve $\Re(\hat{\Omega})=0$ lies entirely within the region $\Upsilon_I>\Upsilon_{3I}$ and $\Upsilon_R<\Upsilon_{3R}$, and thus does not intersect with the curve $\Im(\hat{\mathcal{I}})=0$.	
    Therefore, we conclude that $\Omega(\lambda)\notin \ii \mathbb{R}$ for any $\lambda \in Q\backslash(Q^{(2)}_{R} \cup Q^{(2)}_{I} \cup  Q^{(2)}_{P})$.
\end{proof-spec-cn}

In virtue of the above Lemma, we know that the related periodic solutions are spectrally unstable.
Therefore, we would like to study the subharmonic perturbation of this solution and to explore the subharmonic perturbation stability.
The set $Q^{(2)}_P$ could also be divided into two subsets $Q^{(2)}_{sub}=Q^{(2)}_{sub,R}\cup Q^{(2)}_{sub,C}$, where 
\begin{equation}\label{eq:Q_Psub}
	Q^{(2)}_{sub,R}:=\{\lambda\in Q^{(2)}_{sub}|\lambda\in  Q^{(2)} _{R}\cup  Q^{(2)} _{I}\}, \qquad 
	Q^{(2)}_{sub,C}:=\{\lambda\in Q^{(2)}_{sub}|\lambda\notin  Q^{(2)} _R\cup  Q^{(2)} _{I}\}.
\end{equation}

\begin{prop}\label{prop:I-M-increase}
	Along the curve ${\Im}(\mathcal{I}(\lambda))=0$, the value of $M(\lambda)$ increases (decreases) in the upper half-plane, and it decreases (increases) in the lower half-plane.
\end{prop}

\begin{proof}
By the equation \eqref{eq:M}, the directional derivative of $M(\lambda)$ along the curve ${\Im}(\mathcal{I}(\lambda))=0$ is given by:
	\begin{equation}\label{eq:muz}
		\left( \frac{\dd  M(\lambda)}{\dd  \lambda_R},\frac{\dd  M(\lambda)}{\dd  \lambda_I} \right) \cdot \left(- {\Re} (\mathcal{I}'(\lambda)),{\Im}(\mathcal{I}'(\lambda)) \right) =-2\alpha T  \left( \left({\Im}(\mathcal{I}'(\lambda))\right)^2+\left({\Re}(\mathcal{I}'(\lambda)) \right)^2 \right),
	\end{equation}
	where $\lambda=\lambda_R+\ii \lambda_I, \lambda_I,\lambda_R\in \mathbb{R}$ and $\lambda\in Q$ in equation \eqref{eq:set-Q}. Since the directional derivative of $M(\lambda)$ with respect to $\lambda$ is nonzero along the curve ${\Im}(\mathcal{I}(\lambda))=0$, the value of $M(\lambda)$ is increasing or decreasing along the curve $\Im(\mathcal{I}(\lambda))=0$.
\end{proof}

\newenvironment{proof-spec-cn-P}{\emph{Proof of \Cref{theorem:spectral-complex-P}.}}{\hfill$\Box$\medskip}
\begin{proof-spec-cn-P}
	By the \Cref{define:spect-P}, to prove the spectral stability of the two-phase solutions under the \ref{case2} with the $P$-subharmonic perturbation, we should get the value of $P$ for all $\lambda\in Q$, $\Omega(\lambda)\in \ii \mathbb{R}$.
	By \Cref{prop:Q-M-cn}, we get $\Omega(\lambda)\in \ii \mathbb{R}$ for any $\lambda\in Q^{(2)}_{sub,R}$. 
	Combining with \Cref{theorem:spectral-complex}, we know that for any $\lambda\in Q \backslash ( Q^{(2)} _{R}\cup  Q^{(2)} _{I}\cup  Q^{(2)} _{P})$,  eigenvalues are not pure imaginary number, i.e., $\Omega(\lambda)\notin \ii \mathbb{R}$.  Thus, the spectral stability is converted into prove $Q^{(2)}_{sub,C}=  Q^{(2)} _{P}$. 
	By the symmetry of the set $Q$ and the function $\Omega(\lambda)$, we just study the spectral parameter $\lambda$ in the first quadrant of the $\lambda$-plane.
	We divided the proof into the following five categories for different conditions of the set $Q$ in \Cref{prop: zr0 zi0} (see the Fig. \ref{fig3}).
	
	When $2 E^{(2)}_2A/  K^{(2)}_2-A>-\lambda_2^2$, since along the curve ${\Im}( \mathcal{I}^{(2)} (\lambda))=0$ from $\lambda=\lambda_0$ to $\lambda=\lambda_1$, the value of $M(\lambda)$ is increasing by \Cref{prop:I-M-increase}. 
	From \Cref{prop:Q-M-cn}, we get that $M(\lambda_1)=-\pi$. 
	We must ensure that no other point in $Q^{(2)}_{sub}$ intersects with the curve ${\Im}( \mathcal{I}^{(2)} (\lambda))=0$ between $\lambda=\lambda_0$ and $\lambda=\lambda_1$. 
	Therefore, when $P\le \frac{4\pi}{\pi+ M(\lambda_0)}$, for any $\lambda\in Q^{(2)}_{sub}$, we get $\Omega(\lambda)\in \ii \mathbb{R}$. The solutions with two pairs of complex branch points are spectrally stable with respect to perturbations of period $2PT, \ P\in \mathbb{N}$.
	When $2 E^{(2)}_2A/  K^{(2)}_2-A=-\lambda_2^2$, the analysis is similar to the case $2 E^{(2)}_2A/  K^{(2)}_2-A>-\lambda_2^2$, so we would not repeat them here anymore.
	
	When $0 < 2 E^{(2)}_2A/  K^{(2)}_2-A<-\lambda_2^2$, since $ \mathcal{I}^{(2)} (\lambda_2)=\pi \sqrt{A}/(2  K^{(2)}_2)=\kappa^{(2)}/2$. Combined with the definition of the function $M(\lambda)$, it is easy to obtain $M(\lambda_2)=\pi/2$. 
	And for any $\lambda\in  Q^{(2)} _I$, along the curve $\Im( \mathcal{I}^{(2)} (\lambda))=0$, the function $M(\lambda)$ is monotonous.
	Similar to the above, we conclude that when $P\le \frac{4\pi}{\pi+ M(\lambda_0)}$, for any $\lambda\in Q^{(2)}_{sub}$, we get $\Omega(\lambda)\in \ii \mathbb{R}$. 
	
When $0=2 E^{(2)}_2A/  K^{(2)}_2-A$, we get $M(\lambda_2)=2\pi$. 
It is easy to know that the solution is $2$-subharmonic perturbation spectrally stable.
Then, we would like to study whether there exists a parameter $P>2$, such that the solution is $P$-subharmonic perturbation spectrally stable.
Due to the monotonicity, if have, there must exist a point $\lambda_0 \in Q \backslash ( Q^{(2)} _R \cup  Q^{(2)} _I \cup  Q^{(2)} _P)$ such that $\Omega(\lambda_0)\in \ii \mathbb{R}$, which is contradict with the \Cref{theorem:spectral-complex}. 
So, we obtain that the solution is $2$-subharmonic perturbation spectrally stable.
	
When $2 E^{(2)}_2A/  K^{(2)}_2-A<0$, we would consider the  upper half-plane since the lower half-plane can be obtained similarly. From \Cref{prop: zr0 zi0}, we know that there exists a curve connecting $\lambda_1$ to $\lambda_3$, satisfying ${\Im}( \mathcal{I}^{(2)} (\Upsilon))=0$. Since $M(\lambda_3)=\pi$, $M(\lambda_1)=3\pi$ (see \Cref{prop:Q-M-cn}) and $M(\lambda)$ is continuous and monotonous, only when $P=1$, the set $Q^{(2)}_{sub,C}=Q^{(2)}_P$ holds. So if $2 E^{(2)}_2/  K^{(2)}_2<1$, solutions are spectrally stable with respect to co-periodic perturbations but no other subharmonic perturbation. 	
\end{proof-spec-cn-P}

\begin{figure}[h]
	\centering
	\subfigure[$\lambda_1=-0.8+0.3\ii$, $\lambda_2=0.5\ii$, $\lambda_3=0.8+0.3\ii$, $k^{(2)}_2\approx0.2669$ ]{\includegraphics[trim=0.5cm 0cm 1.1cm 0.5cm,width=0.45\textwidth]{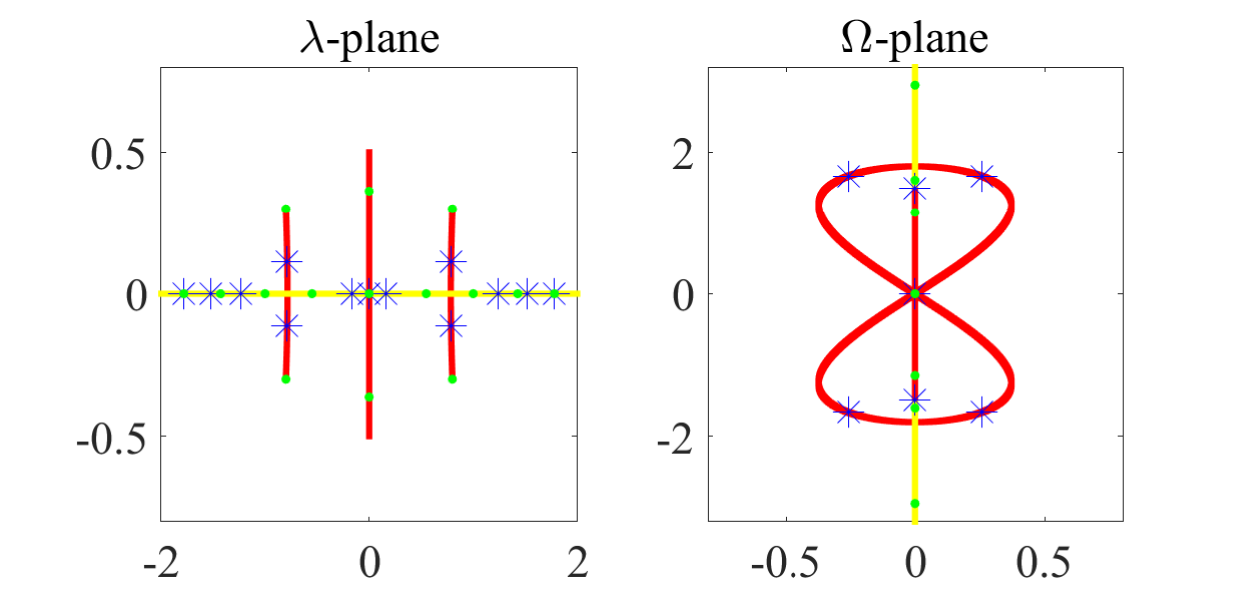}}
	\subfigure[$\lambda_1=-0.5+0.6\ii$, $\lambda_2=0.8961\ii$, $\lambda_3=0.5+0.6\ii$, $k^{(2)}_2\approx0.3493$]{\includegraphics[trim=0.5cm 0cm 1.1cm 0.5cm,width=0.45\textwidth]{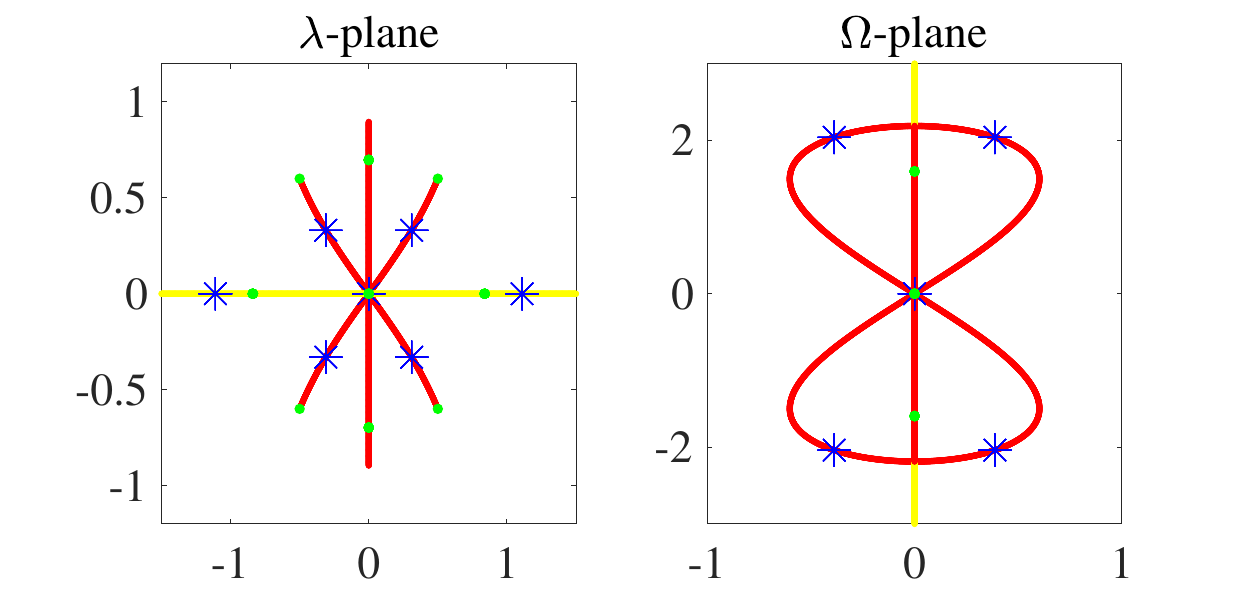}}
	\subfigure[$\lambda_1=-0.4+0.5\ii$, $\lambda_2=1.05\ii$, $\lambda_3=0.4+0.5\ii$, $k^{(2)}_2\approx  0.1870 $]{\includegraphics[trim=0.5cm 0cm 1.1cm 0.5cm,width=0.45\textwidth]{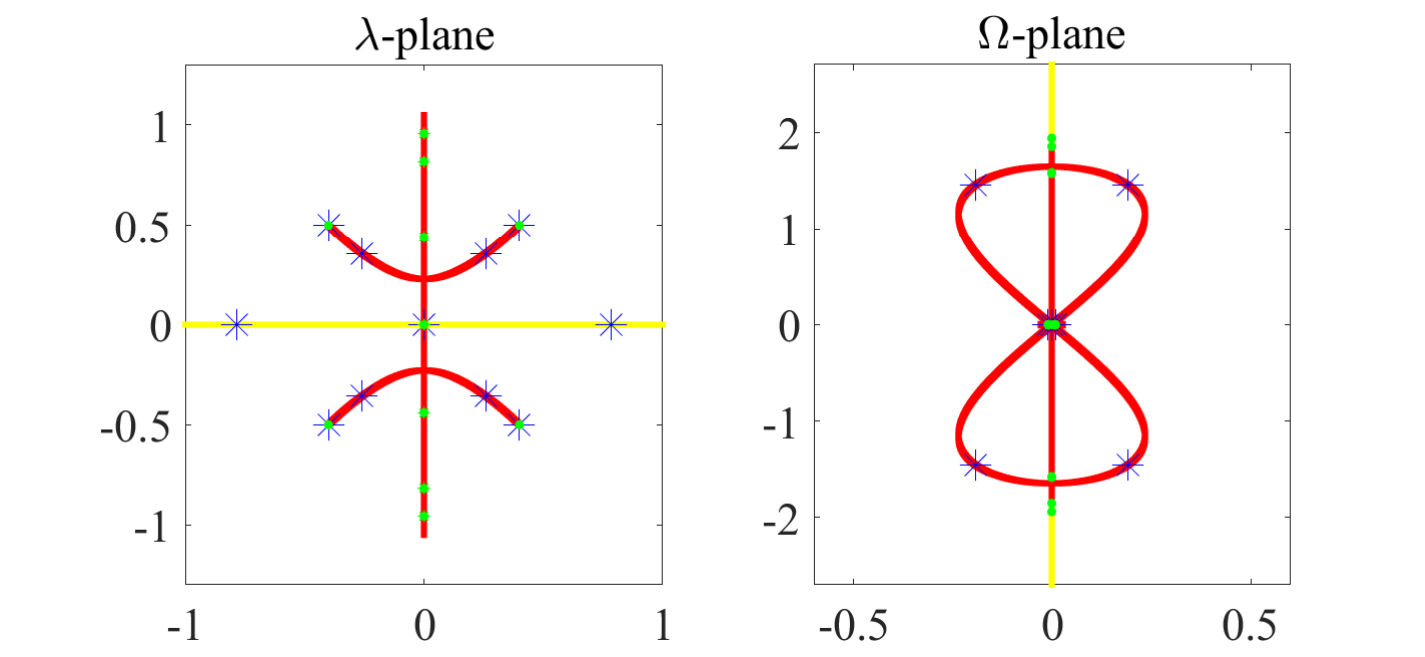}}
	\subfigure[$\lambda_1=-0.5+1.3\ii$, $\lambda_2=0.5669\ii$, $\lambda_3=0.5+1.3\ii$, $k^{(2)}_2\approx0.9089$]{\includegraphics[trim=0.5cm 0cm 1.1cm 0.5cm,width=0.45\textwidth]{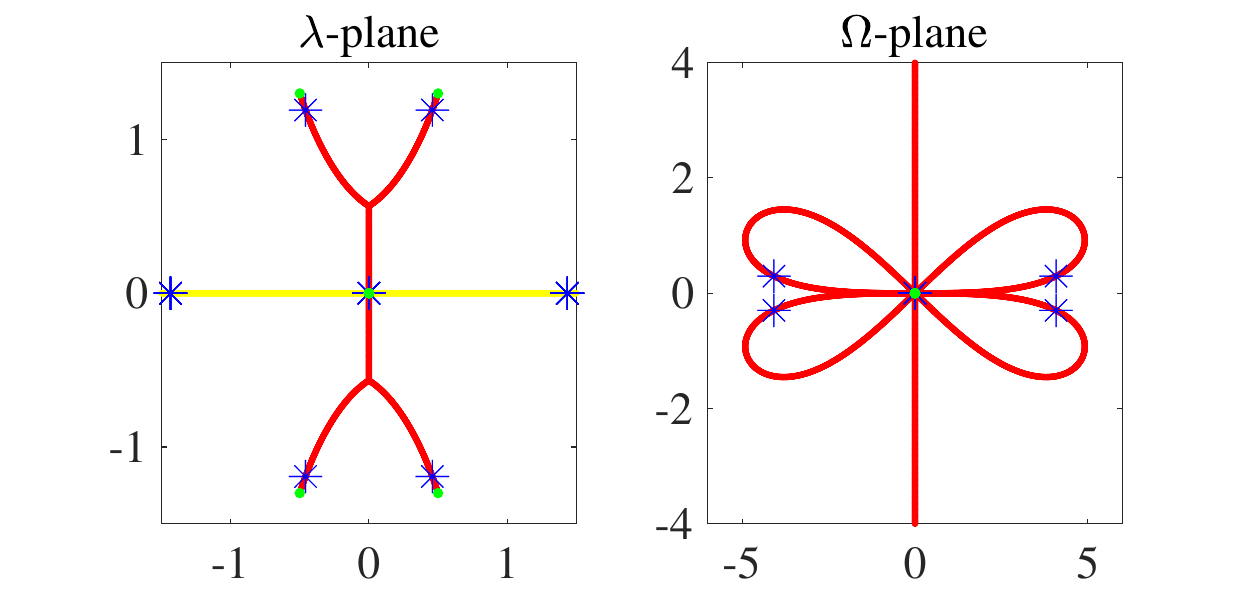}}
	\subfigure[$\lambda_1=-0.5+1.3\ii$, $\lambda_2=0.3\ii$, $\lambda_3=0.5+1.3\ii$, $k^{(2)}_2\approx0.9274$]{\includegraphics[trim=0.5cm 0cm 1.1cm 0.5cm,width=0.45\textwidth]{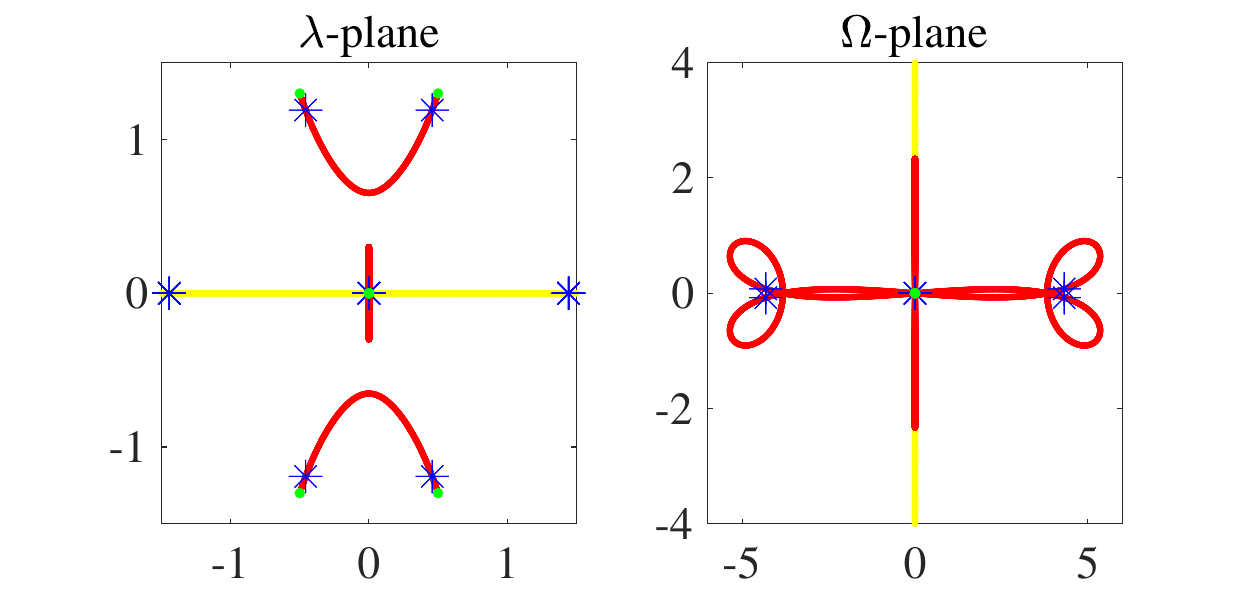}}
	\caption{The set $Q$ and the related eigenvalues $\Omega(\lambda)$ under the \ref{case2}. The related numerical spectrums are consistent with the recent work \cite{CuiP-2025}.}
	\label{fig3}
\end{figure}


Based on the above analysis, we would like to study the maximum value of the parameter $P$.
The function $M(\lambda_0)$ is 
\begin{equation}\label{eq:M-lambda_0}
	\begin{split}
		\!\!\!\!\! \ M(\lambda_0)
		&\ \xlongequal[\eqref{eq:u-parameters-cn-k}]{\eqref{eq:M},\eqref{eq:I-def-C}}-2\ii K^{(2)}_2  \sqrt{A}Z\left(F\left(\frac{(1-2 E^{(2)}_2/  K^{(2)}_2)^{1/2}}{1- E^{(2)}_2/  K^{(2)}_2},  k^{(2)}_2\right),  k^{(2)}_2\right)+2\pi\\
		+&\  2\ii((K^{(2)}_2- 2E^{(2)}_2)((K^{(2)}_2- E^{(2)}_2)^2- 4 (k^{(2)}_2)^2(K^{(2)}_2-2E^{(2)}_2)))^{1/2}/  ((K^{(2)}_2- E^{(2)}_2) (K^{(2)}_2)^{1/2}),
	\end{split}
\end{equation}
with $M(\lambda_0)\in [\pi,2\pi]$.
Thus, we obtain that the value of the parameter $\max(P)$ is dependent on the modulus $  k^{(2)}_2$. We provide the related figure as follows (see the Fig. \ref{fig:P}), by the above equation \eqref{eq:M-lambda_0}.

\begin{figure}[h]
	\centering
	\includegraphics[width=0.8\textwidth]{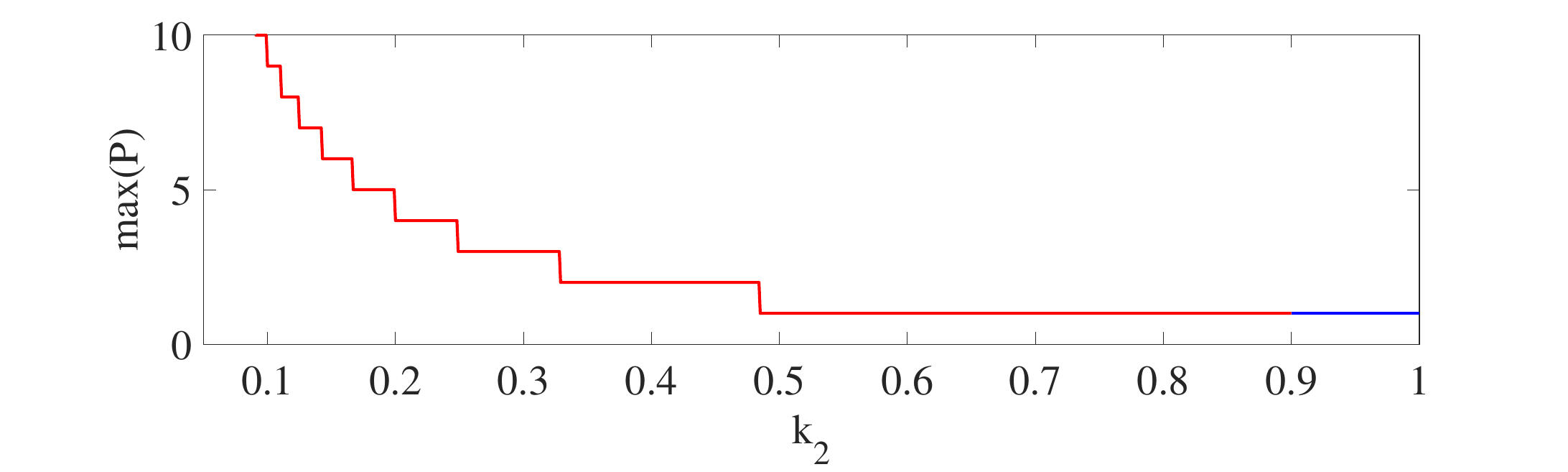}
	\caption{The maximum value of the parameter $P$ under the modulus $k_2^{(2)}$.}
	\label{fig:P}
\end{figure}

\section{The orbital stability of the two-phase solutions }\label{sec:orbital-stability}

The previous section provides the conditions for the spectral stability and P-subharmonic stability of two-phase solutions. Building on these results, we study the orbital stability of the above two-phase solutions in this section.
As we all know, the orbital stability, defined in \Cref{define:orbital stable}, is characterized in terms of the spectrum of the second variation. Since the Krein signature can evaluate the second variation, we convert studying the spectrum of the second variation to considering the Krein signature, which was used to establish the orbital stability of the periodic solutions in the defocusing mKdV equation \cite{Deconinck-10}, the cnoidal waves of the KdV equation \cite{DeconinckK-10} and the elliptic function solutions of the mKdV equation \cite{LingS-23-mKdV-stability}. 
To study the orbital stability, we elaborate on the corresponding Hamiltonian functionals and their variations. 

We consider some corresponding ordinary differential equations of the two-phase solutions.
If $u$ is the stationary solution of the equation $\mathcal{J}\hat{\mathcal{H}}_1^{\prime}(u)=0$ defined in equation \eqref{eq:H-mKdV}, combined with the solutions provided in \Cref{theorem:solution-u}, we obtain 
$\mathcal{JH}^{\prime}_3(u)-v\mathcal{JH}^{\prime}_1(u)=0$, $\mathcal{H}^{\prime}_3(u)-v\mathcal{H}^{\prime}_1(u)=\hat{c}_1$, 
where $\hat{c}_1=-32\lambda_1^2\lambda_2^2\lambda_3^2$.
Moreover, we obtain 
$\mathcal{F}(\mathcal{H}^{\prime}_3(u)-v\mathcal{H}^{\prime}_1(u)-\hat{c}_1)=0$, which deduce that
\begin{equation}\label{eq:hat-c}
	\mathcal{F}(\mathcal{H}^{\prime}_3(u)-v\mathcal{H}^{\prime}_1(u))=2\hat{c}_2u=2\hat{c}_2\mathcal{H}^{\prime}_1(u), \qquad \hat{c}_2=-4v_1,
\end{equation}
 by equations \eqref{eq:det-L-s} and \eqref{eq:define-R-u}.
Therefore, we obtain 
\begin{equation}\nonumber
	\begin{split}
		0=&\, \mathcal{J}\hat{\mathcal{H}}_3^{\prime}(u)=\mathcal{JH}_5^{\prime}(u)+c_{5,3}\mathcal{JH}_3^{\prime}(u)+c_{5,1}\mathcal{JH}_1^{\prime}(u)\\
		=&\, \mathcal{JFH}_3^{\prime}(u)+(c_{5,3}v+c_{5,1})\mathcal{JH}_1^{\prime}(u)\\
		=&\, \mathcal{JF}(\mathcal{H}_3^{\prime}(u)-v\mathcal{H}_1^{\prime}(u))+v\mathcal{JFH}_1^{\prime}(u)+(c_{5,3}v+c_{5,1})\mathcal{JH}_1^{\prime}(u)\\
		=&\, (2\hat{c}_2+v^2+c_{5,3}v+c_{5,1})\mathcal{JH}_1^{\prime}(u),
	\end{split}
\end{equation}
 which implies 
\begin{equation}\label{eq:c1c2}
	c_{5,1}=-v^2-c_{5,3}v-2\hat{c}_2, \qquad c_{5,3}\in \mathbb{R}.
\end{equation}
Similarly, the solution $u$ also satisfies the higher-order stationary equations $\mathcal{J}\hat{\mathcal{H}}'_n(u)=0, n=2,3,\cdots$. Their further details are provided in \cite{LingS-23-mKdV-stability}.

When $W(\xi;\Omega_1)$ satisfies $\Omega_1W(\xi;\Omega_1)=\mathcal{JL}_1W(\xi;\Omega_1)$ with $\Omega_1\in \ii\mathbb{R}$ (defined in equation \eqref{eq:spectral-Omega}), we consider the Krein signature $\mathcal{K}_1(\lambda)$. 
We first study a special case that $\Omega_1=0$, i.e., $\lambda=0,\lambda_{1,2,3},\lambda_{1,2,3}^*$. 
It is easy to know that when $\lambda=0$, the eigenfunction could be written as $W(\xi;0)=\partial_{\xi}u$, and the Krein signature is $\mathcal{K}_1=\left\langle\partial_{\xi}u,\mathcal{L}_1 \partial_{\xi}u \right\rangle_{L^2}=0$. 
When $\lambda=\lambda_{1,2,3},\lambda_{1,2,3}^*$, by analyzing the exponent part of functions $\Phi_{1}$ and $\Phi_{2}$, we know that the function $W(\xi;\Omega_1)$ is not a periodic one. 
Now, we consider the value of $\mathcal{K}_1(\lambda)$ when $\lambda\in \mathbb{R}$ and $\Omega_1\in \ii \mathbb{R}$. 
The function $W(\xi;\Omega_1)=2\lambda(\Phi_{1}^2-\Phi_{2}^2)\exp(-\Omega_1 \eta),\lambda\in \mathbb{R}\backslash\{0\}$ is the eigenfunction of the linearized spectral problem \eqref{eq:spectral-Omega} associated with the eigenvalue $\Omega_1$. 
From \cite{LingS-23-mKdV-stability} and $\Omega_1=8\ii y\in \ii \mathbb{R}$, we obtain
\begin{equation}\label{eq:WLW}
	\begin{split}
		W^*(\xi;\Omega_1)\mathcal{L}_1 W(\xi;\Omega_1)
		=2\ii \lambda\Omega_1 (\Phi_{1}^2+\Phi_{2}^2)(\Phi_{1}^{*2}-\Phi_{2}^{*2}).
	\end{split}
\end{equation}	
From the results of the previous section, we conclude that the exponential factors of the functions $\Phi_{1}(\xi,\eta)$ and $\Phi_{2}(\xi,\eta)$ are purely imaginary for all real variables $\xi,\eta\in \mathbb{R}$,
if and only if $\Re( \mathcal{I}^{(2)} (\lambda))=0$ and $\Re(\Omega)=0$. 
However, due to the lack of suitable additional formulas, it is not easy to obtain a simple form of the equation \eqref{eq:WLW} to proceeding the following calculations.
Then, we would like to study the function $(\Phi_{1}^2+\Phi_{2}^2)(\Phi_{1}^2-\Phi_{2}^2)^*$. 

\begin{lemma}\label{lemma:Phi+-}
	For any $\Omega(\lambda)\in \ii \mathbb{R}$ and $\lambda \in \mathbb{R}$ or $\lambda \in \ii \mathbb{R}$, we obtain
	\begin{equation}\label{eq:phi-11-21}
			\!\! (\Phi_{1}^2+\Phi_{2}^2)(\Phi_{1}^2-\Phi_{2}^2)^* \!
			=\! \pm4(y(2\lambda u^2-\beta_1-\beta_2)-8\lambda u_{\xi}(u-u_1)(u-u_2)(u-u_3)(u-u_4)/u),
	\end{equation}
	where $\beta_{1,2}=2\lambda^3 - v\lambda/2 \mp 2y$. 
	When $\lambda \in  \mathbb{R}$, choose ``$+$";
	when $\lambda \in  \ii\mathbb{R}$, choose ``$-$".
\end{lemma}

\begin{proof}
	We will divide them into the following two cases. 
	One is $\lambda\in  \mathbb{R}$ and the other one is $\lambda\in  \ii \mathbb{R}$.
	 
	When $\lambda\in  \mathbb{R}$ and $\Omega(\lambda)\in \ii \mathbb{R}$, we get $y\in \mathbb{R}$, $\beta_{1,2}=2\lambda^3 - v\lambda/2 \mp 2y\in \mathbb{R}$ since $\Omega(\lambda)=8\ii y\in \ii \mathbb{R}$.
	Since $\pm \ii y$ are two eigenvalues of the matrix function $\mathbf{L}(\xi,\eta;\lambda)$, combining equations \eqref{eq:define-r-12} with \eqref{eq:relation-Phi-r}, we obtain that the vector solutions of the Lax pair could be expressed as
	\begin{equation}\nonumber
		\Phi_{1}= \sqrt{\lambda u^2(\xi)-\beta_1}\exp\left( \ii \lambda \xi+\int_{0}^{\xi}\frac{4\ii\lambda \beta_1- \ii L_0 u(x)}{2\lambda u^2(x)-2\beta_1}\dd x+4\ii y \eta\right), \qquad
		\Phi_{2}= r_1\Phi_{1},	
	\end{equation}
	where $L_0=u_{xx}+2u (u^2-2\alpha_1)=4\sqrt{s_0}$ is a constant and the function $r_1$ is given by the first equation in \eqref{eq:define-r-12}. 
	The detailed process of constructing functions $\Phi_{1}$ and $\Phi_{2}$ is provided in \cite{LingS-23-mKdV-stability}.
	Since $y\in \mathbb{R}$, $\beta_{1,2}\in \mathbb{R}$ and the solution $u(\xi)\in \mathbb{R}$, we obtain 
	\begin{equation}\nonumber
		\ii \lambda \xi+\int_{0}^{\xi}\frac{4\ii\lambda \beta_1- \ii L_0 u(x)}{2\lambda u^2(x)-2\beta_1}\dd x+4\ii y \eta\in \ii \mathbb{R}, \qquad \text{for any} \quad (x,t)\in \mathbb{R}^2.
	\end{equation}
	Therefore, from equations \eqref{eq:L-elements}, \eqref{eq:define-r-12} and \eqref{eq:define-R-u-mu},
we obtain 
	\begin{equation}\nonumber
		\begin{split} 
			(\Phi_{1}^2+\Phi_{2}^2)(\Phi_{1}^2-\Phi_{2}^2)^*
			\xlongequal{\eqref{eq:define-R-u-mu}} 4y(2\lambda u^2-\beta_1-\beta_2)-32\lambda u_{\xi}(u-u_1)(u-u_2)(u-u_3)(u-u_4)/u,
		\end{split}
	\end{equation}
  which implies that the equation \eqref{eq:phi-11-21} holds with the opposite sign.

	Then, we consider another case $\lambda \in  \ii \mathbb{R}$ and $\Omega(\lambda)\in \ii \mathbb{R}$.
	Combining equations \eqref{eq:define-r-12} with \eqref{eq:relation-Phi-r}, we obtain 
	\begin{equation}\nonumber
		\begin{split}
			\frac{\Phi_{1,\xi}}{\Phi_{1}}
			=&-\ii \lambda+\frac{\ii y +\ii \lambda^3+\ii \lambda(\alpha_1-u^2/2)}{\lambda^2+\ii u_{\xi}\lambda/(2u)-u_{\xi\xi}/(4u)+\alpha_1-u^2/2} 
			=\frac{2\lambda^2 u_{\xi}+\ii \lambda u_{\xi\xi}+4\ii uy }{4\lambda^2u+2\ii\lambda u_{\xi}-4|\lambda_1\lambda_2\lambda_3|}, 
		\end{split}
	\end{equation}
	and $\frac{\Phi_{1,\eta}}{\Phi_{1}}=4\ii y$,	 which implies that the fundamental solutions of the Lax pair could be expressed as
	\begin{equation}\nonumber
		\begin{split}
			\Phi_{1}
			=\sqrt{2 u(\lambda-\mu_1)(\lambda-\mu_2)}\exp\left(\int_{0}^{\xi}\frac{2\ii u(x)y \dd x}{2\lambda^2u(x)+\ii\lambda u_x(x)-2|\lambda_1\lambda_2\lambda_3|}+4\ii y \eta\right), 
		\end{split}
	\end{equation}
	and $\Phi_{2}=r_1\Phi_{1}$.
	Since $\lambda\in  Q^{(2)} _I$, we obtain $y\in \mathbb{R}$ and $\beta_1^*=-\beta_2$.
	So, it is easy to obtain 
	\begin{equation}\nonumber
		\int_{0}^{\xi}\frac{2\ii u(x)y \dd x}{2\lambda^2u(x)+\ii\lambda u_x(x)-2|\lambda_1\lambda_2\lambda_3|}+4\ii y \eta\in \ii \mathbb{R}, \qquad \text{for any} \quad (x,t)\in \mathbb{R}^2. 
	\end{equation}
Therefore, the equation \eqref{eq:phi-11-21} holds with the negative sign.
\end{proof}

\subsection{The orbital stability analysis of \ref{case2}} \label{subsec:orbital-cn}

\begin{lemma}\label{lemma:Krein}
	The Krein signature (\Cref{defin:Krein}) with respect to solution $u(x,t)$ defined in equation \eqref{eq:u2-elliptic} is 
	\begin{equation}\label{eq:Krein-value}
		\begin{split}
			\mathcal{K}_1(\lambda)
			=&\ 16|\lambda|^2 \Omega_1^2P  K^{(2)}_2 \left((\lambda_2^2-A-\lambda^2)+2A E^{(2)}_2/  K^{(2)}_2\right)/\alpha.
		\end{split}
	\end{equation} 
\end{lemma}

\begin{proof}
Considering \ref{case2} with the solution $u(x,t)$ expressed in equation \eqref{eq:u2-elliptic} under the transformation \eqref{eq:transformation-xi-eta},
we consider the integration 
\begin{equation}\label{eq:u2-solution-int}
	\begin{split}
		\int_{-PT}^{PT} u^2 \dd \xi 
		\xlongequal[\eqref{eq:equ-2}]{\eqref{eq:u2-elliptic}}
		&\ \frac{2P}{\alpha}\int_{0}^{2  K^{(2)}_2} \left(\frac{A+B}{\ii \lambda_2}-\frac{2A}{\ii \lambda_2}\frac{1-\alpha_1\cn(\xi,  k^{(2)}_2)}{1-\alpha_1^2+\alpha_1^2\sn^2( \xi,  k^{(2)}_2)}\right)^2 \dd \xi \\
		\xlongequal[\eqref{eq:u-parameters-cn-k},\eqref{eq:1-1}]{\eqref{eq:define-third-integral},\eqref{eq:integral}}
		& -\frac{4P}{\alpha}\left((\lambda_1^2-\lambda_2^2+\lambda_3^2+2A)  K^{(2)}_2-4A E^{(2)}_2\right), \qquad \alpha_1=\frac{\delta-1}{1+\delta}.
	\end{split}
\end{equation}
When $\lambda\in  Q^{(2)} _{R}$, we obtain 
\begin{equation}\nonumber
	\begin{split}
		\mathcal{K}_1(\lambda)
		\xlongequal[\eqref{eq:phi-11-21}]{\eqref{eq:define-krein}}&\int_{-PT}^{PT}2\ii \lambda \Omega_1 4y\left(2\lambda u^2(\xi)-\beta_1-\beta_2\right)\dd \xi\\
	\xlongequal[\eqref{eq:define-Omega}]{\eqref{eq:u2-solution-int}}
		& 16\lambda^2 \Omega_1^2P  K^{(2)}_2 \left((\lambda_2^2-A-\lambda^2)+2A E^{(2)}_2/  K^{(2)}_2\right)/\alpha.
	\end{split}  
\end{equation}
Similarly, when $\lambda \in  Q^{(2)} _I$,  we can also obtain the equation \eqref{eq:Krein-value}.
\end{proof}
From the \Cref{lemma:Krein}, we obtain the Krein signature $\mathcal{K}_1(\lambda)$.
Then, by the integrability we deduce the Krein signature $\mathcal{K}_{2}(\lambda)$.

\newenvironment{proof-kerin}{\emph{Proof of \Cref{prop:kerin}.}}{\hfill$\Box$\medskip}
\begin{proof-kerin}
From \Cref{prop: zr0 zi0} and \Cref{lemma:Krein}, the value of $\mathcal{K}_1(\lambda)$ could be classified into the following cases:
\begin{itemize}
	\item[(i)] When $2 E^{(2)}_2A/  K^{(2)}_2-A>-\lambda_2^2$, there exists a point $0<\lambda_0\in \mathbb{R}\backslash\{0\}$ such that $\mathcal{K}_1(\pm\lambda_0)=0$. And then for any $\lambda\in  Q^{(2)} _R$, it is easy to obtain that: for
	$\lambda \in (-\infty,-\lambda_0)\cup (\lambda_0,+\infty)$, $\mathcal{K}_1(\lambda)>0$; for $\lambda \in (-\lambda_0,0)\cup (0,\lambda_0)$, $\mathcal{K}_1(\lambda)<0$ and for $\lambda =0,\pm\lambda_0,$  $\mathcal{K}_1(\lambda)=0$. 
	For $\lambda\in  Q^{(2)} _I$, it is found that: for $\lambda =\pm\lambda_2$,  $\mathcal{K}_1(\lambda)=0$; while for $\lambda\in  Q^{(2)} _I\backslash\{0,\pm\lambda_2\}$, $\mathcal{K}_1(\lambda)<0$.
	\item[(ii)] When $2 E^{(2)}_2A/  K^{(2)}_2-A=-\lambda_2^2$, we obtain that for any $\lambda\in  Q^{(2)} _R\backslash\{0\}$, $\mathcal{K}_1(\lambda)>0$ and for $\lambda\in  Q^{(2)} _I\backslash\{0,\pm \lambda_{2}\}$, $\mathcal{K}_1(\lambda)<0$. 
	If and only if $\lambda=0$ or $\lambda=\lambda_2$, we have $\mathcal{K}_1(\lambda)=0$.
	
	\item[(iii)] When $0<2 E^{(2)}_2A/  K^{(2)}_2-A<-\lambda_2^2$, there exists a point $\lambda_0\in \ii \mathbb{R}\backslash\{0\}$ such that $\mathcal{K}_1(\pm\lambda_0)=0$. 
	Then for any $\lambda\in  Q^{(2)} _R$, it is easy to obtain 
	$\mathcal{K}_1(\lambda)\ge 0$.  
	When $\lambda\in  Q^{(2)} _I$, we obtain if $\Im(\lambda) \in (-\Im(\lambda_0),0)\cup(0,\Im(\lambda_0))$, then $\mathcal{K}_1(\lambda)> 0$ and if $\Im(\lambda) \in (-\Im(\lambda_2),-\Im(\lambda_0))\cup(\Im(\lambda_0),\Im(\lambda_2))$, then $\mathcal{K}_1(\lambda)<0$. If and only if $\lambda=0,\pm \lambda_0,\pm \lambda_2$, we have $\mathcal{K}_1(\lambda)=0$.
	\item[(iv)] When $2 E^{(2)}_2A/  K^{(2)}_2-A \le 0$, the same as above, for any $\lambda\in  Q^{(2)} _R$, it is easy to obtain 
	$\mathcal{K}_1(\lambda) \ge 0$. And  
	for any $\lambda\in  Q^{(2)} _I$, $\mathcal{K}_1(\lambda)\ge 0$.
\end{itemize} 
Thus, when $2 E^{(2)}_2A/  K^{(2)}_2-A\le 0$, we obtain that for any $\lambda \in  Q^{(2)} _R\cup  Q^{(2)} _I$, $\mathcal{K}_1(\lambda)\ge 0$, and when $2 E^{(2)}_2A/  K^{(2)}_2-A> 0$, not all $\lambda \in  Q^{(2)} _R\cup  Q^{(2)} _I$ such that $\mathcal{K}_1(\lambda)\ge 0$.

We invoke to calculate the value of $\mathcal{K}_2(\lambda)$. By equation \eqref{eq:spectral-Omega}, we get
\begin{equation} \nonumber
	\mathcal{K}_n(\lambda)
	=\left\langle W,\mathcal{L}_nW \right\rangle_{L^2}
	=\left\langle W,\Omega_n \mathcal{J}^{-1} W \right\rangle_{L^2}
	=\frac{\Omega_n}{\Omega_1}\left\langle W,\mathcal{L} W \right\rangle_{L^2}
	=\frac{\Omega_n}{\Omega_1} \mathcal{K}_1(\lambda).
\end{equation}
The relationship between $c_{5,3}$ and $c_{5,1}$ is obtained in equation \eqref{eq:c1c2}. 
By the AKNS hierarchy in equation \eqref{eq:H-mKdV}, we derive the $n$-th order mKdV equation, where $n$ is odd.
The related Lax pair could be expressed as $\Phi_{\eta_{n}}=\mathbf{\hat{V}}_{n}\Phi$, where $\hat{\mathbf{V}}_n=\mathbf{V}_{2n+1}+\sum_{i=0}^{2n} c_{n,i} \mathbf{V}_i$, are defined in equation \eqref{eq:Theta-i-expression-diag}.
The corresponding Lax pair of the \ref{eq:mKdV} equation could be rewritten as $\hat{\mathbf{V}}_1=4\mathbf{V}_3-v\mathbf{V}_1$, where $\mathbf{V}_i$ are defined in equation \eqref{eq:Theta-i-expression-diag}.
The eigenvalue $\Omega$ is determined by the solution $\Phi(\xi,\eta;\lambda)$ of the Lax pair \eqref{eq:Lax-pair} under the transformation \eqref{eq:transformation-xi-eta}. 
Furthermore, we obtain $\det(\hat{\mathbf{V}}_1-\Omega_1/2)=0$.
Considering the Lax pair $\Phi_{\xi}=\mathbf{U}\Phi,\Phi_{\eta}=\hat{\mathbf{V}}_1\Phi$,
	since the solution $u(x,t)$ under the transformation \eqref{eq:transformation-xi-eta} could be rewritten as $u(\xi)$, which is independent with respect to the variable $\eta$, the solution of the $\Phi_{\eta_2}=\hat{\mathbf{V}}_2\Phi$ must satisfy $\Phi_{\eta_2}=(\Omega_2/2)\Phi=\hat{\mathbf{V}}_2\Phi$, where $\hat{\mathbf{V}}_2=16\mathbf{V}_5+4c_{5,3}\mathbf{V}_3+c_{5,1}\mathbf{V}_1$ and $\Omega_2/2$ is the eigenvalue of the matrix $\hat{\mathbf{V}}_2$,
	i.e.,
	$\det(\hat{\mathbf{V}}_5-\Omega_2/2)=0$.
	For the higher-order Hamiltonian $\hat{\mathcal{H}}_n(u)$, defined in equation \eqref{eq:H-mKdV} satisfying $\mathcal{J}\hat{\mathcal{H}}_n^{\prime}(u)=0$, together with the associated Lax pair $\Phi_{\xi}=\mathbf{U}\Phi,\Phi_{\eta_n}=\hat{\mathbf{V}}_n\Phi$. Assume $m>n$, it is also a stationary solution of the $m$-th mKdV flow, i.e. $\mathcal{J}\hat{\mathcal{H}}_m^{\prime}(u)=0$. 
	We consider the eigenvalues $\Omega_n/2$ and $\Omega_m/2$ of the matrices $\hat{\mathbf{V}}_n$ and $\hat{\mathbf{V}}_m$.
	These eigenvalues must satisfy the following relationship:
	\begin{lemma}\label{lemma:Omega-n-m}
		The eigenvalues of $\Omega_m=\Omega_m(\lambda)$ and $\Omega_n=\Omega_n(\lambda)$ must satisfy $\Omega_m^2=p_m^2(\lambda)\Omega_n^2$, where $p_m(\lambda)$ is a polynomial of degree $m-n>0$ with respect to the spectral parameter $\lambda$.
	\end{lemma}
	\begin{proof}
		In virtue of the definition of matrices 
		\begin{equation}\nonumber
			\hat{\mathbf{V}}_{n}=\begin{bmatrix}
				A_n & B_n \\ C_n & -A_n
			\end{bmatrix}, \quad \text{and} \quad
			\hat{\mathbf{V}}_{m}=\begin{bmatrix}
				A_m & B_m \\ C_m & -A_m
			\end{bmatrix},
		\end{equation}
		we know that eigenvalues $\pm\Omega_{n,m}(\lambda)/2$ of matrices $\hat{\mathbf{V}}_{n}$ and $\hat{\mathbf{V}}_{m}$ satisfy
		\begin{equation}\label{eq:eigenvalue}\nonumber
			\begin{split}
					&\ \det(\hat{\mathbf{V}}_n\mp \Omega_n(\lambda)/2)=\Omega_n^2(\lambda)/4-(A_n^2+B_nC_n)=0, \\
				&\ \det(\hat{\mathbf{V}}_m\mp \Omega_m(\lambda)/2)=\Omega_m^2(\lambda)/4-(A_m^2+B_mC_m)=0.
			\end{split}
		\end{equation}
		Since $\Phi_{\eta_n}=(\Omega_{n}(\lambda)/2)\Phi=\hat{\mathbf{V}}_n\Phi$ and the matrix $\hat{\mathbf{V}}_n$ is independent of the variable $\eta_n$, the solution of this ordinary differential equation could be expressed as $[\psi_1(\lambda) \,\, \psi_2(\lambda)]\ee^{\sigma_3\Omega_{n}(\lambda)/2\eta_n}$, where $\psi_{1,2}(\lambda)$ are two linearly independent functions and do not depend on the variables $\eta_n$.
		Therefore, it follows that $\psi_{1,2}(\lambda)$ are eigenfunctions of the matrix $\hat{\mathbf{V}}_n$ with respect to the eigenvalue $\pm\Omega_{n}(\lambda)/2$, respectively.
		Moreover, when $m>n$, $[\psi_1(\lambda) \,\, \psi_2(\lambda)]\ee^{\sigma_3\Omega_{m}(\lambda)/2\eta_m}$ is the solution of the equation  $\Phi_{\eta_m}=(\Omega_{m}(\lambda)/2)\Phi=\hat{\mathbf{V}}_m\Phi$, which implies $\psi_{1,2}(\lambda)$ are also eigenfunctions of the matrix $\hat{\mathbf{V}}_m$ with respect to eigenvalues $\pm\Omega_{m}(\lambda)/2$.
		Without loss of generality,	we set $\Omega_n^2(\lambda)/4=A_n^2+B_nC_n=s_n\prod_{i=1}^{4n+2} (\lambda-\lambda_i)$, where $s_n$ is independent of $\lambda$. 
		If $ \Omega_n(\lambda_0)=0$, we get $\psi_1(\lambda)=\psi_2(\lambda)$.
		Since $\psi_1(\lambda)=\psi_2(\lambda)$ is also the eigenfunction of the matrix $\hat{\mathbf{V}}_m$ with respect to eigenvalue $\pm \Omega_{m}(\lambda_0)$, 
		it follows that $\Omega_m(\lambda_0)=0$. So  the eigenvalue could be rewritten as $\Omega_m^2=g_m(\lambda)\Omega_n^2$.
		If $\Omega_m(\lambda_0)=0$ and $\Omega_n(\lambda_0) \neq 0$, 
		there must exist two linearly independent eigenfunctions, corresponding to an eigenvalue of multiplicity two.
		So, the order of $\lambda_0$ such that $\Omega_m(\lambda_0)=0$ is two. Therefore, we conclude $\Omega_m^2(\lambda)=p_m^2(\lambda) \Omega_n^2(\lambda)$.
		\end{proof}
For the KdV equation, the above related theorem was given in the previous literature \cite{NivalaD-2010periodic}. Here we provide an alternative proof to the mKdV case. In what follows, we will provide the required explicit polynomials $p_m(\lambda)$.
Through the compatibility condition of the linear system: $\Phi_{\xi\eta}=\Phi_{\eta\xi}$, we obtain the zero-curvature equation with respect to $\hat{\mathbf{V}}_{1}$: $\mathbf{U}_{\eta}-\hat{\mathbf{V}}_{1,\xi}+[\mathbf{U},\hat{\mathbf{V}}_{1}]=0$. 
Collecting the coefficients of the spectral parameter $\lambda$ for this zero-curvature equation, we obtain that all coefficients vanish.
For the constant term, we get
\begin{equation}\label{eq:stationary-2}
	   \begin{split}
		0=&\ 4\ii \Psi_{3,\xi}-v\ii \Psi_{1,\xi} -[\Psi_1,4\Psi_{3}]\\	
		\xlongequal[\eqref{eq:Theta-i-expression-diag}]{\eqref{eq:Theta-i-expression-off}}&\ 
		 8\sigma_3\Psi_4^{\off}-2v\sigma_3\Psi_2^{\off}-[\Psi_1,4\Psi_{3}^{\off}]-2\ii \sigma_3(\Psi_{1,\xi}\Psi_2^{\off}+\Psi_2^{\off}\Psi_{1,\xi})-2\ii \sigma_3(\Psi_{2,\xi}^{\off}\Psi_1+\Psi_1\Psi_{2,\xi}^{\off}) \\
		\xlongequal{\eqref{eq:Theta-i-expression-off}}&\ 2\sigma_3\left(4\Psi_4^{\off}-v\Psi_2^{\off}\right),
	\end{split}
\end{equation}
where $[\Psi_1,4\Psi_3]=0$ by equation \eqref{eq:Psi-i}.
Then we consider the equation 
\begin{equation}\label{eq:stationary-4}
	\begin{split}
		-(8\Psi_5-2v\Psi_3)^{\diag}\xlongequal{\eqref{eq:Theta-i-expression-diag}}&\ 4\sigma_3(\Psi_4\Psi_1+\Psi_1\Psi_4+\Psi_2\Psi_3+\Psi_3\Psi_2)^{\diag}-v\sigma_3(\Psi_2\Psi_1+\Psi_1\Psi_2)^{\diag}\\
		=&\ \sigma_3(4\Psi_4^{\off}-v\Psi_2^{\off})\Psi_1+\sigma_3\Psi_1(4\Psi_4^{\off}-v\Psi_2^{\off})\xlongequal{\eqref{eq:stationary-2}}0,
	\end{split}
\end{equation}
where $(\Psi_2\Psi_3+\Psi_3\Psi_2)^{\diag}=0$ by equation \eqref{eq:Psi-i}.

Moreover, we consider some stationary equations deduced by the Hamiltonian functional and the recursion operator $\mathcal{F}$ in equations \eqref{eq:H-mKdV} and \eqref{eq:c1c2}. 
Since $\mathcal{F}(\mathcal{H}^{\prime}_3(u)-v\mathcal{H}^{\prime}_1(u))-2\hat{c}_2u=0$ by equation \eqref{eq:hat-c} and $\mathcal{J}\mathcal{F}(\mathcal{H}^{\prime}_3(u)-v\mathcal{H}^{\prime}_1(u))-2\hat{c}_2\mathcal{J}u=0$, we obtain
\begin{equation}\label{eq:stationary-31}
	0= 16\ii \Psi_{5,\xi}^{\off}-4v\ii \Psi_{3,\xi}^{\off}-2\hat{c}_2\ii \Psi_{1,\xi}^{\off}-[\Psi_1,16 \Psi_{5}-4v \Psi_{3}]^{\off}
	\xlongequal{\eqref{eq:stationary-4}}16\ii \Psi_{5,\xi}^{\off}-4v\ii \Psi_{3,\xi}^{\off}-2\hat{c}_2\ii \Psi_{1,\xi}^{\off}.
\end{equation}
Together with equation \eqref{eq:stationary-31} and $\mathcal{F}(\mathcal{H}^{\prime}_3(u)-v\mathcal{H}^{\prime}_1(u))-2\hat{c}_2u=0$ in equation \eqref{eq:hat-c}, we obtain
\begin{equation}\label{eq:stationary-32}
	\begin{split}
		0=&\ 16\ii \Psi_{5}^{\off}-4v\ii \Psi_{3}^{\off}-2\hat{c}_2\ii \Psi_{1}^{\off}\\
		\xlongequal{\eqref{eq:Theta-i-expression-off}}&\  -8( \Psi_{4,x}^{\off}+\ii[\Psi_1,\Psi_4^{\diag}])+2v( \Psi_{2,x}^{\off}+\ii[\Psi_1,\Psi_2^{\diag}]) +\ii\hat{c}_2 [\Psi_1,\Psi_0]\\
		\xlongequal{\eqref{eq:stationary-2}} &\-\frac{\ii}{2} [\Psi_1,16\Psi_4^{\diag}-4v\Psi_2^{\diag}-2 \hat{c}_2 \Psi_0],
	\end{split}
\end{equation}
which implies $16\Psi_4^{\diag}-4v\Psi_2^{\diag}-2 \hat{c}_2 \Psi_0=0$.
Combining with equation \eqref{eq:stationary-2}, we deduce
\begin{equation}\label{eq:stationary-5}
	16\Psi_4-4v\Psi_2-2 \hat{c}_2 \Psi_0=0.
\end{equation}
By equations \eqref{eq:stationary-4} and \eqref{eq:stationary-32},
we get 
\begin{equation}\label{eq:stationary-6}
	16\Psi_5-4v\Psi_3-2 \hat{c}_2 \Psi_1=0.
\end{equation}
Based on equations \eqref{eq:H-mKdV} and \eqref{eq:c1c2}, we get
\begin{equation}\nonumber
	\begin{split}
		\hat{\mathbf{V}}_3
		\xlongequal{\eqref{eq:c1c2}} &\ 16\mathbf{V}_5+4c_{5,3}\mathbf{V}_3+c_{5,1}\mathbf{V}_1 \\
		\xlongequal[\eqref{eq:Theta-expand-lambda-infty}]{\eqref{eq:U-V-n}} &\ 
		 (4 \lambda^2+c_{5,3}+v) \hat{\mathbf{V}}_2 +(v^2+c_{5,1}+c_{5,3} v)\mathbf{V}_1+4\ii v\left( \lambda \Psi_2+\Psi_3\right)-16\ii (\lambda\Psi_4+\Psi_5)\\
		\xlongequal[\eqref{eq:c1c2}]{\eqref{eq:stationary-5},\eqref{eq:stationary-6}}
		&\  (4 \lambda^2+c_{5,3}+v) \hat{\mathbf{V}}_2.
	\end{split}
\end{equation} 
By the linear algebra, the eigenvalue $\Omega_2$ of the fifth-order mKdV equation demonstrates $\Omega_2
=(v+4\lambda^2+c_{5,3})\Omega_1$, which had also been proven in \cite{Deconinck-10}.
Therefore, choosing $c_{5,3}=4(A-\lambda_2^2)-v-8A E^{(2)}_2/  K^{(2)}_2$,  the Krein signature $\mathcal{K}_2(\lambda)$ is linearly related to the function $\mathcal{K}_1(\lambda)$ via the equation 
\begin{equation}\label{eq:K_2(z)}
	\mathcal{K}_2(\lambda)
	=-64|\lambda|^2 \Omega_1^2P  K^{(2)}_2 \left((\lambda_2^2-A-\lambda^2)+2A E^{(2)}_2/  K^{(2)}_2\right)^2/\alpha .
\end{equation}
We have $\mathcal{K}_2(\lambda)\ge 0$, for any $\lambda\in  Q^{(2)} _I\cup  Q^{(2)} _R$ with $\Omega_1\in \ii \mathbb{R}$, the equality is valid if and only if $\lambda=0$, $ \pm \lambda_2$ or $\pm \lambda_0$, where $\lambda_0$ is defined in equation \eqref{eq:define-lambda-0}.
\end{proof-kerin}

\begin{lemma}\label{lemma:alpha-0}
	If the two-phase solutions of the mKdV equation with branch points satisfying the \ref{case2} are spectrally stable with respect to perturbations of the period $2PT, P\in \mathbb{N}$, we get the following cases:
	\begin{enumerate}
		\item[\rm (a)] If $2 E^{(2)}_2\ge  K^{(2)}_2$ and $M(\lambda_0)<-\pi(P-2)/P+2\pi$, all $2PT$ periodic eigenfunctions except $\partial_{\xi}u$ satisfy 
		\begin{equation}\label{eq:bounded away}
			\left\langle \mathcal{L}_2 W, W \right\rangle_{L^2}\ge \alpha_0 \|W\|_{H^2([-PT,PT])}^2, \qquad \alpha_0>0;
		\end{equation} 
		\item[\rm(b)] If $2 E^{(2)}_2\ge  K^{(2)}_2$ with $M(\lambda_0)=-\pi(P-2)/P+2\pi$, all $2PT$ periodic eigenfunctions except $\partial_{\xi}u$ and $W(\xi; \Omega(\pm \lambda_0))$ satisfy the inequality \eqref{eq:bounded away}.	
		\item[\rm(c)] If $2 E^{(2)}_2<  K^{(2)}_2$, all $2T$ periodic eigenfunctions except $\partial_{\xi}u(\xi)$ satisfy
		\begin{equation}\label{eq:bounded away-1}
			\left\langle \mathcal{L}_1 W, W \right\rangle_{L^2} \ge \alpha_0 \|W\|_{H^1([-T,T])}^2.
		\end{equation} 
	\end{enumerate}
\end{lemma}

From \Cref{prop:Q-M-cn}, we get $M(\lambda_3)=2T\mathcal{I}^{(2)}(\lambda_3)=\pi$ and $M(\lambda_1)=2T\mathcal{I}^{(2)}(\lambda_1)=3\pi$, with $2T=\pi/\kappa^{(2)}$. 
In this Lemma, we consider the special case that $M(\lambda)=2\pi(-P+1)/P+(2n+1)\pi$.
The detailed proof process is provided in \cite{LingS-23-mKdV-stability}; hence, we omit the details here.
Furthermore, it is also easy to obtain the following Lemma.

\begin{lemma}\label{lemma:H<u}
	$\hat{\mathcal{H}}_2$ is continuous in $H^2_{per}([-PT,PT])$ on the bounded sets; in other words, for any $\epsilon>0$, there exist constants $M_1,\delta>0$, if $\| u-v\|_{H^2}\le\delta$ and $\| u\|_{H^2}\le M_1$, we have $|\hat{\mathcal{H}}_2(u)-\hat{\mathcal{H}}_2(v)|<\epsilon$.
\end{lemma}

From these results, we aim to conduct the orbital stability analysis of genus-two traveling wave solutions .

\newenvironment{proof-orbital-cn}{\emph{Proof of \Cref{theorem:orbital-cn}.}}{\hfill$\Box$\medskip}
\begin{proof-orbital-cn}
	The proof is similar to that in \cite{DeconinckyU-20,LingS-23-mKdV-stability}. 
	Here, we provide only a brief overview of the relevant results and highlight a different case. 
	First, the following conclusion ensures the global well-posedness of the periodic solutions we need.
	Kappeler and Topalov \cite{KappelerT-2005} proved that the \ref{eq:mKdV} equation is globally well-posed in $L^2(T)$.
	Colliander et al. \cite{CollianderKSTT-03} studied that the Cauchy problem for the mKdV equation with the periodic boundary condition is globally well-posed for the initial data $u(\xi,0)\in H^{s}(T),s>1/2$. 
		
	Second, we consider the disturbance
	\begin{equation}\label{eq:h-dis}
		h(\xi,\eta):=\hat{u}(\xi,\eta)-\mathcal{T}(\gamma(\eta))u, \qquad h(\xi,\eta)\in H^2([-PT,PT]),
	\end{equation}
	where $\mathcal{T}$ is in \Cref{define:orbital stable}. The perturbation $h(\xi,\eta)$ belongs to the nonlinear set $\mathcal{A}:=\{h\in H^2([-PT,PT])|\mathcal{H}_0(h(\xi,\eta)+u)=\mathcal{H}_0(u),\left\langle h(\xi,\eta),\partial_{\xi}u\right \rangle_{L^2}=0\}.$
	Define the linear admissible space 
	$\mathcal{A}_1:=\{h_1\in H^2([-PT,PT])| \left\langle h_1(\xi,\eta),\partial_{\xi}u\right \rangle_{L^2}=\left\langle u,h_1 (\xi,\eta)\right\rangle_{L^2}=0\}$. 
	For any $h(\xi,\eta)\in \mathcal{A}$ with $\|h\|_{H^2}$ sufficiently small, we can decompose $h(\xi,\eta)=h_1(\xi,\eta)+\hat{c}u(\xi),$ where $\hat{c}=\hat{c}(h)=-\|h\|_2^2/2\|u\|_2^2$ and $h_1\in \mathcal{A}_1$, $h\in \mathcal{A}$, (See \cite{LingS-23-mKdV-stability}).
		
	Then, we expand the function $\mathcal{\hat{H}}_2(u+h)-\mathcal{\hat{H}}_2(u)$ in powers of $h$.
	In combination with \Cref{lemma:alpha-0}, $\left\langle \mathcal{L}_2h_1,h_1 \right\rangle_{L^2} \ge \alpha_0 \|h_1\|_{H^2}^2$, we obtain that if $2 E^{(2)}_2\ge   K^{(2)}_2$ and $P<\frac{4\pi}{\pi+M(\lambda_0)}$. 
	Utilizing H\"{o}lder inequality, we get
	\begin{equation}\nonumber
		\begin{split}
			\left\langle \mathcal{L}_2(u)h,h \right\rangle_{L^2}
			\ge &\ \left\langle \mathcal{L}_2(u)h_1,h_1\right\rangle_{L^2}+2\hat{c}\left\langle \mathcal{L}_2(u)u,h_1\right\rangle_{L^2}+\hat{c}^2\left\langle \mathcal{L}_2(u)u,u\right\rangle_{L^2} \\	
			\ge &\ \alpha_0 \|h_1\|_{H^2}^2-|\hat{c}_3|\|h\|^2_{H^2}\|h_1\|_{H^2}/|\|u\|^2_{2}-|\hat{c}_3|\|h\|^4_{H^2}\|u\|_{\infty}/(4\|u\|^4_{2}),
		\end{split}
	\end{equation}
	where $\hat{c}_3:=\mathcal{L}_2(u)u=-4\hat{c}_1(A-\lambda_2^2-2A E^{(2)}_2/  K^{(2)}_2)$, and $\|h_1\|_{\infty}\le \|h_1\|_{H^1}\le\|h_1\|_{H^2}$.
	For the genus-2 cases, the parameter $\hat{c}_1$ is not zero (i.e., $\hat{c}_1\neq 0$). For the genus-1 traveling wave solutions , the parameter $\hat{c}_1$ is zero. This constitutes the main difference from the proof in \cite{LingS-23-mKdV-stability}. 
	Using Minkowski inequality, we know $\|h_1\|_{H^2}^2\ge \|h\|_{H^2}^2-\hat{c}^2\|u\|_{H^2}^2\ge \|h\|_{H^2}^2-c\|h\|_{H^2}^4$, $c=\|u\|_{H^2}^2/(4\|u\|_{2}^4)$. 
	For $\|h\|_{H^2}^2<1/(2c)$ sufficiently small, it follows that $\|h_1\|_{H^2}^2\ge \frac{1}{2}\|h\|_{H^2}^2$.
	Furthermore, $\|h_1\|_{H^2}^2\le \|h\|_{H^2}^2+\hat{c}^2\|u\|_{H^2}^2\le \|h\|_{H^2}^2+\|h\|_{H^2}^4\|u\|_{H^2}^2/\|u\|_2^4$, which implies $\|h_1\|_{H^2}\le  \|h\|_{H^2}+\|h\|_{H^2}^2\|u\|_{H^2}/\|u\|_2^2$.
	Therefore,
	\begin{equation}\nonumber
		\begin{split}
			\left\langle \mathcal{L}_2(u)h,h \right\rangle_{L^2}
			\ge 	
			&\ \alpha_0 \|h\|_{H^2}^2/2-|\hat{c}_3|\|h\|^3_{H^2}/|\|u\|^2_{2}-|\hat{c}_3|(\|u\|_{H_2}/|\|u\|^2_{2}+\|u\|_{\infty}/\|u\|^4_{2})\|h\|^4_{H^2}.
		\end{split}
	\end{equation}
	Since $\mathcal{\hat{H}}_2(u+h)-\mathcal{\hat{H}}_2(u)= \left\langle \mathcal{L}_2(u)h,h \right\rangle_{L^2}/2 +\mathcal{O}(\|h\|_{H^2}^3)$, 
	we obtain
	\begin{equation}\label{eq:H_2-3}
		\begin{split}
			|\mathcal{\hat{H}}_2(u+h)-\mathcal{\hat{H}}_2(u)|\ge \frac{\alpha_0}{4}\|h\|_{H^2}^2-\beta\|h\|_{H^2}^3, \qquad \beta>0.
		\end{split}
	\end{equation}
		
	At last, by a proof similar to that of \cite{LingS-23-mKdV-stability} (Theorem 5),
	we obtain that for any $\epsilon>0$, by choosing $\hat{\delta}(\epsilon)=-\beta\epsilon^3+\frac{\alpha_0^2}{4}\epsilon^2$, the inequality $\|h(\xi,0)\|_{H_2}<\epsilon$ holds, which further implies $\|h(\xi,\eta)\|\le\nu_2(\Delta)<\epsilon$.
	From \Cref{lemma:H<u}, we know that for the above fixed  $\hat{\delta}(\epsilon)>0$, there exists $\delta(\hat{\delta})$ with $\min\{\epsilon,\frac{1}{2c}\}>\delta(\hat{\delta})>0$, when $\|(u(\xi,0)+h(\xi,0))-u(\xi,0)\|_{H^2}\le \delta(\hat{\delta})$, the inequality
	$|\mathcal{\hat{H}}_2(u+h)-\mathcal{\hat{H}}_2(u)| 
	=|\mathcal{\hat{H}}_2(u(\xi,0)+h(\xi,0))-\mathcal{\hat{H}}_2(u(\xi,0))| \le  \hat{\delta}(\epsilon)$ holds.
	In summary, we obtain that for any $\epsilon>0$, there exists $\delta(\epsilon)>0$, such that if $\|v(\xi,0)-T(\gamma)u(\xi,0)\|_{H^2}\le \delta(\epsilon)$ and $t\in\mathbb{R}$, then $\inf_{\gamma\in \mathbb{R}} \|v(\xi,\eta)-T(\gamma)u(\xi,\eta)\|_{H^2}< \epsilon$, which implies
	\begin{equation}\nonumber
		\sup_{t\in\mathbb{R}} \inf_{\gamma\in \mathbb{R}} \|\hat{u}(\xi,\eta)-T(\gamma)u(\xi,\eta)\|_{H^2} < \epsilon. 
	\end{equation}
	By \Cref{define:orbital stable}, we conclude that the solution $u(\xi)$ in equation \eqref{eq:u-parameters-cn} is orbitally stable in the space $H^2([-PT,PT])$, $P<\frac{4\pi}{\pi+M(\lambda_0)}$.
\end{proof-orbital-cn}

	When $2 E^{(2)}_2<   K^{(2)}_2$, we know that for any $\lambda\in ( Q^{(2)} _I\cup  Q^{(2)} _R)$, the inequality $\mathcal{I}'(\lambda)\neq0$ holds. And $\mathcal{K}_1(\lambda)\ge 0$, only when $\Omega_1=0$, $\mathcal{K}_1(\lambda)=0$. Based on \Cref{lemma:alpha-0}, we use a similar proof as the condition $2 E^{(2)}_2<   K^{(2)}_2$ and obtain that the solution $u(\xi)$ in equation \eqref{eq:u-parameters-cn} is orbitally stable in the space $H^1([-T,T])$ when $2E^{(2)}_2<K^{(2)}_2$.

\subsection{The orbital stability for the \ref{case1}}\label{subsec:orbital-dn}

Similar to the proof in \Cref{subsec:orbital-cn}, we first consider the Krein signature $\mathcal{K}_{1,2}(\lambda)$ under the \ref{case1}. 
Based on the exact expressions of the function $u(\xi)$ defined in equation \eqref{eq:u1-elliptic}, we obtain 
	\begin{equation}\nonumber
		\begin{split}
			 \int_{-PT}^{PT}u^2(\xi)\dd \xi 
			\xlongequal[\eqref{eq:define-u-lambda}]{\eqref{eq:u1-elliptic-1}}	&\
			\frac{P}{\alpha}\int_{-K^{(1)}_2}^{K^{(1)}_2}\left(\ii(\lambda_1+\lambda_2-\lambda_3) + \frac{2\ii (\lambda_3-\lambda_2)}{1-\frac{\lambda_2-\lambda_1}{\lambda_3-\lambda_1}\sn^2(\xi,  k^{(1)}_2)}\right)^2\dd \xi\\
			\xlongequal[\eqref{eq:integral}]{\eqref{eq:define-first-integral}, \eqref{eq:define-third-integral}}&\ 
			-\frac{2P}{\alpha}\left((\lambda_1^2+\lambda_2^2-\lambda_3^2)  K^{(1)}_2+2(\lambda_3^2-\lambda_1^2) E^{(1)}_2\right).
		\end{split}
	\end{equation}
	For the solution in \eqref{eq:u1-elliptic-2}, the above integral result also holds. 
	Applying the method of the  \Cref{lemma:Phi+-}, we get
	$	\mathcal{K}_1(\lambda)	=-8|\lambda^2| \Omega_1^2P  K^{(1)}_2 ((\lambda^2+\lambda_1^2+\lambda_2^2)+(\lambda_3^2-\lambda_1^2) E^{(1)}_2/  K^{(1)}_2)/\alpha$,
	when $\lambda\in \mathbb{R}$. 
	On the other hand, when $\lambda \in \ii \mathbb{R}$, the above result holds.
	For all $\lambda\in  Q^{(1)} $ satisfying $\Omega_1(\lambda)\in \ii \mathbb{R}$, the inequality $\mathcal{K}_1(\lambda)\ge 0$ does not hold uniformly.
	Choosing $c_{5,3}=2(\lambda_1^2+\lambda_2^2-\lambda_3^2+2(\lambda_3^2-\lambda_1^2)E_2^{(1)}/K_2^{(1)})$, we obtain that for any $\lambda\in Q^{(1)}$, $\mathcal{K}_2(\lambda)=(v+4\lambda^2+c_{5,3})\mathcal{K}_1(\lambda)\ge 0$.

\newenvironment{proof-orbital-dn}{\emph{Proof of \Cref{theorem:orbital-dn}.}}{\hfill$\Box$\medskip}
\begin{proof-orbital-dn}
	Similar to Theorem \ref{theorem:orbital-cn}, we find that for all $\lambda\in  Q^{(1)} $ satisfying $\Omega_1(\lambda)\in \ii \mathbb{R}$, the statement $\mathcal{K}_1(\lambda)\ge 0$ does not always hold.  We therefore consider the Krein signature $\mathcal{K}_2(\lambda)$.
	Following the same arguments as in the proof of \Cref{lemma:alpha-0} and  \Cref{theorem:orbital-cn}, we conclude that the solution in \ref{case1} is orbitally stable in the space $H^2([-PT,PT])$, $P\in \mathbb{Z}_+$.
\end{proof-orbital-dn}

\section*{Acknowledgments}
The authors would like to express their sincere gratitude to Professor Dmitry Pelinovsky for his insightful discussions on this work.
Liming Ling is supported by the National Natural Science Foundation of China (No. 12471236), the Guangzhou Municipal Science and Technology Project (Guangzhou Science and Technology Plan, No. 2024A04J6245) and Guangdong Natural Science Foundation grant (No. 2025A1515011868).
Xuan Sun is sponsored by the National Natural Science Foundation of China (No. 12501328) and the Natural Science Foundation of Shanghai Basic Research Funding (No. 25ZR1402009).

\subsection*{Data availability statement}
Data sharing not applicable to this article as no datasets were generated or analyzed during the current study.

\appendix

\titleformat{\section}[display]
{\centering\LARGE\bfseries}{ }{11pt}{\LARGE}

\titleformat{\subsection}[display]
{\large\bfseries}{ }{10pt}{\large}

\renewcommand{\appendixname}{Appendix \, \Alph{section}}

\section{\appendixname. The definitions of Elliptic functions}\label{appendix:Elliptic-integrals-formulas}

\setcounter{equation}{0}
\setcounter{define}{0}
\setcounter{prop}{0}
\setcounter{lemma}{0}

\renewcommand\theequation{\Alph{section}.\arabic{equation}}
\renewcommand\thedefine{\Alph{section}.\arabic{define}}
\renewcommand\theprop{\Alph{section}.\arabic{prop}}
\renewcommand\thelemma{\Alph{section}.\arabic{lemma}}

\begin{define}\label{define:elliptic-function}
	The three standard canonical forms of elliptic integrals are given by:
	\begin{itemize}
		\item The normal elliptic integral of the first kind is defined by
		\begin{equation}\label{eq:define-first-integral}
			F(z,k) = \int_{0}^{z}\frac{\dd t}{\sqrt{(1-t^2)(1-k^2t^2)}}=\int_0^{u} \dd \nu, \qquad u = F(z,k).
		\end{equation}
		The associated complete elliptic integrals $K\equiv F(1,k)$ and $K'=K(k'), k'=\sqrt{1-k^2}$ are also commonly used. 
		\item The normal elliptic integral of the second kind is defined by:
		\begin{equation}\label{eq:define-second-integral}
			E(z,k) = \int_{0}^{z}\sqrt{\frac{1-k^2 t^2}{1-t^2}}\dd t=\int_0^{u}\dn^2(\nu,k)\dd \nu, \qquad 
			E\equiv E(1,k)=E(k).
		\end{equation}
		\item The normal elliptic integral of the third kind is defined by:
		\begin{equation}\label{eq:define-third-integral}
		 \!	\Pi(z,\alpha^2,k) \! =\! \int_{0}^{z}\frac{\dd t}{(1-\alpha^2 t^2)\sqrt{(1-t^2)(1-k^2t^2)}} \! =\! \int_0^{u}\frac{\dd \nu}{1-\alpha^2\sn^2(\nu,k)},\,\, \Pi(\alpha^2,k)\equiv \Pi(1,\alpha^2,k).	
		\end{equation}
	\end{itemize}
	The inverse function of the elliptic integral of the first kind \eqref{eq:define-first-integral} is denoted by $z=\sn(u,k)$, $u = F(z,k)$. 
	Based on this, two additional functions are introduced: $\cn(u,k)=\sqrt{1-z^2}$, $\dn(u,k)=\sqrt{1-k^2z^2}$, with the initial conditions $cn(0,k)=1$ and $\dn(0,k)=1$.
\end{define}

The formulas between the above normal elliptic integrals  \cite{ByrdF-54}:
\begin{itemize}
	\item Special addition formulas for the first kind of elliptic functions:
	\begin{equation}\label{eq:add-first}
		F(\theta,k)+F(\beta,k)=K, \qquad \cot(\beta)=k^{\prime} \tan(\theta);
	\end{equation}
	\item Differential equations with respect to $u$:
	\begin{equation}\label{eq:diff-elliptic}
		 \dn^{\prime}(u,k)=-k^2\sn(u,k)\cn(u,k),\quad
		 \sn^{\prime}(u,k)=\cn(u,k)\dn(u,k);
	\end{equation}
	\item Inequality \cite{LingS-23-mKdV-stability}:
	\begin{equation}\label{eq:inequality}
		\begin{split}
			&\ E(k) - (k^{\prime})^2 K(k)> \lim_{k\rightarrow 0} (E(k) - (k^{\prime})^2 K(k))=0;\\
			&\ K(k) - E(k)> \lim_{k\rightarrow 0} (K(k) - E(k))=0;\\
			&\  (1+(k^{\prime})^2 )K(k) -2 E(k)> \lim_{k\rightarrow 0} [(1+(k^{\prime})^2 )K(k) -2 E(k)]=0;\\
		\end{split}
	\end{equation}
	\item The transformation between the first and the third kind of elliptic functions: 
	\begin{equation}\label{eq:Pi-K-alpha}
		\Pi\left(\alpha^2,k\right)+\Pi\left(\frac{k^2}{\alpha^2},k\right)=K+\frac{\pi}{2}\sqrt{\frac{\alpha^2}{\left(1-\alpha^2\right)\left(\alpha^2-k^2\right)}},  \qquad K=F(1,k),
	\end{equation}
	with $0<k^2<\alpha^2<1$ or $0<-\alpha^2<\infty$.
\end{itemize}

\begin{define}[Jacobi theta functions {\cite[p.302]{Farkas-92-Riemann}}]\label{define:theta}
	Jacobi theta functions are defined as:
	\begin{equation}\nonumber
		\begin{split}
			&\vartheta_1(u, \tau):=\Theta\begin{bmatrix}
				1 \\ 1
			\end{bmatrix}(u,\tau), \qquad
			\vartheta_2(u, \tau):=\Theta\begin{bmatrix}
				1 \\ 0
			\end{bmatrix}(u,\tau), \\
			&\vartheta_3(u, \tau):=\Theta\begin{bmatrix}
				0 \\ 0
			\end{bmatrix}(u,\tau), \qquad
			\vartheta_4(u, \tau):=\Theta\begin{bmatrix}
				0 \\ 1
			\end{bmatrix}(u,\tau),
		\end{split}
	\end{equation}
	where $\tau=\ii K'/K$, parameters $K$ and $K'$ are first kind complete elliptic integrals, and
	\begin{equation}\nonumber
		\Theta\begin{bmatrix}
			\epsilon \\
			\epsilon'
		\end{bmatrix}(u,\tau)
		=\sum_{n=-\infty}^{+\infty}\exp\left\{ \left[\frac{1}{2}\left(n+\frac{\epsilon}{2}\right)^2\tau+\left(n+\frac{\epsilon}{2}\right)\left(u+\frac{\epsilon'}{2}\right)\right]\right\}.
	\end{equation}
\end{define}

\begin{define}\label{define:zeta}
	The Jacobi Zeta function is defined by 
	\begin{equation}\label{eq:Zeta-define}
		Z(u,k)\equiv \int_0^u \left( \dn^2(z,k)-\frac{E}{K} \right) \dd z, 
		\qquad \text{or} \qquad 
		Z(u,k)=\frac{\partial }{\partial u}\ln\left(\vartheta_4\left(\frac{\ii u\pi }{ K}\right)\right)\! ,
	\end{equation}
	where $E\equiv E(k)$, $K\equiv K(k)$ are the complete elliptic integrals defined in equations \eqref{eq:define-second-integral} and \eqref{eq:define-first-integral}, respectively.
\end{define}

Indeed, the above elliptic functions can be transformed into one another, enabling us to leverage the properties of their various forms to support our analytical work.
Here, we provide several useful formulas required in this work.
In combination with the definitions of Jacobi theta functions, the transformation between Jacobi theta functions and elliptic functions is defined as follows:
\begin{equation}\label{eq:formula-trans-theta-elliptic}
	\begin{split}
		&\sn(u,k)=\frac{\vartheta_3(0,\tau)\vartheta_1(\pi\ii u/K,\tau)}{\vartheta_2(0,\tau)\vartheta_4(\pi\ii u/K,\tau)},\qquad
		\cn(u,k)=\frac{\vartheta_4(0,\tau)\vartheta_2(\pi\ii u/K,\tau)}{\vartheta_2(0,\tau)\vartheta_4(\pi\ii u/K,\tau)}, \\
		&\dn(u,k)=\frac{\vartheta_4(0,\tau)\vartheta_3(\pi\ii u/K,\tau)}{\vartheta_3(0,\tau)\vartheta_4(\pi\ii u/K,\tau)}, \qquad 
		k=\frac{\vartheta_2^2(0,\tau)}{\vartheta_3^2(0,\tau)}, \qquad
		k^{\prime}=\frac{\vartheta_4^2(0,\tau)}{\vartheta_3^2(0,\tau)}.
	\end{split}
\end{equation}
The elliptic integral of the third kind can be expressed in terms of the Jacobi Zeta function, elliptic functions, and Jacobi theta functions as follows:
\begin{equation}\label{eq:formula-trans-Pi-Zeta-theta}
	\Pi(u,\alpha^2,k)=\frac{\sn(a,k)}{\cn(a,k)\dn(a,k)}\left(\frac{1}{2}\ln \frac{\vartheta_1(\ii(a+u)\pi/K,\tau)}{\vartheta_1(\ii(a-u)\pi/K,\tau)}-uZ(a,k)\right),\quad \alpha=\frac{1}{\sn(a,k)}.
\end{equation}

Some useful formulas about Jacobi elliptic function \cite{ByrdF-54}:
\begin{itemize}
	\item Shift formulas:
	\begin{equation}\label{eq:Jacobi-shift}
		\begin{split}
			&\sn(u+K)=\cd(u),\qquad \sn(u+\ii K')=\ns(u)/k, \qquad \quad  \sn(u+K+\ii K')=\dc(u)/k, \\
			&\cn(u+K)=-k'\sd(u), \quad \cn(u+\ii K')=-\ii \ds(u)/k,  \, \quad 
			\cn(u+K+\ii K')=-\ii k'\nc(u)/k, \\ 
			&\dn(u+K)=k'\nd(u), \qquad \dn(u+\ii K')=-\ii \cs(u),\quad \quad 
			\dn(u+K+\ii K')=\ii k' \tn(u),
		\end{split}
	\end{equation}
	where $\sn(\cdot)=\sn(\cdot,k)$, $\cn(\cdot)=\cn(\cdot,k)$, and $\dn(\cdot)=\dn(\cdot,k)$;
	\item Double arguments:
	\begin{equation}\label{eq:Jacobi-double}
		\cn(2u)=\frac{\cn^2(u)-\sn^2(u)\dn^2(u)}{1-k^2\sn^4(u)}, \quad 
		\dn(2u)=\frac{\dn^2(u)-k^2\sn^2(u)\cn^2(u)}{1-k^2\sn^4(u)};
	\end{equation}
	\item Addition formulas:
	\begin{equation}\label{eq:add-app}
		\begin{split}
			\sn( u \pm  v)= &\ \frac{\sn(u)\cn(v)\dn(v)\pm \sn(v)\cn(u)\dn(u)}{1-k^2\sn^2(u)\sn^2(v)},\\
			\sn(u+v)\sn(u-v)=&\ \frac{\sn^2(u)-\sn^2(v)}{1-k^2\sn^2(v)\sn^2(u)},\\
			\cn(u+v)\sn(u-v)=&\ \frac{\sn(u)\cn(u)\dn(v)-\sn(v)\cn(v)\dn(u)}{1-k^2\sn^2(v)\sn^2(u)},\\
			Z(u\pm v)=&\ Z(u)\pm Z(v) \mp k^2 \sn(u)\sn(v)\sn(u\pm v),\\
			Z(u+ \ii K^{\prime})=&\ Z(u)+\frac{\cn(\nu,k)\dn(\nu,k)}{\sn(\nu,k)}-\frac{\ii\pi}{2K},\\
			Z(u+ 2\ii K^{\prime})=&\ Z(u)-\frac{\ii\pi}{K},\\
			Z(u+2K)=&\ Z(u).
		\end{split}
	\end{equation}
	
\end{itemize}

\begin{prop}\label{prop:Riemann-2-1}
	The dimension-$2$ Riemann theta function with related to the mKdV genus-two algebraic curves could be expressed by the elliptic functions 
	\begin{equation}\label{eq:formula-Rieman-shift-1}
		\begin{split}
			&\Theta(\mathbf{z}, \mathbf{B}^{(2)})=\vartheta_3(z_1-2z_2, \tau^{(2)}_1)\vartheta_3(z_1, \tau^{(2)}_2)+\vartheta_1(z_1-2z_2, \tau^{(2)}_1)\vartheta_1(z_1, \tau^{(2)}_2), \quad
			\mathbf{z}=[z_1,z_2]^{\top}, \\
			& \Theta(\mathbf{z}, \mathbf{B}^{(1)} )=\vartheta_3(z_1-z_2,2  \tau^{(1)}_1)\vartheta_3(z_1+z_2,2  \tau^{(1)}_2)+\vartheta_2(z_1-z_2,2  \tau^{(1)}_1)\vartheta_2(z_1+z_2,2  \tau^{(1)}_2),
		\end{split}
	\end{equation}
	where the matrix $\mathbf{B}^{(1)}$ and $ \mathbf{B}^{(2)}$ are defined in equations \eqref{eq:u-parameters-dn-B} and \eqref{eq:u-parameters-cn-B}, respectively.
\end{prop}
\begin{proof}
	Introduce $n_2=2(m_2-t) \in \mathbb{N}$ and $n_1=m_1 \in \mathbb{N}$. 
	When we consider the variable $n_1$ from $-\infty$ to $\infty$,
	it is easy to obtain the parameter $m_1$ should choose from $-\infty$ to $\infty$ and the parameter $t$ have two conditions $t=0$ and $t=1/2$.
	It is easy to obtain
	\begin{equation}\nonumber
		\begin{split}
			\Theta(\mathbf{z}, \mathbf{B}^{(2)})
			\xlongequal[\eqref{eq:u-parameters-cn-B}]{\eqref{eq:define-Riemann-Theta}}&\  \sum_{n_1,n_2=-\infty}^{+\infty}\exp\left\{ \frac{n_1^2\mathbf{B}^{(2)}_{11}+2n_1n_2 \mathbf{B}^{(2)}_{12}+n_2^2\mathbf{B}^{(2)}_{22}}{2} +(n_1z_1+n_2z_2)\right\}\\
			\xlongequal[n_1=m_1]{n_2=2(m_2-t)}
			&\  \!\!\! \sum_{t=0,\frac{1}{2}}\sum_{m_1,m_2=-\infty}^{\infty}\!\!\exp\left\{  \frac{(m_1-m_2+t)^2\mathbf{B}^{(2)}_{11}+2(m_2-t)(m_1-m_2+t)( 2\mathbf{B}^{(2)}_{12}+ \mathbf{B}^{(2)}_{11})}{2} \right\}\\
			&\cdot \exp\left\{\frac{(m_2-t)^2(\mathbf{B}^{(2)}_{11}+4 \mathbf{B}^{(2)}_{12}+4\mathbf{B}^{(2)}_{22})}{2}+(m_1-m_2+t)z_1+(m_2-t)(2z_2+z_1) \right\}\\
			\xlongequal[m_2=-n_2]{m_1-m_2=n_1}&\ \sum_{t=0,\frac{1}{2}}\sum_{n_1,n_2=-\infty}^{\infty}\exp\left\{ \frac{(n_1+t)^2\mathbf{B}^{(2)}_{11}-2(n_2+t)(n_1+t)( 2\mathbf{B}^{(2)}_{12}+ \mathbf{B}^{(2)}_{11})}{2} \right\}\\
			&\exp\left\{ \frac{(-n_2-t)^2(\mathbf{B}^{(2)}_{11}+4 \mathbf{B}^{(2)}_{12}+4\mathbf{B}^{(2)}_{22})}{2}+(n_1+t)z_1-(n_2+t)(2z_2+z_1)\right\}.
		\end{split}
	\end{equation}
	Plugging the definition of the parameter $ \mathbf{B}^{(2)}$ into the above equations, we obtain 
	\begin{equation}\nonumber
		\begin{split}
			\Theta(\mathbf{z}, \mathbf{B}^{(2)})
			\xlongequal{\eqref{eq:u-parameters-cn-B}}& \!\!\! \sum_{t=0,\frac{1}{2}} \sum_{n_1,n_2=-\infty}^{\infty} \!\!\! \exp\left\{ \frac{(n_1+t)^2}{2} \tau^{(2)}_2+\frac{(n_2+t)^2}{2} \tau^{(2)}_1+(n_1+t)z_1+(n_2+t)(2z_2+z_1)\right\}\\
			=&\ \vartheta_3(2z_2+z_1, \tau^{(2)}_1)\vartheta_3(z_1, \tau^{(2)}_2)+\vartheta_1(2z_2+z_1, \tau^{(2)}_1)\vartheta_1(z_1, \tau^{(2)}_2).
		\end{split}
	\end{equation}
	Hence, the first equation of this proposition is established, and the second can be derived in a similar manner.
\end{proof}

The transformation formulas about the Jacobi theta functions \cite{Kharchev-2015-theta}:
\begin{itemize}
	\item Three-term bilinear identities:
	\begin{equation}\label{eq:formula-theta-2tau-1}
		\begin{split}
			&2\vartheta_1(u+v,2\tau)\vartheta_1(u-v,2\tau)= \vartheta_4(u,\tau)\vartheta_3(v,\tau)-\vartheta_3(u,\tau)\vartheta_4(v,\tau),\\
			&2\vartheta_2(u+v,2\tau)\vartheta_2(u-v,2\tau)= \vartheta_3(u,\tau)\vartheta_3(v,\tau)-\vartheta_4(u,\tau)\vartheta_4(v,\tau),\\
			&2\vartheta_3(u+v,2\tau)\vartheta_3(u-v,2\tau)= \vartheta_3(u,\tau)\vartheta_3(v,\tau)+\vartheta_4(u,\tau)\vartheta_4(v,\tau),\\
			&2\vartheta_2(u+v,2\tau)\vartheta_3(u-v,2\tau)= \vartheta_2(u,\tau)\vartheta_2(v,\tau)-\vartheta_1(u,\tau)\vartheta_1(v,\tau);
		\end{split}
	\end{equation}
	
	
	\item Mixed identity:
	\begin{equation}\label{eq:formula-theta-2-1}
		\begin{split}
			&\ \vartheta_1(u+v,\tau)\vartheta_2(u-v,\tau)\vartheta_3(0,\tau)\vartheta_4(0,\tau)\\
			= &\ \vartheta_1(u,\tau)\vartheta_2(u,\tau)\vartheta_3(v,\tau)\vartheta_4(v,\tau)+\vartheta_1(v,\tau)\vartheta_2(v,\tau)\vartheta_3(u,\tau)\vartheta_4(u,\tau);
		\end{split}
	\end{equation}
%
	\item Shift formulas among four theta functions in \cite{ArmitageE-06-elliptic-function}:
	\begin{equation}\label{eq:Jacobi-Theta-K-iK}
		\begin{split}
			\vartheta_1(z,\tau)&=-\vartheta_2\left(z+\ii\pi,\tau\right)=-\ii M \vartheta_3\left(z+\ii \pi+ \ii\pi\tau,\tau \right)=-\ii M\vartheta_4\left(z+\ii\pi\tau,\tau \right),\\
			\vartheta_2(z,\tau)&=\vartheta_1\left(z+\ii\pi ,\tau\right)=M \vartheta_4\left(z+\ii\pi +\ii\pi \tau,\tau \right)= M\vartheta_3\left(z+\ii\pi \tau ,\tau\right),\\
			\vartheta_3(z,\tau)&=\vartheta_4\left(z+\ii\pi ,\tau\right)= M \vartheta_1\left(z+\ii\pi +\ii\pi\tau,\tau \right)= M\vartheta_2\left(z+\ii\pi \tau,\tau\right),\\
			\vartheta_4(z,\tau)&=\vartheta_3\left(z+\ii\pi ,\tau \right)=\ii M \vartheta_2\left(z+\ii\pi+\ii\pi \tau,\tau \right)=-\ii M\vartheta_1\left(z+\ii\pi \tau,\tau \right),
		\end{split}
	\end{equation}
	where $M=\ee^{z/2+\ii \tau \pi /4}$.
\end{itemize}

\section{\appendixname. The elliptic integrals}\label{appendix:map}

\setcounter{equation}{0}
\setcounter{define}{0}
\setcounter{prop}{0}
\setcounter{lemma}{0}

\renewcommand\theequation{\Alph{section}.\arabic{equation}}
\renewcommand\thedefine{\Alph{section}.\arabic{define}}
\renewcommand\theprop{\Alph{section}.\arabic{prop}}
\renewcommand\thelemma{\Alph{section}.\arabic{lemma}}

We aim to transform the elliptic integrals (the right-hand sides of equations \eqref{eq:hyper-1} and \eqref{eq:hyper-2}) into the standard form of the elliptic integrals.
In this process, we consider two types of elliptic integrals corresponding to  \ref{case1} and \ref{case2}.

\vspace{0.1cm}
\noindent $\bullet$\quad
\textbf{\large \ref{case1}}

For ease of representation, we set $\beta_i=\lambda_{4-i}^2\in \mathbb{R}$, $i=1,2,3$, and $\beta_4=0$.

\begin{prop}\label{prop:elliptic-int-1}
	When $\beta_{1,2,3,4}\in \mathbb{R}$, the elliptic integrals defined in equation \eqref{eq:hyper} with $n=0$ can be transformed into the elliptic integrals of the first kind as defined in \Cref{define:elliptic-function} as follows:
	\begin{subequations}\label{eq:elliptic-int-1}
		\begin{align}
			&\int \frac{((\beta_4-\beta_2)(\beta_3-\beta_1))^{1/2} \,}{\prod_{i=1}^{4}(\Lambda-\beta_i)^{1/2}}\, \dd \Lambda
			= \int\frac{2 \, \dd z }{\sqrt{(1-z^2)(1-k^2z^2)}},  \qquad 	k^2=\frac{(\beta_3-\beta_2)(\beta_4-\beta_1)}{(\beta_3-\beta_1)(\beta_4-\beta_2)}, \label{eq:elliptic-int-1-4}\\
			& \int\frac{(\beta_1-\beta_3)^{1/2} \, \dd \Lambda}{\prod_{i=1}^{3}(\Lambda-\beta_i)^{1/2}} = \int\frac{2 \, \dd z}{\sqrt{(1-z^2)(1-k^2z^2)}}, \qquad
			k^2=\frac{\beta_1-\beta_3}{\beta_2-\beta_3}. \label{eq:elliptic-int-1-3}
		\end{align}	
	\end{subequations}
\end{prop}

\begin{proof}
	To express these integrals in the standard form of the elliptic integral of the first kind, we introduce a linear fractional transformation between $\Lambda$ and $z$, such that the four real branch points $\beta_{1,2,3,4}\in \mathbb{R}$ on the $\Lambda$-plane correspond to the points $0, 1, 1/k, \infty$ on the $z$-plane.
	Without loss of generality, we assign the correspondence $\beta_4\leftrightarrow 0$, $\beta_3\leftrightarrow \infty$, $\beta_1\leftrightarrow 1$, and $\beta_2\leftrightarrow 1/k$. 
	The corresponding linear fractional transformation is then given by
	\begin{equation}\label{eq:trans-z-Lambda-1}
		z^2=\frac{(\beta_3 - \beta_1)(\Lambda - \beta_4)}{(\beta_4 - \beta_1)(\Lambda - \beta_3)}.
	\end{equation}
	Then, we obtain 
	\begin{equation}\label{eq:trans-L-z-1-1}
		\Lambda=\beta_3+\frac{(\beta_3-\beta_1)(\beta_3-\beta_4)}{(\beta_4-\beta_1)z^2-(\beta_3-\beta_1)} \quad \text{and} \quad 
		\dd \Lambda=\frac{2(\beta_3-\beta_1)(\beta_4-\beta_3)(\beta_4-\beta_1)z}{((\beta_4-\beta_1)z^2-(\beta_3-\beta_1))^2}\, \dd z.
	\end{equation} 
	Plugging the above formula \eqref{eq:trans-L-z-1-1} into the left hand-side of equation \eqref{eq:elliptic-int-1-4}, we obtain $(\beta_4-\beta_2)^{1/2}(\beta_3-\beta_1)^{1/2}\prod_{i=1}^{4}(\Lambda-\beta_i)^{-1/2} \dd \Lambda=2  ((1-z^2)(1-k^2z^2))^{-1/2} \dd z$,
	where $k^2$ is defined in equation \eqref{eq:elliptic-int-1-4}. 
	Thus, we obtain the equation \eqref{eq:elliptic-int-1-4}.
	
	Under a suitable transformation, the second integral can also be reduced to the standard form of the first kind of elliptic integral.
	Introducing another linear fractional transformation, we obtain the following relations:
	\begin{equation}\label{eq:trans-L-z-1-2}
		z^{2}=\frac{\beta_1-\beta_3}{\Lambda-\beta_3}, \qquad \Lambda = \beta_3 +\frac{\beta_1-\beta_3}{z^2}, \qquad 
		\dd \Lambda = \frac{2(\beta_3-\beta_1)}{z^3} \, \dd z.
	\end{equation}
	Here, the correspondence between $\beta_{1,2,3}$ on the $\Lambda$-plane and points $1, 1/k, \infty$ on the $z$-plane is given by $\beta_3\leftrightarrow \infty$, $\beta_1\leftrightarrow 1$, and $\beta_2\leftrightarrow 1/k$.
	Substituting the relations in \eqref{eq:trans-L-z-1-2} into equation \eqref{eq:elliptic-int-1-3}, we obtain $(\beta_1-\beta_3)^{1/2} \prod_{i=1}^{3}(\Lambda-\beta_i)^{-1/2} \dd \Lambda=2((1-z^2)(1-k^2z^2))^{-1/2} \dd z$,
	where $k^2$ is defined in equation \eqref{eq:elliptic-int-1-3}.
	Therefore, it follows that \eqref{eq:elliptic-int-1-3}.
\end{proof}

\begin{prop}\label{prop:case-1-P}
	The hyperelliptic integrals in equation \eqref{eq:hyper-1} can be represented as a linear combination of some elementary integrals and the three types of standard elliptic integrals, namely $F(z,k)$, $E(z,k)$ and $\Pi(z,\alpha^2,k)$, as defined in \Cref{define:elliptic-function}, where $k=k^{(1)}_1$ in equation \eqref{eq:u-parameters-dn-k} and
	\begin{equation}\label{eq:z-k-alpha-1}
		z=\frac{\lambda(\lambda_1^2-\lambda_3^2)^{1/2}}{\lambda_3(\lambda_1^2-\lambda^2)^{1/2}}, \qquad  \alpha^2=\frac{\lambda_1^2-\lambda_2^2}{\lambda_2^2}.
	\end{equation}
\end{prop} 

\begin{proof}
	By combining the hyperelliptic integrals in equation \eqref{eq:hyper-1} with the definitions of the elliptic integrals in \eqref{eq:define-first-integral}-\eqref{eq:define-third-integral}, we introduce suitable transformations that convert the hyperelliptic integrals into Legendre’s standard forms. Subsequently, appropriate recursive relations are applied to reduce higher-order elliptic integrals to lower-order ones, thereby rewriting the entire expression in terms of the three canonical standard forms of elliptic integrals.
	
	\textbf{Step 1.}
	Combining with \Cref{prop:elliptic-int-1} and \Cref{define:elliptic-function}, 
	we obtain 
	\begin{equation}\nonumber
		\begin{split}
			\int_{0}^{\lambda}\frac{\dd \chi}{\prod_{i=1}^3(\chi^2-\lambda_i^2)^{1/2}}\xlongequal[\eqref{eq:elliptic-int-1-4}]{\eqref{eq:hyper-1}}&\ \frac{1}{\lambda_2(\lambda_3^2-\lambda_1^2)^{1/2}}
			\int_{0}^{z}\frac{\dd t}{\sqrt{(1-t^2)(1-(k^{(1)}_1)^2t^2)}}
		\xlongequal{\eqref{eq:define-first-integral}}
			\frac{\nu}{\lambda_2(\lambda_3^2-\lambda_1^2)^{1/2}},
		\end{split}	
	\end{equation}
	with $\nu=F(z,k^{(1)}_1)$ and $z$ defined in equation \eqref{eq:z-k-alpha-1}.
	To establish the precise correspondence between the entire $\lambda$-plane and $\nu$-plane exactly,
	we introduce the function
	\begin{equation}\label{eq:fs-2}
		\begin{split}
			& \sn^2(\nu,  k^{(1)}_1)=\frac{\lambda^2(\lambda_1^2-\lambda_3^2)}{\lambda_3^2(\lambda_1^2-\lambda^2)}, \qquad \lambda^2=\frac{\lambda_1^2\lambda_3^2\sn^2(\nu,  k^{(1)}_1)}{\lambda_1^2-\lambda_3^2\cn^2(\nu,  k^{(1)}_1)}, \\
			& S^{(1)}_1:=\left\{\nu\in \mathbb{C}\left| |\Re(\nu)| \le   K^{(1)}_1, |\Im(\nu)| \le   K^{(1)\prime}_1 \right.\right\}.
		\end{split}
	\end{equation}
	This function maps the periodic region $S^{(1)} _{1}$
	in the $\nu$-plane onto the spectral parameter $\lambda \in \mathbb{C} \cup \{\infty\}$ in the entire $\lambda$-plane. 
	The detailed proof is provided in \Cref{prop:appendenx-f-1-2}. 
	Specifically, the three pairs of ``vertical" cuts $[\lambda_3^*,\lambda_2^*]$, $[\lambda_1^*,\lambda_1]$, and $[\lambda_2, \lambda_3]$ in the $\lambda$-plane are mapped onto the corresponding rectangular region in the $\nu$-plane with cuts $[-K^{(1)}_1, -K^{(1)}_1-\ii K^{(1)\prime}_1]$, $[-\ii K^{(1)\prime}_1, \ii K^{(1)\prime}_1]$, $[K^{(1)}_1+\ii K^{(1)\prime}_1, K^{(1)}_1]$, respectively.
	
	\textbf{Step 2.}
	From the mapping in equation \eqref{eq:fs-2}, together with equation \eqref{eq:define-curve-algebro}, we obtain 
	\begin{subequations}\label{eq:fs-2-y}
		\begin{align}
			y=&\ \pm \ii \frac{\lambda_1\lambda_2\lambda_3(\lambda_1^2-\lambda_3^2)\sqrt{\lambda_1^2-\lambda_3^2}\,\cn(\nu,  k^{(1)}_1)\dn(\nu,  k^{(1)}_1)}{\left(\lambda_1^2-\lambda_3^2\cn^2(\nu,  k^{(1)}_1)\,\right)^{3/2}},  \label{eq:fs-2-y-1} \\
			\dd \lambda=&\ \frac{\lambda_1\lambda_3(\lambda_1^2-\lambda_3^2) \cn(\nu,  k^{(1)}_1)\dn(\nu,  k^{(1)}_1)}{\left(\sqrt{\lambda_1^2-\lambda_3^2\cn^2(\nu,  k^{(1)}_1)}\,\right)^3}\, \dd \nu .  \label{eq:fs-2-y-2}
		\end{align}
	\end{subequations}
	Using equations \eqref{eq:fs-2}, \eqref{eq:fs-2-y-1}, \eqref{eq:y-d-k2-a1} and \eqref{eq:Jacobi-shift}, the hyperelliptic integral in equation \eqref{eq:hyper-1} can be expressed as 
	\begin{equation}\label{eq:elliptic-funda-1}
		\int_{0}^{\lambda} \frac{ \chi^{2n} }{y} \dd  \chi \xlongequal[\eqref{eq:fs-2-y-1},\eqref{eq:fs-2-y-2}]{\eqref{eq:fs-2},\eqref{eq:y-d-k2-a1}}  \int_{0}^{\nu} \pm   \left(\frac{\lambda_1^2\lambda_3^2\sn^2(u,  k^{(1)}_1)}{\lambda_1^2-\lambda_3^2\cn^2(u,  k^{(1)}_1)}\right)^n\frac{\dd u}{\ii \lambda_2\sqrt{\lambda_1^2-\lambda_3^2}},
	\end{equation}
	where $\lambda$ and $\nu$ satisfy the equation \eqref{eq:fs-2}.
	
	\textbf{Step 3.}
	Consider the integration path defined in equation \eqref{eq:elliptic-funda-1}.
	The $a_1$-circle (illustrated in \Cref{fig:genus-two-figure}) can be constructed in the $\lambda$-plane as a straight line running along the right-hand edge of the branch cut $[\lambda_2,\lambda_1]$, and then returning along its left-hand edge.
	Along the right-hand edge, we have $\Im(\dd \lambda)<0$. 
	By equation \eqref{eq:fs-2}, this implies $\Re(\dd \nu)<0$, and therefore we obtain
	\begin{equation}\label{eq:y-d-k2-a1}
		\Im\left(\frac{\cn(\nu,  k^{(1)}_1)\dn(\nu,  k^{(1)}_1)}{(\lambda_1^2-\lambda_3^2\cn^2(\nu,  k^{(1)}_1)\,)^{3/2}}\right)>0.
	\end{equation}
	Hence, along this path, the sign of the parameter $y$ in equation \eqref{eq:fs-2-y-1} is chosen as $``+"$.
	Using equations \eqref{eq:fs-2},  \eqref{eq:fs-2-y-1}, \eqref{eq:y-d-k2-a1} and \eqref{eq:Jacobi-shift}, the hyperelliptic integral defined in equation \eqref{eq:elliptic-funda-1} can then be expressed as 
	\begin{equation}\label{eq:elliptic-funda-1-ge}
		\int_{\lambda_2}^{\lambda} \frac{\chi^{2n}\dd \chi}{y}\xlongequal[\eqref{eq:fs-2-y-1}]{\eqref{eq:fs-2},\eqref{eq:y-d-k2-a1}}  \int_{K^{(1)}_1+\ii K^{(1)\prime}_1}^{\nu}\left(\frac{\lambda_1^2\lambda_3^2\sn^2(u,k^{(1)}_1)}{\lambda_1^2-\lambda_3^2\cn^2(u,k^{(1)}_1)}\right)^n\frac{\dd u}{\ii \lambda_2\sqrt{\lambda_1^2-\lambda_3^2}}.
	\end{equation}
	The parameter $y$, as defined in equation \eqref{eq:fs-2-y-1} depends on the chosen integral path.
	When the path runs along the left-hand edge of the branch cut from $\lambda_2$ to $\lambda_1$, $y$ takes the negative value; that is, the sign $``-"$ is assigned in equation \eqref{eq:fs-2-y-1}.
	
	\textbf{Step 4.}
	According to step 2, the hyperelliptic integral in equation \eqref{eq:hyper-1} satisfies
	\begin{equation}\label{eq:integral-dn-1}
		\begin{split}
			\int_{\lambda_2}^{\lambda}\frac{\chi^{2n}}{y}\dd \chi
			\xlongequal[\eqref{eq:elliptic-funda-2}]{\eqref{eq:fs-2}} &-\int_{  K^{(1)}_1+\ii   K^{(1)\prime}_1}^{\nu}\ii\left(\frac{\lambda_1^2\lambda_3^2\sn^2(u,  k^{(1)}_1)}{\lambda_1^2-\lambda_3^2\cn^2(u,  k^{(1)}_1)}\right)^n \frac{\dd u}{\lambda_2\sqrt{\lambda_1^2-\lambda_3^2}}\\
			\xlongequal{\eqref{eq:Jacobi-shift}}&\frac{\ii}{\lambda_2\sqrt{\lambda_1^2-\lambda_3^2}}\int_{\nu-\ii K^{(1)\prime}_1}^{   K^{(1)}_1}\left(\frac{\lambda_1^2\lambda_3^2}{\lambda_3^2+  (k^{(1)}_1)^2(\lambda_1^2-\lambda_3^2)\sn^2(u,  k^{(1)}_1)}\right)^n \dd u.
		\end{split}
	\end{equation}
	The recursive formula \cite{ByrdF-54} is given by  
	\begin{equation}\label{eq:integral}
		\begin{split}
			J_{n+1}=&\ \frac{(2n-1)(\alpha^4-2\alpha^2(k^{(1)}_1)^2-2\alpha^2+3(k^{(1)}_1)^2)J_{n}+2(n-1)(\alpha^2(k^{(1)}_1)^2+\alpha^2-3(k^{(1)}_1)^2)J_{n-1}}{2n(1-\alpha^2)((k^{(1)}_1)^2-\alpha^2)}\\
			&\ +\frac{(2n-1)(k^{(1)}_1)^2J_{n-2}}{2n(1-\alpha^2)((k^{(1)}_1)^2-\alpha^2)}+\frac{\alpha^4\sn(u,k^{(1)}_1)\cn(u,k^{(1)}_1)\dn(u,k^{(1)}_1)}{2n(1-\alpha^2)((k^{(1)}_1)^2-\alpha^2)(1-\alpha^2\sn^2(u,k^{(1)}_1))^n}+C
		\end{split}	
	\end{equation}
	where $n\in \mathbb{Z}^+$, $C$ is the integration constant, and $\alpha^2=(k^{(1)}_1)^2(\lambda_1^2-\lambda_3^2)/\lambda_3^2=(\lambda_1^2-\lambda_2^2)/\lambda_2^2$, 
	\begin{equation}
		\begin{split}
			& J_n=\int\frac{\dd u}{(1-\alpha^2\sn^2(u,k^{(1)}_1))^{n}}, \qquad
			 J_{-1}=\int \frac{(k^{(1)}_1)^2-\alpha^2}{(k^{(1)}_1)^2} +\frac{\alpha^2}{(k^{(1)}_1)^2}\dn^2(u,k^{(1)}_1) \dd u.
		\end{split}	
	\end{equation}	
	Thus, we complete the proof.
\end{proof}

\begin{prop}\label{prop:case-2-P}
	The hyperelliptic integrals  in equation \eqref{eq:hyper-2} can also be expressed as a linear combination of some elementary integrals and the elliptic integrals $F((\lambda_1^2-\lambda_3^2)^{1/2}(\lambda_1^2-\lambda^2)^{-1/2},k^{(1)}_2)$ and $E((\lambda_1^2-\lambda_3^2)^{1/2}(\lambda_1^2-\lambda^2)^{-1/2},k^{(1)}_2)$.
\end{prop}

\begin{proof}
	The method we use here is similar to that employed in the proof of \Cref{prop:case-1-P}; therefore, we only highlight the differences. 
	
	\textbf{Step 1.}
	Combining with \Cref{prop:elliptic-int-1} and \Cref{define:elliptic-function}, 
	we obtain 
	\begin{equation}\nonumber
		\begin{split}
			\int_{0}^{\lambda}\frac{\chi \dd \chi}{\prod_{i=1}^3(\chi^2-\lambda_i^2)^{1/2}}\xlongequal[\eqref{eq:elliptic-int-1-3}]{\eqref{eq:hyper-2}}&\ 
			\int_{0}^{z}\frac{(\lambda_3^2-\lambda_1^2)^{-1/2}\dd t}{((1-t^2)(1-(k^{(1)}_2)^2t^2))^{1/2}}
			\xlongequal{\eqref{eq:define-first-integral}} 
			 \frac{F\left(\left(\frac{\lambda_1^2-\lambda_3^2}{\lambda_1^2-\lambda^2}\right)^{\frac{1}{2}}\! \! ,k^{(1)}_2\right)}{(\lambda_3^2-\lambda_1^2)^{1/2}} . 
		\end{split}	
	\end{equation}
	We introduce the function
	\begin{equation}\label{eq:fs-1}
		\begin{split}
			& \sn^2(\nu,  k^{(1)}_2)=\frac{\lambda_1^2-\lambda_3^2}{\lambda_1^2-\lambda^2}, \qquad \lambda^2=\frac{\lambda_3^2-\lambda_1^2\cn^2(\nu,  k^{(1)}_2)}{\sn^2(\nu,  k^{(1)}_2)}, \\
			& S^{(1)} _2:=\left\{\nu\in \mathbb{C}\left| |\Re(\nu)| \le   K^{(1)}_2,  |\Im(\nu)| \le   K^{(1)\prime}_2 \right.\right\},
		\end{split} 
	\end{equation}
	where $k^{(1)}_2$ is defined in equation \eqref{eq:u-parameters-dn-k}.
	The function \eqref{eq:fs-1} maps the periodic region $ S^{(1)} _{2}$ in the $\nu$-plane onto the spectral parameter $\lambda\in \mathbb{C}\cup \{\infty\}$ in the entire $\lambda$-plane. 
	A detailed proof is given in \Cref{prop:appendenx-f-1-2}. 
	Specifically, the three pairs of cuts $[\lambda_1, \lambda_1^*]$, $[\lambda_2^*, \lambda_2]$, and $[\lambda_3, \lambda_3^*]$ in the $\lambda$-plane are mapped onto the rectangular region in the $\nu$-plane corresponding to the cuts $[\ii K^{(1)\prime}_2, -\ii K^{(1)\prime}_2]$, $[K^{(1)}_2+\ii K^{(1)\prime}_2, -K^{(1)}_2-\ii K^{(1)\prime}_2 ]$, $[K^{(1)}_2, -K^{(1)}_2]$.
	
	\textbf{Step 2.}	
	By combining
	equations \eqref{eq:define-curve-algebro} and \eqref{eq:fs-1}, we obtain 
	\begin{equation}\label{eq:fs-1-y}
			y= \pm \ii \frac{(\lambda_1^2-\lambda_3^2)^{3/2}\cn(\nu,  k^{(1)}_2)\dn(\nu,  k^{(1)}_2)}{\sn^3(\nu,  k^{(1)}_2)}, \qquad
			\lambda\dd \lambda= \frac{(\lambda_1^2-\lambda_3^2)\cn(\nu,  k^{(1)}_2)\dn(\nu,  k^{(1)}_2)}{\sn^3(\nu,  k^{(1)}_2)}\, \dd \nu.
	\end{equation}
	Using equations \eqref{eq:fs-1} and \eqref{eq:fs-1-y}, the hyperelliptic integral defined in equation \eqref{eq:hyper-2} can then be expressed as 
	\begin{equation}\label{eq:elliptic-funda-2}
		\int \frac{\lambda^{2n+1} \dd \lambda}{y}\xlongequal[\eqref{eq:fs-1-y}]{\eqref{eq:fs-1}} \int \pm \left(\frac{\lambda_3^2-\lambda_1^2\cn^2(\nu,k^{(1)}_2)}{\sn^2(\nu,k^{(1)}_2)}\right)^n \frac{\dd \nu}{\ii \sqrt{\lambda_1^2-\lambda_3^2}}.
	\end{equation} 
	
	\textbf{Step 3.}
	The $a_1$-circle (shown in \Cref{fig:genus-two-figure}) can be constructed in the $\lambda$-plane as a straight line running along on the right-hand edge of the branch cut, from $\lambda_2$ to $\lambda_1$, and then returns along the left-hand edge.
	Along the right-hand edge, we get $\Re(\lambda \dd \lambda)>0$. 
	By equation \eqref{eq:fs-2}, this implies $\Re(\dd \nu)<0$, and consequently $\Re(\cn(\nu,  k^{(1)}_2)\dn(\nu,  k^{(1)}_2)/\sn^3(\nu,  k^{(1)}_2))<0$.
	Hence, along this path, the parameter $y$ in equation \eqref{eq:fs-1-y} takes the sign $``-"$.
	By equations \eqref{eq:Jacobi-shift}, \eqref{eq:fs-1} and \eqref{eq:fs-1-y}, the hyperelliptic integral defined in equation \eqref{eq:hyper-2} can be expressed as 
	\begin{equation}\label{eq:elliptic-funda-2-ge}
		\int_{\lambda_2}^{\lambda}\frac{\chi^{2n+1} \dd \chi}{y}\xlongequal[\eqref{eq:fs-1-y}]{\eqref{eq:fs-1}}\int_{K^{(1)}_2+\ii K^{(1)\prime}_2}^{\nu}-\left(\frac{\lambda_3^2-\lambda_1^2\cn^2(u,k^{(1)}_2)}{\sn^2(u,k^{(1)}_2)}\right)^n \frac{\ii\dd u}{ \sqrt{\lambda_1^2-\lambda_3^2}}.
	\end{equation}

	\textbf{Step 4.}
	The hyperelliptic integral defined in equation \eqref{eq:hyper-2} can be expressed as
	\begin{equation}\label{eq:integral-dn-2}
		\begin{split}
			\int_{\lambda_2}^{\lambda}\frac{\chi^{2n+1}}{y}\dd \chi
			\xlongequal[\eqref{eq:Jacobi-shift}]{\eqref{eq:elliptic-funda-2-ge}}&\frac{\ii}{\sqrt{\lambda_1^2-\lambda_3^2}}\int_{\nu-\ii K^{(1)}_2}^{   K^{(1)}_2}\left(\lambda_3^2+(\lambda_1^2-\lambda_3^2)\dn^2(u,  k^{(1)}_2)\right)^n \dd u.
		\end{split}
	\end{equation}
	Since integrals $I_{2n}$ satisfy the following results:
	\begin{equation}\label{eq:recursive-formula-dn}
		\begin{split}
			\!\! \! I_{2n+4}=&  \frac{2(1+n)(1+(k^{(1)}_2)^2)I_{2n+2}-(2n+1)I_{2n}+\sn^{2n+1}(u,k^{(1)}_2)\dn(u,k^{(1)}_2)\cn(u,k^{(1)}_2)+C}{(2n+3)(k^{(1)}_2)^2},\\	I_{2}=&\ \int \frac{1-\dn^{2}(u,k^{(1)}_2)}{k^2}\dd u
			,\qquad I_{2n}= \int \sn^{2n}(u,k^{(1)}_2)\dd u, \qquad n\ge 0.
			\\ 
		\end{split}
	\end{equation}
	Utilizing the above recursive formula, we obtain the final conclusion.	
\end{proof}

\begin{prop}\label{prop:appendenx-f-1-2}
	Functions defined in equations \eqref{eq:fs-2} and \eqref{eq:fs-1} map the entire complex $\lambda$-plane onto the periodic region $S^{(1)} _1$ (from equation \eqref{eq:fs-2}) and $S^{(1)}_2$ (from equation \eqref{eq:fs-1}) in the $\nu$-plane, respectively.
\end{prop}

\begin{proof}
	Consider the function defined in equation \eqref{eq:fs-2}.
	From the definition of Jacobi elliptic functions, it follows that $\sn(\nu,  k_1^{(1)})$ maps the rectangular region $S^{(1)}_1$ (defined in equation \eqref{eq:fs-2}) onto the entire complex plane, including the point at infinity.
	Without loss of generality, we set $\sn(\nu,  k_1^{(1)})=r\ee^{\ii \theta}$, with $r\ge 0$ and $\theta\in (0,2\pi]$. 
	For any $\nu\in  S^{(1)} _1$, there exists a unique $r$ and $\theta$ such that $\sn(\nu,  k_1^{(1)})=r\ee^{\ii \theta}$.
	
	We consider the function $r\ee^{\ii \theta}$, $r\ge 0$ and $\theta\in (0,\pi]$. 
	From the definition of the elliptic function $\sn(\nu,  k_1)$, it follows that the corresponding region lies in the upper half part of the rectangular domain $ S^{(1)} _1$, namely $ S^{(1)u} _1:=\{\nu \in  S^{(1)} _1| \Im(\nu)\ge 0\}$.
	By applying a fractional linear transformation, we obtain a one-to-one correspondence between the parameter $\Lambda_1=r_1\ee^{\ii \theta_1}\in \mathbb{C}$ and $r\ee^{\ii \theta}$, given by 
	\begin{equation}\nonumber
		(r_1\ee^{\ii \theta_1})^2=\frac{\lambda_1^2\lambda_3^2(r\ee^{\ii \theta})^2}{\lambda_1^2-\lambda_3^2+\lambda_3^2(r\ee^{\ii \theta})^2}, \qquad r_1,r\ge 0, \quad \theta_1, \theta\in (0,\pi].
	\end{equation}
	For this, we deduce that the function defined in equation \eqref{eq:fs-2} maps the rectangular area $ S^{(1)} _1$ onto the entire complex plane, with the upper half plane of $ S^{(1)} _1$ mapping onto the upper half plane of $\mathbb{C}$.
	
	Similarly, the function defined in equation \eqref{eq:fs-1} maps the rectangular region $ S^{(1)} _2$ onto the whole complex plane.
\end{proof}
\Cref{fig:confromal-1} illustrates the correspondence between the entire $\lambda$-plane and the periodic region $ S^{(1)} _{1,2}$ of the $\nu$-plane.

\begin{figure}[h!]
	\centering
	\subfigure[The conformal map of the function defined in equation \eqref{eq:fs-2}]{\includegraphics[width=1\textwidth]{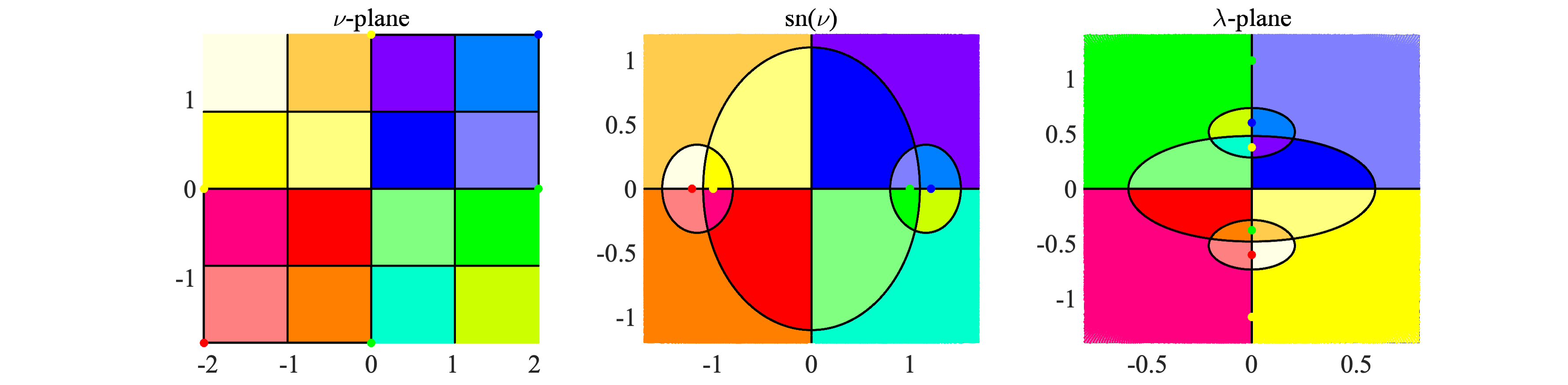}}
	\subfigure[The conformal map of the function defined in equation \eqref{eq:fs-1}]{\includegraphics[width=1\textwidth]{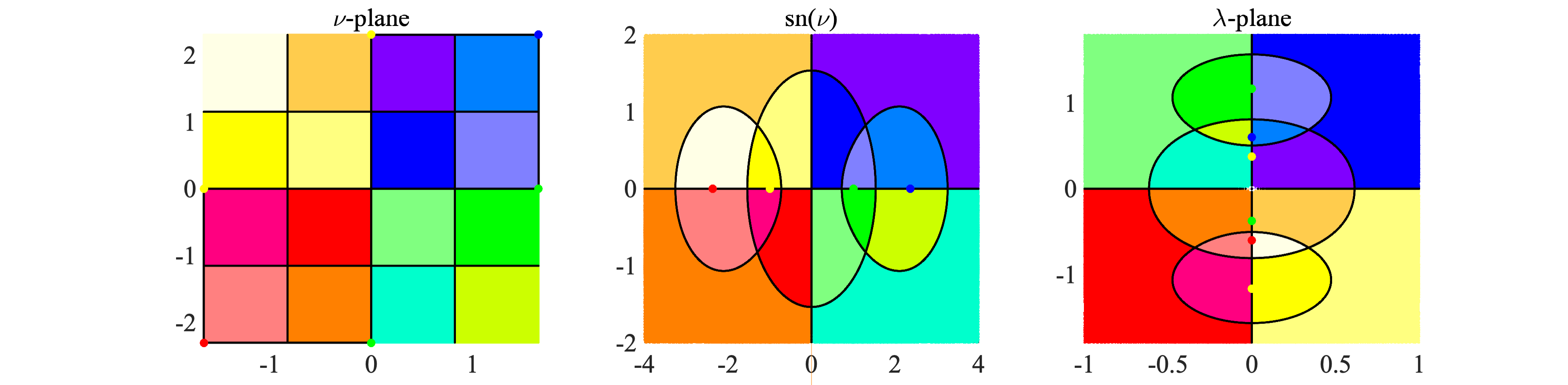}}
	\caption{The correspondence between the $\nu$-plane and $\lambda$-plane with branch points satisfying \ref{case1}. From top to bottom, the green, blue, yellow, green, red, yellow points in the $\lambda$-plane correspond to the branch points $\lambda_3$, $\lambda_2$, $\lambda_1$, $\lambda_1^*$, $\lambda_2^*$ and $\lambda_3^*$, respectively.}
	\label{fig:confromal-1}
\end{figure}

From \Cref{prop:appendenx-f-1-2}, we conclude that the conformal maps defined in equations \eqref{eq:fs-2} and \eqref{eq:fs-1} can transform the integrals into the canonical forms of elliptic integrals.
Together, propositions \ref{prop:elliptic-int-1}-\ref{prop:appendenx-f-1-2} establish the fundamental theorem of elliptic integrals. 
Applying these results to the hyperelliptic integral \eqref{eq:hyper}, we derive the following two propositions.

\vspace{0.1cm}
\noindent $\bullet$\quad
\textbf{\large \ref{case2}}

Consider the case where the parameters satisfy $\beta_{1,2}\in \mathbb{C}\backslash(\ii \mathbb{R}\cup \mathbb{R})$ and $\beta_{3,4}\in \mathbb{R}$, where $\beta_1=\lambda_1^2$, $\beta_2=\lambda_3^2$, $\beta_3=\lambda_2^2$ and $\beta_4=0$.
Since $\beta_{1,2}$ are not real numbers, we adopt an alternative approach to transform the elliptic integrals into the standard form of the first-kind elliptic integral given in equation \eqref{eq:define-first-integral}.
 
\begin{prop}\label{prop:elliptic-int-2}
	The elliptic integrals can be transformed into the first-kind elliptic integrals in \Cref{define:elliptic-function} as follows:
	\begin{subequations}\label{eq:elliptic-int-2}
		\begin{align}
			&\int \frac{\dd \Lambda }{\prod_{i=1}^{4}(\Lambda-\beta_i)^{1/2}}
			=\int \frac{- \ii \dd s }{(AB(1-s^2)(1-k^2s^2))^{1/2}},  \qquad 	k^2=\frac{\beta_2^2-(A-B)^2}{4AB}, \label{eq:elliptic-int-2-4}\\
			& \int\frac{\dd \Lambda }{\prod_{i=1}^{3}(\Lambda-\beta_i)^{1/2}} =\int \frac{  \dd s }{\sqrt{A(1-s^2)(1-k^2s^2)}}, \qquad
			k^2=\frac{2(A+\beta_3)-\beta_1-\beta_2}{4A}, \label{eq:elliptic-int-2-3}
		\end{align}	
	\end{subequations}
	where $\beta_{3,4}\in \mathbb{R}$ and $\beta_1=\beta_2^*\in \mathbb{C}\backslash(\mathbb{R}\cup \ii \mathbb{R})$.
\end{prop}
\begin{proof}
	Let the correspondence between $\beta_{3,4}$ in the
 $\Lambda$-plane and $\pm 1$ in the $z$-plane  be established.
 Introducing the transformation
	\begin{equation}\label{eq:trans-L-z-2-1}
		z=\frac{(A+B)\Lambda-(A\beta_4+B\beta_3)}{(A-B)\Lambda-(A\beta_4-B\beta_3)}, \qquad \Lambda=\frac{(A\beta_4+B\beta_3)-z(A\beta_4-B\beta_3)}{(A+B)-z(A-B)}, 
	\end{equation}
	together with $\dd \Lambda=2AB(\beta_3-\beta_4)((A+B)-z(A-B))^{-2} \dd z$,
 $A=|\beta_3-\beta_2|$, $B=|\beta_4-\beta_2|$,
	we obtain 
	\begin{equation}\label{eq:int-1-trans}\nonumber
		\int \frac{  \dd \Lambda}{\prod_{i=1}^{4}(\Lambda-\beta_i)^{1/2}}
		=\int \frac{ \ii  \dd z }{(AB(1-z^2)(k^{\prime2}+k^2z^2))^{1/2}}\xlongequal{z^2=1-s^2}\int \frac{- \ii \dd s }{(AB(1-s^2)(1-k^2s^2))^{1/2}}.
	\end{equation}
	Thus, we obtain the equation \eqref{eq:elliptic-int-2-4}.
	
	Similarly, consider the linear fractional transformation $z(\Lambda)$, its inverse $\Lambda(z)$ and the corresponding differential form:
	\begin{equation}\label{eq:trans-z-Lambda-3}
		z=\frac{\Lambda+A-\beta_3}{\Lambda-A-\beta_3},\qquad \Lambda=\beta_3+A\frac{z+1}{z-1}, \qquad 
		\dd \Lambda=\frac{-2A}{(z-1)^2}\, \dd z.
	\end{equation} 
	Substituting above formula into the equation \eqref{eq:elliptic-int-2-4}, we obtain 
	\begin{equation}
		\int \frac{  \dd \Lambda}{\prod_{i=1}^{3}(\Lambda-\beta_i)^{1/2}}
		= \int \frac{ - \dd z }{(A(1-z^2)(k^{\prime2}+k^2z^2))^{1/2}}\xlongequal{z^2=1-s^2}\int \frac{  \dd s }{\sqrt{A(1-s^2)(1-k^2s^2)}},
	\end{equation}
	where $k^2$ is defined in equation \eqref{eq:elliptic-int-2-3}. 
\end{proof}

\begin{prop}\label{prop:case-1-C}
The hyperelliptic integrals defined in equation \eqref{eq:hyper-1} can be expressed as a linear combination of some elementary integrals and the three types of standard elliptic integrals, namely $F(z,k_1^{(2)})$, $E(z,k_1^{(2)})$ and $\Pi(z,\alpha^2,k_1^{(2)})$, where $z=2\lambda(AB(\lambda_2^2-\lambda^2))^{1/2}/(\lambda^2(A-B)+\lambda_2^2B)$ and $\alpha^2 = (B-A)^2/(B+A)^2$.
\end{prop}

\begin{proof}
	Analogous to the proofs of Propositions \ref{prop:case-1-P} and \ref{prop:case-2-P}, the evaluation of the hyperelliptic integrals \eqref{eq:hyper-1} is carried out in four steps.
	
	\textbf{Step 1.}
	By combining \Cref{prop:elliptic-int-2} with \Cref{define:elliptic-function}, 
	we obtain 
	\begin{equation}\nonumber
		\begin{split}
			\int_{0}^{\lambda}\frac{\dd \chi}{\prod_{i=1}^3(\chi^2-\lambda_i^2)^{1/2}}\xlongequal[\eqref{eq:elliptic-int-2-4},\eqref{eq:define-first-integral}]{\eqref{eq:hyper-1}}&\ 
			 \frac{ \ii }{(\lambda_3^2-\lambda_1^2)^{1/2}}\int_{0}^{\nu} \dd u
			= \frac{\ii \nu}{(\lambda_3^2-\lambda_1^2)^{1/2}}, \quad \nu=F(z,  k^{(2)}_1).
		\end{split}	
	\end{equation}
	By combining \Cref{prop:elliptic-int-2}, we obtain that the elliptic integrals in equation \eqref{eq:hyper-1} with $n=0$ can be expressed in terms of the elliptic integrals of the first kind defined in equation \eqref{eq:define-first-integral}.
	We now introduce the function
	\begin{equation}\label{eq:fs-4}
		\begin{split}
			&\cn(\nu,  k^{(2)}_1)=\frac{\lambda^2(A+B)-\lambda_2^2B}{\lambda^2(A-B)+\lambda_2^2B},
			\qquad 
			\lambda^2=\frac{\lambda_2^2B(1+\cn(\nu,  k^{(2)}_1))}{(A+B)+(B-A)\cn(\nu,  k^{(2)}_1)},\\
			&S_1^{(2)}:=\{\nu\in \mathbb{C}\left| |\Re(\nu)|\le 4  K^{(2)}_1, \,\, 0\le \Im(\nu)\le   K^{(2)\prime}_1 \right.\},
		\end{split}
	\end{equation}
	where parameters $k^{(2)}_1$, $A$, $B$ are defined in equation \eqref{eq:u-parameters-cn}.
	This function maps the periodic region $S^{(2)}_1$ in the $\nu$-plane onto the spectral parameter $\lambda$ over the entire complex plane $\mathbb{C}\cup \{ \infty \}$, which contains three cuts.
	Further details are provided in \Cref{prop:appendenx-f-3-4}. 
	In particular, the three pairs of cuts $[\lambda_1^*, \lambda_1]$, $[\lambda_2^*, \lambda_2]$, and $[\lambda_3^*, \lambda_3]$ in the $\lambda$-plane are mapped to the rectangular region in the $\nu$-plane with cuts $[-3K^{(2)}_1+\ii K^{(2)\prime}_1 , -K^{(2)}_1+\ii K^{(2)\prime}_1]$, $[-4 K^{(2)}_1 , 0]$, and $[3K^{(2)}_1+\ii K^{(2)\prime}_1, K^{(2)}_1+\ii K^{(2)\prime}_1]$ corresponding to the periodic region bounded by $-4K^{(2)}_1$, $4K^{(2)}_1$, $4K^{(2)}_1+\ii K^{(2)\prime}_1$ and $-4K^{(2)}_1+\ii K^{(2)\prime}_1$.

	\textbf{Step 2.}
	For the conformal map defined in equation \eqref{eq:fs-4}, and in combination with equations \eqref{eq:define-curve-algebro} and \eqref{eq:fs-4}, we obtain 
	\begin{subequations}\label{eq:fs-4-y}
		\begin{align}
			y=&\ \pm \frac{2\ii \lambda_2A^{3/2} B \, \sn(\nu,  k^{(2)}_1) \dn(\nu,  k^{(2)}_1)}{(A+B+(B-A)\cn(\nu,  k^{(2)}_1))^{3/2}(1+\cn(\nu,  k^{(2)}_1))^{1/2}},\label{eq:fs-4-y-1}\\
			\dd \lambda =&\  \frac{-\lambda_2 A B^{1/2} \, \sn(\nu,  k^{(2)}_1)\dn(\nu,  k^{(2)}_1)}{(A+B+(B-A)\cn(\nu,  k^{(2)}_1))^{3/2}(1+\cn(\nu,  k^{(2)}_1))^{1/2}} \, \dd \nu ,\label{eq:fs-4-y-2}
		\end{align}
	\end{subequations}
	where $(k^{(2)\prime }_1)^2=1-  (k^{(2)}_1)^2=((A+B)^2-\lambda_2^4)/(4AB)$.
	By applying equations \eqref{eq:fs-4} and \eqref{eq:fs-4-y-1}, we obtain 
	\begin{equation}\label{eq:elliptic-funda-1-cn}
		\begin{split}
			 \int \frac{\lambda^{2n}}{y}\dd \lambda
			\xlongequal[\eqref{eq:fs-4-y}]{\eqref{eq:fs-4}}  \int \frac{\pm \ii}{2(AB)^{1/2}} \left(\frac{\lambda_2^2B(1+\cn(\nu,  k^{(2)}_1))}{(A+B)+(B-A)\cn(\nu,  k^{(2)}_1)}\right)^n \dd \nu.
		\end{split}
	\end{equation}

	\textbf{Step 3.}
	Consider the integration path in equation \eqref{eq:hyper-1}.
	The $a_1$-circle (shown in \Cref{fig:genus-two-figure}) can be represented in the $\lambda$-plane as a path consisting of two straight-line segments: one along the left-hand edge of the branch cut from $\lambda_1^*$ to $\lambda_1$, and the other returning along the right-hand edge of the branch cut.
	Along the left-hand edge, we have $\Im(\dd \lambda)>0$. 
	By equation \eqref{eq:fs-4}, this implies $\Re(\dd \nu)>0$, so that the inequality $\Re( \sn(\nu,  k^{(2)}_1)\dn(\nu,  k^{(2)}_1)(A+B+(B-A)\cn(\nu,  k^{(2)}_1))^{-3/2}(1+\cn(\nu,  k^{(2)}_1))^{-1/2})<0$	holds.
	Hence, along this path, the parameter $y$ in equation \eqref{eq:fs-4-y-1} takes the sign $``+ "$, that is, we choose the positive sign in equation \eqref{eq:elliptic-funda-1-cn}.

	\textbf{Step 4.}
	By utilizing the following recursive formula \cite{ByrdF-54} 
	\begin{equation}\nonumber
		\begin{split}
			G_{m+1} = &\ \ \frac{(1-2m)((1-2k^2)a^2+2k^2)}{m(a^2-1)((k^{\prime })^2a^2+k^2)}G_m+\frac{(1-m)((1-2k^2)a^2+6k^2)}{m(a^2-1)((k^{\prime })^2a^2+k^2)}G_{m-1}\\
			&\ +\frac{2(3-2m)k^2 G_{m-2}}{m(a^2-1)((k^{\prime })^2a^2+k^2)}
			+\frac{(2-m)k^2G_{m-3}}{m(a^2-1)((k^{\prime })^2a^2+k^2)}
			+\frac{\sn(u,k)\dn(u,k)}{m(1+a\cn(u,k))^{m}}+C,\quad m\ge 0, \\
			G_m=&\  \int \frac{\dd u}{(1+\alpha\cn(u,k))^m},\\ \qquad 
			G_1=&\   \int\frac{\dd u}{1+\alpha \cn(u,k)}=\frac{1}{1-\alpha^2}\Pi\left(z, \frac{\alpha^2}{\alpha^2-1},k\right)
			 - \frac{\alpha}{2m(1-\alpha^2)} \ln\left(\frac{\dn(u,k)+m\sn(u,k)}{\dn(u,k)-m\sn(u,k)}\right), 
		\end{split}
	\end{equation}
	with $k=k^{(2)}_1$, $z=2\lambda(AB(\lambda_2^2-\lambda^2))^{1/2}/(\lambda^2(A-B)+\lambda_2^2B)$, $\alpha=(B-A)/(B+A)$ and $m=\alpha^2/(\alpha^2-1)-k^2$,
	we get the final result.	
\end{proof}

\begin{prop}\label{prop:case-2-C}
	For the \ref{case2}, the hyperelliptic integrals in \eqref{eq:hyper-2} can also be expressed as a linear combination of some elementary integrals and the normal elliptic integrals of the first and second kinds, $F(z,k_2^{(2)})$ and $E(z,k_2^{(2)})$, with $z=2A^{1/2}(\lambda_2^2-\lambda^2)^{1/2}/(\lambda_2^2+A-\lambda^2)$.
\end{prop}

\begin{proof}
	The proof follows a process similar to that of \Cref{prop:case-1-C}.
	
	\textbf{Step 1.}
	By \Cref{prop:elliptic-int-2} and \Cref{define:elliptic-function}, since
	\begin{equation}\nonumber
		\begin{split}
			\int_{0}^{\lambda}\frac{\chi \dd \chi}{\sqrt{\prod_{i=1}^3(\chi^2-\lambda_i^2)}}\xlongequal[\eqref{eq:elliptic-int-2-3}]{\eqref{eq:hyper-2}}&\ 
			\int_{0}^{z}\frac{- \dd t}{\sqrt{A(1-t^2)(k^{\prime2}+k^2t^2)}}
			\xlongequal{\eqref{eq:define-first-integral}} \frac{ -1 }{\sqrt{A}}\int_{0}^{\nu} \dd u
			= \frac{-\nu}{\sqrt{A}}, 
		\end{split}	
	\end{equation}
	where $k=k_2^{(2)}$ and $\nu=F(z,k_2^{(2)})$,
	we introduce the function
	\begin{equation}\label{eq:fs-3}
		\begin{split}
			&\cn(\nu,  k^{(2)}_2)=\frac{\lambda_2^2-A-\lambda^2}{\lambda_2^2+A-\lambda^2},
			\qquad 
			\lambda^2=\frac{(\lambda_2^2-A)-(\lambda_2^2+A)\cn(\nu,  k^{(2)}_2)}{1-\cn(\nu,  k^{(2)}_2)},
			\\
			& S^{(2)}_2=\left\{\nu \in \mathbb{C}\left|\, | \Re(\nu) | \le 2  K^{(2)}_2, |\Im(\nu)|\le   K^{(2)\prime}_2 \right|\right\},
		\end{split}
	\end{equation}
	with $ k^{(2)}_2$ in equation \eqref{eq:u-parameters-cn}.
	This function maps the periodic region $ S^{(2)} _{2}$ in the $\nu$-plane onto the spectral parameter $\lambda\in \mathbb{C}\cup \{\infty\}$ in the entire $\lambda$-plane. 
	The three pairs of cuts $[\lambda_1^*, \lambda_1]$, $[\lambda_2^*, \lambda_2]$, and $[\lambda_3^*, \lambda_3]$ in the $\lambda$-plane are mapped to the rectangular region in the $\nu$-plane with cuts $[-  K^{(2)}_2-\ii   K^{(2)\prime}_2 , K^{(2)}_2-\ii   K^{(2)\prime}_2]$, $[-2  K^{(2)}_2, 2  K^{(2)}_2]$, and $[-  K^{(2)}_2+\ii   K^{(2)\prime}_2 , K^{(2)}_2+\ii   K^{(2)\prime}_2]$, 
	within the rectangular region $S^{(2)}_2$.
	The detailed proof of this result is given in \Cref{prop:appendenx-f-3-4}. 
	
	\textbf{Step 2.}
	For the transformation \eqref{eq:fs-3}, together with equation \eqref{eq:define-curve-algebro}, we obtain 
	\begin{equation}\label{eq:fs-3-y}
			y=  \pm \ii \frac{2A\sqrt{A}\, \sn(\nu,  k^{(2)}_2)\dn(\nu,  k^{(2)}_2)}{\left(1-\cn(\nu,  k^{(2)}_2)\right)^2},
		\qquad 
			\lambda	\dd \lambda =  \frac{A\sn(\nu,  k^{(2)}_2)\dn(\nu,  k^{(2)}_2)}{\left(1-\cn(\nu,  k^{(2)}_2)\right)^2} \dd \nu,
	\end{equation}
	where $(k^{(2)\prime }_2)^2=1-  (k^{(2)}_2)^2=(2A-2\lambda_2^2+\lambda_1^2+\lambda_3^2)/(4A)$.
	Furthermore, by utilizing equations \eqref{eq:fs-3} and \eqref{eq:fs-3-y}, we obtain
	\begin{equation}\label{eq:Y-a-2n+1-cn-calculate}
		\begin{split}
			\int\frac{\lambda^{2n+1}}{y}\dd \lambda
			\xlongequal[\eqref{eq:fs-3-y}]{\eqref{eq:fs-3}} &\ \int \pm \frac{\ii}{2\sqrt{A}} \left(\lambda_2^2-\frac{A+A\cn(\nu,  k^{(2)}_2)}{1-\cn(\nu,  k^{(2)}_2)}\right)^n \dd \nu.
		\end{split}
	\end{equation}

	\textbf{Step 3.}
	Consider the integration in \eqref{eq:Y-a-2n+1-cn-calculate}.
	The $a_1$-circle (\Cref{fig:genus-two-figure}) can be represented in the $\lambda$-plane as a path along the left edge of the branch cut from $\lambda_1^*$ to $\lambda_1$, returning along the right edge. 
	Along the left edge,  $\Im(\lambda \dd \lambda)<0$ and $\Re(\dd \nu)>0$, so that $\Im(\sn(\nu,  k^{(2)}_2)\dn(\nu,  k^{(2)}_2)/(1-\cn(\nu,  k^{(2)}_2))^2)<0.$
	Hence, along this path the sign of parameter $y$ in equation \eqref{eq:fs-3-y} is $``-"$.
	By equations \eqref{eq:Jacobi-shift}, \eqref{eq:fs-3} and \eqref{eq:fs-3-y}, the hyperelliptic integral defined in equation \eqref{eq:hyper-2} can be expressed as 
	\begin{equation}\label{eq:Y-a-2n+1-cn-calculate-1}
		\begin{split}
			 \int_{\lambda_1^*}^{\lambda_1}\frac{\lambda^{2n+1}}{y}\dd \lambda 
			\xlongequal[\eqref{eq:fs-3-y}]{\eqref{eq:fs-3}}&\ \int_{-  K^{(2)}_2-\ii   K^{(2)\prime}_2}^{  K^{(2)}_2-\ii   K^{(2)\prime}_2}-\frac{\ii}{2\sqrt{A}} \left(\lambda_2^2-\frac{A+A\cn(u,  k^{(2)}_2)}{1-\cn(u,  k^{(2)}_2)}\right)^n \dd u\\
			\xlongequal{\eqref{eq:Jacobi-shift}}&\ \int_{-  K^{(2)}_2}^{ K^{(2)}_2}\frac{\ii\left(\lambda_2^2-A+2A\dn^2(u,  k^{(2)}_2)-2\ii Ak\sn(u,  k^{(2)}_2)\dn(u,  k^{(2)}_2)\right)^n}{-2\sqrt{A}}  \dd u.
		\end{split}
	\end{equation}
	
	\textbf{Step 4.}
	It is straightforward to show that $\int \dn^{2n+1}(u,k^{(2)}_2)\sn^{2m+1}(u,k^{(2)}_2) \dd u 
	= \int -((k^{(2)\prime}_2)^2+(k^{(2)}_2)^2\cn^2(u,k^{(2)}_2))^{n}(1-\cn^2(u,k^{(2)}_2))^{m} \dd (\cn(u,k^{(2)}_2))$.
	By combining this with the recursive formulas in \eqref{eq:recursive-formula-dn}, 
	we complete the proof.
\end{proof}

\begin{prop}\label{prop:appendenx-f-3-4}
	The functions in equations \eqref{eq:fs-4} and \eqref{eq:fs-3}
	map the entire complex $\lambda$-plane onto the rectangular region  $ S_1^{(2)}$ and $ S_2^{(2)}$ in the $\nu$-plane, respectively.
\end{prop}

\begin{proof}
	The function $\lambda^2$ maps the right half  $\lambda$-plane onto the entire complex plane. 
	Next, we consider the first function of equation \eqref{eq:fs-4},
	which is a fractional linear transformation.
	By the properties of fractional linear transformations, there exists a conformal mapping between the periodic region $(0,4  K_1^{(1)}, 4  K_1^{(1)}+\ii   K_1^{(1)\prime},\ii   K_1^{(1)\prime})$ and  the right half $\lambda$-plane.
	Similarly, for the left half $\lambda$-plane, we obtain the  periodic region $(0,-4  K_1^{(1)},-4  K_1^{(1)}+\ii   K_1^{(1)\prime},\ii   K_1^{(1)\prime})$.
	In summary, the function defined in equation \eqref{eq:fs-4} maps the complex $\lambda$-plane to the rectangular region $ S_1^{(2)}$ of the $\nu$-plane.
	Likewise, the function defined in equation \eqref{eq:fs-1} maps the rectangular area $ S_2^{(2)}$ onto the entire complex plane.
\end{proof}
The Figure \ref{fig:confromal-2} illustrates the correspondence between the $\lambda$-plane and the periodic region $ S_{1,2}^{(2)}$ in the $\nu$-plane.

\begin{figure}[h!]
	\centering
	\subfigure[The conform of the functions \eqref{eq:fs-4}]{\includegraphics[width=0.8\textwidth]{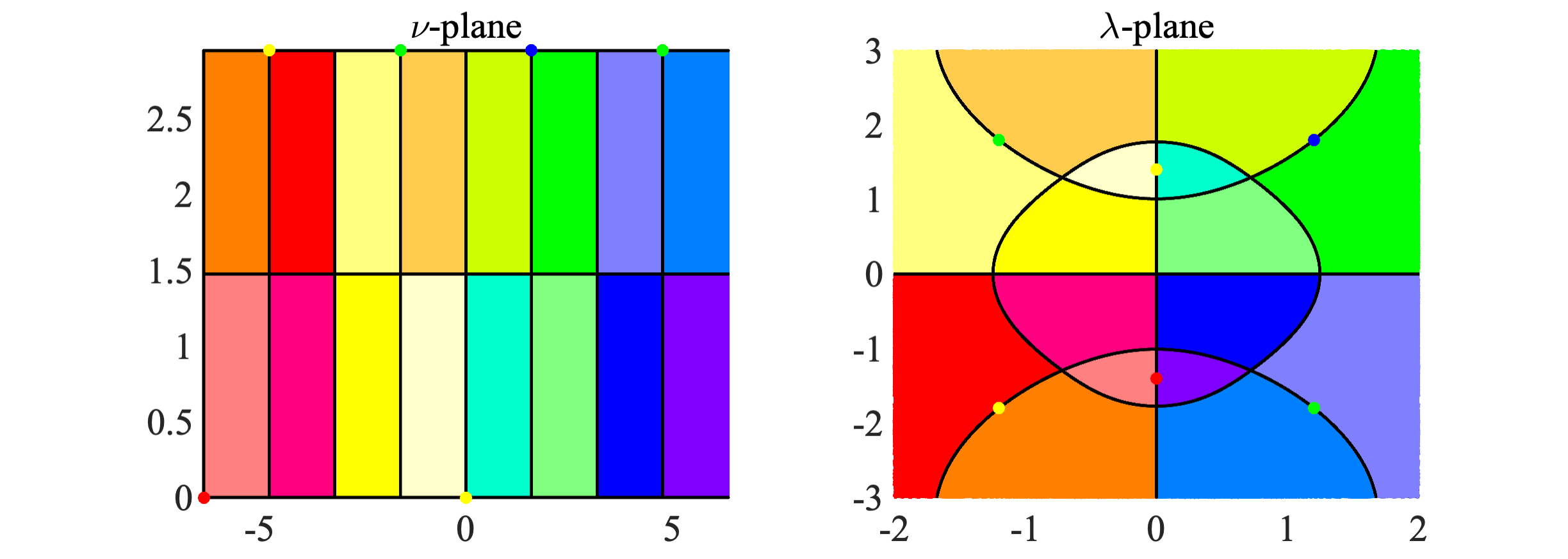}}
	\subfigure[The conform of the functions \eqref{eq:fs-3}]{\includegraphics[width=0.8\textwidth]{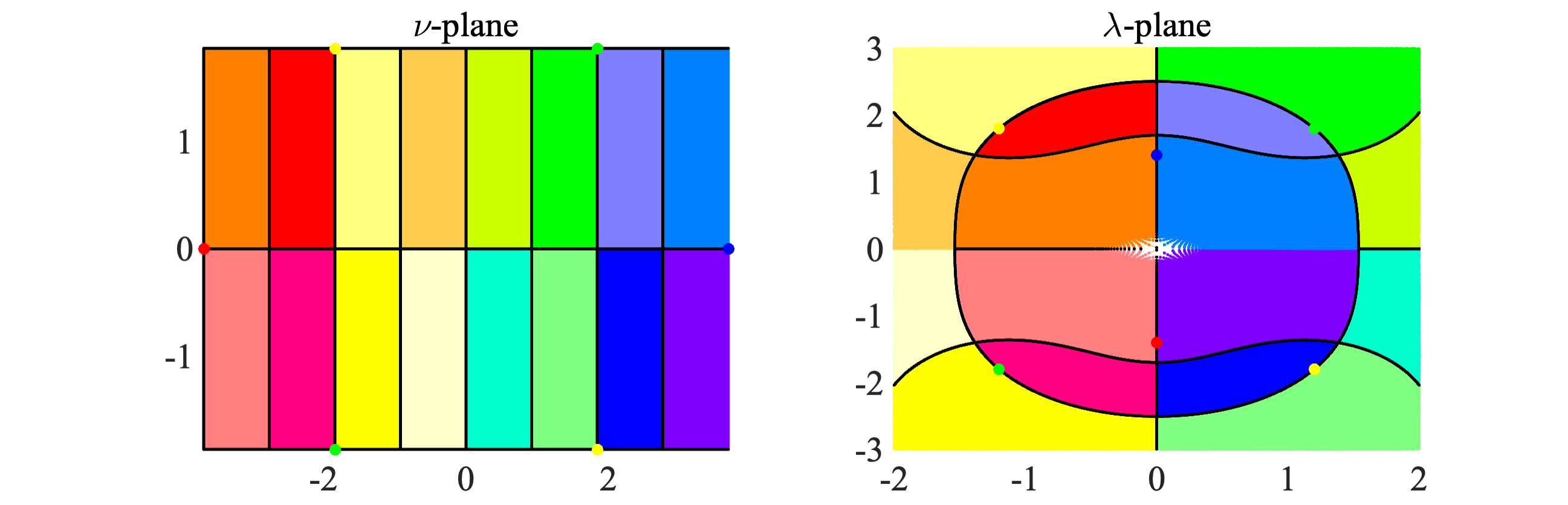}}
	\caption{The correspondence between the $\nu$-plane and $\lambda$-plane with branch points satisfying \ref{case2} and $\Re(\lambda_1)<0=\Re(\lambda_2)<\Re(\lambda_3)$. 
		The points shown in the $\lambda$-plane are correspondent with the branch points $\lambda_{1,2,3}$.}
	\label{fig:confromal-2}
\end{figure}

\bibliographystyle{siam}
\bibliography{reference}
\end{document}